\documentclass[a4paper,12pt]{article}%

\let\ssection=\section
\renewcommand{\section}{\setcounter{equation}{0}\ssection}

\usepackage{amsfonts}
\usepackage{amsmath}
\usepackage{amssymb}
\usepackage{graphicx}%
\usepackage{makeidx}
\usepackage{multicol}
\usepackage{color}

\newtheorem{definition}{Definition}[section]
\newtheorem{theorem}{Theorem}[section]
\newtheorem{lemma}[theorem]{Lemma}
\newtheorem{remark}[theorem]{Remark}
\newtheorem{corollary}[theorem]{Corollary}
\newenvironment{proof}[1][Proof]{\noindent\textbf{#1.} }
{\ \rule{0.5em}{0.5em}}

\makeatletter

\def\diagram{\m@th\leftwidth=\z@ \rightwidth=\z@ \topheight=\z@
\botheight=\z@ \setbox\@picbox\hbox\bgroup}

\def\enddiagram{\egroup\wd\@picbox\rightwidth\unitlength
\ht\@picbox\topheight\unitlength \dp\@picbox\botheight\unitlength
\hskip\leftwidth\unitlength\box\@picbox}

\def\bfig{\begin{diagram}}
\def\efig{\end{diagram}}
\newcount\wideness \newcount\leftwidth \newcount\rightwidth
\newcount\highness \newcount\topheight \newcount\botheight

\def\ratchet#1#2{\ifnum#1<#2 \global #1=#2 \fi}

\def\putbox(#1,#2)#3{%
\horsize{\wideness}{#3} \divide\wideness by 2 {\advance\wideness
by #1 \ratchet{\rightwidth}{\wideness}} {\advance\wideness by -#1
\ratchet{\leftwidth}{\wideness}} \vertsize{\highness}{#3}
\divide\highness by 2 {\advance\highness by #2
\ratchet{\topheight}{\highness}} {\advance\highness by -#2
\ratchet{\botheight}{\highness}} \put(#1,#2){\makebox(0,0){$#3$}}}

\def\putlbox(#1,#2)#3{%
\horsize{\wideness}{#3} {\advance\wideness by #1
\ratchet{\rightwidth}{\wideness}} {\ratchet{\leftwidth}{-#1}}
\vertsize{\highness}{#3} \divide\highness by 2 {\advance\highness
by #2 \ratchet{\topheight}{\highness}} {\advance\highness by -#2
\ratchet{\botheight}{\highness}}
\put(#1,#2){\makebox(0,0)[l]{$#3$}}}

\def\putrbox(#1,#2)#3{%
\horsize{\wideness}{#3} {\ratchet{\rightwidth}{#1}}
{\advance\wideness by -#1 \ratchet{\leftwidth}{\wideness}}
\vertsize{\highness}{#3} \divide\highness by 2 {\advance\highness
by #2 \ratchet{\topheight}{\highness}} {\advance\highness by -#2
\ratchet{\botheight}{\highness}}
\put(#1,#2){\makebox(0,0)[r]{$#3$}}}

\def\adjust[#1]{} 

\newcount \coefa
\newcount \coefb
\newcount \coefc
\newcount\tempcounta
\newcount\tempcountb
\newcount\tempcountc
\newcount\tempcountd
\newcount\xext
\newcount\yext
\newcount\xoff
\newcount\yoff
\newcount\gap%
\newcount\arrowtypea
\newcount\arrowtypeb
\newcount\arrowtypec
\newcount\arrowtyped
\newcount\arrowtypee
\newcount\height
\newcount\width
\newcount\xpos
\newcount\ypos
\newcount\run
\newcount\rise
\newcount\arrowlength
\newcount\halflength
\newcount\arrowtype
\newdimen\tempdimen
\newdimen\xlen
\newdimen\ylen
\newsavebox{\tempboxa}%
\newsavebox{\tempboxb}%
\newsavebox{\tempboxc}%

\newdimen\w@dth

\def\setw@dth#1#2{\setbox\z@\hbox{\m@th$#1$}\w@dth=\wd\z@
\setbox\@ne\hbox{\m@th$#2$}\ifnum\w@dth<\wd\@ne \w@dth=\wd\@ne \fi
\advance\w@dth by 1.2em}

\def\t@^#1_#2{\allowbreak\def\n@one{#1}\def\n@two{#2}\mathrel
{\setw@dth{#1}{#2} \mathop{\hbox to
\w@dth{\rightarrowfill}}\limits \ifx\n@one\empty\else
^{\box\z@}\fi \ifx\n@two\empty\else _{\box\@ne}\fi}}
\def\t@@^#1{\@ifnextchar_{\t@^{#1}}{\t@^{#1}_{}}}
\def\to{\@ifnextchar^{\t@@}{\t@@^{}}}

\def\t@left^#1_#2{\def\n@one{#1}\def\n@two{#2}\mathrel{\setw@dth{#1}{#2}
\mathop{\hbox to \w@dth{\leftarrowfill}}\limits
\ifx\n@one\empty\else ^{\box\z@}\fi \ifx\n@two\empty\else
_{\box\@ne}\fi}}
\def\t@@left^#1{\@ifnextchar_{\t@left^{#1}}{\t@left^{#1}_{}}}
\def\toleft{\@ifnextchar^{\t@@left}{\t@@left^{}}}

\def\two@^#1_#2{\allowbreak
\def\n@one{#1}\def\n@two{#2}\mathrel{\setw@dth{#1}{#2}
\mathop{\vcenter{\lineskip\z@\baselineskip\z@
                 \hbox to \w@dth{\rightarrowfill}%
                 \hbox to \w@dth{\rightarrowfill}}%
       }\limits
\ifx\n@one\empty\else ^{\box\z@}\fi \ifx\n@two\empty\else
_{\box\@ne}\fi}}
\def\tw@@^#1{\@ifnextchar _{\two@^{#1}}{\two@^{#1}_{}}}
\def\two{\@ifnextchar ^{\tw@@}{\tw@@^{}}}

\def\tofr@^#1_#2{\def\n@one{#1}\def\n@two{#2}\mathrel{\setw@dth{#1}{#2}
\mathop{\vcenter{\hbox to \w@dth{\rightarrowfill}\kern-1.7ex
                 \hbox to \w@dth{\leftarrowfill}}%
       }\limits
\ifx\n@one\empty\else ^{\box\z@}\fi \ifx\n@two\empty\else
_{\box\@ne}\fi}}
\def\t@fr@^#1{\@ifnextchar_ {\tofr@^{#1}}{\tofr@^{#1}_{}}}
\def\tofro{\@ifnextchar^ {\t@fr@}{\t@fr@^{}}}

\def\mon{\mathop{\m@th\hbox to
      14.6\P@{\lasyb\char'51\hskip-2.1\P@$\arrext$\hss
$\mathord\rightarrow$}}\limits} 
\def\leftmono{\mathrel{\m@th\hbox to
14.6\P@{$\mathord\leftarrow$\hss$\arrext$\hskip-2.1\P@\lasyb\char'50%
}}\limits} 
\mathchardef\arrext="0200       

\def\settypes(#1,#2,#3){\arrowtypea#1 \arrowtypeb#2 \arrowtypec#3}
\def\settoheight#1#2{\setbox\@tempboxa\hbox{#2}#1\ht\@tempboxa\relax}%
\def\settodepth#1#2{\setbox\@tempboxa\hbox{#2}#1\dp\@tempboxa\relax}%
\def\settokens`#1`#2`#3`#4`{%
     \def\tokena{#1}\def\tokenb{#2}\def\tokenc{#3}\def\tokend{#4}}
\def\setsqparms[#1`#2`#3`#4;#5`#6]{%
\arrowtypea #1 \arrowtypeb #2 \arrowtypec #3 \arrowtyped #4 \width
#5 \height #6 }
\def\setpos(#1,#2){\xpos=#1 \ypos#2}

\def\settriparms[#1`#2`#3;#4]{\settripairparms[#1`#2`#3`1`1;#4]}%

\def\settripairparms[#1`#2`#3`#4`#5;#6]{%
\arrowtypea #1 \arrowtypeb #2 \arrowtypec #3 \arrowtyped #4
\arrowtypee #5 \width #6 \height #6 }

\def\resetparms{\settripairparms[1`1`1`1`1;500]\width 500}

\resetparms

\def\mvector(#1,#2)#3{
\put(0,0){\vector(#1,#2){#3}}%
\put(0,0){\vector(#1,#2){26}}%
}
\def\evector(#1,#2)#3{{
\arrowlength #3
\put(0,0){\vector(#1,#2){\arrowlength}}%
\advance \arrowlength by-30
\put(0,0){\vector(#1,#2){\arrowlength}}%
}}

\def\horsize#1#2{%
\settowidth{\tempdimen}{$#2$}%
#1=\tempdimen \divide #1 by\unitlength }

\def\vertsize#1#2{%
\settoheight{\tempdimen}{$#2$}%
#1=\tempdimen
\settodepth{\tempdimen}{$#2$}%
\advance #1 by\tempdimen \divide #1 by\unitlength }

\def\putvector(#1,#2)(#3,#4)#5#6{{%
\ifnum3<\arrowtype \putdashvector(#1,#2)(#3,#4)#5\arrowtype \else
\ifnum\arrowtype<-3 \putdashvector(#1,#2)(#3,#4)#5\arrowtype \else
\xpos=#1 \ypos=#2 \run=#3 \rise=#4 \arrowlength=#5 \ifnum
\arrowtype<0
    \ifnum \run=0
        \advance \ypos by-\arrowlength
    \else
        \tempcounta \arrowlength
        \multiply \tempcounta by\rise
        \divide \tempcounta by\run
        \ifnum\run>0
            \advance \xpos by\arrowlength
            \advance \ypos by\tempcounta
        \else
            \advance \xpos by-\arrowlength
            \advance \ypos by-\tempcounta
        \fi
    \fi
    \multiply \arrowtype by-1
    \multiply \rise by-1
    \multiply \run by-1
\fi \ifcase \arrowtype
\or \put(\xpos,\ypos){\vector(\run,\rise){\arrowlength}}%
\or \put(\xpos,\ypos){\mvector(\run,\rise)\arrowlength}%
\or \put(\xpos,\ypos){\evector(\run,\rise){\arrowlength}}%
\fi\fi\fi }}

\def\putsplitvector(#1,#2)#3#4{
\xpos #1 \ypos #2 \arrowtype #4 \halflength #3 \arrowlength #3
\gap 140 \advance \halflength by-\gap \divide \halflength by2
\ifnum\arrowtype>0
   \ifcase \arrowtype
   \or \put(\xpos,\ypos){\line(0,-1){\halflength}}%
       \advance\ypos by-\halflength
       \advance\ypos by-\gap
       \put(\xpos,\ypos){\vector(0,-1){\halflength}}%
   \or \put(\xpos,\ypos){\line(0,-1)\halflength}%
       \put(\xpos,\ypos){\vector(0,-1)3}%
       \advance\ypos by-\halflength
       \advance\ypos by-\gap
       \put(\xpos,\ypos){\vector(0,-1){\halflength}}%
   \or \put(\xpos,\ypos){\line(0,-1)\halflength}%
       \advance\ypos by-\halflength
       \advance\ypos by-\gap
       \put(\xpos,\ypos){\evector(0,-1){\halflength}}%
   \fi
\else \arrowtype=-\arrowtype
   \ifcase\arrowtype
   \or \advance \ypos by-\arrowlength
       \put(\xpos,\ypos){\line(0,1){\halflength}}%
       \advance\ypos by\halflength
       \advance\ypos by\gap
       \put(\xpos,\ypos){\vector(0,1){\halflength}}%
   \or \advance \ypos by-\arrowlength
       \put(\xpos,\ypos){\line(0,1)\halflength}%
       \put(\xpos,\ypos){\vector(0,1)3}%
       \advance\ypos by\halflength
       \advance\ypos by\gap
       \put(\xpos,\ypos){\vector(0,1){\halflength}}%
   \or \advance \ypos by-\arrowlength
       \put(\xpos,\ypos){\line(0,1)\halflength}%
       \advance\ypos by\halflength
       \advance\ypos by\gap
       \put(\xpos,\ypos){\evector(0,1){\halflength}}%
   \fi
\fi }

\def\putmorphism(#1)(#2,#3)[#4`#5`#6]#7#8#9{{%
\run #2 \rise #3 \ifnum\rise=0
  \puthmorphism(#1)[#4`#5`#6]{#7}{#8}#9%
\else\ifnum\run=0
  \putvmorphism(#1)[#4`#5`#6]{#7}{#8}#9%
\else
\setpos(#1)%
\arrowlength #7 \arrowtype #8 \ifnum\run=0 \else\ifnum\rise=0
\else \ifnum\run>0
    \coefa=1
\else
   \coefa=-1
\fi \ifnum\arrowtype>0
   \coefb=0
   \coefc=-1
\else
   \coefb=\coefa
   \coefc=1
   \arrowtype=-\arrowtype
\fi \width=2 \multiply \width by\run \divide \width by\rise \ifnum
\width<0  \width=-\width\fi \advance\width by60 \if l#9
\width=-\width\fi
\putbox(\xpos,\ypos){#4}
{\multiply \coefa by\arrowlength
\advance\xpos by\coefa \multiply \coefa by\rise \divide \coefa
by\run \advance \ypos by\coefa
\putbox(\xpos,\ypos){#5} }%
{\multiply \coefa by\arrowlength
\divide \coefa by2 \advance \xpos by\coefa \advance \xpos by\width
\multiply \coefa by\rise \divide \coefa by\run \advance \ypos
by\coefa
\if l#9%
   \putrbox(\xpos,\ypos){#6}%
\else\if r#9%
   \putlbox(\xpos,\ypos){#6}%
\fi\fi }%
{\multiply \rise by-\coefc
\multiply \run by-\coefc \multiply \coefb by\arrowlength \advance
\xpos by\coefb \multiply \coefb by\rise \divide \coefb by\run
\advance \ypos by\coefb \multiply \coefc by70 \advance \ypos
by\coefc \multiply \coefc by\run \divide \coefc by\rise \advance
\xpos by\coefc \multiply \coefa by140 \multiply \coefa by\run
\divide \coefa by\rise \advance \arrowlength by\coefa
\ifcase\arrowtype
\or \put(\xpos,\ypos){\vector(\run,\rise){\arrowlength}}%
\or \put(\xpos,\ypos){\mvector(\run,\rise){\arrowlength}}%
\or \put(\xpos,\ypos){\evector(\run,\rise){\arrowlength}}%
\fi}\fi\fi\fi\fi}}

\newcount\numbdashes \newcount\lengthdash \newcount\increment

\def\howmanydashes{
\numbdashes=\arrowlength \lengthdash=40 \divide\numbdashes by
\lengthdash \lengthdash=\arrowlength \divide\lengthdash by
\numbdashes
\increment=\lengthdash \multiply\lengthdash by 3
\divide\lengthdash by 5 }

\def\putdashvector(#1)(#2,#3)#4#5{%
\ifnum#3=0 \putdashhvector(#1){#4}#5 \else \ifnum#2=0
\putdashvvector(#1){#4}#5\fi\fi}

\def\putdashhvector(#1,#2)#3#4{{%
\arrowlength=#3 \howmanydashes
\multiput(#1,#2)(\increment,0){\numbdashes}%
{\vrule height .4pt width \lengthdash\unitlength} \arrowtype=#4
\xpos=#1 \ifnum\arrowtype<0 \advance\arrowtype by 7 \fi
\ifcase\arrowtype \or \advance\xpos by 10
    \put(\xpos,#2){\vector(-1,0){\lengthdash}}
    \advance\xpos by 40
    \put(\xpos,#2){\vector(-1,0){\lengthdash}}
\or \advance \xpos by 10
    \put(\xpos,#2){\vector(-1,0){\lengthdash}}
    \advance\xpos by  \arrowlength
    \advance\xpos by  -50
    \put(\xpos,#2){\vector(-1,0){\lengthdash}}
\or \advance\xpos by 10
    \put(\xpos,#2){\vector(-1,0){\lengthdash}}
\or \advance\xpos by \arrowlength
    \advance\xpos by -\lengthdash
    \put(\xpos,#2){\vector(1,0){\lengthdash}}
\or {\advance\xpos by 10
    \put(\xpos,#2){\vector(1,0){\lengthdash}}}
    \advance\xpos by \arrowlength
    \advance\xpos by -\lengthdash
    \put(\xpos,#2){\vector(1,0){\lengthdash}}
\or \advance\xpos by \arrowlength
    \advance\xpos by -\lengthdash
    \put(\xpos,#2){\vector(1,0){\lengthdash}}
    \advance\xpos by -40
    \put(\xpos,#2){\vector(1,0){\lengthdash}}
   \fi
}}

\def\putdashvvector(#1,#2)#3#4{{%
\arrowlength=#3 \howmanydashes \ypos=#2 \advance\ypos by
-\arrowlength
\multiput(#1,#2)(0,\increment){\numbdashes}%
    {\vrule width .4pt height \lengthdash\unitlength}
\arrowtype=#4 \ypos=#2 \ifnum\arrowtype<0 \advance\arrowtype by 7
\fi \ifcase\arrowtype \or \advance\ypos by \arrowlength
\advance\ypos by -40
    \put(#1,\ypos){\vector(0,1){\lengthdash}}
    \advance\ypos by -40
    \put(#1,\ypos){\vector(0,1){\lengthdash}}
\or \advance\ypos by 10
    \put(#1,\ypos){\vector(0,1){\lengthdash}}
    \advance\ypos by \arrowlength \advance\ypos by -40
    \put(#1,\ypos){\vector(0,1){\lengthdash}}
\or \advance\ypos by \arrowlength \advance\ypos by -40
    \put(#1,\ypos){\vector(0,1){\lengthdash}}
\or \advance\ypos by 10
    \put(#1,\ypos){\vector(0,-1){\lengthdash}}
\or \advance\ypos by 10
    \put(#1,\ypos){\vector(0,-1){\lengthdash}}
    \advance\ypos by \arrowlength \advance\ypos by -40
    \put(#1,\ypos){\vector(0,-1){\lengthdash}}
\or \advance\ypos by 10
    \put(#1,\ypos){\vector(0,-1){\lengthdash}}
    \advance\ypos by 40
    \put(#1,\ypos){\vector(0,-1){\lengthdash}}
\fi }}

\def\puthmorphism(#1,#2)[#3`#4`#5]#6#7#8{{%
\xpos #1 \ypos #2 \width #6 \arrowlength #6 \arrowtype=#7
\putbox(\xpos,\ypos){#3\vphantom{#4}}%
{\advance \xpos by\arrowlength
\putbox(\xpos,\ypos){\vphantom{#3}#4}}%
\horsize{\tempcounta}{#3}%
\horsize{\tempcountb}{#4}%
\divide \tempcounta by2 \divide \tempcountb by2 \advance
\tempcounta by30 \advance \tempcountb by30 \advance \xpos
by\tempcounta \advance \arrowlength by-\tempcounta \advance
\arrowlength by-\tempcountb
\putvector(\xpos,\ypos)(1,0)\arrowlength\arrowtype \divide
\arrowlength by2 \advance \xpos by\arrowlength
\vertsize{\tempcounta}{#5}%
\divide\tempcounta by2 \advance \tempcounta by20
\if a#8 %
   \advance \ypos by\tempcounta
   \putbox(\xpos,\ypos){#5}%
\else
   \advance \ypos by-\tempcounta
   \putbox(\xpos,\ypos){#5}%
\fi}}

\def\putvmorphism(#1,#2)[#3`#4`#5]#6#7#8{{%
\xpos #1 \ypos #2 \arrowlength #6 \arrowtype #7
\settowidth{\xlen}{$#5$}%
\putbox(\xpos,\ypos){#3}%
{\advance \ypos by-\arrowlength
\putbox(\xpos,\ypos){#4}}%
{\advance\arrowlength by-140 \advance \ypos by-70 \ifdim\xlen>0pt
   \if m#8%
      \putsplitvector(\xpos,\ypos)\arrowlength\arrowtype
   \else
   \putvector(\xpos,\ypos)(0,-1)\arrowlength\arrowtype
   \fi
\else
   \putvector(\xpos,\ypos)(0,-1)\arrowlength\arrowtype
\fi}%
\ifdim\xlen>0pt
   \divide \arrowlength by2
   \advance\ypos by-\arrowlength
   \if l#8%
      \advance \xpos by-40
      \putrbox(\xpos,\ypos){#5}%
   \else\if r#8%
      \advance \xpos by40
      \putlbox(\xpos,\ypos){#5}%
   \else
      \putbox(\xpos,\ypos){#5}%
   \fi\fi
\fi }}

\def\putsquarep<#1>(#2)[#3;#4`#5`#6`#7]{{%
\setsqparms[#1]%
\setpos(#2)%
\settokens`#3`%
\puthmorphism(\xpos,\ypos)[\tokenc`\tokend`{#7}]{\width}{\arrowtyped}b%
\advance\ypos by \height
\puthmorphism(\xpos,\ypos)[\tokena`\tokenb`{#4}]{\width}{\arrowtypea}a%
\putvmorphism(\xpos,\ypos)[``{#5}]{\height}{\arrowtypeb}l%
\advance\xpos by \width
\putvmorphism(\xpos,\ypos)[``{#6}]{\height}{\arrowtypec}r%
}}

\def\putsquare{\@ifnextchar <{\putsquarep}{\putsquarep%
   <\arrowtypea`\arrowtypeb`\arrowtypec`\arrowtyped;\width`\height>}}
\def\square{\@ifnextchar< {\squarep}{\squarep
   <\arrowtypea`\arrowtypeb`\arrowtypec`\arrowtyped;\width`\height>}}
\def\squarep<#1>[#2`#3`#4`#5;#6`#7`#8`#9]{{
\setsqparms[#1]
\diagram
\putsquarep<\arrowtypea`\arrowtypeb`\arrowtypec`
\arrowtyped;\width`\height>
(0,0)[#2`#3`#4`{#5};#6`#7`#8`{#9}]
\enddiagram
}}                                                 
\def\putptrianglep<#1>(#2,#3)[#4`#5`#6;#7`#8`#9]{{%
\settriparms[#1]%
\xpos=#2 \ypos=#3 \advance\ypos by \height
\puthmorphism(\xpos,\ypos)[#4`#5`{#7}]{\height}{\arrowtypea}a%
\putvmorphism(\xpos,\ypos)[`#6`{#8}]{\height}{\arrowtypeb}l%
\advance\xpos by\height
\putmorphism(\xpos,\ypos)(-1,-1)[``{#9}]{\height}{\arrowtypec}r%
}}

\def\putptriangle{\@ifnextchar <{\putptrianglep}{\putptrianglep
   <\arrowtypea`\arrowtypeb`\arrowtypec;\height>}}
\def\ptriangle{\@ifnextchar <{\ptrianglep}{\ptrianglep
   <\arrowtypea`\arrowtypeb`\arrowtypec;\height>}}
\def\ptrianglep<#1>[#2`#3`#4;#5`#6`#7]{{
\settriparms[#1]
\diagram
\putptrianglep<\arrowtypea`\arrowtypeb`
\arrowtypec;\height>
(0,0)[#2`#3`#4;#5`#6`{#7}]
\enddiagram
}}                                            

\def\putqtrianglep<#1>(#2,#3)[#4`#5`#6;#7`#8`#9]{{%
\settriparms[#1]%
\xpos=#2 \ypos=#3 \advance\ypos by\height
\puthmorphism(\xpos,\ypos)[#4`#5`{#7}]{\height}{\arrowtypea}a%
\putmorphism(\xpos,\ypos)(1,-1)[``{#8}]{\height}{\arrowtypeb}l%
\advance\xpos by\height
\putvmorphism(\xpos,\ypos)[`#6`{#9}]{\height}{\arrowtypec}r%
}}

\def\putqtriangle{\@ifnextchar <{\putqtrianglep}{\putqtrianglep
   <\arrowtypea`\arrowtypeb`\arrowtypec;\height>}}
\def\qtriangle{\@ifnextchar <{\qtrianglep}{\qtrianglep
   <\arrowtypea`\arrowtypeb`\arrowtypec;\height>}}
\def\qtrianglep<#1>[#2`#3`#4;#5`#6`#7]{{
\settriparms[#1]
\width=\height                                
\diagram
\putqtrianglep<\arrowtypea`\arrowtypeb`
\arrowtypec;\height>
(0,0)[#2`#3`#4;#5`#6`{#7}]
\enddiagram
}}

\def\putdtrianglep<#1>(#2,#3)[#4`#5`#6;#7`#8`#9]{{%
\settriparms[#1]%
\xpos=#2 \ypos=#3
\puthmorphism(\xpos,\ypos)[#5`#6`{#9}]{\height}{\arrowtypec}b%
\advance\xpos by \height \advance\ypos by\height
\putmorphism(\xpos,\ypos)(-1,-1)[``{#7}]{\height}{\arrowtypea}l%
\putvmorphism(\xpos,\ypos)[#4``{#8}]{\height}{\arrowtypeb}r%
}}

\def\putdtriangle{\@ifnextchar <{\putdtrianglep}{\putdtrianglep
   <\arrowtypea`\arrowtypeb`\arrowtypec;\height>}}
\def\dtriangle{\@ifnextchar <{\dtrianglep}{\dtrianglep
   <\arrowtypea`\arrowtypeb`\arrowtypec;\height>}}
\def\dtrianglep<#1>[#2`#3`#4;#5`#6`#7]{{
\settriparms[#1]
\width=\height                                
\diagram
\putdtrianglep<\arrowtypea`\arrowtypeb`
\arrowtypec;\height>
(0,0)[#2`#3`#4;#5`#6`{#7}]
\enddiagram
}}

\def\putbtrianglep<#1>(#2,#3)[#4`#5`#6;#7`#8`#9]{{%
\settriparms[#1]%
\xpos=#2 \ypos=#3
\puthmorphism(\xpos,\ypos)[#5`#6`{#9}]{\height}{\arrowtypec}b%
\advance\ypos by\height
\putmorphism(\xpos,\ypos)(1,-1)[``{#8}]{\height}{\arrowtypeb}r%
\putvmorphism(\xpos,\ypos)[#4``{#7}]{\height}{\arrowtypea}l%
}}

\def\putbtriangle{\@ifnextchar <{\putbtrianglep}{\putbtrianglep
   <\arrowtypea`\arrowtypeb`\arrowtypec;\height>}}
\def\btriangle{\@ifnextchar <{\btrianglep}{\btrianglep
   <\arrowtypea`\arrowtypeb`\arrowtypec;\height>}}
\def\btrianglep<#1>[#2`#3`#4;#5`#6`#7]{{
\settriparms[#1]
\width=\height                               
\diagram
\putbtrianglep<\arrowtypea`\arrowtypeb`
\arrowtypec;\height>
(0,0)[#2`#3`#4;#5`#6`{#7}]
\enddiagram
}}

\def\putAtrianglep<#1>(#2,#3)[#4`#5`#6;#7`#8`#9]{{%
\settriparms[#1]%
\xpos=#2 \ypos=#3 {\multiply \height by2
\puthmorphism(\xpos,\ypos)[#5`#6`{#9}]{\height}{\arrowtypec}b}%
\advance\xpos by\height \advance\ypos by\height
\putmorphism(\xpos,\ypos)(-1,-1)[#4``{#7}]{\height}{\arrowtypea}l%
\putmorphism(\xpos,\ypos)(1,-1)[``{#8}]{\height}{\arrowtypeb}r%
}}

\def\putAtriangle{\@ifnextchar <{\putAtrianglep}{\putAtrianglep
   <\arrowtypea`\arrowtypeb`\arrowtypec;\height>}}
\def\Atriangle{\@ifnextchar <{\Atrianglep}{\Atrianglep
   <\arrowtypea`\arrowtypeb`\arrowtypec;\height>}}
\def\Atrianglep<#1>[#2`#3`#4;#5`#6`#7]{{
\settriparms[#1]
\width=\height                                     
\diagram
\putAtrianglep<\arrowtypea`\arrowtypeb`
\arrowtypec;\height>
(0,0)[#2`#3`#4;#5`#6`{#7}]
\enddiagram
}}

\def\putAtrianglepairp<#1>(#2)[#3;#4`#5`#6`#7`#8]{{%
\settripairparms[#1]%
\setpos(#2)%
\settokens`#3`%
\puthmorphism(\xpos,\ypos)[\tokenb`\tokenc`{#7}]{\height}{\arrowtyped}b%
\advance\xpos by\height
\puthmorphism(\xpos,\ypos)[\phantom{\tokenc}`\tokend`{#8}]%
{\height}{\arrowtypee}b%
\advance\ypos by\height
\putmorphism(\xpos,\ypos)(-1,-1)[\tokena``{#4}]{\height}{\arrowtypea}l%
\putvmorphism(\xpos,\ypos)[``{#5}]{\height}{\arrowtypeb}m%
\putmorphism(\xpos,\ypos)(1,-1)[``{#6}]{\height}{\arrowtypec}r%
}}

\def\putAtrianglepair{\@ifnextchar <{\putAtrianglepairp}{\putAtrianglepairp%
   <\arrowtypea`\arrowtypeb`\arrowtypec`\arrowtyped`\arrowtypee;\height>}}
\def\Atrianglepair{\@ifnextchar <{\Atrianglepairp}{\Atrianglepairp%
   <\arrowtypea`\arrowtypeb`\arrowtypec`\arrowtyped`\arrowtypee;\height>}}

\def\Atrianglepairp<#1>[#2;#3`#4`#5`#6`#7]{{
\settripairparms[#1]
\settokens`#2`
\width=\height                                
\diagram
\putAtrianglepairp                            
<\arrowtypea`\arrowtypeb`\arrowtypec`
\arrowtyped`\arrowtypee;\height>
(0,0)[{#2};#3`#4`#5`#6`{#7}]
\enddiagram
}}

\def\putVtrianglep<#1>(#2,#3)[#4`#5`#6;#7`#8`#9]{{%
\settriparms[#1]%
\xpos=#2 \ypos=#3 \advance\ypos by\height {\multiply\height by2
\puthmorphism(\xpos,\ypos)[#4`#5`{#7}]{\height}{\arrowtypea}a}%
\putmorphism(\xpos,\ypos)(1,-1)[`#6`{#8}]{\height}{\arrowtypeb}l%
\advance\xpos by\height \advance\xpos by\height
\putmorphism(\xpos,\ypos)(-1,-1)[``{#9}]{\height}{\arrowtypec}r%
}}

\def\putVtriangle{\@ifnextchar <{\putVtrianglep}{\putVtrianglep
   <\arrowtypea`\arrowtypeb`\arrowtypec;\height>}}
\def\Vtriangle{\@ifnextchar <{\Vtrianglep}{\Vtrianglep
   <\arrowtypea`\arrowtypeb`\arrowtypec;\height>}}
\def\Vtrianglep<#1>[#2`#3`#4;#5`#6`#7]{{
\settriparms[#1]
\width=\height                                 
\diagram
\putVtrianglep<\arrowtypea`\arrowtypeb`
\arrowtypec;\height>
(0,0)[#2`#3`#4;#5`#6`{#7}]
\enddiagram
}}

\def\putVtrianglepairp<#1>(#2)[#3;#4`#5`#6`#7`#8]{{
\settripairparms[#1]%
\setpos(#2)%
\settokens`#3`%
\advance\ypos by\height
\putmorphism(\xpos,\ypos)(1,-1)[`\tokend`{#6}]{\height}{\arrowtypec}l%
\puthmorphism(\xpos,\ypos)[\tokena`\tokenb`{#4}]{\height}{\arrowtypea}a%
\advance\xpos by\height
\puthmorphism(\xpos,\ypos)[\phantom{\tokenb}`\tokenc`{#5}]%
{\height}{\arrowtypeb}a%
\putvmorphism(\xpos,\ypos)[``{#7}]{\height}{\arrowtyped}m%
\advance\xpos by\height
\putmorphism(\xpos,\ypos)(-1,-1)[``{#8}]{\height}{\arrowtypee}r%
}}

\def\putVtrianglepair{\@ifnextchar <{\putVtrianglepairp}{\putVtrianglepairp%
    <\arrowtypea`\arrowtypeb`\arrowtypec`\arrowtyped`\arrowtypee;\height>}}
\def\Vtrianglepair{\@ifnextchar <{\Vtrianglepairp}{\Vtrianglepairp%
    <\arrowtypea`\arrowtypeb`\arrowtypec`\arrowtyped`\arrowtypee;\height>}}
\def\Vtrianglepairp<#1>[#2;#3`#4`#5`#6`#7]{{
\settripairparms[#1]
\settokens`#2`
\diagram
\putVtrianglepairp                             
<\arrowtypea`\arrowtypeb`\arrowtypec`
\arrowtyped`\arrowtypee;\height>
(0,0)[{#2};#3`#4`#5`#6`{#7}]
\enddiagram
}}

\def\putCtrianglep<#1>(#2,#3)[#4`#5`#6;#7`#8`#9]{{%
\settriparms[#1]%
\xpos=#2 \ypos=#3 \advance\ypos by\height
\putmorphism(\xpos,\ypos)(1,-1)[``{#9}]{\height}{\arrowtypec}l%
\advance\xpos by\height \advance\ypos by\height
\putmorphism(\xpos,\ypos)(-1,-1)[#4`#5`{#7}]{\height}{\arrowtypea}l%
{\multiply\height by 2
\putvmorphism(\xpos,\ypos)[`#6`{#8}]{\height}{\arrowtypeb}r}%
}}

\def\putCtriangle{\@ifnextchar <{\putCtrianglep}{\putCtrianglep
    <\arrowtypea`\arrowtypeb`\arrowtypec;\height>}}
\def\Ctriangle{\@ifnextchar <{\Ctrianglep}{\Ctrianglep
    <\arrowtypea`\arrowtypeb`\arrowtypec;\height>}}
\def\Ctrianglep<#1>[#2`#3`#4;#5`#6`#7]{{
\settriparms[#1]
\width=\height                               
\diagram
\putCtrianglep<\arrowtypea`\arrowtypeb`
\arrowtypec;\height>
(0,0)[#2`#3`#4;#5`#6`{#7}]
\enddiagram
}}                                           
\def\putDtrianglep<#1>(#2,#3)[#4`#5`#6;#7`#8`#9]{{%
\settriparms[#1]%
\xpos=#2 \ypos=#3 \advance\xpos by\height \advance\ypos by\height
\putmorphism(\xpos,\ypos)(-1,-1)[``{#9}]{\height}{\arrowtypec}r%
\advance\xpos by-\height \advance\ypos by\height
\putmorphism(\xpos,\ypos)(1,-1)[`#5`{#8}]{\height}{\arrowtypeb}r%
{\multiply\height by 2
\putvmorphism(\xpos,\ypos)[#4`#6`{#7}]{\height}{\arrowtypea}l}%
}}

\def\putDtriangle{\@ifnextchar <{\putDtrianglep}{\putDtrianglep
    <\arrowtypea`\arrowtypeb`\arrowtypec;\height>}}
\def\Dtriangle{\@ifnextchar <{\Dtrianglep}{\Dtrianglep
   <\arrowtypea`\arrowtypeb`\arrowtypec;\height>}}
\def\Dtrianglep<#1>[#2`#3`#4;#5`#6`#7]{{
\settriparms[#1]
\width=\height                              
\diagram
\putDtrianglep<\arrowtypea`\arrowtypeb`
\arrowtypec;\height>
(0,0)[#2`#3`#4;#5`#6`{#7}]
\enddiagram
}}                                          
\def\setrecparms[#1`#2]{\width=#1 \height=#2}%

\def\recursep<#1`#2>[#3;#4`#5`#6`#7`#8]{{\m@th
\width=#1 \height=#2 \settokens`#3`
\settowidth{\tempdimen}{$\tokena$} \ifdim\tempdimen=0pt
  \savebox{\tempboxa}{\hbox{$\tokenb$}}%
  \savebox{\tempboxb}{\hbox{$\tokend$}}%
  \savebox{\tempboxc}{\hbox{$#6$}}%
\else
  \savebox{\tempboxa}{\hbox{$\hbox{$\tokena$}\times\hbox{$\tokenb$}$}}%
  \savebox{\tempboxb}{\hbox{$\hbox{$\tokena$}\times\hbox{$\tokend$}$}}%
  \savebox{\tempboxc}{\hbox{$\hbox{$\tokena$}\times\hbox{$#6$}$}}%
\fi \ypos=\height \divide\ypos by 2 \xpos=\ypos \advance\xpos by
\width \bfig
\putCtrianglep<-1`1`1;\ypos>(0,0)[`\tokenc`;#5`#6`{#7}]%
\puthmorphism(\ypos,0)[\tokend`\usebox{\tempboxb}`{#8}]{\width}{-1}b%
\puthmorphism(\ypos,\height)[\tokenb`\usebox{\tempboxa}`{#4}]{\width}{-1}a%
\advance\ypos by \width
\putvmorphism(\ypos,\height)[``\usebox{\tempboxc}]{\height}1r%
\efig }}

\def\recurse{\@ifnextchar <{\recursep}{\recursep<\width`\height>}}

\def\puttwohmorphisms(#1,#2)[#3`#4;#5`#6]#7#8#9{{%
\puthmorphism(#1,#2)[#3`#4`]{#7}0a \ypos=#2 \advance\ypos by 20
\puthmorphism(#1,\ypos)[\phantom{#3}`\phantom{#4}`#5]{#7}{#8}a
\advance\ypos by -40
\puthmorphism(#1,\ypos)[\phantom{#3}`\phantom{#4}`#6]{#7}{#9}b }}

\def\puttwovmorphisms(#1,#2)[#3`#4;#5`#6]#7#8#9{{%
\putvmorphism(#1,#2)[#3`#4`]{#7}0a \xpos=#1 \advance\xpos by -20
\putvmorphism(\xpos,#2)[\phantom{#3}`\phantom{#4}`#5]{#7}{#8}l
\advance\xpos by 40
\putvmorphism(\xpos,#2)[\phantom{#3}`\phantom{#4}`#6]{#7}{#9}r }}

\def\puthcoequalizer(#1)[#2`#3`#4;#5`#6`#7]#8#9{{%
\setpos(#1)%
\puttwohmorphisms(\xpos,\ypos)[#2`#3;#5`#6]{#8}11%
\advance\xpos by #8
\puthmorphism(\xpos,\ypos)[\phantom{#3}`#4`#7]{#8}1{#9} }}

\def\putvcoequalizer(#1)[#2`#3`#4;#5`#6`#7]#8#9{{%
\setpos(#1)%
\puttwovmorphisms(\xpos,\ypos)[#2`#3;#5`#6]{#8}11%
\advance\ypos by -#8
\putvmorphism(\xpos,\ypos)[\phantom{#3}`#4`#7]{#8}1{#9} }}

\def\putthreehmorphisms(#1)[#2`#3;#4`#5`#6]#7(#8)#9{{%
\setpos(#1) \settypes(#8)
\if a#9 %
     \vertsize{\tempcounta}{#5}%
     \vertsize{\tempcountb}{#6}%
     \ifnum \tempcounta<\tempcountb \tempcounta=\tempcountb \fi
\else
     \vertsize{\tempcounta}{#4}%
     \vertsize{\tempcountb}{#5}%
     \ifnum \tempcounta<\tempcountb \tempcounta=\tempcountb \fi
\fi \advance \tempcounta by 60
\puthmorphism(\xpos,\ypos)[#2`#3`#5]{#7}{\arrowtypeb}{#9}
\advance\ypos by \tempcounta
\puthmorphism(\xpos,\ypos)[\phantom{#2}`\phantom{#3}`#4]{#7}{\arrowtypea}{#9}
\advance\ypos by -\tempcounta \advance\ypos by -\tempcounta
\puthmorphism(\xpos,\ypos)[\phantom{#2}`\phantom{#3}`#6]{#7}{\arrowtypec}{#9}
}}

\def\setarrowtoks[#1`#2`#3`#4`#5`#6]{%
\def\toka{#1}
\def\tokb{#2}
\def\tokc{#3}
\def\tokd{#4}
\def\toke{#5}
\def\tokf{#6}
}
\def\hex{\@ifnextchar <{\hexp}{\hexp<1000`400>}}
\def\hexp<#1`#2>[#3`#4`#5`#6`#7`#8;#9]{%
\setarrowtoks[#9] \yext=#2 \advance \yext by #2 \xext=#1
\advance\xext by \yext \bfig
\putCtriangle<-1`0`1;#2>(0,0)[`#5`;\tokb``\tokd] \xext=#1 \yext=#2
\advance \yext by #2
\putsquare<1`0`0`1;\xext`\yext>(#2,0)[#3`#4`#7`#8;\toka```\tokf]
\advance \xext by #2
\putDtriangle<0`1`-1;#2>(\xext,0)[`#6`;`\tokc`\toke] \efig }
\makeatother

\newcommand\eq[1]{(\ref{#1})}
\newcommand\eqs[2]{(\ref{#1}--\ref{#2})}

\newcommand\mapright[1]{\smash{
        \mathop{\mbox{\large{$\longrightarrow$}}}\limits^{#1}}}

\newcommand\mathC{\mkern1mu\raise2.2pt\hbox{$\scriptscriptstyle|$}
        {\mkern-7mu\rm C}}           
\newcommand{\mathR}{\mathbb{R}}      
\newcommand{\mathN}{\mathbb{N}}
\newcommand{\mathZ}{\mathbb{Z}}

\newcommand\vs[1]{\vspace{#1pt}}
\newcommand\hs[1]{\hspace{#1pt}}
\newcommand\displayE[4]{
    \begin{center}
    \framebox{\begin{minipage}[t]{#1cm}\vs{#2}{\ #4}\vs{#3}
                \end{minipage}}
    \end{center}}

\newcommand{\di}{\diamond}
\newcommand{\ga}{\gamma}
\newcommand{\ka}{\kappa}
\newcommand{\om}{\omega}
\renewcommand\l{\lambda}
\newcommand\s{\sigma}
\renewcommand\a{\alpha}

\newcommand\De{\Delta}
\newcommand{\Ga}{\Gamma}
\renewcommand{\O}{\Omega}
\newcommand\Si{\Sigma}

\newcommand\sq{\rightsquigarrow}
\newcommand\la{\langle}
\newcommand{\map}{\rightarrow}
\newcommand\ra{\rangle}

\newcommand\brak[2]{\la#1,#2\ra}

\newcommand{\sa}{{\rm sa}}
\newcommand{\id}{{\rm id}}
\newcommand\ie{{i.e.},}
\newcommand{\op}{{\rm op}}
\newcommand{\picl}{\pi_{{\rm cl}}}
\newcommand{\piqt}{\pi_{{\rm qt}}}
\newcommand\pr{{\rm pr}}
\newcommand\varin{\,\varepsilon\,}
\newcommand\tr[1]{{\rm tr}(#1)}

\newcommand{\A}{{\hat A}}
\newcommand{\U}{{\hat U}}
\renewcommand{\P}{{\hat P}}
\renewcommand{\S}{{\cal S}}
\newcommand\Q[1]{{\cal Q}(#1)}
\newcommand{\Hi}{{\cal H}}
\newcommand{\R}{{\cal R}}

\newcommand\BH{B\mathcal{(H)}}
\newcommand\PH{\mathcal{P(H)}}
\newcommand\BlH{\mathcal{B}l({\Hi})}
\newcommand\PV{\mathcal{P}(V)}
\newcommand\UH{\mathcal{U(H)}}
\newcommand\BV{{\rm BV}(\Ob{\V{}},\mathR)}

\newcommand\q[1]{`#1\mbox{'}}
\newcommand\bra[1]{\langle #1|\,}
\newcommand\ket[1]{\,|#1\rangle}
\newcommand\ketbra[1]{\ket{#1}\bra{#1}}
\renewcommand\sp{{\rm sp}}  

\newcommand\dasmap{\delta}
\newcommand\delo{\delta^o}
\newcommand\deli{\delta^i}

\newcommand\das[1]{\delta(\hat{#1})}
\newcommand\daso[1]{\delta^o(\hat{#1})}
\newcommand\dasi[1]{\delta^i(\hat{#1})}

\newcommand\dasto[2]{\delta(\hat{#2})_{#1}}
\newcommand\dastoo[2]{\delta^o(\hat{#2})_{#1}}
\newcommand\dastoi[2]{\delta^i(\hat{#2})_{#1}}

\newcommand\dasB[1]{\breve{\delta}(\hat{#1})}                               
\newcommand\dasBV[2]{\breve{\delta}(\hat{#2})_{#1}}                 
\newcommand\dasBi[1]{\breve{\delta}^i(\hat{#1})}                        
\newcommand\dasBVi[2]{\breve{\delta}^i(\hat{#2})_{#1}}          
\newcommand\dasBo[1]{\breve{\delta}^o(\hat{#1})}                        
\newcommand\dasBVo[2]{\breve{\delta}^o(\hat{#2})_{#1}}          

\newcommand\GT[1]{\overline {#1}}

\newcommand\Hom[3]{{\rm Hom}_{#1}\big(#2,#3\big)}
\newcommand\name[1]{\ulcorner #1\urcorner}   
\newcommand\cha[1]{\chi_{#1}}                 
\newcommand\Ob[1]{{\rm Ob(#1)}}
\newcommand\Sub[1]{{\rm Sub}(#1)}             
\newcommand\Subcl[1]{{\rm Sub}_{{\rm cl}}(#1)} 
\newcommand\fu[1]{#1}

\newcommand\F[1]{F_{\L{#1}}\big(\Sigma,\R\big)}
\renewcommand\L[1]{\mathcal{L}({#1})}
\newcommand\PL[1]{{\cal PL}(#1)}
\newcommand\LeftDB{[\mkern-3mu[}
\newcommand\RightDB{]\mkern-3mu]}
\newcommand\Val[1]{\LeftDB\,#1\,\RightDB}
\newcommand\TVal[2]{\nu\big(#1;#2\big)}         
\newcommand\TValM[1]{\nu\big(\,#1\,\big)}               
\newcommand\typeTime{{\cal T}}
\newcommand\SAin[1]{\mbox{``}A\,\varepsilon\,#1\mbox{''}}
\newcommand\Ain[1]{A\,\varepsilon\,#1}
\newcommand\va[1]{\tilde{#1}}
\newcommand\so[1]{#1}

\newcommand\ps[1]{\underline{#1}}
\newcommand\cps[1]{\overline{#1}}
\newcommand\TO{\mathbb{T}}   
\newcommand\WO{\mathbb{W}}    
\newcommand\w{\mathfrak{w}}
\newcommand{\Om}{\ps{\Omega}}
\newcommand{\G}{\ps{O}}                   
\renewcommand{\H}{\ps{I}}                 
\newcommand{\dG}{\ps{\mkern1mu\raise2.5pt\hbox{$\scriptscriptstyle|$}
        {\mkern-7mu\rm O}}}              
\newcommand{\dH}{\ps{{\rm I\! I}}}        
\newcommand{\dOU}{\ps{\mkern1mu\raise2.5pt\hbox{$\scriptscriptstyle|$}
        {\mkern-7mu\rm U}}}              

\newcommand{\Sig}{\ps{\Sigma}}            
\newcommand{\PSig}{P_{{\rm cl}}\Sig}       
\newcommand{\SR}{\ps{{\mathR}^{\succeq}}} 
\newcommand{\kSR}{k(\SR)}

\newcommand{\SpA}{\ps{\sp(\A)^{\succeq}}}
\newcommand\Sh[1]{{\rm Sh}(#1)}

\newcommand{\Loc}{{\bf Loc}}
\newcommand\Set{{\bf Sets}}                     
\newcommand\SetH[1]{\Set^{{\V{#1}}^{\rm op}}}   
\newcommand\SetC[1]{\Set^{{#1}^{\rm op}}}       

\newcommand{\Sys}{{\bf Sys}}                    
\newcommand\V[1]{{\cal V}(\Hi_{#1})}            

\newcommand{\OP}{\ps{{\mathR}^\preceq}}
\newcommand{\PR}[1]{\ps{#1^\leftrightarrow}}

\begin{document}

\begin{center}
{\large\bf `What is a Thing?': Topos Theory in the Foundations of
Physics\footnote{To appear in \emph{New Structures in Physics}, ed
R.~Coecke, Springer (2008).} }
\end{center}

\begin{center}
        Andreas~D\"oring\footnote{email: a.doering@imperial.ac.uk}\\[10pt]

\begin{center}                      and
\end{center}

        Chris Isham\footnote{email:  c.isham@imperial.ac.uk}\\[20pt]

        The Blackett Laboratory\\ Imperial College of Science,
        Technology \& Medicine\\ South Kensington\\ London SW7 2BZ\\
\end{center}

\begin{center}
      2 March 2008
     \end{center}


\begin{abstract}

\begin{quote}
\emph{``From the range of the basic questions of metaphysics we
shall here ask this \emph{one} question: ``What is a thing?'' The
question is quite old. What remains ever new about it is merely
that it must be asked again and again \cite{HeidThing}.''}
\end{quote}
\hfill{\bf Martin Heidegger}

\bigskip
The goal of this paper is to summarise the first steps in
developing a fundamentally new way of constructing theories of
physics. The motivation comes from a desire to address certain
deep issues that arise when contemplating quantum theories of
space and time. In doing so we provide a new answer to Heidegger's
timeless question ``What is a thing?''.

Our basic contention is that constructing a theory of physics is
equivalent to finding a representation in a topos of a certain
formal language that is attached to the system. Classical physics
uses the topos of sets. Other theories involve a different topos.
For the types of theory discussed in this paper, a key goal is to
represent any physical quantity $A$ with an arrow
$\breve{A}_\phi:\Si_\phi\map\R_\phi$ where $\Si_\phi$ and
$\R_\phi$ are two special objects (the `state-object' and
`quantity-value object') in the appropriate topos, $\tau_\phi$.

We discuss two different types of language that can be attached to
a system, $S$. The first, $\PL{S}$, is a propositional language;
the second, $\L{S}$, is a higher-order, typed language. Both
languages provide deductive systems with an intuitionistic logic.
With the aid of $\PL{S}$ we  expand and develop some of the
earlier work\footnote{By CJI and collaborators.} on topos theory
and quantum physics.  A key step is a process we term
`daseinisation' by which a projection operator is mapped to a
sub-object of the spectral presheaf $\Sig$---the topos quantum
analogue of a classical state space. The topos concerned is
$\SetH{}$: the category of contravariant set-valued functors on
the category (partially ordered set) $\V{}$ of commutative
sub-algebras of the algebra of bounded operators on the quantum
Hilbert space $\Hi$.

There are two types of daseinisation, called `outer' and `inner':
they involve approximating a  projection operator by projectors
that are, respectively, larger and smaller in the lattice of
projectors on $\Hi$.

We then introduce the more sophisticated language $\L{S}$ and use
it to study `truth objects' and `pseudo-states' in the topos.
These objects  topos play the role of states: a necessary
development as the spectral presheaf has no global elements, and
hence there are no microstates in the sense of classical physics.

One of the main mathematical achievements is finding  a topos
representation for self-adjoint operators. This involves showing
that, for any bounded, self-adjoint operator $\A$, there is a
corresponding arrow $\dasBo{A}:\Sig\map\SR$ where $\SR$ is the
quantity-value object for this theory. The construction of
$\dasBo{A}$ is an extension of the daseinisation of projection
operators.

The object $\SR$ is  a monoid-object only in the topos,
$\tau_\phi=\SetH{}$, of the theory, and to enhance the
applicability of the formalism we discuss another candidate,
$\PR{\mathR}$, for the quantity-value object. In this presheaf,
both inner- and outer-daseinisation are used in a symmetric way.
Another option is to apply to $\SR$ a topos analogue of the
Grothendieck extension of a monoid to a group. The resulting
object, $\kSR$, is an abelian group-object in $\tau_\phi$.

Finally we turn to considering a \emph{collection} of systems: in
particular, we are interested in the relation between the topos
representation of a composite system, and the representations of
its constituents.  Our approach to these matters is to construct a
\emph{category} of systems and to find  coherent topos
representations of the entire category.

\end{abstract}

\tableofcontents

\section{Introduction}
Many people who work in quantum gravity would  agree that a deep
change in our understanding of foundational issues will occur at
some point along the path. However, opinions differ greatly on
whether a radical revision  is necessary at the very
\emph{beginning} of the process, or if it will emerge `along the
way' from an existing, or future, research programme that is
formulated using the current paradigms. For example, many (albeit
not all) of the current generation of string theorists seem
inclined to this view, as do a, perhaps smaller, fraction of those
who work in loop quantum gravity.

In this article we take the iconoclastic view that a radical step
is needed at the very outset.  However, for anyone in this camp
the problem is always knowing where to start. It is easy to talk
about a `radical revision of current paradigms'---the phrase slips
lightly off the tongue---but converting this pious hope into a
concrete theoretical structure is a problem of the highest order.

For us, the starting point is quantum theory itself. More
precisely, we believe that this theory needs to be radically
revised, or even completely replaced, before a satisfactory theory
of quantum gravity can be obtained.

In this context,  a striking feature of the various current
programmes for quantising gravity---including superstring theory
and loop quantum gravity---is that, notwithstanding their
disparate views on the nature of space and time, they almost all
use more-or-less standard quantum theory. Although understandable
from a pragmatic viewpoint (since all we have \emph{is}
more-or-less standard quantum theory) this situation is
nevertheless questionable when viewed from a wider perspective.

For us, one of the most important issues is the use in the
standard quantum formalism of  critical mathematical ingredients
that are taken for granted and yet which, we claim, implicitly
assume certain properties of space and/or time. Such an \emph{a
priori} imposition of spatio-temporal concepts would be a major
category\footnote{The philosophy of Kant runs strongly in our
veins.} error  if they turn out to be fundamentally incompatible
with what is needed for a theory of quantum gravity.

A prime example is the use of the \emph{continuum}\footnote{When
used in this rather colloquial way, the word `continuum' suggests
primarily the cardinality of the sets concerned, and, secondly,
the topology that is conventionally placed on these sets.} by
which, in this context, is meant the real and/or complex numbers.
These are a central ingredient in all the various mathematical
frameworks in which quantum theory is commonly discussed. For
example, this is clearly so with the use of (i) Hilbert spaces or
$C^*$-algebras; (ii) geometric quantisation; (iii) probability
functions on a non-distributive quantum logic; (iv) deformation
quantisation; and (v) formal (\ie\ mathematically ill-defined)
path integrals and the like. The \emph{a priori} imposition of
such continuum concepts could be radically incompatible with a
quantum-gravity formalism in which, say, space-time is
fundamentally discrete: as, for example,  in  the causal-set
programme.

As we shall argue later, this issue is closely connected with the
question of what is meant by the `value' of a physical quantity.
In so far as the concept is meaningful at all at the Planck scale,
why should the value be a real number defined mathematically in
the usual way?

Another significant reason  for aspiring to change the quantum
formalism is the peristalithic problem of deciding how  a `quantum
theory of cosmology' could be interpreted  if one was lucky enough
to find one. Most people who worry about foundational issues in
quantum gravity would probably place the
quantum-cosmology/closed-system problem at, or near, the top of
their list of reasons for re-envisioning quantum theory. However,
although we are deeply interested in such conceptual issues,  the
primary motivation for our research programme is not to find a new
interpretation of quantum theory. Rather, our main goal is to find
a novel structural framework within which new  \emph{types} of
theories of physics can  be constructed.

However, having said that, in the context of quantum cosmology it
is certainly true that the lack of any external `observer'  of the
universe `as a whole' renders inappropriate the standard
Copenhagen interpretation with its instrumentalist use of
counterfactual statements about what \emph{would} happen \emph{if}
a certain measurement is performed. Indeed, the Copenhagen
interpretation is inapplicable for \emph{any}\footnote{The
existence of the long-range, and all penetrating, gravitational
force means that, at a fundamental level, there is only \emph{one}
truly  closed system, and that is the universe itself.} system
that is truly `closed' (or `self-contained') and for which,
therefore, there is no `external' domain in which an observer can
lurk. This problem has motivated much research over the years and
continues to be of wide interest.

The philosophical questions that arise are profound, and look back
to the birth of Western philosophy in ancient Greece, almost three
thousand years ago.  Of course, arguably, the longevity of these
issues suggests that these questions are ill-posed in the first
place, in which case the whole enterprise is a complete waste of
time! This is probably the view of  most, if not all, of our
colleagues at Imperial College; but we beg to differ\footnote{Of
course, it is also possible that our colleagues are right.}.

When considering a closed system, the inadequacy of the
conventional instrumentalist interpretation of quantum theory
encourages the search for an interpretation that is more `realist'
in some way. For over eighty years, this has been a recurring
challenge for those concerned with the conceptual foundations of
modern physics. In rising to this challenge we join our Greek
ancestors in confronting once more the fundamental
question:\footnote{``What is a thing?''\ is the title of one of
the more comprehensible of Heidegger's works \cite{HeidThing}. By
this, we mean comprehensible to the authors of the present
article. We cannot speak for our colleagues across the channel:
from some of them we may need to distance ourselves.}
\vs{-6}\displayE{3.5}{4}{2}{``What is a thing?''}\vs{7}

Of course, as written, the question is itself questionable. For
many philosophers, including Kant, would assert that the correct
question is not ``What is a thing?'' but rather ``What is a thing
as it appears to \emph{us}?'' However, notwithstanding Kant's
strictures, we seek the thing-in-itself, and, therefore, we
persevere with Heidegger's form of the question.

Nevertheless, having said that,  we can hardly ignore the last
three thousand years of philosophy. In particular, we must defend
ourselves against the charge of being `naive
realists'.\footnote{If we were professional philosophers this
would be a terrible insult. :-)} At this point it become clear
that theoretical physicists have a big advantage over professional
philosophers. For we are permitted/required to study such issues
in the context of specific mathematical frameworks for addressing
the physical world; and one of the great fascinations of this
process is the way in which various philosophical positions are
implicit in the ensuing structures. For example, the exact meaning
of `realist' is infinitely debatable but, when used by a classical
physicist, it invariably means the following:
\begin{enumerate}
\item The idea of `a property of the system' (for example, `the value of a
physical quantity at a certain time') is meaningful, and
mathematically representable  in the theory.

\item Propositions about the system (typically asserting that the system has this or that property) are handled using  Boolean
logic. This requirement is compelling in so far as we humans are
inclined to think in a Boolean way.

\item There is  a space of `microstates' such  that specifying a
microstate\footnote{In simple non-relativistic systems, the state
is specified  at any given moment of time. Relativistic systems
(particularly quantum gravity!) require a more sophisticated
understanding of `state', but the general idea is the same.} leads
to unequivocal  truth values for all propositions  about the
system: \ie\ a state\footnote{We are a little slack in our use of
language here and in what follows by frequently referring to a
microstate as just a `state'. The distinction only becomes
important if one wants to introduce things like mixed states (in
quantum theory), or macrostates (in classical physics) all of
which are often just known as `states'. Then one must talk about
microstates (pure states) to distinguish them from the other type
of state.} encodes ``the way things are''. This is a natural way
of ensuring that the  first two conditions above  are satisfied.
\end{enumerate}
The standard interpretation of classical physics satisfies these
requirements and provides the paradigmatic example of a realist
philosophy in science. Heidegger's answer to his own question
adopts a similar position \cite{HeidThing}:
\begin{quote}
\emph{``A thing is always something that has such and such
properties, always something that is constituted in such and such
a way. This something is the bearer of the properties; the
something, as it were, that underlies the qualities.''}
\end{quote}

In quantum theory, the situation is very different. There,  the
existence of  any such realist interpretation is foiled by the
famous Kochen-Specker theorem \cite{KS67}. This asserts that it is
impossible to assign values to all physical quantities at once if
this assignment is to satisfy the consistency condition that the
value of a function of a physical quantity is that function of the
value. For example, the value of `energy squared' is the square of
the value of energy.

Thus, from a conceptual perspective, the challenge is to find a
quantum formalism that is `realist enough' to provide an
acceptable alternative to the Copenhagen interpretation, with its
instrumentally-construed intrinsic probabilities, whilst taking on
board the implications of the Kochen-Specker theorem.

So, \emph{in toto} what we seek is a formalism that is (i) free of
\emph{prima facie} prejudices about the nature of the values of
physical quantities---in particular, there should be no
fundamental use of the real or complex numbers; and (ii)
`realist', in at least  the minimal  sense that propositions are
meaningful, and  are assigned `truth values', not just
instrumentalist probabilities of what would happen if appropriate
measurements are made.

However, finding such a formalism is not easy: it is notoriously
difficult to modify the mathematical framework of quantum theory
without destroying the entire edifice. In particular, the Hilbert
space structure is very rigid and cannot easily be changed; and
the formal path-integral techniques do not fare much better.

To seek inspiration let us return  briefly to the situation in
classical physics. There, the concept of realism (as asserted in
the three statements above)\ is encoded mathematically in the idea
of a space of states, $\S$,  where specifying a particular state
(or `micro-state'), $s\in\S$, determines entirely `the way things
are' for the system. In particular, this suggests that each
physical quantity $A$ should be associated with a real-valued
function $\breve{A}:\S\map\mathR$ such that when the state of the
system is $s$, the value of $A$ is $\breve{A}(s)$. Of course, this
is indeed precisely how the formalism of classical physics works.

In the spirit of general abstraction, one might one wonder if this
formalism can be generalised to a structure in which $A$ is
represented by an arrow $\breve{A}:\Si\map\R$ where $\Si$ and $\R$
are objects in some category, $\tau$,  other than the category of
sets, $\Set$? In such a theory, one would seek to represent
propositions about the `values' (whatever that might  mean) of
physical quantities with sub-objects of $\Si$, just as in
classical physics propositions are represented by subsets of the
state space $\S$ (see Section \ref{SubSec:GenesisToposIdeas} for
more detail of this). \displayE{10}{5}{5}{Our central conceptual
idea is that such a categorial structure constitutes a
generalisation of the concept of `realism' in which the `values'
of a physical quantity are coded  in the arrow
$\breve{A}:\Si\map\R$. }\vs{7}

Clearly the propositions  will play a key role in any such theory,
and, presumably, the minimum required  is that the associated
sub-objects of $\Si$ form some sort of `logic', just as the
subsets of $\S$ form a Boolean algebra.

This rules out most categories since, generically, the sub-objects
of an object do not have any logical structure. However,  if the
category $\tau$ is a  `topos' then the sub-objects of any object
\emph{do} have this property, and hence the current research
programme.

Our  suggestion, therefore, is to try to construct physical
theories that are formulated in a topos other than $\Set$ . This
topos will depend on both the theory-type and the system. More
precisely, if  a theory-type (such as classical physics, or
quantum physics) is applicable to a certain class of systems,
then, for each system in this class, there is a topos in which the
theory is to be formulated. For some theory-types the topos is
system-independent: for example, classical physics always uses the
topos of sets. For other theory-types, the topos varies from
system to system:  as we shall see, this is the case in quantum
theory.

In somewhat more detail, any particular example of  our suggested
scheme will have the following ingredients:
\begin{enumerate}

\item There are two special objects in the topos $\tau_\phi$:  the
`state-object'\footnote{The meaning of the subscript `$\phi$' is
explained in the main text. It refers to a particular
topos-representation of a formal language attached to the
system.}, $\Si_\phi$ and the `quantity-value object', $\R_\phi$.
Any physical quantity, $A$, is represented by an arrow
$A_\phi:\Si_\phi\map\R_\phi$ in the topos. Whatever meaning can be
ascribed to the concept of the `value' of a physical quantity is
encoded in (or derived from) this representation.

\item Propositions about a system are represented by sub-objects
of the state-object $\Si_\phi$. These sub-objects form a Heyting
algebra (as indeed do the  sub-objects of any object in a topos):
a distributive lattice that differs from a Boolean algebra only in
that the \emph{law of excluded middle} need not hold, \ie\ $
\alpha\lor\lnot \alpha\preceq1$. A Boolean algebra is a Heyting
algebra with strict equality: $\alpha\lor\lnot\alpha=1$.

\item Generally speaking (and unlike in set theory), an object in a topos may not be determined
by its `points'. In particular, this may be so for the
state-object, in which case the concept of a microstate is not so
useful.\footnote{In  quantum theory, the state-object has no
points/microstates  at all. As we shall see, this statement is
equivalent to the Kochen-Specker theorem.} Nevertheless, truth
values can be assigned to propositions with the aid of a `truth
object' (or `pseudo-state').  These truth values lie in another
Heyting algebra.

\end{enumerate}

Of course, it is not instantly obvious that quantum theory can be
written in this way. However, as we shall see, there \emph{is} a
topos reformulation of quantum theory, and this has two immediate
implications. The first is that we acquire a new type of `realist'
interpretation of standard quantum theory. The second is that this
new approach suggests ways of generalising  quantum theory that
make no fundamental reference to Hilbert spaces, path integrals,
etc. In particular, there is no \emph{prima facie} reason for
introducing standard continuum quantities. As emphasised above,
this is one of our main motivations for developing the topos
approach. We shall say more about this later.

From a conceptual perspective, a central feature of our scheme is
the `neo-realist' structure reflected mathematically  in the three
statements above. This neo-realism is the conceptual fruit of  the
 fact that, from a categorial perspective, a physical
theory  expressed in a topos `looks'  like \emph{classical}
physics expressed in the topos of sets.

The fact that (i) physical quantities are  represented by arrows
whose domain is the state-object, $\Si_\phi$; and (ii)
propositions are represented by sub-objects of $\Si_\phi$,
suggests strongly that $\Si_\phi$ can be regarded as the
topos-analogue of a classical state space. Indeed, for any
classical system the topos is just the category of sets, $\Set$,
and the ideas above  reduce to the familiar picture in which (i)
there is a state space (set) $\S$; (ii) any physical quantity,
$A$,  is represented by a real-valued functions
$\breve{A}:\S\map\mathR$; and (iii) propositions are represented
by subsets of $\S$  with a logical structure given by the
associated Boolean algebra.

Evidently the suggested mathematical structures could be used in
two different ways. The first is that of the `conventional'
theoretical physicist with little interest in conceptual matters.
For him/her, what we and our colleagues are developing is a new
tool-kit with which to construct novel types of theoretical model.
Whether or not Nature has chosen such models  remains to be seen,
but, at the very least, the use of topoi certainly suggests new
techniques.

For those physicists who are interested in conceptual issues, the
topos framework gives a radically new way of thinking about the
world. The neo-realism inherent in the formalism is described
mathematically using the \emph{internal} language that is
associated with any topos. This describes how things look from
`within' the topos: something that should be particularly useful
in the context of quantum cosmology\footnote{In this context see
the work of Markopoulou who considers a topos description of the
universe as seen by different observers who live inside it
\cite{Fotini00}.}.

On the other hand, the pragmatic theoretician with no interest in
conceptual matters can use the `external' description of the topos
in which the category of sets provides a metalanguage with which
to formulate the theory. From a mathematical perspective, the
interplay between the internal and external languages of a topos
is one of the fascinations of the subject. However, much remains
to be said about the significance of this interaction for real
theories of physics.

This present article is partly an amalgam of a series of four
papers that we placed on the ArXiv server\footnote{These are due
to published in \emph{Journal of Mathematical Physics} in the
Spring of 2008.} in March, 2007 \cite{DI(1),DI(2),DI(3),DI(4)}.
However, we have  added a fair amount of new material, and also
made a few minor corrections (mainly typos).\footnote{Some of the
more technical theorems have been placed in the Appendix with the
hope that this makes the article a little easier to read.} We have
also added some remarks about developments made by researchers
other than ourselves since the ArXiv preprints were written. Of
particular importance to our general programme is the work of
Heunen and Spitters \cite{HeuSpit07} which adds some powerful
ingredients to the topoi-in-physics toolkit. Finally, we have
included some background material from the earlier papers that
formed the starting point for the current research programme
\cite{IB98,IB99,IB00,IB02}.

We must emphasise that this is \emph{not}\ a review article about
the general application of topos theory to physics; this would
have made the article far too long. For example, there has been a
fair amount of study of the use of synthetic differential geometry
in physics. The reader can find references to much of this on the,
so-called, `Siberian toposes' web site\footnote{This is
http://users.univer.omsk.su/\~\,topoi/. See also Cecilia Flori's
website that deals more generally with topos theory and physics:
http://topos-physics.org/}.  There is also the work by Mallios and
collaborators on `Abstract Differential Geometry'
\cite{Mallios1,Mallios2,Raptis,MalliosZafiris}. Of course, as
always these days, Google will speedily reveal all that we have
omitted.

But even less is this paper a review of the use of category theory
in general in physics. For there any many important topics that we
do not mention at all. For example, Baez's advocation of
$n$-categories \cite{Baez1,Baez2}; `categorial quantum theory'
\cite{AC04,Vic06}; Takeuti's theory\footnote{Takeuti's work is not
exactly about category theory applied to quantum theory: it is
more about the use  of formal logic, but the spirit is similar.
For a recent paper in this genre see \cite{Ozawa06}.} of `quantum
sets' \cite{Takeuti81}; and Crane's work on categorial models of
space-time \cite{Crane}.

Finally, a word about the style in which this article is written.
We spent much time pondering on this, as we did before writing the
four ArXiv preprints. The intended audience is our colleagues who
work in theoretical physics, especially those whose interests
included foundational issues in quantum gravity and quantum
theory. However, topos theory is not an easy branch of
mathematics, and this poses  the dilemma of how much background
mathematics should be assumed of the reader, and how much  should
be explained as we go along.\footnote{The references that we have
found most helpful in our research are
\cite{McL71,Gol84,LamScott86,Bell88,MM92,Jst02}.} We have
approached this problem by including a short mathematical appendix
on topos theory. However, reasons of space precluded a thorough
treatment, and we hope that, fairly soon, someone  will write an
introductory review of topos theory in a style that is accessible
to a typical theoretical-physicist reader.

This article is structured in the following way. We begin with a
discussion of some of the conceptual background, in particular the
role of the real numbers in conventional theoretical physics. Then
in Section \ref{Sec:ToposLogic} we introduce the idea of attaching
a propositional language, $\PL{S}$, to each physical system $S$.
The intent is that each theory of $S$ corresponds to a particular
representation of $\PL{S}$. In particular, we show how classical
physics satisfies this requirement in a very natural way.

Propositional languages have limited scope (they lack the
quantifiers `$\forall$' and `$\exists$'), and  in Section
\ref{Sec:TypedLanguage} we propose the use of a higher-order
language $\L{S}$. Languages of this type are a central feature of
topos theory and it is natural to consider the  idea of
representing $\L{S}$ in different topoi. Classical physics always
takes place in the topos, $\Set$, of sets but our expectation is
that other areas of physics will use a different topos.

This expectation is confirmed in Section \ref{Sec:QuPropSpec}
where we discuss in detail the representation of $\PL{S}$ for a
quantum system (the representation of $\L{S}$ is discussed in
Section \ref{Sec:psSR}). The central idea is to represent
propositions as sub-objects of the `spectral presheaf' $\Sig$
which belongs to the topos, $\SetH{}$, of presheaves (set-valued,
contravariant functors) on the category, $\V{}$, of abelian
sub-algebras of the algebra $\BH$ of all bounded operators on
$\Hi$. This representation employs the idea of `daseinisation' in
which any given projection operator $\P$ is represented at each
context/stage-of-truth $V$ in $\V{}$ by the `closest' projector to
it in $V$. There are two variants of this: (i) `outer'
daseinisation, in which $\P$ is approached from above (in the
lattice of projectors in $V$); and (ii) `lower' daseinisation, in
which $\P$ is approached from below.

The next key move is to discuss the `truth values' of propositions
in a quantum theory. This requires the introduction of some
analogue of the microstates of classical physics. We say
`analogue' because the spectral presheaf $\Sig$---which is the
quantum topos equivalent of a classical state space---has no
global elements, and hence there are no microstates at all: this
is equivalent to the Kochen-Specker theorem. The critical idea is
that of a `truth object', or `pseudo-state' which, as we  show in
Section \ref{Sec:TruthValues}, is the closest one can get  in
quantum theory to a microstate.

In Section \ref{Sec:deG} we introduce the `de Groote' presheaves
and the associated ideas that lead to the concept of daseinising
an arbitrary bounded self-adjoint operator, not just a projector.
Then, in Section \ref{Sec:psSR},  the spectral theorem is used to
construct several possible models for the quantity-value presheaf
in quantum physics. The simplest choice is $\SR$, but this uses
only outer daseinisation, and a more balanced choice is
$\PR{\mathR}$ which uses both inner and outer daseinisation.
Another possibility is $k(\SR)$: the Grothendieck topos extension
of the monoid object $\SR$. A\ key result  is the `non-commutative
spectral theorem' which involves showing how each bounded,
self-adjoint operator $\A$ can be represented by an arrow
$\breve{A}:\Sig\map\PR{\mathR}$.

In Section \ref{Sec:Unitary} we discuss the way in which unitary
operators act on the quantum topos objects. Then, in Sections
\ref{Sec:CatSys}, \ref{Sec:ToposAxioms} and \ref{Sec:ReviewQT} we
discuss the problem of handling `all' possible systems in a single
coherent scheme. This involves introducing a category of systems
which, it transpires, has a natural monoidal structure. We show in
detail how this scheme works in the case of classical and quantum
theory.

Finally, in Section \ref{Sec:CharPropsObjects} we
discuss/speculate on some properties of the state object,
quantity-value object, and  truth objects that might be present in
any topos representation of a physical system.

To facilitate reading this long article, some of the more
technical material has been put in Appendix 1. In Appendix 2 there
is a short introduction to some of the relevant parts of topos
theory.

%
\section{The Conceptual Background of our Scheme}
\label{Sec:ConceptualBackground}
\subsection{The Problem of Using Real Numbers a Priori}
As mentioned in the Introduction, one of the main goals of our
work is to find new tools with which to develop theories that are
significant extensions of, or developments from, quantum theory
but without being tied \emph{a priori} to the use of the standard
real or complex numbers.

In this context we note  that  real numbers  arise in theories of
physics in three different (but related) ways: (i)  as the values
of physical quantities; (ii) as the values of probabilities; and
(iii) as a fundamental ingredient in models of space and time
(especially in those based on differential geometry). All three
are of direct concern vis-a-vis our worries about making
unjustified, \emph{a priori} assumptions in quantum theory. We
shall now examine them in detail.

\subsubsection{Why Are Physical Quantities Assumed to be Real-Valued?}
One reason for assuming physical quantities are real-valued is
undoubtedly  grounded in the remark that, traditionally (\ie\ in
the pre-digital age), they are measured with rulers and pointers,
or they are defined operationally in terms of such measurements.
However, rulers and pointers are taken to be classical objects
that exist in the physical space of classical physics, and this
space is modelled using the reals. In this sense there is a direct
link between the space in which physical quantities take their
values (what we call the `quantity-value space') and the nature of
physical space or space-time \cite{Isham03}.

If  conceded, this claim means  the assumption that physical
quantities are real-valued is problematic in any theory in which
space, or space-time, is not modelled by a smooth manifold.
Admittedly, if the theory employs a \emph{background} space, or
space-time---and if this background is a manifold---then the use
of real-valued physical quantities \emph{is} justified in so far
as their value-space can be related to this background. Such a
stance is particularly appropriate in situations where the
background plays a central role in giving meaning to concepts like
`observers' and  `measuring devices', and thereby provides a basis
for an instrumentalist interpretation of the theory.

But even here caution is needed since many theoretical physicists
have claimed that the notion of a `space-time point in a manifold'
is intrinsically flawed. One argument (due to Penrose) is based on
the observation that any attempt to localise a `thing' is bound to
fail beyond a certain point because of the quantum production of
 pairs of particles from the energy/momentum uncertainty
caused by the spatial localisation. Another argument concerns the
artificiality\footnote{The integers, and associated rationals,
have a `natural' interpretation from a physical perspective since
we can all count. On the other hand,  the Cauchy-sequence and/or
the Dedekind-cut definitions of the reals are distinctly
un-intuitive from a physical perspective.} of the use of real
numbers as coordinates with which to identify a space-time point.
There is also Einstein's famous `hole argument' in general
relativity which asserts that the notion of a space-time point (in
a manifold) has no physical meaning in a theory that is invariant
under the group of space-time diffeomorphisms.

Another cautionary caveat concerning  the invocation of a
background  is that this background structure may arise only in
some `sector' of the theory; or  it may exist only in some
limiting, or approximate, sense. The associated instrumentalist
interpretation would then be similarly limited in scope.  For this
reason, if no other, a `realist' interpretation is more attractive
than an instrumentalist one.

In fact, in  such circumstances, the phrase `realist
interpretation' does not really do justice to the situation since
it tends to imply that there are other interpretations of the
theory, particularly instrumentalism, with which the realist one
can contend on a more-or-less equal footing. But, as we just
argued, the instrumentalist interpretation may be severely limited
 as compared   to the realist one. To flag this point,  we
will sometimes refer to a `realist formalism', rather than a
`realist interpretation'.\footnote{Of course, such discussions are
unnecessary in classical physics since, there, if knowledge of the
value of a physical quantity is gained by making a (ideal)
measurement, the reason why we obtain the result that we do, is
because the quantity \emph{possessed} that value immediately
before the measurement was made. In other words, ``epistemology
models ontology''.}

\subsubsection{Why Are Probabilities Required to Lie in the
Interval $[0,1]$?}  The motivation for using  the subset $[0,1]$
of the real numbers as the value space for probabilities comes
from the relative-frequency interpretation of probability. Thus,
in principle, an experiment is to be repeated a large number, $N$,
times, and the probability associated with a particular result is
defined to be the ratio $N_i/N$, where $N_i$ is the number of
experiments in which that result was obtained. The rational
numbers $N_i/N$ necessarily lie between $0$ and $1$, and if the
limit $N\map\infty$ is taken---as is appropriate for a
hypothetical `infinite ensemble'---real numbers in the closed
interval $[0,1]$ are obtained.

The relative-frequency interpretation of probability is natural in
instrumentalist theories of physics,  but it is not meaningful if
there is no classical spatio-temporal background in which the
necessary  measurements could be made; or, if there is  a
background, it is one to which the relative-frequency
interpretation cannot be adapted.

In the absence of   a relativity-frequency interpretation, the
concept of `probability' must be understood in a different way. In
the physical sciences, one of the most discussed approaches
involves  the concept of `potentiality', or `latency', as favoured
by Heisenberg \cite{Heisenberg52}, Margenau \cite{Margenau49}, and
Popper \cite{Popper82} (and, for good measure, Aristotle). In this
case there is no  compelling reason why the probability-value
space should necessarily be a subset of the real numbers. The
minimal requirement is that this value-space is an ordered set, so
that  one proposition can be said to be more or less probable than
another. However, there is no \emph{prima facie} reason why this
set should be \emph{totally} ordered: \ie\ there may be  pairs of
propositions whose potentialities  cannot be compared---something
that seems eminently plausible in the context of non-commensurable
quantities in quantum theory.

By invoking the idea of `potentiality', it becomes feasible to
imagine a quantum-gravity theory with no spatio-temporal
background but where probability is still a fundamental concept.
However,  it could also be that the  concept of probability plays
no fundamental role in such circumstances, and can be given a
meaning \emph{only} in the context of a sector, or limit, of the
theory where a background does exist. This background could then
support a limited instrumentalist interpretation which  would
include  a (limited) relative-frequency understanding of
probability.

In fact, most modern approaches  to quantum gravity aspire to  a
formalism that is background independent
\cite{Bae00,CM05,Smo05,Smo06}. So, if a background space  does
arise, it will  be in one of the restricted senses mentioned
above.  Indeed, it is often asserted that a proper theory of
quantum gravity will not involve \emph{any} direct spatio-temporal
concepts, and that  what we commonly call `space' and `time' will
`emerge' from the formalism only in some appropriate limit
\cite{BI01}. In  this case, any instrumentalist interpretation
could only `emerge' in the same limit, as would the associated
relative-frequency interpretation of probability.

In a theory of this type, there will be no \emph{prima facie} link
between the values of physical quantities and the nature of space
or space-time, although, of course, this cannot be totally ruled
out. In any event, part of the fundamental specification of the
theory will involve deciding what the `quantity-value space'
should be.

These considerations suggest that quantum theory must be radically
changed if one wishes to  accommodate situations where there is no
background  space/space-time, manifold within which an
instrumentalist interpretation can be formulated. In such a
situation, some sort of `realist' formalism is essential.

These reflections also suggest  that the quantity-value space
employed in  an instrumentalist realisation of a theory---or a
`sector', or `limit', of the theory---need not be the same as the
quantity-value space in a neo-realist formulation. At first sight
this may seem strange but, as is shown in Section \ref{Sec:psSR},
this is precisely what happens in the topos reformulation of
standard quantum theory.

\subsection{The Genesis of Topos Ideas in Physics}
\label{SubSec:GenesisToposIdeas}
\subsubsection{Why are Space and Time Modelled with Real Numbers?}
\label{SubSubSec:WhyReals} Even setting aside the more exotic
considerations of quantum gravity, one can still query the use of
real numbers to model space and/or time. One might argue that (i)
the use of (triples of) real numbers to model space is based on
empirically-based reflections about the nature of `distances'
between objects; and (ii) the use of real numbers to model time
reflects our experience that `instants of time' appear to be
totally ordered, and that intervals of time are always
divisible\footnote{These remarks are expressed in the context of
the Newtonian view of space and time, but it is easy enough to
generalise them to special relativity.}.

However, what does it really mean to say that two particles are
separated by a distance of, for example, $\sqrt{2}$cms? From an
empirical perspective, it would be impossible to make a
measurement that could unequivocally reveal precisely that value
from among the continuum of real numbers that lie around it. There
will always be experimental errors of some sort: if nothing else,
there are thermodynamical fluctuations in  the measuring device;
and, ultimately, uncertainties arising from quantum
`fluctuations'. Similar remarks apply to attempts to measure time.

Thus, from an operational perspective, the use of real numbers to
label `points' in space and/or time is a theoretical abstraction
that can never be realised in practice.  But if the notion of a
space/time/space-time `point' in a continuum, is an abstraction,
why do we use it? Of course it works well in theories used in
normal physics, but at a fundamental level it must be seen as
questionable.

These operational remarks say nothing about the structure of space
(or time) `in itself', but, even assuming that this concept makes
sense, which is debatable, the use of real numbers is still a
metaphysical assumption with no fundamental justification.

Traditionally, we teach our students that  measurements of
physical quantities that are represented theoretically by real
numbers, give results that fall into `bins',  construed as being
subsets of the real line. This suggests that, from an operational
perspective, it would be more appropriate  to base mathematical
models of space or time on a theory of `regions', rather than the
real numbers themselves.

But then one asks ``What is a region?'', and if we answer ``A
subset of triples of real numbers for space, and a subset of real
numbers for time'', we are thrown back  to the real numbers. One
way of avoiding this circularity is to focus on relations between
these `subsets' and see if they can be axiomatised in some way.
The natural operations to perform on regions are (i) take
intersections, or unions, of pairs of regions; and (ii) take the
complement of a region. If the regions are modelled on Borel
subsets of $\mathR$, then the intersections and unions could be
extended to countable collections. If they are modelled on open
sets, it would be arbitrary unions and finite intersections.

From a physical perspective, the use of open subsets as models of
regions is attractive as it leaves a certain, arguably desirable,
`fuzziness' at the edges, which is absent for closed sets. Thus,
following this path, we would axiomatise that a mathematical model
of space or time (or space-time) involves an algebra of entities
called `regions', and with operations that are the analogue of
unions and intersections for subsets of a set. This algebra would
allow arbitrary `unions' and finite `intersections', and would
distribute\footnote{If the distributive law is dropped we could
move towards the quantum-set ideas of \cite{Takeuti81}; or,
perhaps, the ideas of non-commutative geometry instigated by Alain
Connes \cite{Connes}.} over these operations. In effect, we are
axiomatising that an appropriate mathematical model of space-time
is an object in the category of locales.

However, a locale is the same thing as a complete Heyting algebra
(for the definition see below), and, as we shall, Heyting algebras
are inexorably linked with topos theory.

\subsubsection{Another Possible Role for Heyting Algebras}
The use of a Heyting algebra to model space/time/space-time is an
attractive possibility, and was the origin of the interest in
topos theory of one of us (CJI) some years ago. However, there is
another motivation which is based more on logic, and the desire to
construct a `neo-realist' interpretation of quantum theory.

To motivate  topos theory as the source of neo-realism let us
first consider classical physics, where everything is defined in
the category, $\Set$, of sets and functions between sets. Then (i)
any physical quantity, $A$, is represented by a real-valued
function $\breve{A}:\S\map\mathR$, where $\S$ is the space of
microstates; and (ii) a proposition of the form $\SAin\De$ (which
asserts that the value of the physical quantity $A$ lies in the
subset $\De$ of the real line $\mathR$)\footnote{In the rigorous
theory of classical physics, the set $\S$ is a symplectic
manifold, and $\De$ is a \emph{Borel} subset of $\mathR$. Also,
the function
 $\breve{A}:{\cal S}\map\mathR$  may be required to be
measurable, or continuous, or smooth, depending on the quantity,
$A$, under consideration.} is represented by the
subset\footnote{Throughout this article we will adopt the notation
in which $A\subseteq B$ means that $A$ is a subset of $B$ that
could equal $B$; while $A\subset B$ means that $A$ is a
\emph{proper} subset of $B$; \ie\ $A$ does not equal $B$. Similar
remarks apply to other pairs of ordering symbols like
$\prec,\preceq$; or $\succ,\succeq$, etc.}
$\breve{A}^{-1}(\De)\subseteq \S$. In fact any proposition $P$
about the system is represented by an associated subset, $\S_P$,
of $\S$: namely, the set of states for which $P$ is true.
Conversely, every (Borel) subset of $\S$ represents a
proposition.\footnote{More precisely, every Borel subset of $\S$
represents \emph{many} propositions about the values of physical
quantities. Two propositions are said to be `physically
equivalent' if they are represented by the same subset of $\S$.}

It is easy to see how the logical calculus of propositions arises
in this picture. For let $P$ and $Q$ be propositions, represented
by the subsets $\S_P$ and $\S_Q$ respectively, and consider the
proposition ``$P$ and $Q$''. This is true if, and only if, both
$P$ and $Q$ are true, and hence the subset of states that
represents this logical conjunction consists of those states that
lie in both $\S_P$ and $\S_Q$---\ie\ the set-theoretic
intersection $\S_P\cap\S_Q$. Thus ``$P$ and $Q$'' is represented
by $\S_P\cap\S_Q$. Similarly, the proposition ``$P$
 or $Q$'' is true if either $P$ or $Q$ (or both) are true,
and hence this logical disjunction is represented by those states
that lie in $\S_P$ plus those states that lie in $\S_Q$---\ie\ the
set-theoretic union $\S_P\cup\S_Q$. Finally, the logical negation
``not $P$'' is represented by all those points in $\S$ that do not
lie in $\S_P$---\ie\ the set-theoretic complement $\S/\S_P$.

In this way, a fundamental relation is established between the
logical calculus of propositions about a physical system, and the
Boolean algebra of subsets of the state space. Thus the
mathematical structure of classical physics is such that, \emph{of
necessity}, it reflects a `realist' philosophy, in the sense in
which we are using the word.

One way to escape from the tyranny of Boolean algebras and
classical realism is via topos theory. Broadly speaking, a topos
is a category that behaves very much like the category of sets; in
particular, the collection of sub-objects of an object forms a
\emph{Heyting algebra}, just as the collection of subsets of a set
form a Boolean algebra. Our intention, therefore, is to explore
the possibility of associating physical propositions with
sub-objects of some object $\Si$ (the analogue of a classical
state space) in some topos.

A Heyting algebra, $\mathfrak{h}$, is  a distributive lattice with
a zero element, $0$, and a unit element, $1$, and with the
property that to each pair $\a,\beta\in\mathfrak{h}$ there is an
implication $\a\Rightarrow\beta$, characterized by
\begin{equation}
\ga\preceq(\a\Rightarrow\beta)\mbox{ if and only if } \ga\land
\a\preceq \beta.
\end{equation}
The negation is defined as $\lnot\a:=(\a\Rightarrow0)$ and has the
property that the \emph{law of excluded middle} need not hold,
\ie\ there may exist $\a\in\mathfrak{h}$, such that $ \a\lor\lnot
\a\prec1$ or, equivalently, there may exist$\a\in\mathfrak{h}$
such that $\lnot\lnot \a\succ\a$. This is the characteristic
property of an intuitionistic logic.\footnote{Here,
$\a\Rightarrow\beta$ is nothing but the category-theoretical
exponential $\beta^\alpha$ and $\ga\land\alpha$ is the product
$\ga\times\alpha$. The definition uses the adjunction between the
exponential and the product,
$\operatorname{Hom}(\ga,\beta^\alpha)=\operatorname{Hom}(\ga\times%
\alpha,\beta)$. A slightly easier, albeit `less categorical'
definition is: a Heyting algebra, $\mathfrak{h}$, is a
distributive lattice such that for any two elements
$\alpha,\beta\in\mathfrak{h}$, the set
$\{\ga\in\mathfrak{h}\mid\ga\land\alpha\leq\beta\}$ has a maximal
element, denoted by $(\alpha\Rightarrow\beta)$.} A Boolean algebra
is the special case of a Heyting algebra in which there is the
strict equality:\ \ie\ $\alpha\lor\lnot\alpha=1$ for all $\a$. It
is known from Stone's theorem \cite{Sto36} that each Boolean
algebra is isomorphic to an algebra of (clopen, i.e., closed and
open) subsets of a suitable (topological) space.

The elements of a Heyting algebra can be manipulated in a very
similar way to those in a Boolean algebra. One of our claims is
that,  as far as theories of physics are concerned, Heyting logic
is a  viable\footnote{The main difference between theorems proved
using Heyting logic and those using Boolean logic is that proofs
by contradiction cannot be used in the former. In particular, this
means that one cannot prove that something exists by arguing that
the assumption that it does not  leads to contradiction; instead
it is necessary to provide a \emph{constructive} proof of the
existence of the entity concerned. Arguably, this does not place
any major restriction on building theories of physics. Indeed,
over the years, various physicists (for example, Bryce DeWitt)
have argued that constructive proofs should always be used in
physics.} alternative to Boolean logic.

To give some idea of the difference between a Boolean algebra and
a Heyting algebra, we note that the paradigmatic example of the
former is the collection of all measurable subsets of a measure
space $X$. Here, if $\alpha\subseteq X$ represents a proposition,
the logical negation, $\neg\alpha $, is just the set-theoretic
complement $X\backslash\alpha$.

On the other hand, the paradigmatic example of a Heyting  algebra
is the collection of all open sets in a topological space $X$.
Here, if $\alpha\subseteq X$ is open, the logical negation
$\neg\alpha$ is defined to be the \emph{interior} of the
set-theoretical complement $X\backslash\alpha$. Therefore, the
difference between $\neg\alpha$ in the topological space $X$, and
$\lnot\alpha$ in the measurable space generated by the topology of
$X$, is just the `thin' boundary of the closed set
$X\backslash\alpha$.

\subsubsection{Our Main Contention about Topos Theory and Physics}
We contend that, for a given theory-type (for example, classical
physics, or quantum physics), each system $S$ to which the theory
is applicable is associated with a particular topos $\tau_\phi(S)$
within whose framework the theory, as applied to $S$, is to be
formulated and interpreted. In this context, the
`$\phi$'-subscript is a label that changes as the theory-type
changes. It signifies the representation of a system-language in
the topos $\tau_\phi(S)$: we will come to this later.

The conceptual interpretation of this formalism is `neo-realist'
in the following sense:
\begin{enumerate}
\item A physical quantity, $A$, is to be represented by an arrow
$A_{\phi,S}:\Si_{\phi,S}\map\R_{\phi,S}$ where $\Si_{\phi,S}$ and
$\R_{\phi,S}$ are two special objects in the topos $\tau_\phi(S)$.
These are the analogues of, respectively, (i) the classical state
space, $\cal S$; and (ii) the real numbers, $\mathR$, in which
 classical physical quantities take their values.

In what follows,  $\Si_{\phi,S}$ and $\R_{\phi,S}$ are called the
`state object', and the `quantity-value object', respectively.

\item Propositions about the system $S$ are represented by
sub-objects of $\Si_{\phi,S}$. These sub-objects form a Heyting
algebra.

\item Once the topos analogue of a state (a `truth object') has
been specified,  these propositions are assigned truth values in
the Heyting logic associated with the global elements of the
sub-object classifier, $\O_{\tau_\phi(S)}$, in the topos
$\tau_\phi(S)$.
\end{enumerate}

Thus a theory expressed in this way \emph{looks} very much like
classical physics except that whereas classical physics always
employs the topos of sets, other theories---including quantum
theory and, we conjecture, quantum gravity---use a different
topos.

One deep result in topos theory is that there is an \emph{internal
language} associated with each topos. In fact, not only does each
topos generate an internal language, but, conversely, a language
satisfying appropriate conditions generates a topos. Topoi
constructed in this way are called `linguistic topoi', and every
topos can be regarded as a linguistic topos. In many respects,
this is one of the profoundest  ways of understanding what a topos
really `is'.\footnote{This aspect of topos theory is discussed at
length in the books by Bell \cite{Bell88}, and Lambek and Scott
\cite{LamScott86}}.

These results are exploited  in Section \ref{Sec:TypedLanguage}
where we introduce the idea that, for any applicable\ theory-type,
each physical system $S$  is associated with a `local' language,
$\L{S}$. The application of the theory-type to $S$ is then
involves finding a representation of $\L{S}$ in an appropriate
topos; this is equivalent to finding a `translation' of $\L{S}$
into the internal language of that topos.

Closely related to the existence of this linguistic structure is
the striking fact that a topos can be used as a \emph{foundation}
for  mathematics itself, just as set theory is used in the
foundations of `normal' (or `classical') mathematics. In this
context, the key remark  is that the internal language of a topos
has a form that is similar in many ways to the formal language on
which normal set theory is based. It is this internal, topos
language that is used to \emph{interpret} the theory in a
`neo-realist' way.

The main difference with classical logic is that the logic of the
topos language does not satisfy the principle of excluded middle,
and  hence proofs by contradiction are not permitted. This has
many intriguing consequences. For example, there are topoi in
which there exist genuine \emph{infinitesimals} that can be used
to construct a rival to normal calculus. The possibility of such
quantities stems from the fact that the normal proof that they do
\emph{not} exist is a proof by contradiction.

Thus each topos carries its own world of mathematics: a world
which, generally speaking, is \emph{not} the same as that of
classical mathematics.

Consequently, by postulating that, for a given theory-type, each
physical system carries its own topos, we are also saying that to
each physical system plus theory-type there is associated a
framework for mathematics itself! Thus classical physics uses
classical mathematics; and quantum theory uses `quantum
mathematics'---the mathematics formulated in the topoi of quantum
theory. To this we might add the conjecture: ``Quantum gravity
uses `quantum gravity' mathematics''!


\section{Propositional Languages and Theories of Physics}
\label{Sec:ToposLogic}
\subsection{Two Opposing Interpretations of Propositions}
\label{SubSec:BackgroundRemarks} Attempts to construct a na\"ive
realist interpretation of quantum theory founder on the
Kochen-Specker theorem. However, if, despite this theorem, some
degree of realism is still sought, there are not that many
options.

One approach is to focus on a particular, maximal commuting subset
of physical quantities and declare by fiat that these are the ones
that `have' values; essentially, this is what is done in  `modal'
interpretations of quantum theory. However, this leaves open the
question of why Nature should select this particular set, and the
reasons proposed vary greatly from one scheme to another.

In our work, we take a completely different approach and try to
formulate a scheme which takes into account \emph{all} these
different choices for commuting sets of physical quantities; in
particular,  equal ontological status is ascribed to all of them.
This scheme is grounded in the topos-theoretic approach that was
first  proposed in \cite{IB98, IB99, IB00, IB02}. This uses a
technique whose first step is to construct a category, $\cal C$,
the objects of which can be viewed as \emph{contexts} in which the
quantum  theory can be displayed: in fact, they are just the
commuting sub-algebras of operators in the theory. All this will
be explained in more detail in Section \ref{Sec:QuPropSpec}.

In this earlier work, it was postulated that the logic for
handling quantum propositions from this perspective is that
associated with the topos of presheaves\footnote{In quantum
theory,  the category $\cal C$ is just a partially-ordered set,
which simplifies many manipulations.} (contravariant functors from
$\cal C$ to $\Set$), $\SetC{\cal C}$. The idea is that a single
presheaf will encode  quantum propositions from the perspective of
\emph{all} contexts at once. However, in the original papers, the
crucial `daseinisation' operation (see Section
\ref{Sec:QuPropSpec}) was not known and, consequently, the
discussion became rather convoluted in places. In addition, the
generality and power of the underlying procedure was not fully
appreciated by the authors.

For this reason, in the present article we   return to the basic
questions and reconsider them in the light of the overall topos
structure that has now become clear.

We start by considering the way in which  propositions arise, and
are manipulated, in physics. For simplicity, we will concentrate
on systems that are associated with `standard' physics. Then, to
each such system $S$ there is associated a set of physical
quantities---such as energy, momentum, position, angular momentum
etc.\footnote{This set does not have to contain `\emph{all}'
possible physical quantities: it suffices to concentrate on a
subset that are deemed to be of particular interest. However, at
some point, questions may arise about the `completeness' of  the
set.}---all of which are real-valued.   The associated
propositions are of the form $\SAin\De$, where $A$ is a physical
quantity, and $\De$ is a subset\footnote{As was remarked earlier,
for various reasons, the subset $\De\subseteq\mathR$ is usually
required to be a \emph{Borel} subset, and for the most part we
will assume this without further comment.} of $\mathR$.

From a conceptual perspective, the proposition $\SAin\De$ can be
read in two, very different, ways:
\begin{enumerate}
\item[(i)] {\bf The (na\"ive) realist interpretation:} ``The
physical quantity $A$ \emph{has} a value, and that value lies in
$\De$.''

\item[(ii)] {\bf The instrumentalist interpretation:}  ``\emph{If}
a measurement is made of $A$, the result will be found to lie in
$\De$.''
\end{enumerate}
The former is the familiar, `commonsense' understanding of
propositions in both classical physics and daily life. The latter
underpins the Copenhagen interpretation of quantum theory. Of
course, the instrumentalist interpretation can also be applied to
classical physics, but it does not lead to anything new. For, in
classical physics, what is measured is what \emph{is} the case:
``Epistemology models ontology''.

We will now study the role of propositions in physics more
carefully, particularly in the context of  `realist'
interpretations.

\subsection{The Propositional Language $\PL{S}$}
\label{SubSec:PropLangPhys}
\subsubsection{Intuitionistic Logic and the Definition of $\PL{S}$}
We are going to construct a formal language, $\PL{S}$, with which
to express propositions about a physical system, $S$, and to make
deductions concerning them. Our intention is to interpret these
propositions  in a `realist' way: an endeavour whose  mathematical
underpinning  lies in constructing a representation of  $\PL{S}$
in a Heyting algebra, $\mathfrak{H}$, that is part of the
mathematical framework involved in the application of a particular
theory-type to $S$.

In constructing $\PL{S}$ we suppose that we have first identified
some set, $\Q{S}$, of \emph{physical quantities}: this plays a
fundamental role in our language. In addition, for \emph{any}
system $S$, we have the set, $P_B\mathR$ of (Borel) subsets of
$\mathR$. We use the sets $\Q{S}$ and $P_B\mathR$ to construct the
`primitive propositions' about the system $S$. These are of the
form $\SAin\De$ where $A\in\Q{S}$ and $\De\in P_B\mathR$.

We denote the set of all such strings by $\PL{S}_0$. Note that
what has been here called a `physical quantity' could better (but
more clumsily) be termed the `name' of the physical quantity. For
example, when we talk about the `energy' of a system, the word
`energy' is the same, and functions in the same way in the formal
language, irrespective of the details of the actual Hamiltonian of
the system.

The primitive propositions $\SAin\De$ are used  to define
`sentences'. More precisely, a new set of symbols
$\{\neg,\land,\lor,\Rightarrow\}$ is added to the language, and
then a \emph{sentence} is defined inductively by the following
rules (see Ch.~6 in \cite{Gol84}):
\begin{enumerate}
\item Each primitive proposition $\SAin\De$ in $\PL{S}_0$ is a
sentence.

\item If $\alpha$ is a sentence, then so is $\neg\alpha$.

\item If $\alpha$ and $\beta$ are sentences, then so are
$\alpha\land\beta$, $\alpha\lor\beta$, and
$\alpha\Rightarrow\beta$.
\end{enumerate}
The collection of all sentences, $\PL{S}$, is an elementary formal
language that can be used to express and manipulate propositions
about the system $S$.  Note that, at this stage, the symbols
$\neg$, $\land$, $\lor$, and $\Rightarrow$ have no explicit
meaning, although of course the implicit intention is that they
should stand for `not', `and', `or' and `implies', respectively.
This implicit meaning becomes explicit when a representation of
$\PL{S}$ is constructed as part of the application of a
theory-type to $S$ (see below). Note also that $\PL{S}$ is a
\emph{propositional} language only: it does not contain the
quantifiers `$\forall$' or `$\exists$'. To include them requires a
higher-order language. We shall return to this  in our discussion
of the language $\L{S}$.

The next step arises because $\PL{S}$ is not only a vehicle for
expressing propositions about the system $S$: we also want to
\emph{reason} with it about the system. To achieve this, a series
of axioms for a deductive logic  must be added to $\PL{S}$. This
could be either classical logic or intuitionistic logic, but we
select the latter since it allows a larger class of
representations/models, including representations in topoi in
which the law of excluded middle fails.

The axioms for intuitionistic logic consist of a finite collection
of sentences in $\PL{S}$ (for example,
$\alpha\land\beta\Rightarrow\beta\land\alpha$), plus  a single
rule of inference, \emph{modus ponens} (the `rule of detachment')
which says that from $\alpha$ and $\alpha\Rightarrow\beta$ the
sentence $\beta$ may be derived.

Others axioms  might be added to $\PL{S}$ to reflect the implicit
meaning of  the primitive proposition $\SAin\De$: \ie\ (in a
realist reading) ``$A$ has a value, and that value lies in
$\De\subseteq\mathR$''. For example, the sentence ``$\Ain{\De_1}
\land \Ain{\De_2}$'' (`$A$ belongs to $\De_1$' \emph{and} `$A$
belongs to $\De_2$')  might seem to be equivalent to ``$A$ belongs
to $\De_1\cap\De_2$'' \ie\ ``$\Ain{\De_1\cap\De_2}$''. A similar
remark applies to ``$\Ain{\De_1}\lor \Ain{\De_2}$''.

Thus, along with the axioms of intuitionistic logic and
detachment, we might be tempted to add the following axioms:
\begin{eqnarray}
\Ain{\De_1}\land \Ain{\De_2}&\Leftrightarrow&
                     \Ain{\De_1\cap\De_2} \label{AinD1andD2}\\
\Ain{\De_1}\lor \Ain{\De_2}&\Leftrightarrow&
                     \Ain{\De_1\cup\De_2} \label{AinD1orD2}
\end{eqnarray}
These axioms are consistent with the intuitionistic logical
structure of $\PL{S}$.

We shall see later the extent to which the axioms
\eqs{AinD1andD2}{AinD1orD2} are compatible with the topos
representations of classical and quantum physics. However, the
other obvious proposition to consider in this way---``It is
\emph{not} the case that $A$ belongs to $\De$''---is clearly
problematical.

In classical logic, this proposition\footnote{The parentheses
$(\;)$ are not  symbols in the language; they are just a way of
grouping letters and sentences.}, ``$\neg(\Ain\De)$'', is
equivalent to ``$A$ belongs to $\mathR\backslash\De$'', where
$\mathR\backslash\De$ denotes the set-theoretic complement of
$\De$ in $\mathR$. This  might suggest augmenting
\eqs{AinD1andD2}{AinD1orD2} with a third axiom
\begin{equation}
    \neg(\Ain\De) \Leftrightarrow \Ain{\mathR\backslash\De}
                                                \label{negAinD}
\end{equation}
However, applying `$\neg$' to both sides of \eq{negAinD} gives
\begin{equation}
        \neg\neg(\Ain\De) \Leftrightarrow \Ain\De
\end{equation}
because of the set-theoretic result
$\mathR\backslash(\mathR\backslash\De)=\De$. But in an
intuitionistic logic we do not have
$\alpha\Leftrightarrow\neg\neg\alpha$ but only
$\alpha\Rightarrow\neg\neg\alpha$, and so \eq{negAinD} could be
false in a Heyting-algebra representation of $\PL{S}$ that is not
Boolean. Therefore, adding \eq{negAinD} as an axiom in $\PL{S}$ is
not indicated if representations are  to be sought in non-Boolean
topoi.

\subsubsection{Representations of $\PL{S}$.}
\label{SubSubSec:RepPLS} To use a language $\PL{S}$  `for real'
for some specific physical system $S$ one must first decide on the
set $\Q{S}$ of physical quantities that are to be used in
describing $S$. This language must then be represented in the
concrete mathematical structure that arises when a theory-type
(for example: classical physics, quantum physics,
DI-physics,\ldots ) is applied to $S$. Such a representation,
$\pi$, maps each primitive proposition, $\a$, in $\PL{S}_0$ to an
element, $\pi(\a),$ of some Heyting algebra (which could be
Boolean), $\mathfrak{H}$, whose specification is part of the
theory of $S$. For example, in classical mechanics, the
propositions are represented in the Boolean algebra of all (Borel)
subsets of the classical state space.

The representation of the primitive propositions can be extended
recursively to all of $\PL{S}$ with the aid of the following rules
\cite{Gol84}:
\begin{eqnarray}
&(a)& \pi(\a\lor\beta):=\pi(\alpha)\lor\pi(\beta)
                                                \label{pi(a)}\\
&(b)& \pi(\a\land\beta):=\pi(\alpha)\land\pi(\beta)
                                                \label{pi(b)}\\
&(c)& \pi(\neg\a):=\neg\pi(\alpha)\hspace{3cm} \label{pi(c)} \\
&(d)&
\pi(\alpha\Rightarrow\beta):=\pi(\alpha)\Rightarrow\pi(\beta)
                                                \label{pi(d)}
\end{eqnarray}
Note that,  on the left hand side of \eqs{pi(a)}{pi(d)}, the
symbols $\{\neg,\land,\lor,\Rightarrow\}$ are elements of the
language $\PL{S}$, whereas on the right hand side they denote the
logical connectives in the Heyting algebra, $\mathfrak{H}$, in
which the representation takes place.

This extension of $\pi$ from $\PL{S}_0$ to $\PL{S}$ is consistent
with the axioms for the intuitionistic, propositional logic  of
the language $\PL{S}$. More precisely, these axioms become
tautologies: \ie\ they are all represented by the maximum element,
$1$, in the Heyting algebra. By construction, the map
$\pi:\PL{S}\map\mathfrak{H}$ is then a representation of $\PL{S}$
in the Heyting algebra $\mathfrak{H}$. A logician would say that
$\pi:\PL{S}\map\mathfrak{H}$ is  an
\emph{$\mathfrak{H}$-valuation}, or \emph{$\mathfrak{H}$-model},
of the language $\PL{S}$.

Note that different systems, $S$, can have the same language. For
example, consider a point-particle moving in one dimension, with a
Hamiltonian function $H(x,p)=\frac{p^2} {2m}+V(x)$ and state space
$T^*\mathR$. Different potentials $V$ correspond to different
systems (in the sense in which we are using the word `system'),
but the physical quantities for these systems---or, more
precisely, the `names' of these quantities, for example, `energy',
`position', `momentum'---are the same for them all. Consequently,
the language $\PL{S}$ is independent of $V$. However, the
\emph{representation} of, say, the proposition
``$E\varepsilon\De$'' (where `E' is the energy), with a specific
subset of the state space \emph{will} depend on the details of the
Hamiltonian.

Clearly, a major consideration in using the language $\PL{S}$ is
choosing the Heyting algebra in which the representation is to
take place. A fundamental result in topos theory is that the set
of all sub-objects of any object in a topos is  a Heyting algebra,
and these are the Heyting algebras with which we will be
concerned.

Of course, beyond  the language, $\S$, and its representation
$\pi$, lies the question of whether or not a proposition is
`true'. This requires the concept of a `state' which, when
specified, yields `truth values' for the primitive propositions in
$\PL{S}$. These can then be extended recursively to the rest of
$\PL{S}$. In classical physics, the possible truth values are just
`true' or `false'. However,  as we shall see, the situation in
topos theory is more complex.

\subsubsection{Using  Geometric Logic}
\label{SubSubSec: PosMoveGeomLog} The inductive definition  of
$\PL{S}$ given above  means that sentences can involve only a
\emph{finite} number of primitive propositions, and therefore only
a finite number of disjunctions (`$\lor$') or conjunctions
(`$\land$'). An interesting variant of this structure is the,
so-called, `propositional geometric logic'. This is characterised
by modifying the language and logical axioms so that:
\begin{enumerate}
        \item There are arbitrary disjunctions, including the empty disjunction (`0').
        \item There are finite conjunctions, including the empty conjunction (`1')
        \item Conjunction distributes over arbitrary disjunctions; disjunction distributes over finite conjunctions.
\end{enumerate}
This structure does \emph{not} include negation, implication, or
infinite conjunctions.

From a conceptual viewpoint, this set of rules is  obtained by
considering what it means to actually `affirm' the propositions in
$\PL{S}$. A careful analysis of this concept is given by Vickers
\cite{Vickers89}; the idea itself goes back to work by Abramsky
\cite{Abramsky87}. The conclusion is that the set of `affirmable'
propositions should satisfy the rules above.

Clearly such  a logic is tailor-made for seeking representations
in the open sets of a topological space---the paradigmatic example
of a Heyting algebra. The phrase `geometric logic' is normally
applied to a \emph{first-order} logic with the properties above,
and we will return to this in our discussion of the typed language
$\L{S}$. What we have here is just the propositional part of this
logic.

The restriction to geometric logic would be  easy to incorporate
into our languages $\PL{S}$: for example, the axiom \eq{AinD1orD2}
(if added) could be extended to read\footnote{Note that  the
bi-implication $\Leftrightarrow$ used in, for example,
\eqs{AinD1andD2}{AinD1orD2}, is not available if there is no
implication symbol. Thus we have assumed that we are now working
with a logical structure in which `equality' is a meaningful
concept; hence the introduction of `$=$' in \eq{GeomLogicBigVee}.
}
\begin{equation}
        \bigvee_{i\in I}(A\varin\De_i)=A\varin\bigcup_{i\in I}\De_i
        \label{GeomLogicBigVee}
\end{equation}
for all index sets $I$.

The  move to geometric logic is motivated by a conception of truth
that is grounded in the actions of making real measurements. This
resonates strongly with the logical positivism that seems still to
lurk in the collective unconscious of the physics profession, and
which, of course, was strongly affirmed by Bohr in his analysis of
quantum theory. However, our drive towards `neo-realism' involves
replacing the idea of observation/measurement with that of `the
way things are', albeit in a more sophisticated interpretation
than that of the ubiquitous cobbler-in-the-market. Consequently,
the conceptual reasons for using `affirmative' logic are less
compelling. This issue deserves further thought: at the moment we
are open-minded about it.

The use of geometric logic becomes more interesting in the context
of the typed language $\L{S}$, and we shall return to this in
Section \ref{SubSubSec:TheoryTrLang}

\subsubsection{Introducing Time Dependence}
\label{SubSec:IntTimeDep} In addition to describing `the way
things are' there is also the question of how the-way-things-are
changes in time. In the form presented above, the language
$\PL{S}$ may seem geared towards a `canonical' perspective in so
far as the propositions concerned are implicitly taken to be
asserted at a particular moment of time. As such, $\PL{S}$ deals
with the values of physical quantities at that time. In other
words, the underlying spatio-temporal perspective seems thoroughly
`Newtonian'.

However, this is only partly true since the phrase `physical
quantity' can have meanings other than the canonical one. For
example, one could talk about the `time average of momentum', and
call that a physical quantity. In this case, the propositions
would be about \emph{histories} of the system, not just `the way
things are' at a particular moment in time.

In practice, the question of time dependence can be addressed in
various ways. One  is to attach a (external) time label, $t$, to
the physical quantities, so that the primitive propositions become
of the form ``$A_t\,\varin\,\De$''. This can be interpreted in two
ways. The first is to think of $\Q{S}$ as including the symbols
$A_t$ for all physical quantities $A$ and all values of time
$t\in\mathR$. The second is to keep $\Q{S}$ fixed, but instead let
the language itself becomes time-dependent, so that we should
write $\PL{S}_t$, $t\in\mathR$.

In the former case, $\PL{S}$ would naturally include
\emph{history} propositions of the form
\begin{equation}
(A_{1t_1}\varin\De_1)\land (A_{2t_2}\varin\De_2)\land\cdots
        \land (A_{nt_n}\varin\De_n)
\label{A1t1in...Antnin}
\end{equation}
and other obvious variants of this. Here we assume that $t_1\leq
t_2\leq\cdots\leq t_n$.

The sequential proposition in \eq{A1t1in...Antnin} is to be
interpreted (in a realist reading) as asserting that `` `The
physical quantity $A_1$ has a value that lies in $\De_1$ at time
$t_1$' \emph{and} `the physical quantity $A_2$ has a value that
lies in $\De_2$ at time $t_2$' \mbox{\emph{and} $\cdots$
\emph{and} }`the physical quantity $A_n$ has a value that lies in
$\De_n$ at time $t_n$' ''. Clearly what we have here is a type of
\emph{temporal} logic. Thus this would be an appropriate structure
with which to discuss the `consistent histories' interpretation of
quantum theory, particularly in the, so-called, HPO (history
projection formalism) \cite{Isham94}. In that context,
\eq{A1t1in...Antnin}  represents a, so-called, `homogeneous'
history.

From a general conceptual perspective, one might prefer to have an
internal time object, rather than adding external time labels in
the language. Indeed, in our later discussion of the higher-order
language $\L{S}$ we will strive to eliminate external entities.
However, in the present case,  $\De\subseteq\mathR$ is already an
`external' (to the language) entity, as indeed is $A\in\Q{S}$, so
there seems no particular objection to adding a time label too.

In the second approach, where there is only one time label, the
representation $\pi$ will map ``$A_t\,\varepsilon\,\De$'' to a
time-dependent element, $\pi(A_t\,\varepsilon\,\De)$,  of the
Heyting algebra, $\mathfrak{H}$; one could say that this is a type
of `Heisenberg picture'.

This suggests another option, which is to keep the language free
of any time labels, but allow the \emph{representation} to be
time-dependent. In this case, $\pi_t(\Ain\De)$ is  a
time-dependent member of $\mathfrak{H}$.\footnote{Perhaps we
should also consider the possibility that the Heyting algebra is
time dependent, in which case  $\pi_t(\Ain\De)$ is  a member of
$\mathfrak{H}_t$.}

A different approach is to ascribe time dependence to the `truth
objects' in the theory: this corresponds to a type of
Schr\"odinger picture. The concept of a truth object is discussed
in detail in Section \ref{Sec:TruthValues}.

\subsubsection{The Representation of $\PL{S}$ in Classical Physics}
\label{SubSubSec:PhysRepPLS} Let us now look at the representation
of $\PL{S}$ that corresponds to classical physics. In this case,
the topos involved is just the category, $\Set$, of sets and
functions between sets.

We will denote by $\picl$ the representation of  $\PL{S}$ that
describes the classical, Hamiltonian mechanics of a system, $S$,
whose state-space is a symplectic (or Poisson) manifold $\S$. We
denote by $\breve{A}:\S\map\mathR$ the real-valued
function\footnote{In practice, $\breve{A}$ is required to be
measurable, or smooth, depending on the type of physical quantity
that $A$ is. However, for the most part, these details of
classical mechanics are not relevant to our discussions, and
usually we will not characterise $\breve{A}:\S\map\mathR$ beyond
just saying that it is a function/map from $\S$ to $\mathR$.} on
$\S$ that represents the physical quantity $A$.

Then the representation $\picl$ maps the primitive proposition
$\SAin\De $ to the subset of $\S$ given by
\begin{eqnarray}
\picl(\Ain\De)&:=&\{s\in\S\mid \breve{A}(s)\in\De\}
                                    \nonumber\\
                  &=& \breve{A}^{-1}(\De).\label{sigmaSAinDelta}
\end{eqnarray}
This representation can be extended to all the sentences in
$\PL{S}$ with the aid of \eqs{pi(a)}{pi(d)}. Note that, since
$\De$ is a Borel subset of $\mathR$, $\breve{A}^{-1}(\De)$ is a
Borel subset of the state-space $\S$. Hence, in this case,
 $\mathfrak{H}$ is equal to the Boolean algebra of all
Borel subsets of $\S$.

We note that, for all (Borel) subsets $\De_1,\De_2$ of $\mathR$ we
have
\begin{eqnarray}
   \breve{A}^{-1}(\De_1)\cap\breve{A}^{-1}(\De_2) &=& \breve{A}^{-1}(\De_1\cap\De_2)
                                                \label{A-1D1andD2}\\
   \breve{A}^{-1}(\De_1)\cup\breve{A}^{-1}(\De_2) &=& \breve{A}^{-1}(\De_1\cup\De_2)
                                                \label{A-1D1orD2}\\
                \neg\breve{A}^{-1}(\De_1)&=&\breve{A}^{-1}(\mathR\backslash\De_1)
                                                \label{notA-1D}
\end{eqnarray}
and hence, in classical physics, all three conditions
\eqs{AinD1andD2}{negAinD} that we discussed earlier can be added
consistently  to the language $\PL{S}$.

Consider now the assignment of truth values to the propositions in
this theory. This involves the idea of a `microstate' which, in
classical physics, is simply an element $s$ of the state space
$\S$. Each microstate $s$ assigns to each primitive proposition
$\SAin\De$, a truth value, $\TVal{\Ain\De}{s}$,  which lies in the
set $\{{\rm false},{\rm true}\}$ (which we identify with
$\{0,1\}$) and is defined as
\begin{equation}\label{Def:[AinD]Class}
        \TVal{\Ain\De}{s}:=
        \left\{\begin{array}{ll}
            1 & \mbox{\ if\ $\breve{A}(s)\in\De$;} \\
            0 & \mbox{\ otherwise}
         \end{array}
        \right.
\end{equation}
for all $s\in\S$.

\subsubsection{The Failure to Represent $\PL{S}$ in Standard Quantum Theory.}
The procedure above that works so easily for classical physics
fails completely if one tries to apply it  to standard quantum
theory.

In quantum physics, a physical quantity $A$ is represented by a
self-adjoint operator $\A$ on a Hilbert space $\Hi$, and the
proposition $\SAin\De$ is represented by the projection operator
$\hat E[A\in\De]$ which projects onto  the subset $\De$ of the
spectrum of $\A$; \ie\
\begin{equation}
        \pi(\Ain\De):=\hat E[A\in\De].
                        \label{Def:rhoQT}
\end{equation}

Of course, the set of all projection operators, $\PH$, in $\Hi$
has a `logic' of its own---the `quantum logic'\footnote{For an
excellent survey of quantum logic see \cite{DCG02}. This includes
a discussion of a first-order axiomatisation of quantum logic, and
with an associated sequent calculus. It is interesting to compare
our work with what the authors of this paper have done. We hope to
return to this at some time in the future.} of the Hilbert space
$\Hi$---but this is incompatible with the intuitionistic logic of
the language $\PL{S},$ and the representation \eq{Def:rhoQT}.

Indeed, since the `logic' $\PH$ is non-distributive, there will
exist non-commuting operators $\A,\hat B,\hat C$, and Borel
subsets  $\De_A,\De_B,\De_C$ of $\mathR$ such that\footnote{There
is a well-known example that uses three rays in $\mathR^2$, so
this phenomenon is not particularly exotic.}
\begin{eqnarray}
\hat E[A\in\De_A]\land\left(\hat E[B\in\De_B]
        \lor \hat E[C\in\De_C]\right)&\neq&\nonumber\\
\left(\hat E[A\in\De_A]\land \hat E[B\in\De_B]\right)&\lor &
 \left(\hat E[A\in\De_A]\land\hat E[C\in\De_C]\right)\hspace{1cm}
\end{eqnarray}
while, on the other hand, the logical bi-implication
\begin{equation}
        \alpha\land(\beta\lor\gamma)\Leftrightarrow
                (\alpha\land\beta)\lor(\alpha\land\gamma)
\end{equation}
can be deduced from the axioms of the language $\PL{S}$.

This failure of distributivity bars any na\"ive realist
interpretation of quantum logic. If an instrumentalist
interpretation is used instead, the spectral projectors $\hat
E[A\in\De]$ now represent propositions about what would happen
\emph{if} a measurement is made, not propositions about what is
`actually the case'.  And, of course, when a state is specified,
this does not yield actual truth values but only the Born-rule
probabilities of getting certain results.


\section{A Higher-Order, Typed  Language for Physics}
\label{Sec:TypedLanguage}
\subsection{The Basics of the Language $\L{S}$}
We want now  to consider the possibility of representing the
physical quantities of a system by arrows in a topos other than
$\Set$.

The physical meaning of such an arrow is not clear, \emph{a
priori}. Nor is it even clear   \emph{what} it is that is being
represented in this way. However, what  \emph{is} clear is that in
such a situation it is not correct to assume that the
quantity-value object is necessarily the real-number object in the
topos (assuming that there is one). Rather, the target-object,
$\R_S$, has to be determined for each topos, and is therefore an
important part of the `representation'.

A powerful technique for allowing the quantity-value object to be
system-dependent is to add a symbol `$\R$' to the system language.
Developing this line of thinking  suggests that `$\Si$', too,
should be added to the language, as should a set of symbols of the
form `$A:\Si\map\R$', to be construed as `what it is' (hopefully a
physical quantity) that is represented by arrows in a topos.
Similarly, there should be a symbol `$\O$', to act as the
linguistic precursor to the sub-object classifier in the topos; in
the topos $\Set$, this is just the set $\{0,1\}$.

The clean way of doing all this is to construct a `local language'
\cite{Bell88}. Our basic assumption is that such a language,
$\L{S}$, can be associated with each system $S$. A physical theory
of $S$ then corresponds to a representation of $\L{S}$ in an
appropriate topos.

\paragraph{The symbols of $\L{S}$.}  We first
consider  the minimal set of symbols needed to handle elementary
physics. For more sophisticated theories in physics it will be
necessary to change, or enlarge, this set of `ground-type'
symbols.

The symbols for the local language, $\L{S}$, are defined
recursively as follows:
\begin{enumerate}
\item
\begin{enumerate}
  \item The basic \emph{type symbols} are $1,\O,\Si,\R$.
  The last two, $\Si$ and $\R$, are known as
  \emph{ground-type symbols}. They are the linguistic precursors
  of the state object, and quantity-value object, respectively.

If $T_1,T_2,\ldots,T_n$, $n\geq1$, are type symbols, then so
is\footnote{By definition, if $n=0$ then $T_1\times
T_2\times\cdots\times T_n:=1$.} $T_1\times T_2\times\cdots\times
T_n$.
        \item  If $T$ is a type symbol, then so is $PT$.
\end{enumerate}
\item
\begin{enumerate}
        \item   For each type symbol, $T$, there is associated a
        countable  set of \emph{variables of type $T$}.

        \item There is a special symbol $*$.
\end{enumerate}
\item
\begin{enumerate}
\item To each pair $(T_1,T_2)$ of type symbols there is associated
a set, $F_{\L{S}}(T_1,T_2)$, of \emph{function symbols}. Such a
symbol, $A$, is said to have \emph{signature} $T_1\map T_2$; this
is indicated  by writing $A:T_1\map T_2$.

\item Some of these sets of function symbols may be empty.
However, in our case, particular importance is attached to the
set, $F_{\L{S}}(\Si,\R)$, of   function symbols
\mbox{$A:\Si\map\R$}, and we assume this set is non-empty.

\end{enumerate}
\end{enumerate}

The function symbols $A:\Si\map\R$ represent the `physical
quantities' of the system, and hence $F_{\L{S}}(\Si,\R)$ will
depend on the system $S$. In fact, the only parts of the language
that are system-dependent are these function symbols. The set
$F_{\L{S}}(\Si,\R)$ is the analogue of the set, $\Q{S}$, of
physical quantities associated with the propositional language
$\PL{S}$.

For example, if $S_1$ is a point particle moving in one dimension,
the set of physical quantities could be chosen to be
$F_{\L{S_1}}(\Si,\R)=\{x,p,H\}$ which represent the position,
momentum, and energy of the system. On the other hand, if $S_2$ is
a particle moving in three dimensions, we could have
$F_{\L{S_2}}(\Si,\R)=\{x,y,z,p_x,p_y,p_z,H\}$ to allow for
three-dimensional position and momentum (with respect to some
given Euclidean coordinate system). Or, we could decide to add
angular momentum too, to give the set $F_{\L{S_2}}(\Si,{\R})
=\{x,y,z,p_x,p_y,p_z,J_x,J_y,J_z,H\}$. A still further extension
would be to add the quantities $\underline{x}\cdot\underline{n}$
and $\underline{p}\cdot \underline{m}$ for all unit vectors
$\underline{n}$ and $\underline{m}$; and so on.

Note that, as with the propositional language $\PL{S}$, the fact
that a given system has a specific Hamiltonian\footnote{It must be
emphasised once more that the use of a local language is
\emph{not} restricted to standard, canonical systems in which the
concept of a `Hamiltonian' is meaningful. The scope of the
linguistic ideas is \emph{much} wider than that and the canonical
systems are only an example. Indeed, our long-term interest is in
the application of these ideas to quantum gravity where the local
language is likely to be very different from that used here.
However, we anticipate that the basic  ideas will be the
same.}---expressed as a particular function of position and
momentum coordinates---is not something that is to be coded into
the language: instead, such system dependence arises in the choice
of \emph{representation} of the language. This means that many
different systems can have the same local language.

Finally, it should be emphasised that this list of symbols is
minimal and one will certainly want to add more. One obvious,
general, example is a type symbol $\mathN$ that is to be
interpreted as the linguistic analogue of the natural numbers. The
language could then be augmented with the axioms of Peano
arithmetic.

\paragraph{The terms of $\L{S}$.}
The next step is to enumerate the `terms' in the language,
together with their associated types \cite{Bell88,LamScott86}:
\begin{enumerate}
\item
\begin{enumerate}
\item For each type symbol $T$, the variables of type $T$ are
terms of type $T$.

\item The symbol $*$ is a term of type $1$.

\item A term of type $\O$ is called a \emph{formula}; a formula
with no free variables is called a \emph{sentence}.
\end{enumerate}

\item If $A$ is function symbol with signature $T_1\map T_2$, and
$t$ is a term of type $T_1$, then $A(t)$ is  term of type $T_2$.

In particular, if $A:\Si\map\R$ is a physical quantity, and $t$ is
a term of type $\Si$, then $A(t)$ is a term of type $\R$.

\item
\begin{enumerate}
        \item If $t_1,t_2,\ldots,t_n$ are terms of type
        $T_1,T_2,\ldots,T_n$, then
        $\langle t_1,t_2,\ldots,t_n\rangle$ is
        a term of type $T_1\times T_2\times\cdots\times T_n$.

       \item If $t$ is a term of type
       $T_1\times T_2\times\cdots\times T_n$,
       and if $1\leq i\leq n$, then $(t)_i$ is a term of type $T_i$.
\end{enumerate}

\item
\begin{enumerate}
        \item If $\om$ is a term of type $\O$, and $\va{x}$ is a
        variable of type $T$, then $\{\va{x}\mid \om\}$ is a term of
        type $PT$.

        \item If $t_1,t_2$ are terms of the same type, then `$t_1=t_2$' is a term of type $\O$.

        \item If $t_1,t_2$ are terms of type $T,PT$ respectively,
        then $t_1\in t_2$ is a term of type $\O$.
\end{enumerate}
\end{enumerate}

Note that the logical operations are not included in the set of
symbols. Instead,  they can all be defined using what is already
given. For example, (i) $true:= (*=*)$; and (ii) if $\alpha$ and
$\beta$ are terms of type $\O$, then\footnote{The parentheses
$(\;)$ are not  symbols in the language, they are just a way of
grouping letters and sentences. The same remark applies to the
inverted commas `'.} $\alpha\land\beta:=\big(\la\alpha,
\beta\ra=\la{\rm true}, {\rm true}\ra\big).$   Thus, in terms of
the original set of symbols, we have
\begin{equation}
  \alpha\land\beta:=\big(\langle\alpha,
   \beta\rangle=\langle*=*,*=*\rangle\big)
\end{equation}
and so on.

\paragraph{Terms of particular interest to us.}
Let  $A$ be a physical quantity in the set $\F{S}$, and therefore
a function symbol of signature $\Si\map\R$. In addition, let
$\va\De$ be a variable (and therefore a term) of type $P\R$; and
let $\va{s}$ be a variable (and therefore a term) of type $\Si$.
Then some terms of particular interest to us are  the following:
\begin{enumerate}
\item $A(\va{s})$ is a term of type $\R$ with a free variable,
$\va{s}$, of type $\Si$.

\item `$A(\va{s})\in\va\De$' is a term of type $\O$ with
free variables (i) $\va{s}$ of type $\Si$; and (ii) $\va{\De}$ of
type $P\R$.

\item $\{\va{s}\mid A(\va{s})\in\va\De\}$ is a term of
type $P\Si$ with a free variable $\va{\De}$ of type $P\R$.
\end{enumerate}
As we shall  see, `$A(\va{s})\in\va\De$' is an analogue of the
primitive propositions $\SAin\De$ in the propositional language
$\PL{S}$. However, there is a crucial difference. In  $\PL{S}$,
the `$\De$' in $\SAin\De$ is a specific subset of the external (to
the language) real line $\mathR$. On the other hand, in the local
language $\L{S}$, the `$\va\De$' in `$A(\va{s})\in\va\De$' is an
\emph{internal} variable within the language.

\paragraph{Adding axioms to the language.}
To make the language $\L{S}$ into  a deductive system we need to
add a set of appropriate axioms and rules of inference. The former
are expressed using \emph{sequents}: defined as expressions of the
form $\Ga:\alpha$ where $\alpha$ is a formula (a term of type
$\O$) and $\Ga$ is a set of such formula. The intention is that
`$\Ga:\alpha$' is to be read intuitively as ``the collection of
formula in $\Ga$ `imply' $\alpha$''. If $\Ga$ is empty we just
write $:\alpha$.

The basic axioms include things like `$\alpha:\alpha$'
(tautology), and  `$: \va{t}\in\{\va{t}\mid\alpha\}
\Leftrightarrow \alpha$' (comprehension) where $\va{t}$ is a
variable of type $T$. These axioms\footnote{The complete set is
\cite{Bell88}:
\begin{eqnarray*}
\mbox{Tautology:}       &&\alpha=\alpha\\[2pt]
\mbox{Unity}:           &&\va{x}_1=* \mbox{\ where $\va{x}_1$
                is a variable of type $1$.}       \\[2pt]
\mbox{Equality:}         && x=y,\alpha(\va{z}/x):
             \alpha(\va{z}/y). \mbox{ Here, $\alpha(\va{z}/x)$
             is the term $\alpha$ with $\va{z}$ replaced} \\
             &&\mbox{by the term $x$ for each free occurrence of the variable
         $\va{z}$. The terms}        \\
        && \mbox{$x$ and $y$ must be of the same type as $\va{z}$}.   \\[2pt]
\mbox{Products:}&& :(\la x_1,\ldots,x_n\ra)_i=x_i\\
                &&:x=\la(x)_1,\ldots,(x)_n\ra           \\[2pt]
\mbox{Comprehension:}&&:\va{t}\in\{\va{t}\mid\alpha\}
\Leftrightarrow \alpha
\end{eqnarray*}
} and the rules of inference (sophisticated analogues of
\emph{modus ponens})  give rise to a deductive system using
intuitionistic logic. For the details see
\cite{Bell88,LamScott86}.

For applications in physics we could, and presumably should, add
extra axioms (in the form of sequents). For example, perhaps the
quantity-value object  should always be an abelian-group object,
or at least a semi-group\footnote{One could go even further and
add the axioms for real numbers.  However, the example of quantum
theory suggests that this is inappropriate: in general, the
quantity-value object will \emph{not} be the real-number object
\cite{DI(3)}.}? This can be coded into the language by adding the
axioms for an abelian group structure for $\R$. This involves the
following steps:
\begin{enumerate}
\item Add the following symbols:
\begin{enumerate}

\item A `unit' function symbol $0:1\map\R$; this will be the
linguistic analogue of the unit element in an abelian group.

\item An `addition' function symbol $+:\R\times\R\map\R$.

\item An `inverse' function symbol $-:\R\map\R$
\end{enumerate}

\item Then add axioms like
`$:\forall\va{r}\big(+\langle\va{r}, 0(*)\rangle= \va{r}\big)$'
where $\va{r}$ is a variable of type $\R$, and so on.
\end{enumerate}

For another example, consider a point particle moving in three
dimensions, with the function symbols
$F_{\L{S}}(\Si,{\R})=\{x,y,z,p_x,p_y,p_z,J_x,J_y,J_z,H\}$. As
$\L{S}$ stands,  there is no way to specify, for example, that
`$J_x=yp_z-zp_y$'. Such relations can only be implemented in a
\emph{representation} of the language. However, if
 this relation is felt to be `universal' (\ie\ if it is expected to hold in all
physically-relevant representations) then it could be added to the
language with the use of extra axioms.

One of the delicate decisions that has to be made about $\L{S}$ is
what extra axioms  to add to the base language. Too few, and the
language lacks content; too many, and representations of potential
physical significance are excluded. This is one of the places in
the formalism where a degree of physical insight is necessary!

\subsection{Representing $\L{S}$ in a Topos}
The construction of a theory of the system $S$  involves choosing
a representation\footnote{The word `interpretation' is often used
in the mathematical literature, but we want to reserve that for
use in discussions of interpretations of quantum theory, and the
like.}/model, $\phi$, of the language $\L{S}$ in a
topos\footnote{A more comprehensive notation is $\tau_\phi(S)$,
which draws attention to the system $S$ under discussion;
similarly, the state object could be written as $\Sigma_{\phi,S}$,
and so on. This extended notation is used in Section
\ref{Sec:CatSys} where we are concerned with the relations between
\emph{different} systems, and then it is essential to indicate
which system is meant. However, in the present article, only one
system at a time is being considered, and so the truncated
notation is fine.} $\tau_\phi$. The choice of both topos and
representation  depend on the theory-type being used.

For example, consider a  system, $S$, that can be treated using
both classical physics and quantum physics, such as   a point
particle moving in three dimensions. Then, for the application of
the theory-type `classical physics', in a representation denoted
$\s$, the topos $\tau_\s$ is $\Set$, and $\Si$ is represented by
the symplectic manifold $\Si_\s:= T^*\mathR^3$; $\R$ is
represented by the usual real numbers $\mathR$.

On the other hand, as we shall see in Section
\ref{Sec:QuPropSpec}, for  the application of the theory-type
`quantum physics', $\tau_\phi$ is the topos, $\SetH{}$, of
presheaves over the category\footnote{We recall that the objects
in $\V{}$ are the unital, commutative von Neumann sub-algebras of
the algebra, $\BH$, of all bounded operators on $\Hi$. We will
explain, and motivate, this later.} $\V{}$, where ${\cal H}\simeq
L^2(\mathR^3,d^3x)$ is the Hilbert space of the system $S$. In
this case, $\Si$ is represented by $\Si_\phi:=\Sig$, where $\Sig$
is the spectral presheaf; this representation is discussed at
length in Sections \ref{Sec:QuPropSpec}. For both theory types,
the \emph{details} of, for example, the Hamiltonian, are coded in
the representation.

We now list the $\tau_\phi$-representation of the  most
significant symbols and terms in our language, $\L{S}$ (we have
picked out only the parts that are immediately relevant to our
programme: for full details see \cite{Bell88, LamScott86}).
\begin{enumerate}

\item \begin{enumerate} \item The ground type symbols $\Si$ and
$\cal R$ are represented by objects $\Si_\phi$ and ${\cal R}_\phi$
in $\tau_\phi$. These are identified physically as the state
object and quantity-value object, respectively.

   \item The symbol $\O$, is represented by
   $\O_\phi:=\O_{\tau_\phi}$, the sub-object classifier of
the topos $\tau_\phi$.

   \item The  symbol $1$, is represented by
$1_\phi:={1}_{\tau_\phi}$,
   the terminal object in $\tau_\phi$.
\end{enumerate}

\item For each type symbol $PT$, we have $(PT)_\phi:=PT_\phi$, the
power object of the object $T_\phi$ in $\tau_\phi$.

In particular, $(P\Si)_\phi=P\Si_\phi$ and $(P{\R})_\phi
=P\R_\phi$.

\item Each function symbol $A:\Si\map\R$ in $\F{S}$ (\ie\ each
physical quantity) is represented by an arrow
$A_\phi:\Si_\phi\map{\R}_\phi$ in $\tau_\phi$.

We will generally require the representation to be
\emph{faithful}: \ie\ the map $A\mapsto A_\phi$\ is one-to-one.

\item  A term of type $\O$ of the form
`$A(\va{s})\in\va\De$' (which has free variables $\va{s},\va{\De}$
of type $\Si$ and $P\R$ respectively) is represented by an arrow
$\Val{A(\va{s})\in\va\De}_\phi :\Si_\phi\times P{\R}_\phi\map
\O_{\tau_\phi}$. In detail, this arrow is
\begin{equation}
\Val{A(\va{s})\in\va\De}_\phi=e_{\R_\phi}\circ
        \la\Val{A(\va{s})}_\phi,\Val{\va\De}_\phi\ra
\end{equation}
where $e_{\R_\phi}:\R_\phi\times P\R_\phi\map\O_{\tau_\phi}$ is
the usual evaluation map;
$\Val{A(\va{s})}_\phi:\Si_\phi\map\R_\phi$ is the arrow $A_\phi$;
and $\Val{\va{\De}}_\phi:P\R_\phi\map P\R_\phi$ is the identity.

Thus $\Val{A(\va{s})\in\va\De}_\phi$ is the chain of arrows:
\begin{equation}
\Si_\phi\times P\R_\phi\mapright{A_\phi\times\id}
        \R_\phi\times P\R_\phi\mapright{e_{\R_\phi}}\O_{\tau_\phi}.
                                        \label{A(s)intildeDeChain}
\end{equation}
We see that the analogue of the `$\De$' used in the
$\PL{S}$-proposition $\SAin\De$ is played by sub-objects of
$\R_\phi$ (\ie\ global elements of $P\R_\phi$) in the domain of
the arrow in \eq{A(s)intildeDeChain}. These objects are, of
 course, representation-dependent (\ie\ they depend on $\phi$).

\item A term of type $P\Si$ of the form $\{\va{s}\mid
A(\va{s})\in\va\De\}$ (which has a free variable $\va\De$ of type
$P\R$) is represented by an arrow $\Val{\{\va{s}\mid
A(\va{s})\in\va\De\}}_\phi : P\R_\phi\map P\Si_\phi$. This arrow
is the power transpose\footnote{One of the basic properties of a
topos is that there is a one-to-one correspondence between arrows
$f:A\times B\map\O$ and arrows $\name{f}:B\map PA:=\O^A$. In
general, $\name{f}$ is called the \emph{power transpose} of $f$.
If $B\simeq 1$ then  $\name{f}$ is known as the \emph{name} of the
arrow $f:A\map\O$.}of
 $\Val{A(\va{s})\in\va\De}_\phi$:
\begin{equation}
\Val{\{\va{s}\mid A(\va{s})\in\va\De\}}_\phi =
\name{\Val{A(\va{s})\in\va\De}_\phi}\label{[]=nametildeDe}
\end{equation}

\item A term, $\om$, of type $\O$ with no free variables is
represented by a global element $\Val{\om}_\phi:1_{\tau_\phi} \map
\O_{\tau_\phi}$. These will typically act as `truth values' for
propositions about the system.

\item Any axioms that have been added to the language are required
to be represented by the arrow
$true:1_{\tau_\phi}\map\O_{\tau_\phi}$.
\end{enumerate}

\subsubsection{The Local Set Theory of a Topos.}
We should emphasise that the decision to focus on the particular
type of language that we have, is not an arbitrary one. Indeed,
there is a deep connection between such languages and topos
theory.

In this context, we first note that to any local language, ${\cal
L}$, there is associated a `local set theory'. This involves
defining an `${\cal L}$-set' to be a term $X$ of power type (so
that expressions of the form $x\in X$ are meaningful) and with no
free variables. Analogues of all the usual set operations can be
defined on $\cal L$-sets. For example, if $X,Y$ are $\cal L$-sets
of type $PT$, one can define $X\cap Y:=\{\va{x}\mid \va{x}\in
X\land \va{x}\in Y\}$ where $\va{x}$ is a variable of type $T$.

Furthermore, each local language, ${\cal L}$, gives rise to an
associated topos, ${\cal C}({\cal L})$, whose objects are
equivalence classes of ${\cal L}$-sets, where $X\equiv Y$ is
defined to mean that the equation $X=Y$ (\ie\ a term of type
$\Omega$ with no free variables) can be proved using the sequent
calculus of the language with its axioms. From this perspective, a
representation of the system-language $\L{S}$ in a topos $\tau$ is
equivalent to a \emph{functor} from the topos ${\cal C}(\L{S})$ to
$\tau$.

\subsubsection{Theory Construction as a Translation of Languages}
\label{SubSubSec:TheoryTrLang} Conversely, for each topos $\tau$
there is a local language, ${\cal L}(\tau)$, whose ground-type
symbols are the objects of $\tau$, and whose function symbols are
the arrows in $\tau$. It then follows that a representation of a
local language, $\cal L$, in $\tau$ is equivalent to a
`translation' of ${\cal L}$ in ${\cal L}(\tau)$. \vspace{5pt}

\displayE{10}{5}{4}{Thus constructing a theory of physics is
equivalent to finding a suitable \emph{translation} of the system
language, $\L{S}$,  to the language, $\L{\tau}$, of an appropriate
topos $\tau$.}

\noindent As we will see later, the idea of translating one local
language into another plays a central role in the discussion of
composite systems and sub-systems.

In the case of spoken languages, one can translate  from, say, (i)
English to German;  or (ii) from English to Greek, and then from
Greek to German. However, no matter how good the translators,
these two ways of going from English to German will generally not
agree. This is partly because the translation process is not
unique, but also because each language possesses certain intrinsic
features that simply do not admit of translation.

There is an interesting analogous question for the representation
of the local languages $\L{S}$. Namely, suppose
$\phi_1:\L{S}\map\L{\tau_{\phi_1}}$ and
$\phi_2:\L{S}\map\L{\tau_{\phi_2}}$ are two different topos
theories of the same system $S$ (these could be, say, classical
physics and quantum physics). The question is if/when will there
be a translation
$\phi_{12}:\L{\tau_{\phi_1}}\map\L{\tau_{\phi_2}}$ such that
\begin{equation}
                        \phi_2=\phi_{12}\circ\phi_1
                                       \label{comptransl}
\end{equation}
In terms of the representation functors from the topos ${\cal
C}(\L{S})$ to the topoi $\tau_{\phi_1}$ and $\tau_{\phi_2}$, the
question is if there exists an interpolating
 functor
from $\tau_{\phi_1}$ to $\tau_{\phi_2}$.

In Section \ref{SubSec:GeneralToposAxioms}, we will introduce a
certain category, ${\cal M}(\Sys)$, whose objects are topoi and
whose arrows are geometric morphisms between topoi. It would be
natural to require the arrow from $\tau_{\phi_1}$ to
$\tau_{\phi_2}$ (if it exists)  to be an arrow in this category.

It is at this point that `geometric logic' enters the scene (cf.\
Section \ref{SubSubSec: PosMoveGeomLog}). A formula in $\L{S}$ is
said to be \emph{positive} if it does not contain the
symbols\footnote{Here, the formula $\a\Rightarrow\beta $ is
defined as $\a\Rightarrow\beta:=(\a\land\beta)=\alpha$; $\forall$
is defined as $\forall x\alpha:=(\{x\mid\a\}=\{x\mid {\rm
true}\})$; where ${\rm true}:=*=*$} $\Rightarrow$ or $\forall$.
These conditions imply that $\neg$ is also absent. In fact, a
positive formula uses only $\exists,\land$ and $\vee$. A
disjunction can have an arbitrary index set, but a conjunction can
have only a finite index set. A sentence of the form $\forall
x(\a\Rightarrow\beta)$ is said to be a \emph{geometric
implication} if both $\a$ and $\beta$ are positive. Then a
\emph{geometric logic} is one in which only geometric implications
are present in the language.

\displayE{12}{5}{4}{The advantage of using just the geometric part
of logic is that geometric implications are \emph{preserved} under
geometric morphisms. This  makes it appropriate to ask for the
existence of `geometric translations'
$\phi_{12}:\L{\tau_{\phi_1}}\map\L{\tau_{\phi_2}}$, as in
\eq{comptransl}, since these will preserve the logical structure
of the language $\L{S}$.}

The notion of `toinvariance' introduced recently by Landsmann
\cite{Land08a} can be interpreted within our structures as
asserting that the translations
$\phi_{12}:\L{\tau_{\phi_1}}\map\L{\tau_{\phi_2}}$ should always
exist; or, at least, they should under appropriate conditions. Of
course, the significance of this depends on how much information
about the system is reflected in the language $\L{S}$ and how much
in the individual representations.

For example, in the case of classical and quantum physics, one
might go so far as  to include information about the dynamics of
the system within the local language $\L{S}$. If the topoi
$\phi_1$ and $\phi_2$ are those for the classical and quantum
physics of $S$ respectively (so that $\phi_1$ is $\Set$ and
$\phi_2$ is $\SetH{}$), then an interpolating translation
$\phi_{12}:\L{\Set}\map\L{\SetH{}}$ would be a nice realisation of
Landsmann's long-term goal of regarding quantisation as some type
of functorial operation.

Of course, introducing dynamics raises interesting questions about
the status of the concept of `time' (cf the discussion in Section
\ref{SubSec:IntTimeDep}). In particular, is time to be identified
as an object in representing topos, or is it an external
parameter, like the `$\De$' quantities in the propositional
languages $\PL{S}$?

\subsection{Classical Physics in the Local Language $\L{S}$}
The quantum theory representation of $\L{S}$ is studied in Section
\ref{Sec:QuPropSpec}. Here we will look at the concrete form of
the expressions above  for the example of classical physics. In
this case, for all systems $S$, and all classical representations,
$\s$, the topos $\tau_\s$ is $\Set$. This representation of
$\L{S}$ has the following ingredients:
\begin{enumerate}
\item
\begin{enumerate}
        \item The ground-type symbol $\Si$ is represented by
        a symplectic manifold, $\Si_\s$, that is the
        state-space for the system $S$.

        \item The ground-type symbol $\R$ is represented by the
        real line, \ie\ $\R_\s:=\mathR$.

        \item The type symbol $P\Si$ is represented by the set,
        $P\Si_\s$, of all\footnote{To be super precise, we really need to use the collection $P_{\rm Bor}\Si_\s$ of all \emph{Borel} subsets of $\Si_\s$.} subsets of the state space $\Si_\s$.

        The type symbol $P\R$ is represented by the set, $P\mathR$,
        of all subsets of $\mathR$.
\end{enumerate}

\item
\begin{enumerate}
   \item The  type symbol $\O$, is represented by
   $\Omega_\Set:=\{0,1\}$:  the sub-object classifier in $\Set$.

   \item The type symbol $1$, is represented by the
   singleton set: \ie\ $1_\Set=\{*\}$, the terminal object in $\Set$.
\end{enumerate}

\item Each function symbol $A:\Si\map\cal R$, and hence each
physical quantity, is represented by a real-valued function,
$A_\s:\Si_\s\map\mathR$, on the state space $\Si_\s$.

\item  The term `$A(\va{s})\in\va\De$' of type $\O$ (where
$\va{s}$ and $\va\De$ are free variables of type $\Si$ and $P\R$
respectively) is represented by the function
$\Val{A(\va{s})\in\va\De}_\s:\Si_\s\times P\mathR \map\{0,1\}$
that is defined by (c.f.\ \eq{A(s)intildeDeChain})
\begin{equation}
 \Val{A(\va{s})\in\va\De}_\s(s,\De)=
        \left\{\begin{array}{ll}
            1 & \mbox{\ if\ $A_\s(s)\in \De$;} \\
            0 & \mbox{\ otherwise.}
         \end{array}
        \right. \label{A(s)intildeDeChainCL}
\end{equation}
for all $(s,\De)\in\Si_\s\times P\mathR$.

\item The term  $\{\va{s}\mid A(\va{s}) \in\va{\De}\}$ of
type $P\Si$ (where $\va{\De}$ is a free variable of type $P\cal
R$) is represented by the function $\Val{\{\va{s}\mid A(\va{s})
\in\va{\De}\}}_\s: P\mathR \map P\Si_\s$ that is defined by
\begin{eqnarray}        \label{Def:sigmaS(A(s)inDelta)=}
        \Val{\{\va{s}\mid A(\va{s})
\in\va{\De}\}}_\s(\De) &:=&
        \{s\in\Si_\phi\mid A_\s(s)\in\De\}\nonumber\\
        &=& A_\s^{-1}(\De)
\end{eqnarray}
for all $\De\in P\mathR$.
\end{enumerate}

\subsection{Adapting the Language $\L{S}$ to Other Types of Physical
System}\label{SubSec:AdaptLanguage} Our central contention in this
work is that (i) each physical system, $S$,  can be equipped with
a local language, $\L{S}$; and (ii) constructing an explicit
theory of $S$ in a particular theory-type is equivalent to finding
a representation of $\L{S}$ in a topos which may well be other
than the topos of sets.

There are many situations in which  the language is independent of
the theory-type, and then, for a given system $S$, the different
topos representations of $\L{S}$ correspond to the application of
the different theory-types to the same system $S$. We gave an
example earlier of a point particle moving in three dimensions:
the classical physics representation is in the topos $\Set$, but
the quantum-theory representation is in the presheaf topos
$\Set^{{\cal V}(L^2(\mathR^3,\,d^3x))}$ .

However, there are other situations where the relationship between
the language and its representations is more complicated than
this. In particular, there is the critical question about what
features of the theory should go into the language, and what into
the representation. The first step in adding  new features is to
augment the set  of ground-type symbols. This is because these
represent the entities that are going to be of generic interest
(such as a state object or quantity-value object). In doing this,
extra axioms may also be introduced to encode the properties that
the new objects are expected to possess in all representations of
physical interest.

For example, suppose we want to use our formalism to discuss
space-time physics: where does the information about the
space-time go? If the subject is classical  field theory in a
curved  space-time, then the topos $\tau$ is  $\Set$, and the
space-time manifold is  part of the \emph{background} structure.
This makes it natural to have the manifold assumed in the
representation; \ie\ the information about the space-time is in
the representation.

Alternatively, one can add  a new ground type symbol, `$M$', to
the language, to serve as the linguistic progenitor of
`space-time'; thus $M$ would have  the same theoretical status as
the symbols $\Si$ and $\R$. In this context, we recall the brief
discussion in Section \ref{SubSubSec:WhyReals} about the use of
the real numbers in modelling space and/or time, and the
motivation this provides for representing space-time as an object
in a topos, and whose sub-objects represent the fundamental
`regions'.

If `$M$' is added to the language, a function symbol
$\psi:M\map\R$ is then  the progenitor of a physical field. In a
representation, $\phi$, the object $M_\phi$ plays the role of
`space-time' in the topos $\tau_\phi$, and
$\psi_\phi:M_\phi\map\R_\phi$ is the representation of the field.

Of course, the language $\L{S}$  says nothing about what sort of
entity $M_\phi$  is, except in so far as such information is
encoded in extra axioms. For example, if the subject is classical
field theory, then $\tau_\phi=\Set$, and $M_\phi$ would be a
standard differentiable manifold. On the other hand, if the topos
$\tau_\phi$ admits `infinitesimals', then $M_\phi$ could be a
manifold according to the language of synthetic differential
geometry \cite{Kock81}.

The same type of argument applies to the status of `time' in a
canonical theory. In particular, it would be possible to add a
ground-type symbol, $\typeTime$, so that, in any representation,
$\phi$, the object $\typeTime_\phi$ in the topos $\tau_\phi$ is
the analogue of the `time-line' for that theory. For standard
physics in $\Set$ we have $\typeTime_\phi=\mathR$, but the form of
$\typeTime_\phi$ in a more general topos, $\tau_\phi$, would be a
rich subject for speculation.

The addition of  a  `time-type' symbol, $\typeTime$, to the
language $\L{S}$ is a prime example of a situation where  one
might want to add extra axioms. These could involve ordering
properties, or algebraic  properties like those of  an abelian
group, and so on. In any topos representation, these properties
would  then be realised as the corresponding type of object in
$\tau_\phi$. Thus abelian group axioms mean that $\typeTime_\phi$
is an abelian-group object in the topos $\tau_\phi$;
total-ordering axioms for the time-type $\typeTime$ mean that
$\typeTime_\phi$ is a totally-ordered object in $\tau_\phi$, and
so on.

As an  interesting extension of this idea, one could have a
space-time ground type symbol $M$, but then add the axioms for a
partial ordering. In that case, $M_\phi$ would be a poset-object
in $\tau_\phi$, which could be interpreted physically as the
$\tau_\phi$-analogue of a causal set \cite{Dow05}.


\section{Quantum Propositions as Sub-Objects of the Spectral Presheaf}
\label{Sec:QuPropSpec}
\subsection{Some Background Remarks}
\label{SubSec:QuPropSpecBackgroundRemarks}
\subsubsection{The Kochen-Specker Theorem}
The idea of representing quantum theory in a topos of presheaves
stemmed originally \cite{IB98} from a desire to acquire a new
perspective on the  Kochen-Specker theorem \cite{KS67}. It will be
helpful at this stage to review some of this older material.

A commonsense belief, and one apparently shared by Heidegger, is
that at any given time any physical quantity \emph{must} have a
value even if we do not know what it is. In classical physics,
this is not problematic  since the underlying mathematical
structure is geared precisely to realise it. Specifically, if $\S$
is the state space of some classical system, and if the physical
quantity $A$ is represented by a real-valued function
$\breve{A}:\S\map\mathR$,  then the value $V_s(A)$ of $A$ in any
state $s\in\S$ is simply
\begin{equation}
    V^s(A)=\breve{A}(s).                      \label{Vs(A)=A(s)}
\end{equation}
Thus all physical quantities possess a value in any state.
Furthermore, if $h:\mathR\map\mathR$ is a real-valued function, a
new physical quantity $h(A)$ can be defined by requiring the
associated function $\breve{h(A)}$ to be
\begin{equation}
    \breve{h(A)}(s):=h(\breve{A}(s))       \label{Def:h(A)}
\end{equation}
for all $s\in\S$; {\em i.e.}, $\breve{h(A)}:=h\circ\breve
A:\S\map\mathR$. Thus the physical quantity $h(A)$ is {\em
defined\/} by saying that its value in any state $s$ is the result
of applying the function $h$ to the value of $A$; hence, by
definition, the values of the physical quantities $h(A)$ and $A$
satisfy the `functional composition principle'
\begin{equation}
    V^s(h(A))=h(V^s(A))                     \label{FUNCT-class}
\end{equation}
for all states $s\in\cal S$.

However, standard quantum theory precludes any such naive realist
interpretation of the relation between formalism and physical
world. And this obstruction comes from the mathematical formalism
itself, in the guise of the famous Kochen-Specker theorem which
asserts the impossibility of assigning values to all physical
quantities whilst, at the same time, preserving the functional
relations between them \cite{KS67}.

In a quantum theory, a physical quantity $A$ is represented by a
self-adjoint operator $\A$ on the Hilbert space of the system, and
the first thing one has to decide is whether to regard a valuation
as a function of the physical quantities themselves, or on the
operators that represent them.  From a mathematical perspective,
the latter strategy is preferable, and we shall therefore define a
valuation to be a real-valued function $V$ on the set of all
bounded, self-adjoint operators, with the properties that : (i)
the value $V(\A)$ of the physical quantity $A$ represented by the
operator $\A$ belongs to the spectrum of $\A$ (the so-called
`value rule'); and (ii) the functional composition principle (or
FUNC for short) holds:
\begin{equation}
    V(\hat B)=h(V(\A))                  \label{funct-rule}
\end{equation}
for any pair of self-adjoint operators $\A$, $\hat B$ such that
$\hat B= h(\A)$ for some real-valued function $h$. If they
existed, such valuations could be used to embed the set of
self-adjoint operators in the commutative ring of real-valued
functions on an underlying space of microstates, thereby laying
the foundations for a hidden-variable interpretation of quantum
theory.

Several important results follow from the definition of a
valuation. For example, if $\A_1$ and $\A_2$ commute, it follows
from the spectral theorem that there exists an operator $\hat C$
and functions $h_1$ and $h_2$ such that $\A_1=h_1(\hat C)$ and
$\A_2=h_2(\hat C)$. It then follows from FUNC that
\begin{equation}
    V(\A_1+\A_2)=V(\A_1)+V(\A_2)
\end{equation}
and
\begin{equation}
    V(\A_1\A_2)=V(\A_1)V(\A_2).\label{VAB}
\end{equation}

The defining equation \eq{funct-rule}) for a valuation makes sense
whatever the nature of the spectrum $\sp(\A)$ of the operator
$\A$. However, if $\sp(\A)$ contains a continuous part, one might
doubt the physical meaning of assigning one of its elements as a
value.  To handle the more general case, we shall view a valuation
as primarily giving {\em truth-values\/} to {\em propositions\/}
about the values of a physical quantity, rather than assigning a
specific value to the quantity itself.

As in Section \ref{Sec:ToposLogic}, the propositions concerned are
of the type $\SAin\De$, which (in a realist reading) asserts that
the value of the physical quantity $A$ lies in the (Borel) subset
$\De$ of the spectrum $\sp(\A)$ of the associated operator $\A$.
This proposition is represented by the spectral projector $\hat
E[A\in\De]$, which motivates studying the general mathematical
problem of assigning truth-values to projection operators.

If $\P$ is a projection operator, the identity $\P=\P^2$ implies
that $V(\P)=V(\P^2)=(V(\P))^2$ (from \eq{VAB}); and hence,
necessarily, $V(\P)=0$ or $1$. Thus $V$ defines a homomorphism
from the Boolean algebra $\{\hat 0, \hat 1, \P,\neg\P\equiv(\hat
1-\P)\}$ to the `false(0)-true(1)' Boolean algebra $\{0,1\}$. More
generally, a valuation $V$ induces a homomorphism
$\chi^V:W\rightarrow\{0,1\}$ where $W$ is any Boolean sub-algebra
of the lattice ${\cal P}(\Hi)$ of projectors on $\Hi$. In
particular,
\begin{equation}
    \hat\alpha\preceq\hat\beta\mbox{\ \ implies \ }
    \chi^V(\hat\alpha)\leq\chi^V(\hat\beta) \label{a<b->V(a)<V(b)}
\end{equation}
where `$\hat\a\preceq\hat\beta$' refers to the partial ordering in
the lattice ${\cal P}(\Hi)$, and
`$\chi^V(\hat\alpha)\leq\chi^V(\hat\beta)$' is the ordering in the
Boolean algebra $\{0,1\}$.

The Kochen-Specker theorem asserts that no global valuations exist
if the dimension of the Hilbert space $\Hi$ is greater than two.
The obstructions to the existence of such valuations typically
arise when trying to assign a single value to an operator $\hat C$
that can be written as $\hat C=g(\A)$ and as $\hat C=h(\hat B)$
with $[\A,\,\hat B]\neq 0$.

The various interpretations of quantum theory that aspire to use
`beables', rather than `observables', are all concerned in one way
or another with addressing this issue. Inherent in such schemes is
a type of `contextuality' in which a value given to a physical
quantity $C$ cannot be part of a global assignment of values but
must, instead, depend on some {\em context\/} in which $C$ is to
be considered. In practice, contextuality is endemic in any
attempt to ascribe properties to quantities in a quantum theory.
For example, as emphasized by Bell \cite{Bel66}, in the situation
where $\hat C=g(\hat A)=h(\hat B)$, if the value of $C$ is
construed counterfactually as referring to what would be obtained
{\em if\/} a measurement of $A$ or of $B$ is made---and with the
value of $C$ then being {\em defined\/} by applying to the result
of the measurement the relation $C=g(A)$, or $C=h(B)$---then one
can claim that the actual value obtained depends on whether the
value of $C$ is determined by measuring $A$, or by measuring $B$.

In the programme to be discussed here, the idea of a contextual
valuation will be developed in a different direction from that of
the existing modal interpretations in which `reality' is ascribed
to only some commutative subset of physical quantities. In
particular, rather than accepting such a limited domain of beables
we shall propose a theory of `generalised' valuations that are
defined globally on {\em all\/} propositions about values of
physical quantities. However, the price of global existence is
that any given proposition may have only a `{\em generalised\/}'
truth-value. More precisely, (i) the truth-value of a proposition
$\SAin\Delta$ belongs to a logical structure that is larger than
$\{0,1\}$; and (ii) these target-logics, and truth values, are
context-dependent.

It is clear that the main task is to formulate mathematically the
idea of a contextual, truth-value in such a way that the
assignment of generalised truth-values is consistent with an
appropriate analogue of the functional composition principle FUNC.

\subsubsection{The Introduction of Coarse-Graining}
In the original paper \cite{IB98}, this task is tackled using a
type of `coarse-graining' operation. The key idea is that,
although in a given situation in quantum theory it may not be
possible to declare a particular proposition $\SAin\De$ to be true
(or false), nevertheless there may be (Borel) functions $f$ such
that the associated propositions ``$f(A)\varin f(\De)$'' {\em
can\/} be said to be true.  This possibility arises for the
following reason.

Let $W_A$ denote the spectral algebra of the operator $\hat A$
that represents a physical quantity $A$. Thus $W_A$ is the Boolean
algebra of projectors $\hat E[A\in\De]$ that project onto the
eigenspaces associated with the Borel subsets $\De$ of the
spectrum $\sp(\A)$ of $\A$; physically speaking, $\hat
E[A\varin\De]$ represents the proposition $\SAin\De$.  It follows
from the spectral theorem that, for all Borel subsets $J$ of the
spectrum of $f(\A)$, the spectral projector $\hat E[f(A)\varin J]$
for the operator $f(\A)$ is equal to the spectral projector $\hat
E[A\varin f^{-1}(J)]$ for $\A$. In particular, if $f(\De)$ is a
Borel subset of $\sp(f(\A))$ then, since $\De\subseteq
f^{-1}(f(\De))$, we have $\hat E[A\varin\De]\preceq \hat E[A\varin
f^{-1}(f(\De))]$; and hence
\begin{equation}
    \hat E[A\varin\De]\preceq\, \hat E[f(A)\varin f(\De)]
                        \label{1EA<EfA}.
\end{equation}

Physically, the inequality in \eq{1EA<EfA} reflects the fact that
the proposition ``$f(A)\varin f(\Delta)$'' is generally weaker
than the proposition $\SAin\De$ in the sense that the latter
implies the former, but not necessarily vice versa. For example,
the proposition ``$f(A)=f(a)$'' is weaker than the original
proposition ``$A=a$'' if the function $f$ is many-to-one and such
that more than one eigenvalue of $\A$ is mapped to the same
eigenvalue of $f(\A)$. In general, we shall say that ``$f(A)\varin
f(\De)$'' is a {\em coarse-graining\/} of $\SAin\De$.

Now, if the proposition $\SAin\De$ is evaluated as `true' then,
from \eq{a<b->V(a)<V(b)} and \eq{1EA<EfA}, it follows that the
weaker proposition ``$f(A)\varin f(\De)$'' is also evaluated as
`true'.

This remark provokes the following observation. There may be
situations in which, although the proposition $\SAin\De$ cannot be
said to be either true or false, the weaker proposition
``$f(A)\varin f(\De)$'' can. In particular, if the latter {\em
can\/} be given the value `true', then---by virtue of the remark
above---it is natural to suppose that any further coarse-graining
to give an operator $g(f(\A))$ will yield a proposition
``$g(f(A))\in g(f(\De))$'' that will also be evaluated as `true'.
Note that there may be more than one possible choice for the
`initial' function $f$, each of which can then be further
coarse-grained in this way.  This multi-branched picture of
coarse-graining is one of the main justifications for our
invocation of the topos-theoretic idea of a presheaf.

It transpires that the key remark above is the statement:
\displayE{10}{3}{5}{
 If ``$f(A)\varin f(\De)$'' is true, then so
 is ``$g(f(A)\varin g(f(\De)$''
 for any function $g:\mathR\map\mathR$.}
This is key because  the property thus asserted can be restated by
saying that the collection of all functions $f$ such that
``$f(A)\varin f(\De)$'' is  `true' is a \emph{sieve}; and sieves
are closely associated with global elements of the sub-object
classifier in a category of presheaves.

To clarify this we start by defining a category $\cal O$ whose
objects are the bounded, self-adjoint operators on $\Hi$. For the
sake of simplicity, we will assume for the moment that $\cal O$
consists only of the operators whose spectrum is discrete. Then we
say that there is a `morphism' from $\hat B$ to $\A$ if there
exists a Borel function (more precisely, an equivalence class of
Borel functions) $f:\sp(\A)\map \mathR$ such that $\hat B=f(\A)$,
where $\sp(\A)$ is the spectrum of $\A$. Any such function on
$\sp(\A)$ is unique (up to the equivalence relation), and hence
there is at most one morphism between any two operators. If $\hat
B=f(\A)$, the corresponding morphism in the category $\cal O$ will
be denoted $f_{\cal O}: \hat B\map\A$. It then becomes clear that
the statement in the box above is equivalent to the statement that
the collection of all functions $f$ such that ``$f(A)\varin
f(\De)$'' is  `true', is a sieve\footnote{It is a matter of
convention whether this is called a sieve or a co-sieve.} on the
object $\A$ in the category $\cal O$.

This motivates very strongly looking at the topos category,
$\SetC{\cal O}$ of contravariant\footnote{Ab initio, we could just
as well have looked at covariant functors, but with our
definitions the contravariant ones are more natural.}, set-valued
functors on $\cal O$. Then, bearing in mind our discussion of
values of physical quantities, it is rather natural to construct
the following object in this topos:
\begin{definition}\label{Defn:spectral-presheaf}
The {\em spectral presheaf} on ${\cal O}$ is the contravariant
functor $\Sig:{\cal O} \map\Set$ defined as follows:
\begin{enumerate}
    \item On objects: $\Sig(\A):=\sp(\A)$.

    \item On morphisms: If $f_{{\cal O}}:\hat B\rightarrow \hat
A$, so that $\hat B=f(\hat A)$, then ${\bf\Sigma}(f_{{\cal
O}}):\sigma(\A)\rightarrow \sigma (\hat B)$ is defined by
${\bf\Sigma}(f_{{\cal O}})(\lambda):=f(\lambda)$ for all
$\lambda\in \sigma(\hat A)$.
\end{enumerate}
\end{definition}
Note that ${\bf\Si}(f_{{\cal O}})$ is well-defined since, if
$\lambda\in\sigma(\hat A)$, then $f(\lambda)$ is indeed an element
of the spectrum of $\hat B$; indeed, for these discrete-spectrum
operators we have $\sigma(f(\hat A))=f(\sigma(\hat A))$.

    The key remark now is the following. If $\cal C$ is any category,
a {\em global element\/}, of a contravariant functor $\ps{X}:{\cal
C}\map\Set$ is defined to be a function $\ga$ that assigns to each
object $A$ in the category $\cal C$ an element $\ga_A\in
\ps{X}(A)$ in such a way that if $f:B\map A$ then
$\ps{X}(f)(\ga_A)=\ga_B$ (see Appendix 2 for more details).

In the case of the spectral functor $\Sig$, a global element is
therefore a function $\ga$ that assigns to each (bounded, discrete
spectrum) self-adjoint operator $\A$, a real number $\ga_A\in
\sp(\A)$ such that if $\hat B=f(\A)$ then $f(\ga_A)=\ga_B$. But
this is precisely the condition FUNC in Eq.\ (\ref{funct-rule})
for a valuation!

\displayE{10}{3}{3}{Thus, the Kochen-Specker theorem is equivalent
to the statement that, if $\dim{\Hi}>2$, the spectral presheaf
$\Sig$ has no global elements.}

It was this observation that motivated the original suggestion by
one of us (CJI) and his collaborators that quantum theory should
be studied from the perspective of topos theory. However, as it
stands, the discussion above works only for operators with a
discrete spectrum. This is fine for finite-dimensional Hilbert
spaces, but in an infinite-dimensional space operators can have
continuous parts in their spectra, and then things get more
complicated.

One powerful way of tackling this problem is to replace the
category of operators with a category, $\V{}$, whose objects are
commutative von Neumann sub-algebras of the algebra $\BH$ of all
bounded operators on $\Hi$. There is a close link with the
category $\cal C$ since each self-adjoint operator generates a
commutative von Neumann algebra, but using $\V{}$ rather than
$\cal C$ solves all the problems associated with continuous
spectra \cite{IB00}.

Of course, this particular motivation for introducing $\V{}$ is
purely mathematical, but there are also very good physics reasons
for this step. As we have mentioned earlier, one approach to
handling the implications of the Kochen-Specker theorem is to
`reify' only a subset of physical variables, as is done in the
various `modal interpretations'. The topos-theoretic extension of
this idea of `partial reification', first proposed in \cite{IB98,
IB99, IB00, IB02}, is to build a structure in which all possible
reifiable sets of physical variables are included on an equal
footing. This involves constructing a category, $\cal C$,  whose
objects are collections of quantum observables that \emph{can} be
simultaneously reified because the corresponding self-adjoint
operators \emph{commute}. The application of this type of topos
scheme to an actual modal interpretation is discussed in the
recent paper by Nakayama \cite{Nakayama07}

From a physical perspective, the objects in the category $\cal C$
can be viewed as \emph{contexts} (or `world-views', or `windows on
reality', or `classical snapshots') from whose perspectives the
quantum  theory can be displayed. This is the physical motivation
for using commutative von Neumann algebras.

 In the normal, instrumentalist interpretation of quantum
theory, a context is therefore a collection of physical variables
that can be measured simultaneously. The physical significance of
this contextual logic is discussed at length in \cite{IB98, IB99,
IB00, IB02, IB00b} and \cite{DI(2),DI(3)}.

\subsubsection{Alternatives to von Neumann Algebras}

It should be remarked that $\V{}$ is not the only possible choice
for the category of concepts. Another possibility is to construct
a category whose objects are the Boolean sub-algebras of the
non-distributive lattice of projection operators on the Hilbert
space; more generally we could consider the Boolean sub-algebras
of \emph{any} non-distributive lattice. This option was discussed
in \cite{IB98}.

Yet another possibility is to consider the abelian
$C^*$-sub-algebras of the algebra $\BH$ of all bounded operators
on $\Hi$. More generally, one could consider the abelian
sub-algebras of \emph{any} $C^*$-algebra; this is the option
adopted by Heunen and Spitters \cite{HeuSpit07} in their very
interesting recent development of our scheme. One disadvantage of
a $C^*$-algebra is that does not contain projectors, and if  one
wants to include
 them it is necessary to move to
$AW*$-algebras, which are the abstract analogue of the concrete
von Neumann algebras that we employ. For each of these choices
there is a corresponding spectral object, and these different
spectral objects are closely related.

It is clear that a similar procedure could be followed for
\emph{any} algebraic quantity $\mathfrak{A}$ that has an
`interesting' collection of commutative sub-algebras. We will
return to this remark  in Section \ref{SubSubSec:AppAlgebra}.

\subsection{From Projections to Global Elements of the Outer
Presheaf}\label{_SubS_FromProjsToGlobSecs}
\subsubsection{The Definition of $\dasto{V}{P}$}\label{SubSub:defdas}
The fundamental thesis of our work is that in constructing
theories of physics one should seek representations of a formal
language in a topos that may be other than $\Set$. We  want now to
study this idea closely in the context of the `toposification' of
standard quantum theory, with particular emphasis on a topos
representation of propositions. Most `standard' quantum systems
(for example, one-dimensional motion with a Hamiltonian
$H=\frac{p^2}{2m}+V(x)$) are obtained by `quantising' a classical
system, and consequently the formal language is the same as it is
for the classical system. Our immediate goal is to represent
physical propositions with sub-objects of the spectral presheaf
$\Sig$.

In this Section we concentrate   on the propositional language
$\PL{S}$ introduced in Section \ref{SubSec:PropLangPhys}. Thus a
key task is to find the map $\piqt:\PL{S}_0\map \Sub{\Sig}$, where
the primitive propositions in $\PL{S}_0$ are of the form
$\SAin\De$. As we shall see, this is where the critical concept of
\emph{daseinisation} arises: the procedure whereby a projector
$\P$ is transformed to a
 sub-object, $\ps{\das{P}}$, of the spectral presheaf, $\Sig$, in
the topos $\SetH{}$ (the precise definition of $\Sig$   is given
in Section \ref{SubSebSec:GlobalElOtoSig}).

In standard quantum theory, a physical quantity is represented by
a self-adjoint operator $\hat{A}$ in the algebra, $\BH,$ of all
bounded operators on $\Hi$. If $\De\subseteq\mathR$ is a Borel
subset, we know from the spectral theorem  that the proposition
$\SAin{\De}$ is represented by\footnote{Note, however, that the
map from propositions to projections is not injective: two
propositions $\SAin{\De_1}$ and ``$B\varepsilon\De_{2}$''
concerning two distinct physical quantities, $A$ and $B$, can be
represented by the same projector: \ie\ $\hat E[A\in\De_1]=\hat
E[B\in\De_2]$.} the projection operator $\hat E[A\in\De]$ in
$\BH$. For typographical simplicity, for the rest of this Section,
$\hat E[A\in\De]$ will be denoted by $\P$.

We are going to consider the projection operator  $\P$ from the
perspective of the `category of contexts'---a keystone of the
topos approach to quantum theory. As we have remarked earlier,
there are several possible choices for this category   most of
which are considered in detail in the original papers
\cite{IB98,IB99,IB00,IB02}. Here we have elected to use the
category $\V{}$ of unital, abelian sub-algebras  of $\BH$. This
partially-ordered set has a  category structure in which (i) the
objects are the abelian sub-algebras of $\BH$; and (ii)  there is
an arrow $i_{V^\prime V}:V^\prime\map V$, where
$V^\prime,V\in\Ob{\V{}}$,\footnote{We denote by $\Ob{{\cal C}}$
the collection of all objects in the category $\cal C$.} if and
only if $V^\prime\subseteq V$.  By definition, the trivial
sub-algebra $V_0=\mathC\hat{1}$ is not included in the objects of
$\V{}$. A context could also be called a `world-view', a
`classical snap-shot',  a `window on reality', or even a
Weltanschauung\footnote{`Weltanschauung' is a splendid German
word. `Welt' means world; `schauen' is a verb and means to look,
to view; `anschauen' is to look at; and `-ung' at the end of a
word can make a noun from a verb. So it's Welt-an-schau-ung.};
mathematicians often refer to it as a `stage of truth'.

The critical question is what can be said about the projector $\P$
`from the perspective' of a particular context $V\in\Ob{\V{}}$? If
$\P$ belongs to $V$ then  a `full' image  of $\P$ is obtained from
this view-point, and there is nothing more to say. However,
suppose the abelian sub-algebra $V$ does \emph{not} contain $\P$:
what then?

We need to `approximate' $\P$ from the perspective of $V$, and an
important ingredient in our work is to define this as meaning the
`smallest' projection operator, $\dasto{V}{P}$, in $V$ that is
greater than, or equal to, $\P$:
\begin{equation}
\dasto{V}{P}:=\bigwedge\big\{\hat\a\in\PV\mid \hat\a\succeq
\P\big\}.           \label{Def:dasouter}
\end{equation}
where `$\succeq$' is the usual  ordering of projection operators,
and where $\PV$ denotes the set of all projection operators in
$V$.

To see what this means, let $\P$ and $\hat Q$ represent the
propositions $\SAin\De$ and $\SAin{\De^\prime}$ respectively with
$\De\subseteq\De^{\prime}$, so that $\P\preceq\hat Q$.
 Since we learn less about the value of
$A$ from the proposition $\SAin{\De^{\prime}}$ than from
$\SAin{\De}$, the former proposition is said to be \emph{weaker}.
Clearly, the weaker proposition $\SAin{\De^{\prime}}$ is implied
by the stronger proposition $\SAin{\De}$.  The construction of
$\dasto{V}{P}$ as the smallest projection in $V$ greater than or
equal to $\P$ thus gives the strongest proposition expressible in
$V$ that is implied by $\P$ (although, if $\hat A\notin V$, the
projection $\dasto{V}{P}$ cannot usually be interpreted as a
proposition \emph{about} $A$).\footnote{Note that the definition
in \eq{Def:dasouter} exploits the fact that the lattice $\PV$ of
projection operators in $V$ is complete. This is the main reason
why we chose von Neumann sub-algebras rather than $C^*$-algebras:
the former contain enough projections, and their projection
lattices are complete.} Note that if $\P$ belongs to $V$, then
$\dasto{V}{P}=\P$. The mapping $\P\mapsto\dasto{V}{P}$ was
originally introduced by de Groote in \cite{deG05},  who called it
the `$V$-support' of $\P$.

The key idea in this part of our  scheme is that rather than
thinking of a quantum proposition, $\SAin\De$,  as being
represented by the \emph{single} projection operator $\hat
E[A\in\De]$, instead we consider the entire \emph{collection}
$\{\delta\big(\hat E[A\in\De]\big)_V\mid V\in\Ob{\V{}}\}$ of
projection operators, one for each context $V$. As we will see,
the link with topos theory is that this collection of projectors
is a global element of a certain presheaf.

This `certain' presheaf is in fact the `outer' presheaf, which is
defined as follows:

{\definition The \emph{outer\footnote{In  the original papers by
CJI\ and collaborators, this was called the `coarse-graining'
presheaf, and was denoted $\ps{G}$. The reason for the change of
nomenclature will become apparent later.} presheaf} $\G$ is
defined over the category $\V{}$ as follows \cite{IB98,IB00}:
\begin{enumerate}
\item[(i)] On objects
$V\in\Ob{\V{}}$: We have  $\G_V:=\PV$

\item[(ii)] On morphisms $i_{V^{\prime}V}:V^{\prime
}\subseteq V:$ The mapping $\G(i_{V^{\prime} V}):\G_V
\map\G_{V^{\prime}}$ is given by
$\G(i_{V^{\prime}V})(\hat{\alpha}):=\dasto{V^\prime}{\alpha}$ for
all $\hat{\alpha}\in\PV$.
\end{enumerate}
}

With this definition, it is clear that, for each projection
operator $\P$, the assignment $V\mapsto \dasto{V}{P}$ defines a
global element of the presheaf $\G$. Indeed, for each context $V$,
we have the projector $\dasto{V}{P}\in\PV=\G_V$, and if
$i_{V^{\prime}V}:V^{\prime }\subseteq V$, then
\begin{equation}
\delta\big(\dasto{V}{P}\big)_{V^\prime}
=\bigwedge\big\{\hat{Q}\in\mathcal{P}(V^{\prime})\mid
\hat{Q}\succeq \dasto{V}{P}\big\}=\dasto{V^\prime}{P}
\end{equation}
and so the elements $\dasto{V}{P}$, $V\in\Ob{\V{}}$, are
compatible with the structure of the outer presheaf. Thus we have
a mapping
\begin{eqnarray}
\dasmap:\PH  &\map&\Ga\G                        \nonumber\\
\P  &  \mapsto&\{\dasto{V}{P}\mid {V\in\Ob{\V{}}\}}
\end{eqnarray}
from the projectors in $\PH$ to the global elements, $\Ga\G$, of
the outer presheaf.\footnote{Vis-a-vis our use of the language
$\L{S}$ a little further on, we should  emphasise that the outer
presheaf has no linguistic precursor, and  in this sense, it has
no fundamental status in the theory. In fact, we could avoid the
outer presheaf altogether and always work directly with the
spectral presheaf, $\Sig$, which, of course, \emph{does} have a
linguistic precursor. However, it is technically convenient to
introduce the outer  presheaf as an intermediate tool.}

\subsubsection{Properties of the Mapping $\dasmap:\PH\rightarrow\Ga$\underline O.}
\label{SubSubSec:PropDas} Let us now note some properties of the
map $\dasmap:\PH\map\Ga\G$ that are relevant to our overall
scheme.
\begin{enumerate}
\item For all contexts $V$, we have $\dasto{V}{0}=\hat0$.

The null projector represents all propositions of the form
$\SAin\De$ with the property that $\sp(\A)\cap\De=\varnothing$.
These propositions are trivially false.

\item For all contexts $V$, we have $\dasto{V}{1}=\hat 1$.

The unit operator $\hat 1$ represents all propositions of the form
$\SAin\De$  with the property that $\sp(\A)\cap\De=\sp(\A)$. These
propositions are trivially true.

\item There exist global elements  of $\G$ that are \emph{not}
of the form $\dasto{}{P}$ for any  projector $\P$. This phenomenon
will be discussed later. However, if $\ga\in\Ga\G$ \emph{is} of
the form $\das{P}$ for some $\P$, then
\begin{equation}\label{ProjIsMinOfDas(P)}
                \P=\bigwedge_{V\in\Ob{\V{}}}\dasto{V}{P},
\end{equation}
because $\dasto{V}{P}\succeq \P$ for all $V\in\Ob{\V{}}$, and
$\dasto{V}{P}=\P$ for any  $V$ that contains $\P$.
\end{enumerate}

The next result is important as it means that `nothing is lost' in
mapping a projection operator $\P$ to its associated global
element, $\das{P}$, of the presheaf $\G$.

\displayE{10.0}{-5}{0}{\begin{theorem} The map
$\dasmap:\PH\map\Ga\G$ is injective.
\end{theorem}}

This simply follows from (\ref{ProjIsMinOfDas(P)}): if
$\das{P}=\das{Q}$ for two projections $\P,\hat Q$, then
\begin{equation}
        \P=\bigwedge_{V\in\Ob{\V{}}}\dasto{V}{P}
        =\bigwedge_{V\in\Ob{\V{}}}\dasto{V}{Q}=\hat Q.
\end{equation}

\subsubsection{A Logical Structure for $\Ga$\G?}
\label{SubSec:LogStructG} We have seen that the quantities
$\das{P}:=\{\dasto{V}{P}\mid V\in\Ob{\V{}}\}$, $\P\in\PH$, are
elements of $\Ga\G$, and if they are to represent  quantum
propositions, one might expect/hope that (i) these global elements
of $\G$ form a Heyting algebra; and (ii) this algebra is related
in some way to the Heyting algebra of sub-objects of $\Sig$. Let
us see how far we can go in this direction.

Our first remark is that any two global elements $\ga_{1},\ga_{2}$
of $\G$ can be compared at each stage $V$ in the sense of logical
implication. More precisely, let $\ga_{1}{}_{V}\in\PV$ denote the
$V$'th `component' of $\ga_1$, and ditto for $\ga_{2}{}_{V}$. Then
we have the following result:
\begin{definition}
A partial ordering on $\Ga\G$ can be constructed in a `local' way
(\ie\ `local' with respect to the objects in the category $\V{}$)
by defining
\begin{equation}
\ga_{1}\succeq\ga_{2} \mbox{  if, and only if, } \forall
V\in\Ob{\V{}},\ \ga_{1}{}_V\succeq\ga_{2}{}_V \label{Def:g1>g2}
\end{equation}
where the ordering on the  right hand side of \eq{Def:g1>g2} is
the usual ordering in the lattice of projectors $\PV$.
\end{definition}
It is trivial to check that \eq{Def:g1>g2} defines a partial
ordering on $\Ga\G$. Thus $\Ga\G$ is a partially ordered set.

Note that if $\P,\hat Q$ are projection operators, then it follows
from \eq{Def:g1>g2} that
\begin{equation}
\das{P}\succeq\das{Q} \mbox{ if and only if } \P\succeq\hat Q
\end{equation}
since $\P\succeq\hat Q$ implies $\dasto{V}{P}\succeq \dasto{V}{Q}$
for all contexts $V$.\footnote{On the other hand, in general,
$\P\succ \hat Q$ does not imply $\dasto{V}{P}\succ \dasto{V}{Q}$
but only $\dasto{V}{P}\succeq \dasto{V}{Q}$.} Thus the mapping
$\delta:\PH\map\Ga\G$ respects the partial order.

The next thing is to see if a  logical $`\lor$'-operation can be
defined on $\Ga\G$. Once again, we try a `local' definition:
\begin{theorem}
A `$\lor$'-structure on $\Ga\G$ can be defined locally by
\end{theorem}
\begin{equation}
(\ga_1\lor\ga_2)_V:=\ga_1{}_V\lor\ga_2{}_V
        \label{Def:g1lorg2}
\end{equation}
for all $\ga_1,\ga_2\in\Ga\G$, and for all $V\in\Ob{\V{}}$.

\begin{proof}
It is not instantly clear that \eq{Def:g1lorg2} defines a global
element of $\G$. However, a key result in this direction is the
following:
\begin{lemma}
For each context $V$, and for all $\hat\alpha,\hat\beta\in\PV$, we
have
\begin{equation}
\G(i_{V^{\prime}V})(\hat{\alpha}\lor\hat\beta)=
\G(i_{V^{\prime}V})(\hat{\alpha})\lor
\G(i_{V^{\prime}V})(\hat{\beta}) \label{G(alorb)=}
\end{equation}
for all contexts $V^\prime$ such that  $V^\prime\subseteq V$.
\end{lemma}
The proof  is a straightforward consequence of the definition of
the presheaf $\G$.

One immediate consequence is that \eq{Def:g1lorg2} defines a
global element\footnote{The existence of the $\lor$-operation on
$\Ga\G$ can be extended to $\G$ itself. More precisely, there is
an arrow $\lor:\G\times \G\map\G$ where $\G\times\G$ denotes the
product presheaf over $\V{}$, whose objects are
$(\G\times\G)_V:=\G_V\times\G_V$. Then the arrow $\lor:\G\times
\G\map\G$ is defined at any context $V$ by
$\lor_{V}(\hat\alpha,\hat\beta):=\hat\alpha\lor\hat\beta$ for all
$\hat\alpha,\hat\beta\in\G_V$.} of $\G$. Hence the theorem is
proved.\end{proof}

\

It is also straightforward to  show that, for any pair of
projectors $\P,\hat Q\in\PH$, we have $\delta(\P\lor \hat
Q)_{V}=\dasto{V}{P}\vee\dasto{V}{Q}$, for all contexts
$V\in\Ob{\V{}}$. This means that, as elements of $\Ga\G$,
\begin{equation}
\delta(\P\vee \hat Q)=\das{P}\vee\das{Q}.
\end{equation}
Thus the mapping $\delta:\PH\map\Ga\G$ preserves the logical
`$\lor$' operation.

However, there is no analogous equation for the logical
`$\land$'-operation. The obvious local definition would be, for
each context $V$,
\begin{equation}
(\ga_1\land\ga_2)_V:=\ga_1{}_V\land\ga_2{}_V\label{Def:g1landg2}
\end{equation}
but this does not define a global element of $\G$ since, unlike
\eq{G(alorb)=}, for the $\land$-operation we have only
\begin{equation}
\G(i_{V^{\prime}V})(\hat{\alpha}\land\hat\beta)\;\preceq\;
\G(i_{V^{\prime}V})(\hat{\alpha})\land
\G(i_{V^{\prime}V})(\hat{\beta}) \label{G(alandb)=}
\end{equation}
for all $V^\prime\subseteq V$. As a consequence, for all $V$, we
have only the inequality
\begin{equation}
\delta(\P\land\hat Q)_{V}\;\preceq\;\dasto{V}{P}\land\dasto{V}{Q}
\label{delta(PlandQ)}
\end{equation}
and hence
\begin{equation}
\delta(\P\land\hat Q)\;\preceq\;\das{P}\land\das{Q}.
\end{equation}

It is easy to find examples where the inequality is strict. For
example, let $\P\neq \hat 0,\hat 1$ and $\hat Q=\hat 1-\P$. Then
$\P\land \hat Q=0$ and hence $\delta_{V}(\P\land \hat Q)=\hat 0$,
while $\dasto{V}{P}\land\dasto{V}{Q}$ can be strictly larger than
$\hat 0$, since $\dasto{V}{P}\succeq \P$ and $\dasto{V}{Q}\succeq
\hat Q$.

\subsubsection{Hyper-Elements of $\Ga$\G.}
We have seen that the global elements of $\G$, \ie\ the elements
of $\Ga\G$, can be equipped with a partial-ordering and a
`$\lor$'-operation, but attempts to define a `$\land$'-operation
in the same way fail because of the inequality in
\eq{delta(PlandQ)}.

However, the form of \eqs{G(alandb)=}{delta(PlandQ)} suggests the
following procedure. Let us define a \emph{hyper-element} of $\G$
to be an association, for each stage $V\in\Ob{\V{}}$, of an
element $\ga_V\in\G_V$ with the property that
\begin{equation}
         \ga_{V^\prime}\;\succeq\; \G(i_{V^\prime V})(\ga_V)
         \label{gammaVpsucceqG}
\end{equation}
for all $V^\prime\subseteq V$. Clearly every element of $\Ga\G$ is
a hyper-element, but not conversely.

Now, if $\ga_1$ and $\ga_2$ are hyper-elements, we can define the
operations `$\lor$' and  `$\land$' locally as:
\begin{eqnarray}
   (\ga_1\lor\ga_2)_V&:=&\ga_1{}_V\lor\ga_2{}_V\\[2pt]
   (\ga_1\land\ga_2)_V&:=&\ga_1{}_V\land\ga_2{}_V
\end{eqnarray}

Because of \eq{G(alandb)=} we have, for all $V^\prime\subseteq V$,
\begin{eqnarray}
\G(i_{V^{\prime}V})\big((\ga_1\land\ga_2)_V\big)
     &=&\G(i_{V^{\prime}V})\big(\ga_1{}_V\land \ga_2{}_V\big)\\
        &\preceq&\G(i_{V^{\prime}V})(\ga_1{}_V)\land
                               \G(i_{V^{\prime}V})(\ga_2{}_V)\\
      &\preceq&\ga_1{}_{V^{\prime}}\land\ga_2{}_{V^{\prime}}\\
        &=&(\ga_1\land\ga_2)_{V^{\prime}}
\end{eqnarray}
so that the hyper-element condition \eq{gammaVpsucceqG} is
preserved.

The occurrence of a logical `$\lor$' and $\land$' structure is
encouraging, but it is not yet what we want. For one thing, there
is no mention of a negation operation; and, anyway, this is not
the expected algebra of sub-objects of a `state space' object. To
proceed further we must study more carefully the sub-objects of
the spectral presheaf.

\subsection{Daseinisation:\ Heidegger Encounters Physics}
\subsubsection{From Global Elements of \G\ to Sub-Objects of \underline{$\Sigma$}.}
\label{SubSebSec:GlobalElOtoSig} The spectral presheaf, $\Sig$,
played a central role in the earlier discussions of quantum theory
from a topos perspective \cite{IB98,IB99,IB00,IB02}. Here is the
formal definition.

{\definition\label{Def_SpectralPresheaf} The \emph{spectral
presheaf}, $\Sig$, is defined as the following functor from
$\V{}^\op$ to $\Set$:
\begin{enumerate}
\item On objects $V$:  $\Sig_V$ is the Gel'fand spectrum of the unital, abelian
sub-algebra $V$ of $\BH$; \ie\  the set of all multiplicative
linear functionals $\l:V\map\mathC$ such that $\brak{\l}{\hat
1}=1$.

\item On morphisms $i_{V^{\prime}V}:V^\prime\subseteq V$:
$\Sig(i_{V^{\prime}V}):\Sig_V\map \Sig_{V^\prime}$ is defined by
$\Sig(i_{V^{\prime}V})(\l):= \l|_{V^\prime}$; \ie\ the restriction
of the functional $\l:V\map\mathC$ to the sub-algebra
$V^\prime\subseteq V$.
\end{enumerate}
} \noindent One central result of spectral theory is that $\Sig_V$
has a topology that is compact and Hausdorff, and with respect to
which the Gel'fand transforms\footnote{If $\hat A\in V$, the
Gel'fand transform, $\GT{A}:\Sig_V\map\mathC$, of $\A$ is defined
by $\GT{A}(\l):=\la\l,\A\ra$ for all $\l\in\Sig_V$.} of the
elements of $V$  are continuous functions from $\Sig_V$ to
$\mathC$. This will be important in what follows \cite{KR83a}.

The spectral presheaf plays a fundamental role in our research
programme as applied to quantum theory. For example, it was shown
in the earlier work that the Kochen-Specker theorem \cite{KS67} is
equivalent to the statement that $\Sig$ has no global elements.
However, $\Sig$ \emph{does} have sub-objects, and these are
central to our scheme: {\definition A \emph{sub-object} $\ps{S}$
of the spectral presheaf $\Sig$ is a functor $\ps{S}:\V{}
^{op}\map\Set$ such that
\begin{enumerate}
\item$\ps{S}_V$ is a subset of $\Sig_V$ for all $V$.

\item  If $V^{\prime}\subseteq V$, then
$\ps{S}(i_{V^{\prime}V}):\ps{S}_V\map \ps{S}_{V^{\prime}}$ is just
the restriction $\l\mapsto\l|_{V^{\prime}}$ (\ie\ the same as for
$\Sig$), applied to the elements $\l\in \ps{S}_V\subseteq\Sig_V$.
\end{enumerate}
}

This definition of a sub-object is standard. However, for our
purposes we need something slightly different, namely  concept of
a `clopen' sub-object. This is defined to be a sub-object $\ps{S}$
of $\Sig$ such that, for all  $V$, the set $\ps{S}_V$ is a
\emph{clopen}\footnote{A `clopen' subset of a topological space is
one that is both open and closed.} subset of the compact,
Hausdorff space $\Sig_V$. We denote by $\Subcl\Sig$ the set of all
clopen sub-objects of $\Sig$. We will show later (in the Appendix)
that, like $\Sub\Sig$, the set $\Subcl\Sig$ is a Heyting algebra.
In Section \ref{SubSec:PresheafP-clSig} we show that there is an
object $\PSig$ whose global elements are precisely the clopen
sub-objects of $\Sig$.

This interest in clopen sets is easy to explain. For, according to
the Gel'fand spectral theory, a projection operator
$\hat\a\in\mathcal{P}(V)$  corresponds to a unique clopen subset,
$S_{\hat\a}$ of the Gel'fand spectrum, $\Sig_V$. Furthermore, the
Gel'fand transform $\GT\alpha:\Sig_V\map\mathC$ of $\hat\alpha$
takes the values $0,1$ only, since the spectrum of a projection
operator is just $\{0,1\}$.

It follows that $\GT{\alpha}$ is the characteristic function of
the subset, $S_{\hat\alpha}$, of $\Sig_V$, defined by
\begin{equation}
S_{\hat\a}:=\{\l\in\Sig_V \mid\brak{\l}{\hat\a}=1\}.
\label{Def:Salphahat}
\end{equation}
The clopen nature of $S_{\hat\a}$  follows from the fact that, by
the spectral theory, the function $\GT{\alpha}:\Sig_V\map\{0,1\}$
is continuous.

In fact, there is a lattice isomorphism between the lattice $\PV$
of projectors in $V$ and the lattice $\mathcal{C}L(\Sig_{V})$ of
clopen subsets of $\Sig_{V}$,\footnote{The lattice structure on
$\mathcal{C}L(\Sig_{V})$ is defined as follows: if $(U_i)_{i\in
I}$ is an arbitrary family of clopen subsets of $\Sig_{V}$, then
the \emph{closure} $\overline{\bigcup_{i\in I}U_i}$ is the
maximum. The closure is necessary since the union of infinitely
many closed sets need not be closed. The \emph{interior}
$\operatorname{int}\bigcap_{i\in I}U_i$ is the minimum of the
family. One must take the interior since $\bigcap_{i\in I}U_i$ is
closed, but not necessarily open.} given by
\begin{equation}
    \hat{\a}\mapsto S_{\hat{\a}}:=
\{ \l\in\Sig_{V}\mid\brak\l{\hat{\a}}=1\}.
\label{LatticeIsomProjsAndClopenSubsets}
\end{equation}
Conversely, given a clopen subset $S\in\mathcal{C}L(\Sig_{V})$, we
get the corresponding projection $\hat{\alpha}$ as the (inverse
Gel'fand transform of the) characteristic function of $S$. Hence,
each $S\in\mathcal{C}L(\Sig_{V})$ is of the form
$S=S_{\hat{\alpha}}$ for some $\hat{\alpha}\in\PV$.

Our claim is the following:
\begin{theorem}
For each projection operator $\P\in\PH$, the collection
\begin{equation}
\ps{\das{P}}:=\{S_{\dasto{V}{P}}\subseteq\Sig_V\mid
V\in\Ob{\V{}}\}
\end{equation}
forms a (clopen) \emph{sub-object} of the spectral presheaf
$\Sig$.
\end{theorem}
\begin{proof} To see this, let $\l\in S_{\dasto{V}{P}}$.
Then if $V^{\prime}$ is some abelian sub-algebra of $V$, we have
$\dasto{V^{\prime}}{P}
=\bigwedge\big\{\hat{\a}\in\mathcal{P}(V^{\prime})\mid
\hat{\a}\succeq \dasto{V}{P}\big\}\succeq \dasto{V}{P}$. Now let
$\hat{\a}:=\dasto{V^{\prime}}{P}-\dasto{V}{P}$. Then
$\brak{\l}{\dasto{V^{\prime}}{P}}
=\brak{\l}{\dasto{V}{P}}+\brak\l{\hat{\a}}=1$, since
$\brak\l{\dasto{V}{P}}=1$ and $\brak\l{\hat{\a}}\in\{0,1\}$. This
shows that
\begin{equation}
\{\l|_{V^{\prime}}\mid \l\in S_{\dasto{V}{P}}\}\subseteq
S_{\dasto{V^{\prime}}{P}}.\label{ldum1}
\end{equation}
However, the left hand side of \eq{ldum1} is the subset
$\G(i_{V^{\prime}V})(S_{\dasto{V}{P}})\subseteq\Sig_{V^\prime}$ of
the outer-presheaf restriction of elements in $S_{\dasto{V}{P}}$
to $\Sig_{V^{\prime}}$, and the restricted elements all lie in
$S_{\dasto{V^{\prime}}{P}}$. It follows that the collection of
sets
\begin{equation}
\ps{\das{P}}:=\{S_{\dasto{V}{P}}\subseteq\Sig_V\mid{V\in\Ob{\V{}}}\}
\end{equation}
forms a (clopen) sub-object of the spectral presheaf $\Sig$.
\end{proof}

By these means we have constructed a mapping
\begin{eqnarray}
\dasmap:\PH &\longrightarrow&\Subcl{\Sig}\nonumber\\
\P  &  \mapsto& \ps{\das{P}}:=\{S_{\dasto{V}{P}} \mid
V\in\Ob{\V{}}\} \label{delP(H)toSub}
\end{eqnarray}
which sends projection  operators on $\Hi$ to clopen sub-objects
of $\Sig$. As a matter of notation, we will denote the clopen
subset $S_{\dasto{V}{P}}\subseteq\Sig_V$ as $\ps{\das{P}}_V$. The
notation $\das{P}_V$ refers to the element (\ie\ projection
operator) of $\G_V$ defined earlier.

\subsubsection{The Definition of Daseinisation}
As usual, the projection $\P$ is regarded  as representing a
proposition about the quantum system. Thus $\delta$  maps
propositions about a quantum system to (clopen) sub-objects of the
spectral presheaf. This is strikingly analogous to the situation
in classical physics, in which propositions are represented by
subsets of the classical state space.

\displayE{12}{-7}{4} {\definition \noindent The map $\delta$ in
\eq{delP(H)toSub} is a fundamental part of our constructions. We
call it the \emph{daseinisation} of $\P$. We shall  use the same
word to refer to the operation in \eq{Def:dasouter} that relates
to the outer presheaf.}

The expression `daseinisation' comes from the German word
\emph{Dasein}, which plays a central role in Heidegger's
existential philosophy. Dasein translates to `existence' or, in
the very literal sense often stressed by Heidegger, to
being-there-in-the-world\footnote{The hyphens are \emph{very}
important.}. Thus daseinisation
`brings-a-quantum-property-into-existence'\footnote{The hyphens
are \emph{very} important.} by  hurling it into the collection of
all possible classical snap-shots of the world provided by the
category of contexts. \

We will summarise here some useful properties of daseinisation.
\begin{enumerate}
    \item The null projection $\hat 0$ is mapped to the empty
     sub-object of $\Sig$:
\begin{equation}
     \ps{\delta(\hat 0)}=\{\varnothing_{V}\mid
     V\in\Ob{\V{}}\}\ \ \ \ \
\end{equation}

    \item The identity projection $\hat 1$ is mapped to the
    unit sub-object of $\Sig$:
\begin{equation}
    \ps{\delta(\hat 1)}=
    \{\Sig_V\mid V\in\Ob{\V{}}\}=\Sig
\end{equation}

    \item Since the daseinisation map $\dasmap:\PH\map\Ga\G$
    is injective (see Section \ref{SubSubSec:PropDas}), and the
    mapping $\Ga\G\map\Ga(\PSig)$ is injective (because there
    is a monic arrow $\G\map\PSig$ in $\SetH{}$; see
    Section \ref{SubSec:MonicGPSig}), it follows that the
    daseinisation map
    $\dasmap:\PH\map\Ga(\PSig)\simeq\Subcl\Sig$
    is also injective. Thus no information about the projector
    $\P$ is lost when it is daseinised to become $\ps{\das{P}}$.
\end{enumerate}

\subsection{The Heyting Algebra Structure on $\operatorname{Sub}_{\operatorname{cl}}$(\underline{$\Sigma$}).}
The reason for daseinising projections is that the set,
$\Sub{\Sig}$, of sub-objects of the spectral presheaf forms a
\emph{Heyting algebra}. Thus the idea is to find a map
$\piqt:\PL{S}_0\map\Sub{\Sig}$ and then extend it to all of
$\PL{S}$ using the simple recursion ideas discussed in Section
\ref{SubSubSec:RepPLS}.

In our case, the act of daseinisation gives  a map from the
projection operators to the clopen sub-objects of $\Sub\Sig$, and
therefore  a map $\piqt:\PL{S}_0\map\Subcl{\Sig}$ can be defined
by
\begin{equation}
   \piqt(\Ain\De):=\ps{\delta\big(\hat E[A\in\De]\big)}
\label{Def:pi(AinDelta)}
\end{equation}
However,  to extend this definition to $\PL{S}$, it is necessary
to show that the set of clopen sub-objects, $\Subcl{\Sig}$, is a
Heyting algebra. This is not completely obvious from the
definition alone. However, it is true, and the proof is given in
Theorem \ref{Th:SubclSig} in the Appendix.

In conclusion: daseinisation can be used to give a
representation/model of the language $\PL{S}$ in the Heyting
algebra $\Subcl\Sig$.\footnote{Since the clopen subobjects of
$\Sig$ correspond bijectively to the global sections of the outer
presheaf $\G$, it is clear that $\Ga\G$ too is a Heyting algebra.}

\subsection{Daseinisation and the Operations of Quantum Logic.}
It is interesting to ask to what extent the map
$\delta:\PH\map\Subcl\Sig$ respects the lattice structure on
$\PH$. Of course, we know that it cannot be \emph{completely}
preserved since the quantum logic $\PH$ is non-distributive,
whereas $\Subcl\Sig$ is a Heyting algebra, and hence distributive.

We saw in Section \ref{SubSec:LogStructG} that, for the mapping
$\delta:\PH\map\Ga\G$, we have
\begin{eqnarray}
  \delta(\P\lor\hat{Q})_{V}&=&\dasto{V}{P}\lor\dasto{V}{Q},\\
 \delta(\P\land\hat{Q})_{V}&\preceq&\dasto{V}{P}\land\dasto{V}{Q}
\end{eqnarray}
for all contexts $V$ in $\Ob{\V{}}$.

The clopen subset of $\Sig_V$ that corresponds  to
$\dasto{V}{P}\lor\dasto{V}{Q}$ is $S_{\dasto{V}{P}}\cup
S_{\dasto{V}{Q}}$. This implies that the daseinisation map
$\delta:\PH\map\Subcl\Sig$ is a morphism of $\lor$-semi-lattices.

On the other hand, $\dasto{V}{P}\land\dasto{V}{Q}$ corresponds to
the subset $S_{\dasto{V}{P}}\cap S_{\dasto{V}{Q}}$ of $\Sig_V$.
Therefore, since $S_{\delta(\P\wedge \hat{Q})_{V}}\subseteq
S_{\dasto{V}{P}}\cap S_{\dasto{V}{Q}}$, daseinisation is
\emph{not} a morphism of $\wedge$-semi-lattices. In summary, for
all projectors $\P,\hat Q$ we have
\begin{eqnarray}
\ps{\delta(\P\lor\hat Q)}&=&\ps{\das{P}}\lor\ps{\das{Q}}
                                            \label{del(PorQ)}\\[3pt]
\ps{\delta(\P\land\hat Q)}&\preceq&\ps{\das{P}}\land\ps{\das{Q}}
                                            \label{del(PandQ)}
\end{eqnarray}
where the logical connectives on the left hand side lie in the
quantum logic $\PH$, and those on the right hand side lie in the
Heyting algebra $\Subcl\Sig$, as do the symbols `$=$' and
`$\preceq$'.

As remarked above, it is not surprising that  \eq{del(PandQ)} is
not an equality. Indeed, the quantum logic $\PH$ is
non-distributive, whereas the Heyting algebra $\Subcl\Sig$
\emph{is} distributive, and so it would be impossible for both
\eq{del(PorQ)} and \eq{del(PandQ)} to be equalities. The
inequality in \eq{del(PandQ)} is the price that must be paid for
liberating the projection operators from the shackles of quantum
logic and transporting  them to the existential world of Heyting
algebras.

\subsubsection{The Status of the Possible Axiom
`$\Ain{\De_1}\land \Ain{\De_2}\Leftrightarrow
                     \Ain{\De_1\cap\De_2}$'}
We have the representation in \eq{Def:pi(AinDelta)},
$\piqt(\Ain\De):=\ps{\delta\big(\hat E[A\in\De]\big)}$, of the
primitive propositions $\Ain\De$, and, as explained in Section
\ref{SubSubSec:RepPLS}, this can be extended to compound sentences
by making the obvious definitions:
\begin{eqnarray}
&(a)&\ \piqt(\alpha\lor\beta):=\piqt(\alpha)\lor\piqt(\beta)
                                                \label{piqt(a)}\\
&(b)&\ \piqt(\alpha\land\beta):=\piqt(\alpha)\land\piqt\beta)
                                                \label{piqt(b)}\\
&(c)&\ \piqt(\neg\alpha):=\neg\piqt(\alpha)\hspace{3cm}
                                                \label{piqt(c)}\\
&(d)&\ \piqt(\alpha\Rightarrow\beta):=
        \piqt(\alpha)\Rightarrow\piqt(\beta)    \label{piqt(d)}
\end{eqnarray}

As a result, we necessarily get a representation of the full
language $\PL{S}$ in the Heyting algebra $\Subcl\Sig$. However, we
then find that:
\begin{eqnarray}
        \piqt(\Ain{\De_1}\land\Ain{\De_2})&:=&
 \piqt(\Ain{\De_1})\land\piqt(\Ain{\De_2})\label{D1andD2a} \\
&=&\ps{\delta(\hat E[A\in\De_1])}\land\ps{\delta(\hat
E[A\in\De_2])}
                                        \label{D1andD2b}\\
&\succeq& \ps{\delta(\hat E[A\in\De_1]\land \hat E[A\in\De_2])}
                                         \label{D1andD2c}\\
&=&\ps{\delta(\hat E[A\in\De_1\cap\De_2)])} \label{D1andD2d}\\
&=&\piqt(\Ain{\De_1\cap\De_2})
\end{eqnarray}
where, \eq{D1andD2c} comes from \eq{del(PandQ)}, and in
\eq{D1andD2d} we have used the property of spectral projectors
that $\hat E[A\in\De_1]\land \hat E[A\in\De_2]= \hat
E[A\in\De_1\cap\De_2)]$. Thus, although by definition,
$\piqt(\Ain{\De_1}\land\Ain{\De_2})=
          \piqt(\Ain{\De_1})\land\piqt(\Ain{\De_2})$,
we only have the inequality
\begin{equation}
\piqt(\Ain{\De_1\cap\De_2})\preceq
\piqt(\Ain{\De_1}\land\Ain{\De_2})
\end{equation}

On the other hand, the same line of argument shows that
\begin{equation}
  \piqt(\Ain{\De_1}\lor\Ain{\De_2})=
                \piqt(\Ain{\De_1\cup\De_2})
\end{equation}
Thus it would be consistent to add the axiom
\begin{equation}
        \Ain{\De_1}\lor \Ain{\De_2}\Leftrightarrow
                     \Ain{\De_1\cup\De_2}
\end{equation}
to the language $\PL{S}$, but not
\begin{equation}
        \Ain{\De_1}\land \Ain{\De_2}\Leftrightarrow
                     \Ain{\De_1\cap\De_2}
\end{equation}
Of, course, both axioms are consistent with the representation of
$\PL{S}$ in classical physics.

It should be emphasised that there is nothing wrong with this
result: indeed, as stated above, it is the necessary price to be
paid for forcing a non-distributive algebra to have a
`representation' in a Heyting algebra.

\subsubsection{Inner Daseinisation and $\delta(\neg\P)$.}
In the same spirit, one might ask about ``$\neg(\Ain{\De})$''. By
definition, as in \eq{pi(d)}, we have
$\piqt(\neg(\Ain{\De})):=\neg\piqt(\Ain\De)
=\neg\ps{\delta\big(\hat E[A\in\De]\big)}$. However, the question
then is how, if at all, this is related  to $\ps{\delta(\hat
E[A\in\mathR/\De])}=\ps{\delta(\neg\hat E[A\in\De])}$, bearing in
mind the axiom
\begin{equation}
 \neg(\Ain\De) \Leftrightarrow \Ain{\mathR\backslash\De}
\end{equation}
that can be added to the classical representation of $\PL{S}$.
Thus something needs to be said about $\ps{\delta(\neg\P)}$, where
$\neg\P=\hat 1-\P$ is the negation operation in the quantum logic
$\PH$.

To proceed further, we need to introduce another operation:
{\definition  The \emph{inner daseinisation}, $\delta^i(\P)$, of
$\P$ is defined for each context $V$ as
\begin{equation}
\dastoi{V}{P}:=\bigvee\big\{\hat{\beta}\in\mathcal{P}(V)\mid
\hat{\beta}\preceq \P\big\}. \label{Def:dasinner}
\end{equation}
} This should be contrasted with the definition of outer
daseinisation in  \eq{Def:dasouter}.

Thus $\dastoi{V}{P}$ is the best approximation that can be made to
$\P$ by taking the `largest' projector in $V$ that implies $\P$.

As with the other daseinisation construction, this operation  was
first introduced by de Groote in \cite{deG05} where he called it
the \emph{core} of the projection operator $\P$. We prefer to use
the phrase `inner daseinisation', and then to refer to
\eq{Def:dasouter} as the `outer daseinisation' operation on $\P$.
The existing notation $\dasto{V}{P}$ will be replaced with
$\dastoo{V}{P}$ if there is any danger of confusing the two
daseinisation operations.

With the aid of inner daseinisation, a new presheaf, $\H$, can be
constructed as an exact analogue of the outer presheaf, $\G$,
defined in Section \ref{SubSub:defdas}. Specifically:

{\definition The \emph{inner presheaf} $\H$ is defined over the
category $\V{}$ as follows:
\begin{enumerate}
\item[(i)] On objects $V\in\Ob{\V{}}$:  We have
$\H_V:=\PV$

\item[(ii)] On morphisms $i_{V^{\prime}V}:V^{\prime }\subseteq
V$:  The mapping $\H(i_{V^{\prime} V}):\H_V \map\H_{V^{\prime}}$
is given by
$\H(i_{V^{\prime}V})(\hat{\alpha}):=\dastoi{V}{\alpha}$ for all
$\hat\alpha\in\PV$.
 \end{enumerate}
} It is easy to see that the collection $\{ \dastoi{V}{P}\mid
V\in\Ob{\V{}}\}$ of projection operators given by
\eq{Def:dasinner} is a global element of $\H$.

It is also straightforward to show that
\begin{equation}
  \G(i_{V^\prime V})(\neg\hat \alpha)=
\neg\, \H(i_{V^{\prime} V})(\hat\alpha) \label{Gi(negalpha)=}
\end{equation}
for all projectors $\hat\alpha$ in $V$, and for all
$V^\prime\subseteq V$. It follows from \eq{Gi(negalpha)=} that
\begin{equation}
\delta^{o}(\neg \P)_V=\hat 1-\dastoi{V}{P} \label{domega=}
\end{equation}
for all projectors $\P$ and all contexts $V$.

It is clear from \eq{Gi(negalpha)=} that the negation operation on
projectors defines a map $\neg:\Ga\G\map\Ga\H$, $\ga\mapsto
\neg\ga$; \ie\ for all contexts $V$, we map
$\ga(V)\mapsto\neg\ga(V):=\hat 1-\ga(V)$. Actually, one can go
further than this and show that the presheaves $\G$ and $\H$ are
isomorphic in the category $\SetH{}$. This  means that, in
principle, we can always work with one presheaf only. However, for
reasons of symmetry it is sometime useful to invoke both
presheaves.

As with outer daseinisation, inner daseinisation can also be used
to define a mapping from projection operators to sub-objects of
the spectral presheaf. Specifically, if $\P$ is a projection,  for
each $V\in\Ob{\V{}}$ define
\begin{equation}
        T_{\dastoi{V}{P}}:=\{\l\in\Sig_V\mid
        \la\l,\dastoi{V}{P}\ra=0\}.
\end{equation}
It is easy to see that these subsets form a clopen subobject,
$\ps{\dastoi{}{P}}$,  of $\Sig$. It follows from (\ref{domega=})
that $T_{\dastoi{V}{P}}=S_{\delta^o(\neg \P)_V}$.

\subsubsection{Using Boolean Algebras as the Base Category}
\label{SubSubSec:BoolAlgBase} As we have mentioned several times
already, the collection, $\V{}$, of all commutative von Neumann
sub-algebras of $\BH$ is not the only possible choice for the base
category over which to construct presheaves. In fact, if we are
only interested in the propositional language\ $\PL{S}$, a
somewhat simpler choice is the collection, $\BlH$  of all Boolean
sub-algebras of the non-distributive lattice, $\PH$, of projection
operators on $\Hi$. More abstractly, for any non-distributive
lattice $\mathfrak{B}$, one could use  the category of Boolean
sub-algebras of $\mathfrak{B}$. This possibility was raised in the
original paper \cite{IB98} but has not been used much thereafter.
However, it does have some interesting features.

The analogue of the (von Neumann algebra) spectral presheaf,
$\Sig$, is the so-called \emph{dual} presheaf, $\ps{D}$:
\begin{definition}
\label{Defn:dual-presheaf-PH} The {\em dual presheaf\/} on $\BlH$
is the contravariant functor $\ps{D}:{\BlH}\map\Set$ defined as
follows:
\begin{enumerate}
\item On objects in $\BlH$: $\ps{D}(B)$ is the {\em dual\/} of $B$;
\ie\ the set ${\rm Hom}(B,\{0,1\})$ of all homomorphisms from the
Boolean algebra $B$ to the Boolean algebra $\{0,1\}$.

\item On morphisms in $\BlH$: If $i_{B_2B_1}:B_2\subseteq B_1$ then
$\ps{D}(i_{B_2B_1}): \ps{D}(B_1)\map \ps{D}(B_2)$ is defined by
$\ps{D}(i_{B_2B_1})(\chi):=\chi|_{B_2}$, where $\chi|_{B_2}$
denotes the restriction of $\chi\in \ps{D} (B_1)$ to the
sub-algebra $B_2\subseteq B_1$.
\end{enumerate}
\end{definition}

    A global element of the functor $\ps{D}:\BlH^\op\map
{\rm Set}$ is then a function $\ga$ that associates to each
$B\in\Ob{\BlH}$ an element $\ga_B$ of the dual of $B$ such that if
$i_{B_2B_1}: B_2\map B_1$ then $\ga_{B_1}|_{B_2}=\ga_{B_2}$; thus,
for all $\hat\a\in B_2$,
\begin{equation}
\ga_{B_2}(\hat\a)=\ga_{B_1}((i_{B_2B_1}(\hat\a)).
\end{equation}

Since each projection operator, $\hat\a$  belongs to at least one
Boolean algebra (for example, the algebra $\{\hat 0,\hat
1,\hat\a,\neg\hat\a\}$) it follows that a global element of the
presheaf $\ps{D}$ associates to each projection operator $\hat\a$
a number $V(\hat\a)$ which is either $0$ or $1$, and is such that,
if $\hat\a\land\hat\beta=\hat 0$, then
$V(\hat\a\lor\hat\beta)=V(\hat\a)+V(\hat\beta)$. These  types of
valuation are often used in the proofs of the Kochen-Specker
theorem that focus on the construction of specific
counter-examples. In fact, it is easy to see the following:

\displayE{11}{3}{3}{The Kochen-Specker theorem is equivalent to
the statement that, if $\dim{\cal H}>2$, the dual presheaf
$\ps{D}:\BlH^\op\map\Set$ has no global elements. }

It is easy to apply the concept of `daseinisation'  to the topos
$\Set^{\BlH^\op}$. In the case of von Neumann algebras, the outer
daseinisation of a projection operator $\P$ was defined as (see
\eq{Def:dasouter})
\begin{equation}
\dasto{V}{P}:=\bigwedge\big\{\hat{\a}\in\PV\mid \hat{\a}\succeq
\P\big\}           \label{Def:dasouter2}
\end{equation}
where $\PV$ denotes the collection of all projection operators in
the commutative von Neumann algebra $V$. In this form, $\das{P}$
appears as a global element of the outer presheaf $\G$.

When using the base category, $\BlH$, of Boolean sub-algebras of
$\PH$, we define
\begin{equation}
\dasto{B}{P}:=\bigwedge\big\{\hat{\a}\in B\mid \hat{\a}\succeq
\P\big\}           \label{Def:dasouterBool}
\end{equation}
for each Boolean sub-algebra $B$ of projection operators on $\Hi$.
Clearly, the (outer) daseinisation, $\das{P}$, is now a global
element of the obvious $\BH{}$-analogue  of the outer presheaf
$\G$. There are parallel remarks for the inner daseinisation and
inner presheaf. The existence of these daseinisation operations
means that  the propositional language $\PL{S}$ can be represented
in the topos $\Set^{\BlH^\op}$ in a way that is closely analogous
to that used above for the topos $\SetH{}$.

Note that (i) each Boolean algebra of projection operators $B$
generates a commutative von Neumann algebra, $B^{\prime\prime}$,
(the double commutant); and, conversely, (ii) to each von Neumann
algebra $V$ there is associated the Boolean algebra $\PV$ of the
projection operators in $V$. This implies that the operation
\begin{eqnarray}
               \phi: \BlH&\map&\V{}    \\
                   B&\mapsto& B^{\prime\prime}
\end{eqnarray}
defines a full and faithful functor between the categories $\BlH$
and $\V{}$. This functor can be used to pull-back the spectral
presheaf, $\Sig$, in $\SetH{}$ to the object
$\phi^*\Sig:=\Sig\circ\phi$ in $\Set^{\BlH^\op}$. This pull-back
is closely related to the dual presheaf $\ps{D}$.

\subsection{The Special Nature of Daseinised Projections}
\label{Sec:SpecNatureDasP}
\subsubsection{Daseinised Projections as Optimal Sub-Objects}
We have shown how  daseinisation leads to an interpretation/model
of the language $\PL{S}$ in the Heyting algebra $\Subcl\Sig$. In
particular, any primitive proposition $\SAin\De$ is represented by
the clopen sub-object $\ps{\delta({\hat E[A\in\De]})}$.

We have seen  that, in general, the `and',
$\ps{\das{P}}\land\ps{\das{Q}}$, of the daseinisation of two
projection operators $\P$ and $\hat Q$, is not itself of the form
$\ps{\das{R}}$ for any projector $\hat R$.  The same applies to
the negation $\lnot\ps{\das{P}}$.

This raises the question of whether the sub-objects of $\Sig$ that
\emph{are} of the form $\ps{\das{P}}$ can be characterised in a
simple way. Rather interestingly, the answer is `yes', as we will
now see.

Let $V^{\prime},V\in\Ob{\V{}}$ be such that $V^{\prime}\subseteq
V$. As would be expected, there is a close connection between the
restriction $\G(i_{V^{\prime}V}):\G_V\map\G_{V^{\prime}}$,
$\dasto{V}{P}\mapsto \dasto{V^{\prime}}{P}$, of the outer
presheaf, and the restriction $\Sig(i_{V^{\prime}V}):
\Sig_V\map\Sig_{V^{\prime}}$, $\l\mapsto\l|_{V^{\prime}}$, of the
spectral presheaf. Indeed, if $\P\in\PH$ is a projection operator,
and $S_{\dasto{V}{P}}\subseteq\Sig_V$  is  defined as in
\eq{Def:Salphahat}, we have the following result:
\begin{equation}
S_{\G(i_{V^{\prime}V})(\dasto{V}{P})}= \Sig (i_{V^{\prime}
V})(S_{\dasto{V}{P}}).
\end{equation}
The proof is given in Theorem \ref{Theorem: SG=SigS} in the
Appendix

This result shows that the sub-objects $\ps{\das{P}}=
\{S_{\dasto{V}{P}}\mid V\in\Ob{\V{}}\}$ of $\Sig$ are of a very
special kind. Namely, they are such that the restrictions
\begin{equation}
\Sig(i_{V^{\prime}V}):S_{\dasto{V}{P}}\map
                        S_{\dasto{V^\prime}{P}}
\end{equation}
are \emph{surjective} mapping of sets.

For an arbitrary sub-object $\ps{K}$ of $\Sig$, this will
 not be the case and $\Sig(i_{V^{\prime}V})$  only
maps $\ps{K}_V$ \emph{into} $\ps{K}_{V^\prime}$. Indeed, this is
essentially the \emph{definition} of a sub-object of a presheaf.
Thus we see that the daseinised projections $\ps{\das{P}}=
\{S_{\dasto{V}{P}}\mid V\in\Ob{\V{}}\}$ are optimal in the
following sense. As we go `down the line' to smaller and smaller
sub-algebras of a context $V$---for example, from $V$ to
$V^{\prime}\subseteq V$, then to $V^{\prime\prime }\subseteq
V^{\prime}$ etc.---then the subsets $S_{\dasto{V^\prime}{P}}$,
$S_{\dasto{V^{\prime\prime}}{P}}$,... are as small as they can be;
\ie\ $S_{\dasto{V^\prime}{P}}$ is the \emph{smallest} subset of
$\Sig_{V^{\prime}}$ such that
$\Sig(i_{V^{\prime}V})(S_{\dasto{V}{P}})\subseteq
S_{\dasto{V^\prime}{P}}$, likewise
$S_{\dasto{V^{\prime\prime}}{P}}$ is the smallest subset of
$\Sig_{V^{\prime\prime}}$ such that
$\Sig(i_{V^{\prime\prime}V^{\prime}})
(S_{\dasto{V^\prime}{P}})\subseteq
S_{\dasto{V^{\prime\prime}}{P}}$, and so on.

It is also clear from this result that there are lots of
sub-objects of $\Sig$ that are \emph{not} of the form
$\ps{\das{P}}$ for any projector $\P\in\PH$.

These more general sub-objects of $\Sig$ show up explicitly in the
representation of the more sophisticated language $\L{S}$. This
will be discussed thoroughly in Section \ref{Sec:psSR} when we
analyse the representation, $\phi$, of the language $\L{S}$ in the
topos $\SetH{}$. This involves constructing the quantity-value
object $\R_\phi$ (to be denoted $\ps{\R}$), and then finding the
representation of a function symbol $A:\Si\map\R$ in $\L{S}$, in
the form of a specific arrow $\breve{A}:\Sig\map\ps{\R}$ in the
topos. The generic sub-objects of $\Sig$ are then of the form
$\breve{A}^{-1}(\ps{\Xi})$ for sub-objects $\ps{\Xi}$ of $\ps\R$.
This is an illuminating way of studying the sub-objects of $\Sig$
that do not come from the propositional language $\PL{S}$.


\section{Truth Values in Topos Physics}
\label{Sec:TruthValues}
\subsection{The Mathematical Proposition ``$x\in K$''}
So far we have concentrated on finding a Heyting-algebra
representation of the propositions in quantum theory, but  of
course there is more to physics than that. We also want to know
if/when a certain proposition is \emph{true}: a question which, in
physical theories, is normally answered by specifying a
\emph{(micro) state} of the system, or something that can play an
analogous role.

In classical physics, the situation is straightforward (see
Section \ref{SubSubSec:PhysRepPLS}). There, a proposition
$\SAin\De$ is represented by the subset
$\picl(\Ain\De):=\breve{A}^{-1}(\Delta)\subseteq\S$ of the state
space $\S$; and then, the proposition is true in a state $s$ if
and only if $s\in\breve{A}^{-1}(\De)$; \ie\ if and only if the
(micro-) state $s$ belongs to the subset, $\picl(\Ain\De)$, of
$\S$ that represents the proposition.

Thus, each state $s$ assigns to any primitive proposition
$\SAin\De$, a truth value, $\TVal{\Ain\De}{s}$,  which lies in the
set $\{{\rm false},{\rm true}\}$ (which we identify with
$\{0,1\}$) and is defined as 
\begin{equation}\label{Def:[AinD]ClassRep1}
        \TVal{\Ain\De}{s}:=
        \left\{\begin{array}{ll}
            1 & \mbox{\ if\ $s\in\picl(\Ain\De):=
                                \breve{A}^{-1}(\De)$;} \\
            0 & \mbox{\ otherwise.}
         \end{array}
        \right.
\end{equation}

However, the situation in quantum theory is very different. There,
the spectral presheaf $\Sig$---which is the analogue of the
classical state space $\S$---has no global elements at all. Our
expectation is that this will be true in any topos-based theory
that goes `beyond quantum theory': \ie\ $\Ga\Si_\phi$ is empty;
or, if $\Si_\phi$ does have global elements, there are not enough
of them to determine $\Si_\phi$ as an object in the topos.  In
this circumstance,  a new concept is required to replace the
familiar idea of a `state of the system'. As we shall see, this
involves the concept of a `truth object', or `pseudo-state'.

In physics, the propositions of interest are of the form
$\SAin\De$, which refers to the value of a \emph{physical}
quantity. However, in constructing a theory of physics, such
physical propositions must first be translated into
\emph{mathematical} propositions. The concept of `truth' is then
studied in the context of the latter.

Let us start with set-theory based mathematics, where the most
basic proposition is of the form ``$x\in K$'', where $K$ is a
subset of a set $X$, and $x$ is an element of $X$. Then the truth
value, denoted $\TValM{x\in K}$, of the proposition ``$x\in K$''
is
\begin{equation}
 \TValM{x\in K}=
        \left\{\begin{array}{ll}
            1 & \mbox{\ if\ $x$ belongs to $K$;} \\
            0 & \mbox{\ otherwise.}
         \end{array}
        \right.                                 \label{TVxinKcl}
\end{equation}
Thus the proposition ``$x\in K$'' is true if, and only if, $x$
 belongs to $K$. In other words,
 $x\mapsto\TValM{x\in K}$ is the characteristic function of the subset $K$\ of $X$; cf.\ \eq{Def:chiK} in the Appendix.

This remark is the foundation of the assignment of truth values in
classical physics. Specifically, if the state is $s\in\S$, the
truth value, $\TVal{\Ain\De}{s}$,  of the \emph{physical}
proposition $\SAin\De$ is defined to be the truth value of the
\emph{mathematical} proposition ``$\breve{A}(s)\in\De$''; or,
equivalently, of the mathematical proposition ``$s\in
\breve{A}^{-1}(\De)$''. Thus, using \eq{TVxinKcl}, we get, for all
$s\in\S$,
\begin{equation}
        \TVal{\Ain\De}{s}:=
        \left\{\begin{array}{ll}
            1 & \mbox{\ if\ $s$ belongs to $\breve{A}^{-1}(\De)$;}
         \\
            0 & \mbox{\ otherwise.}\label{Def:[AinD]Class2}
         \end{array}
        \right.
\end{equation}
which reproduces \eq{Def:[AinD]ClassRep1}.

We now consider the analogue of the above in a general topos
$\tau$. Let $X$ be an object in $\tau$, and let $\so{K}$ be a
sub-object of $X$. Then $\so{K}$ is determined by a characteristic
arrow $\cha{K}: X \map\O_\tau$, where $\O_\tau$ is the sub-object
classifier; equivalently, we have an arrow $\name{K}:1_\tau\map
PX$.

Now suppose that $x:1_\tau\map X$ is a global element of $X$; \ie\
$x\in\Ga X:=\Hom{\tau}{1_\tau}{X}$. Then the truth value of the
mathematical proposition ``$x\in\so{K}$'' is defined to be
\begin{equation}
        \TValM{x\in\so{K}}:=\chi_{\so{K}}\circ {x}
        \label{VxinpsK}
\end{equation}
where  $\cha{K}\circ {x}:1_\tau\map\O_\tau$. Thus
$\TValM{x\in\so{K}}$ is an element of $\Ga\O_\tau$; \ie\ it is a
global element of the sub-object classifier $\O_\tau$.

The connection with the result  \eq{TVxinKcl} (in the topos
$\Set$) can be seen by noting that, in \eq{TVxinKcl}, the
characteristic function of the subset $K\subseteq X$ is the
function $\cha{K}:X\map\{0,1\}$ such that $\chi_K(x)=1$ if $x\in
K$, and $\chi_K(x)=0$ otherwise. It follows that \eq{TVxinKcl} can
be rewritten as
\begin{eqnarray}
        \TValM{x\in K}&=&\chi_K(x)          \\
                    &=&\chi_K\circ{x}\label{TVxinKcl(b)}
\end{eqnarray}
where in \eq{TVxinKcl(b)}, ${x}$ denotes the function
${x}:\{*\}\map X$ that is defined by ${x}(*):=x$. The link with
\eq{VxinpsK} is clear when one remembers that, in the topos
$\Set$, the terminal object, $1_\Set$, is just the singleton set
$\{*\}$.

In quantum theory,  the topos is $\SetH{}$, and so the objects are
all presheaves. In particular, at each stage $V$, the sub-object
classifier $\Om:=\O_{\SetH{}}$  is the set of sieves on $V$. In
this case, if $\ps{K}$ is a sub-object of $\ps{X}$, and
$x\in\Ga\ps{X}$, the explicit form for \eq{TVxinKcl(b)} is the
sieve
\begin{equation}
  \TValM{x\in\ps{K}}_V:=\{V^\prime\subseteq V
                \mid x_{V^\prime}\in \ps{K}_{V^\prime}\}
        \label{ValxinKpresheaf}
\end{equation}
at each stage $V\in\Ob{\V{}}$. In other words, at each
stage/context V, the truth value of the mathematical proposition
``$x\in\ps{K}$'' is defined to be all those stages
$V^\prime\subseteq V$ `down the line' such that the `component',
$x_{V^\prime}$ of $x$ at that stage \emph{is} an element of the
component, $\ps{K}_{V^\prime}\subseteq \ps{X}_{V^\prime}$, of
$\ps{K}$ at that stage.

The definitions \eq{VxinpsK} and \eq{ValxinKpresheaf} play a
central role in constructing truth values in out quantum topos
scheme. However, as $\Sig$ has no global elements, these truth
values cannot be derived from some expression $\TValM{s\in\ps{K}}$
with ${s}:1_{\SetH{}}\map\Sig$. Therefore, we must proceed
differently, as will become clear by the end of the following
Section.

However, before we do so, let us make one final remark concerning
\eq{TVxinKcl}. Namely, in normal set theory the proposition
``$x\in K$'' is true if, and only if,
\begin{equation}
                        \{x\}\subseteq K  \label{{x}subsetK}
\end{equation}\ie\ if an only if the set $\{x\}$ is a subset of $K$.
The transition from the proposition ``$x\in K$'' to the
proposition ``$\{x\}\subseteq K$'' is  seemingly trivial, but in a
topos other than sets it takes on a new significance. In
particular, as we shall see shortly, although the spectral
presheaf, $\Sig$, has no global elements, it does have certain
`minimal' sub-objects that are as `close' as one can get to a
global element, and then the topos analogue of \eq{{x}subsetK} is
very important.

\subsection{Truth Objects}\label{SubSub:TruthObjects}
\subsubsection{Linguistic Aspects of Truth Objects.}
To understand how `truth values' of physical propositions arise we
return again to our earlier discussion of  local languages. In
this Section we will employ the local language $\L{S}$ rather than
the  propositional language, $\PL{S}$,  that was used earlier in
this article.

Thus, let $\L{S}$ be the local language for a system $S$. This is
a typed language whose minimal set of ground-type symbols is $\Si$
and $\R$. In addition, there is a non-empty set, $\F{S}$, of
function symbols $A:\Si\map\R$ that correspond to the physical
quantities of $S$.

Now  consider a  representation, $\phi$, of $\L{S}$ in a topos
$\tau_\phi$.  As discussed earlier, the propositional aspects of
the language $\L{S}$ are captured in the term
$\q{A(\va{s})\in\va\De}$ of type $\O$, where $\va{s}$ and
$\va{\De}$ are  variables of type $\Si$ and $P\R$ respectively
\cite{DI(1)}. In  a topos representation, $\phi$, the
representation, $\Val{A(\va{s})\in\va\De}_\phi$, of the term
$\q{A(\va{s})\in\va\De}$  is given by the chain of
arrows\footnote{In \eq{A(s)intildeDeChainRep1}, $e_{\R_\phi}:
\R_\phi\times P\R_\phi\map\O_{\tau_\phi}$ is the evaluation arrow
associated with the power object $P\R_\phi$.}\cite{Bell88} (cf
\eq{A(s)intildeDeChain})
\begin{equation}
\Si_\phi\times P\R_\phi\mapright{A_\phi\times\id}
        \R_\phi\times P\R_\phi\mapright{e_{\R_\phi}}\O_{\tau_\phi}
                        \label{A(s)intildeDeChainRep1}
\end{equation}
in the topos $\tau_\phi$. Then, if $\name{\Xi}:1_{\tau_\phi}\map
P\R_\phi$ is the name of a sub-object, $\Xi$, of the
quantity-value object $\R_\phi$, we get the chain
\begin{equation}
\Si_\phi\simeq\Si_\phi\times 1_{\tau_\phi}\mapright{\id\times
\name{\Xi}} \Si_\phi\times P\R_\phi\mapright{A_\phi\times\id}
        \R_\phi\times P\R_\phi\mapright{e_{\R_\phi}}\O_{\tau_\phi}.
\end{equation}
which is  the characteristic arrow of the sub-object of $\Si_\phi$
that represents the physical proposition $\SAin \Xi$.

Equivalently,  we can use the term, $\{\va{s}\mid
A(\va{s})\in\va\De\}$, which has a free variable $\va\De$ of type
$P\R$ and is of type $P\Si$. This term is represented by the arrow
$\Val{\{\va{s}\mid A(\va{s})\in\va\De\}}_\phi : P\R_\phi\map
P\Si_\phi$, which  is the power transpose of
 $\Val{A(\va{s})\in\va\De}_\phi$ (cf \eq{[]=nametildeDe}):
\begin{equation}
\Val{\{\va{s}\mid A(\va{s})\in\va\De\}}_\phi =
\name{\Val{A(\va{s})\in\va\De}_\phi}\label{[]=nametildeDeRep1}
\end{equation}
The proposition $\SAin\Xi$ is then represented by the arrow
$\Val{\{\va{s}\mid
A(\va{s})\in\va\De\}}_\phi\circ\name{\Xi}:1_{\tau_\phi}\map
P\Si_\phi$; this is the name of the sub-object of $\Si_\phi$ that
represents $\SAin\Xi$.

We note an important difference from the analogous situation for
the language $\PL{S}$. In propositions of the type $\SAin\De$, the
symbol `$\Delta$' is a specific subset of $\mathR$ and is hence
\emph{external} to the language. In particular, it is
\emph{independent} of the representation of $\PL{S}$. However, in
the case of $\L{S}$, the variable $\va{\De}$ is \emph{internal} to
the language, and the quantity $\Xi$ in the proposition $\SAin
\Xi$ is a sub-object of $\R_\phi$ in a \emph{specific} topos
representation,  $\phi$, of $\L{S}$.

So, this is how physical propositions are represented
mathematically. But  how are truth values to be assigned to these
propositions? In the topos $\tau_\phi$ a truth value is an element
of the Heyting algebra $\Ga\O_{\tau_\phi}$. Thus the challenge is
to assign a global element of $\O_{\tau_\phi}$ to each proposition
associated with the representation of the term  $\{\va{s}\mid
A(\va{s})\in\va\De\}$ of type $P\Si$;  (or, equivalently, the
representation of the term `$A(\va{s})\in\va\De$').

Let us first pose this question at a linguistic level. In a
representation $\phi$, an element of $\Ga\O_{\tau_\phi}$ is
associated with a representation of a term of type $\O$ with no
free variables. Hence the question can be rephrased as asking how
a term, $t$, in $\L{S}$ of type $P\Si$ can be `converted' into a
term of type $\O$? At this stage, we are happy to have free
variables, in which case the desired term will be represented by
an arrow in $\tau_\phi$ whose co-domain is $\O_{\tau_\phi}$, but
whose domain is other than $1_{\tau_\phi}$. This would be an
intermediate stage to obtaining a global element of
$\O_{\tau_\phi}$.

In the context of the language $\L{S}$ there are three obvious
ways of `converting' the term $t$ of type $P\Si$ to a term of type
$\O$:
\begin{enumerate}
        \item Choose a term, $s$, of type $\Si$; then  the
        term `$s\in t$' is of type $\O$.
We will call this the `micro-state' option.
        \item Choose a term,  $\TO$, of type $PP\Si$; then
        the term `$t\in \TO$' is of type $\O$. We shall refer to this as the `truth-object' option.

        \item Choose a term, $\w$, of type $P\Si$; then
        the term `$\w\subseteq t $' is of type
        $\O$.\footnote{In general, if $t$ and $s$ are
        set-like terms (\ie\ terms of power type, $PX$, say),
        then `$t\subseteq s$' is defined as the term
        `$\forall \va{x}\in t(\va{x}\in s)$'; here,
        $\va{x}$ is a variable of type $X$.} For reasons that will become clear later we shall refer to this as the `pseudo-state' option.
\end{enumerate}

\subsubsection{The Micro-State Option}
In regard to the first option, the simplest example of a term of
type $\Si$ is a variable $\va{s_1}$ of type $\Si$. Then, the term
`$\va{s_1}\in \{\va{s}\mid A(\va{s})\in\va{\De}\}$' is of type
$\O$ with the free variables $\va{s_1}$ and $\va{\De}$ of type
$\Si$ and $P\R$ respectively. However, the axiom of comprehension
in $\L{S}$ says that
\begin{equation}
\va{s_1}\in \{\va{s}\mid A(\va{s})\in\va{\De}\}\Leftrightarrow
        A(\va{s_1})\in\va{\De}
\end{equation}
and so we are back with the term `$A(\va{s})\in\va{\De}$', which
is of type $\O$ and with the free variable $\va{s}$ of type $\Si$.

As stated above,  the $\phi$-representation,
$\Val{A(\va{s})\in\va\De}_\phi$, of $\q{A(\va{s})\in\va{\De}}$ is
the chain of arrows in \eq{A(s)intildeDeChainRep1}. Now, suppose
the representation, $\phi$, is such that there exist global
elements, ${s}:1_{\tau_\phi}\map\Si_\phi$, of $\Si_\phi$. Then
each such element can be regarded as a `(micro)-state' of the
system in that topos representation. Furthermore, let
$\name{\Xi}:1_{\tau_\phi}\map P\R_\phi$ be the name of a
sub-object, $\Xi$, of the quantity-value object $\R_\phi$. Then,
by the basic property of the product $\Si_\phi\times P\R_\phi$,
there is an arrow $\la
s,\name\Xi\ra:1_{\tau_\phi}\map\Si_\phi\times P\R_\phi$. This can
be combined with the arrow
$\Val{A(\va{s})\in\va{\De}}_\phi:\Si_\phi\times
P\R_\phi\map\O_{\tau_\phi}$ to give the arrow
\begin{equation}
\Val{A(\va{s})\in\va{\De}}_\phi\circ \la{s},\name\Xi\ra:
1_{\tau_\phi}\longrightarrow
      \O_{\tau_\phi} \label{AinDelcircs}
\end{equation}
This is the desired global element of $\O_{\tau_\phi}$.

In other words, when the `state of the system' is
$s\in\Ga\Si_\phi$, the `truth value' of the proposition $\SAin
\Xi$ is the global element of $\O_{\tau_\phi}$ given  by the arrow
$\Val{A(\va{s})\in\va{\De}}_\phi\circ\la {s},\name{\Xi}\ra:
1_{\tau_\phi}\map\O_{\tau_\phi}$.

This is the procedure that is adopted in classical physics when a
truth value is assigned to propositions by specifying a
micro-state, $s\in\Si_\s$, where $\Si_\s$ is the classical state
space in the representation $\s$ of $\L{S}$. Specifically, for all
$s\in\Si_\s$, the truth value of the proposition $\SAin\De$ as
given by \eq{AinDelcircs} is (c.f. \eq{Def:[AinD]ClassRep1})
\begin{equation}
\TVal{\Ain\De}{s}
        =\Val{A(\va{s})\in\va{\De}}_\s(s,\De)=
        \left\{\begin{array}{ll}
            1 & \mbox{\ if\ $A_\s(s)\in\De$;}\\
            0 & \mbox{\ otherwise.}
         \end{array}
        \right.                 \label{nuAinDe;s}
\end{equation}
where $\Val{A(\va{s})\in\va{\De}}_\s:\Si_\s\times
P\mathR\map\O_{\tau_\s}\simeq\{0,1\}$. Thus we recover the earlier
result \eq{Def:[AinD]Class2}.

\subsubsection{The Truth Object Option.}
By hindsight, we know that the option to use global elements of
$\Si_\phi$ is not available in the quantum case. For there the
state object, $\Sig$, is the spectral presheaf, and this has no
global elements by virtue of the Kochen-Specker theorem. The
absence of global elements of the state object $\Si_\phi$ could
well be true in many other topos models of physics  (particularly
those that go `beyond quantum theory'), and therefore an
alternative general strategy is needed to that employing
micro-states $\name{s}:1_{\tau_\phi}\map\Si_\phi$.

This takes us to the second possibility: namely, to introduce a
term, $\TO$,    of type $PP\Si$, and then work with the term
`$\{\va{s}\mid A(\va{s})\in\va{\De}\} \in\TO$', which is  of type
$\O$, and has whatever free variables are contained in $\TO$, plus
the variable $\va{\De}$ of type $P\R$.

The simplest choice is to let the term of type $PP\Si$ be a
variable, $\va{\TO}$, of type $PP\Si$, in which case the term
$\q{\{\va{s}\mid A(\va{s})\in\va\De\}\in\va\TO}$ has variables
$\va\De$ and $\va{\TO}$ of type $P\cal R$ and $PP\Si$
respectively. Therefore, in a topos representation  it is
represented by an arrow $\Val{\{\va{s}\mid A(\va{s})\in
\va\De\}\in\va{\TO}}_\phi:P{\cal R}_\phi\times
P(P\Si_\phi)\map\O_{\tau_\phi}$. In detail (see \cite{Bell88}) we
have that
\begin{equation}
\Val{\{\va{s}\mid A(\va{s})\in \va\De\}\in\va{\TO}}_\phi =
        e_{P\Si_\phi}\circ \Val{\{\va{s}\mid A(\va{s})\in
\va\De\}}_\phi\times\Val{\va{\TO}}_\phi
\end{equation}
where $e_{P\Si_\phi}:P\Si_\phi\times
P(P\Si_\phi)\map\O_{\tau_\phi}$ is the usual evaluation arrow. In
using this expression we need the $\phi$-representatives:
\begin{eqnarray}
    \Val{\{\va{s}\mid A(\va{s})\in\va\De\}}_\phi
        :P\R_\phi&\map& P\Si_\phi\\[3pt]
    \Val{\va{\TO}}_\phi:P(P\Si_\phi) &
    \overset{\id}{\longrightarrow}& P(P\Si_\phi)
\end{eqnarray}

Finally, let $\la\name{\Xi},\name{\TO}\ra$ be a pair of global
elements in $P\R_\phi$ and $P(P\Si_\phi)$ respectively, so that
$\name{\Xi}:1_{\tau_\phi}\map P\R_\phi$ and
$\name{\TO}:1_{\tau_\phi}\map P(P\Si_\phi)$. Thus, $\name\TO$ is
the name of a
 `truth object', $\TO$, in $\tau_\phi$. Then, for the physical
proposition $\SAin\Xi$, we have the truth value
\begin{equation}
\TVal{\Ain\Xi}{\TO}=\Val{\{\va{s}\mid A(\va{s})\in
\va\De\}\in\va{\TO}}_\phi\circ\la\name{\Xi},\name{\TO}\ra
:1_{\tau_\phi}\map\O_{\tau_\phi}\label{AinXiTO}
\end{equation}
 where $\la\name{\Xi},\name\TO\ra:1_{\tau_\phi}\map
P\R_\phi\times P(P\Si_\phi)$.

\paragraph{A\ small generalisation:}
Slightly more  generally, if $\va{J}$ and $\va\TO$ are variables
of type $P\Si$ and $P(P\Si)$ respectively, the term of interest is
`$\va{J}\in\va\TO$'. In the representation, $\phi$, of $\L{S}$,
this term maps to an arrow
$\Val{\va{J}\in\va\TO}_\phi:P\Si_\phi\times
P(P\Si_\phi)\map\O_{\tau_\phi}$.  Here,
$\Val{\va{J}\in\va{\TO}}_\phi =
        e_{P\Si_\phi}\circ \Val{\va{J}}_\phi\times
\Val{\va{\TO}}_\phi$ where $\Val{\va{J}}_\phi
        :P\Si_\phi\overset{\id}{\map} P\Si_\phi$ and
$\Val{\va{\TO}}_\phi:P(P\Si_\phi) \overset{\id} {\longrightarrow}
P(P\Si_\phi)$. Let $\name{J}$, $\name\TO$ be  global elements of
$P\Si_\phi$ and $P(P\Si_\phi)$ respectively, so that
$\name{J}:1_{\tau_\phi}\map P\Si_\phi$ and
$\name\TO:1_{\tau_\phi}\map P(P\Si_\phi)$.  Then  the truth of the
(mathematical) proposition ``${J}\in\TO$'' is
\begin{eqnarray}
\TValM{{J}\in\TO}&=&
        \Val{\va{J}\in\va\TO}_\phi\circ \la
        \name{J},\name\TO\rangle                \nonumber\\
        &=&e_{P\Si_\phi}\circ\la\name{J},\name{\TO}\ra:
                1_{\tau_\phi}\map \O_{\tau_\phi}
        \label{TValMgainTO}
\end{eqnarray}

\subsubsection{The Example of Classical Physics.}
If classical physics is studied this way, the general formalism
simplifies, and the term `$\{\va{s}\mid A(\va{s})\in\va{\De}\}
\in\va\TO$' is represented by the function $\TVal{\Ain\De}{\TO}
 :=\Val{\{\va{s}\mid
A(\va{s})\in\va{\De}\} \in\va\TO}_\s :P\mathR\times
P(P\Si_\s)\map\O_{\Set}\simeq\{0,1\}$ defined by
\begin{eqnarray}
\TVal{\Ain\De}{\TO}:= \Val{\{\va{s}\mid A(\va{s})\in\va{\De}\}
\in\va\TO}_\s(\De,\TO)
                &=&
 {\left\{\begin{array}{ll}
            1 & \mbox{\ if\ $\{s\in\Si_\s\mid A_\s(s)
                        \in\De\} \in \TO$;} \\
            0 & \mbox{\ otherwise}
         \end{array}
        \right.}\nonumber\\[5pt]
&=&
 {\left\{\begin{array}{ll}
            1 & \mbox{\ if\ $A_\s^{-1}(\De) \in \TO$;} \\
            0 & \mbox{\ otherwise}
         \end{array}
        \right.}\label{Def:nu(AinD;T}
\end{eqnarray}
for all $\TO\in P(P\Si_\s)$. We can clearly see the sense in which
the truth object $\TO$ is playing the role of a state. Note that
the result \eq{Def:nu(AinD;T} of classical physics is a special
case of \eq{AinXiTO}.

To recover the usual truth values given in  \eq{nuAinDe;s},  an
appropriate truth object, $\TO^s$, must be associated with each
micro-state $s\in\Si_\s$. The correct choice is
\begin{equation}
        \TO^s:=\{J\subseteq\Si_\s\mid s\in J\}
                        \label{Def:ClassTO}
\end{equation}
for each $s\in\Si_\s$. It is clear that $s\in A_\s^{-1}(\De)$ (or,
equivalently, $A_\s(s)\in\De$) if, and only if,
$A_\s^{-1}(\De)\in\TO^s$. Hence \eq{Def:nu(AinD;T} can be
rewritten as
\begin{equation}\label{Def:[AinD]Class(2)}
        \TVal{\Ain\De}{\TO^s}=
        \left\{\begin{array}{ll}
            1 & \mbox{\ if\ $s\in A_\s^{-1}(\De)$;} \\
            0 & \mbox{\ otherwise.}
         \end{array}
        \right.
\end{equation}
which reproduces  \eq{nuAinDe;s} once $\TVal{\Ain\De}{s}$ is
identified with $\TVal{\Ain\De}{\TO^s}$.

\subsection{Truth Objects in Quantum Theory}
\label{SubSubSec:TOQT}
\subsubsection{Preliminary Remarks}
We can now start to discuss the application of these ideas to
quantum theory.  In order to use  \eq{AinXiTO} (or
\eq{TValMgainTO}) we need to construct  concrete truth objects,
${\ps\TO}$, in the topos $\tau_\phi:=\SetH{}$. Thus the presheaf
$\ps\TO$ is a sub-object of $P\Sig$; equivalently,
$\name{\ps\TO}:1_{\tau_\phi}\map P(P\Sig)$.

However, we have to keep in mind the need to restrict to
\emph{clopen} sub-objects of $\Sig$. In particular, we must show
that there is a well-defined presheaf $\PSig$ such that
\begin{equation}
        \Subcl{\Sig}\simeq\Ga(\PSig)      \label{Subcl=GPSig}
\end{equation}
We will prove this in Section \ref{SubSec:PresheafP-clSig}. Given
\eq{Subcl=GPSig} and  $\ps{J}\in\Subcl\Sig$, it is then clear that
a truth object, $\ps{\TO}$, actually has to be a sub-object of
$\PSig$ in order that the valuation $\TValM{{\ps J}\in\ps\TO}$ in
\eq{TValMgainTO} is meaningful.

This truth value, $\TValM{{\ps J}\in\ps\TO}$, is a global element
of $\Om$, and in the topos of presheaves, $\SetH{}$, we have (see
\eq{ValxinKpresheaf})
\begin{equation}
\TValM{{\ps{J}}\in\ps\TO}_V:=\{V^\prime\subseteq V\mid
                        \ps{J}_{V^\prime}\in\ps\TO_{V^\prime}\}
\end{equation}
for each context $V$.

There are various examples of  the presheaf $\ps{J}$ that are of
interest to us. In particular, let $\ps{J}=\ps{\das{P}}$ for some
projector $\P$. Then, using the propositional language $\PL{S}$
introduced earlier, the `truth' of the proposition represented by
$\P$ (for example,  $\SAin\De$) is
\begin{equation}
\TValM{{\ps{\das{P}}}\in\ps\TO}_V =\{V^\prime\subseteq V\mid
\ps{\dasto{}{P}}_{V^\prime}\in\ps\TO_{V^\prime}\}\label{ValDasFinal}
\end{equation}
for all stages $V$.

When using the local language $\L{S}$, an important class of
examples of the sub-object $\ps{J}$ of $\Sig$ are of the form
$A_\phi^{-1}(\ps\Xi)$, for some sub-object $\ps\Xi$ of $\ps{\R}$.
This will yield the truth value, $\TVal{\Ain\ps\Xi}{\ps\TO}$, in
\eq{AinXiTO}. However, to discuss this further requires  the
representation of function symbols $A:\Si\map\R$ in the topos
$\SetH{}$, and this is deferred until Section \ref{Sec:deG}.

\subsubsection{The Truth Objects \ps{$\TO$}$^{\ket\psi}$.}
The definition of truth objects in quantum theory was studied in
the original papers \cite{IB98,IB99,IB00,IB02}. It was shown there
that to each quantum state $\ket\psi\in\cal H$, there corresponds
a truth object, $\ps\TO^{\ket\psi}$, which was defined as the
following sub-object of the outer presheaf, $\G$:
\begin{eqnarray}
\ps\TO^{\ket\psi}_V&:=&\{\hat\alpha\in \G_V\mid
                {\rm Prob}(\hat\alpha;\ket\psi)=1\}\label{TOpsi1}
                                                \nonumber\\[2pt]
        &=&\{\hat\alpha\in \G_V\mid
                \bra\psi\hat\alpha\ket\psi=1\}     \label{TOpsi2}
\end{eqnarray}
for all stages $V\in\Ob{\V{}}$. Here, ${\rm
Prob}(\hat\alpha;\ket\psi)$ is the usual expression for the
probability that the proposition represented by the projector
$\hat\alpha$ is true, given that the quantum state is the
(normalised) vector $\ket\psi$.

It is easy to see that \eq{TOpsi2} defines a genuine sub-object
$\ps\TO^{\ket\psi}=\{\ps\TO^{\ket\psi}_V\mid V\in\Ob{\V{}}\}$ of
$\G$. Indeed, if $\hat\beta\succeq\hat\alpha$, then
$\bra\psi\hat\beta\ket\psi \geq\bra\psi\hat\alpha\ket\psi$, and
therefore, if $V^\prime\subseteq V$ and $\hat\alpha\in \G_V$, then
$\bra\psi \G(i_{V^\prime V})(\hat\alpha)\ket\psi
\geq\bra\psi\hat\alpha\ket\psi$. In particular, if
$\bra\psi\hat\alpha\ket\psi=1$ then $\bra\psi \G(i_{V^\prime
V})(\hat\alpha)\ket\psi =1$.

The next step is to define the presheaf $\PSig$, and show that
there is a monic arrow $\G\map \PSig$, so that $\G$ is a
sub-object of $\PSig$. Then, since $\ps\TO^{\ket\psi}$ is a
sub-object of $\G$, and $\G$ is a sub-object of $\PSig$, it
follows that $\ps\TO^{\ket\psi}$ is a sub-object of $\PSig$, as
required. The discussion of the construction of $\PSig$ is
deferred to Section \ref{SubSec:PresheafP-clSig} so as not to
break the flow of the presentation.

With this definition of $\ps\TO^{\ket\psi}$, the truth value,
\eq{ValDasFinal}, for the propositional language $\PL{S}$ becomes
\begin{equation}
\TValM{{\ps{\das{P}}}\in\ps\TO^{\ket\psi}}_V =\{V^\prime\subseteq
V\mid \bra\psi\dasto{V^\prime}{P}\ket\psi=1\}\label{nudPinT}
\end{equation}

It is easy to see that the definition of a truth object in
\eq{TOpsi2} can be extended to a mixed state with a density-matrix
operator $\hat\rho$:\ simply replace the definition in \eq{TOpsi2}
with
\begin{eqnarray}
\ps\TO^{\hat\rho}_V&:=&\{\hat\alpha\in \G_V\mid
                {\rm Prob}(\hat\alpha;\rho)=1\}      \nonumber\\[2pt]
        &=&\{\hat\alpha\in \G_V\mid
               {\rm tr}(\hat\rho\hat\alpha)=1\}     \label{TOrho}
\end{eqnarray}

However there is an important difference between the truth object
associated with a vector state, $\ket\psi$, and the one associated
with a density matrix, $\rho$. In the vector case, it is easy to
see that the mapping $\ket\psi\map\ps\TO^{\ket\psi}$ is one-to-one
(up to a phase factor on $\ket\psi$) so that, in principle, the
state $\ket\psi$ can be \emph{recovered} from $\ps\TO^{\ket\psi}$
(up to a phase-factor). On the other hand, there are simple
counterexamples which show  that, in general, the density matrix,
$\rho$ \emph{cannot} be recovered from $\ps\TO^{\hat\rho}$.

In a sense, this should not surprise us. The analogue of a density
matrix in classical physics is a probability measure $\mu$ defined
on the classical state space $\cal S$. Individual microstates
$s\in\cal S$ are in one-to-one correspondence with probability
measures of the form $\mu_s$ defined by $\mu_s(J)=1$ if $s\in J$,
$\mu_s(J)=0$ if $s\not\in J$.

However, one of the main claims of our programme is that any
theory can be made to `look like' classical physics in the
appropriate topos. This suggests that, in the topos version of
quantum theory, a density matrix should be represented by some
sort of measure on the state object $\Sig$ in the topos
$\tau_\phi$; and  this should relate in some way to an `integral'
of `vector truth objects'. The recent work by Heunen and Spitters
provides the mathematical basis for such a construction
\cite{HeuSpit07}. We shall return to some of their ideas later.

\subsection{The  Pseudo-state Option}\label{SubSec:IdeaWurst}
\subsubsection{Some Background Remarks}

We turn now to the third way mentioned above whereby a term, $t$,
of type $P\Si$ in $\L{S}$ can be `converted' to a term of type
$\O$. Namely, choose a term, $\w$, of type $P\Si$ and then use
`$\w\subseteq t$'. As we shall  see, this idea is easy to
implement in the case of quantum theory and leads to an
alternative way of thinking about truth objects.

Let us start by considering once more the case of classical
physics. There, for each microstate $s$ in the symplectic state
manifold $\Si_\s$, there is an associated truth object, $\TO^s$,
defined by $\TO^s:=\{J\subseteq\Si_\s\mid s\in J\}$, as in
\eq{Def:ClassTO}. It is clear that the state $s$ can be uniquely
recovered from the collection of sets $\TO^s$ as
\begin{equation}
  s=\bigcap \{J\subseteq\Si_\s\mid s\in J\}
                \label{s=bigcapK}
\end{equation}
Note that \eq{s=bigcapK} implies that $\TO^s$ is an
\emph{ultrafilter} of subsets of $\Si_\s$\footnote{Let
$\mathbb{L}$ be a lattice with zero element $0$. A subset
$F\subset\mathbb{L}$ is a `filter base' if (i) $0\notin F$ and
(ii) for all $a,b\in F$, there is some $c\in F$ such that $c\leq
a\wedge b$. A subset $D\subset\mathbb{L}$ is called a `(proper)
dual ideal' or a `filter' if (i) $0\notin D$, (ii) for all $a,b\in
D$, $a\wedge b\in D$ and (iii) $a\in D$ and $b>a$ implies $b\in
D$. A maximal dual ideal/filter $F$ in a complemented,
\emph{distributive} lattice $\mathbb{L}$ is called an
`ultrafilter'. It has the property that for all $a\in\mathbb{L}$,
either $a\in F$ or $a^{\prime}\in F$, where $a^{\prime}$ is the
complement of $a$.}. As we shall shortly see, there is an
intriguing analogue of this property for the quantum truth
objects.

The analogue of \eq{s=bigcapK} in the case of quantum theory is
rather interesting. Now, of course, there are no microstates, but
we do have the truth objects defined in \eq{TOpsi1}, one for each
vector state $\ket\psi\in\Hi$. To proceed further we note that
$\bra\psi\hat\a\ket\psi=1$ if and only if
$\ketbra\psi\preceq\hat\a$. Thus $\ps\TO^{\ket\psi}$ can be
rewritten as
\begin{equation}
        \ps\TO^{\ket\psi}_V:=\{\hat\a\in\G_V\mid
                  \ketbra\psi\preceq\hat\a\}  \label{TOpsi2b}
\end{equation}
for each stage $V$. Note that, as defined in \eq{TOpsi2b},
$\ps\TO^{\ket\psi}$ is a sub-object of $\G$; \ie\ it is defined in
terms of projection operators. However,  as will be shown in
Section \ref{SubSec:MonicGPSig}, there is a monic arrow
$\G\map\PSig$, and  by using this arrow, $\ps\TO^{\ket\psi}$ can
be regarded as a sub-object of $\PSig$; hence
$\Ga\ps\TO^{\ket\psi}$ is a collection of clopen sub-objects of
$\Sig$. In this form, the definition of $\ps\TO^{\ket\psi}$
involves clopen subsets of the spectral sets $\Sig_V$,
$V\in\Ob{\V{}}$.

It is clear from \eq{TOpsi2b} that, for each $V$,
$\ps\TO^{\ket\psi}_V$ is a \emph{filter} of projection operators
in $\G_V\simeq\PV$; equivalently, it is a filter of clopen
sub-sets of $\Sig_V$.

These ordering properties are associated with the following
observation. If $\ket\psi$ is any vector state, we can collect
together all the projection operators that are `larger' or equal
to $\ketbra\psi$ and define:
\begin{equation}
        T^{\ket\psi}:=\{\hat\alpha\in\PH\mid
        \ketbra\psi\preceq\hat\alpha\}\label{Def:TpsiGlobal}
\end{equation}
It is clear that, for all stages/contexts $V\in\Ob{\V{}}$, we have
\begin{equation}
\ps\TO^{\ket\psi}_V=T^{\ket\psi}\cap V
\end{equation}
Thus the presheaf $\ps\TO^{\ket\psi}$ is obtained by `localising'
$T^{\ket\psi}$ at each context $V$.

The significance of this localisation property is that
$T^{\ket\psi}$ is a \emph{maximal} (proper) filter in the
non-distributive lattice, $\PH$, of all projection operators on
$\Hi$. Such maximal filters in the projection lattices of von
Neumann algebras were extensively discussed by de Groote
\cite{deG05c} who called them `quasi-points'. In particular,
$T^{\ket\psi}$ is a, so-called, `atomic' quasi-point in $\PH$.
Every pure state $\ket\psi$ gives rise to an atomic quasi-point,
$T^{\ket\psi}$, and vice versa. We will return to these entities
in Section \ref{SubSubSec:PhysIntArrow}.

\subsubsection{Using Pseudo-States in Lieu of Truth Objects}
\label{PseudoStatesLieu} The equation \eq{s=bigcapK} from
classical physics suggests
 that, in the quantum case, we look at the set-valued function on $\Ob{\V{}}$ defined by
\begin{equation}
       V\mapsto\bigwedge\{\hat\alpha\in\ps\TO^{\ket\psi}_V\}=
       \bigwedge \{\hat\a\in\G_V\mid
                  \ketbra\psi\preceq\hat\a\}\label{Def:wpsi}
\end{equation}
where we have used \eq{TOpsi2b} as the definition of
$\ps\TO^{\ket\psi}$. It is easy to check that this is a global
element of $\G$; in fact, the right hand side of \eq{Def:wpsi} is
nothing but the outer daseinisation $\dasmap(\ketbra\psi)$ of the
projection operator $\ketbra\psi$! Evidently, the quantity
\begin{equation}
 \w^{\ket\psi}:= \dasmap(\ketbra\psi)=V\mapsto \bigwedge
                 \{\hat\a\in\G_V\mid\ketbra\psi\preceq\hat\a\}
                 \label{Def:wpsi_2}
\end{equation}
is of considerable interest. We shall refer to it as a
`pseudo-state' for reasons that appear below.

Note that  $\w^{\ket\psi}$ is defined by \eq{Def:wpsi_2}  as an
element of $\Ga\G$. However,  because of the monic $\G\map\PSig$
we can also regard $\w^{\ket\psi}$ as an element of
$\Ga(\PSig)\simeq \Subcl\Sig$. The corresponding (clopen)
sub-object of $\Sig$ will be denoted
$\ps\w^{\ket\psi}:=\ps{\delta(\ketbra\psi)}$.

We know that the map $\ket\psi\mapsto\ps\TO^{\ket\psi}$ is
injective. What can be said about the map
${\ket\psi}\mapsto\ps\w^{\ket\psi}$? In this context, we note that
$\ps\TO^{\ket\psi}$ is readily recoverable from
$\w^{\ket\psi}\in\Ga\G$ as
\begin{equation}
        \ps\TO^{\ket\psi}_V=\{\hat\alpha\in\G_V\mid\hat\alpha\succeq
        \w^{\ket\psi}_V\}\label{Trecov_w}
\end{equation}
for all contexts $V$. From these relations it follows that is
${\ket\psi}\mapsto\ps\w^{\ket\psi}$ is injective.

Note that, \eq{Trecov_w} essentially follows from the fact that,
for each $V$, the collection, $\ps\TO^{\ket\psi}_V$ of projectors
in $\G_V$ is an \emph{upper} set (in fact, as remarked earlier, it
is a filter). In this respect, the projectors/clopen subsets
$\ps\TO^{\ket\psi}_V$ behave like the filter of \emph{clopen
neighbourhoods} of a subset in a topological space. This remark
translates globally to the relation of the collection,
$\Ga\ps\TO^{\ket\psi}$, of sub-objects of $\Sig$ to the specific
sub-object $\ps{\w}^{\ket\psi}$.

It follows that there is a one-to-one correspondence between truth
objects, $\ps\TO^{\ket\psi}$, and pseudo-states,
$\ps\w^{\ket\psi}$. However, the former is (a representation of) a
term of type $P(P\Si)$, whereas the latter is of type $P\Si$. So
how is this reflected in the assignment of generalised truth
values?

Note first that, from the definition of $\w^{\ket\psi}$, it
follows that if $\hat\a\in\ps\TO_V^{\ket\psi}$ then
$\hat\a\succeq\w^{\ket\psi}_{V}$. On the other hand, from
\eq{Trecov_w} we have that if $\hat\a\succeq\w^{\ket\psi}_{V}$
then $\hat\a\in\ps\TO_V^{\ket\psi}$. Thus we have the simple, but
important, result:
\begin{equation}
      \hat\a\in\ps\TO_V^{\ket\psi} \mbox{ if, and only if }
      \hat\a\succeq\w^{\ket\psi}_V\label{betainTiff}
\end{equation}
In particular, for any projector $\P$ we have $
    {\das{P}}_V\in\ps\TO_V^{\ket\psi} \mbox{ if, and only if }
         {\das{P}}_V\succeq\w^{\ket\psi}_V.
$

In terms of sub-objects of $\Sig$, we have ${\das{P}}_V
\succeq\w^{\ket\psi}_V$ if and only if $\ps{\das{P}}_V\supseteq
\ps{\w}^{\ket\psi}_V$. Hence, \eq{betainTiff} can be rewritten as
\begin{equation}
{\das{P}}_V\in\ps\TO_V^{\ket\psi} \mbox{ if, and only if }
         \ps{\das{P}}_V\supseteq \ps{\w}_V
\end{equation}
and so \eq{ValDasFinal} can be written as
\begin{equation}
\TValM{{\ps{\das{P}}}\in\ps\TO}_V =\{V^\prime\subseteq V\mid
 \ps{\das{P}}_V\supseteq \ps{\w}^{\ket\psi}_V\}
\label{ValDasFinal(2)}
\end{equation}

However, the right hand side of \eq{ValDasFinal(2)} is just the
topos truth value, $\nu(\ps\w^{\ket\psi}\subseteq\ps{\das{P}})$.
It follows that
\begin{equation}
 ``{\ps{\das{P}}}\in\ps\TO^{\ket\psi}\mbox{''} \mbox{ is equivalent to }
``\ps\w^{\ket\psi}\subseteq
\ps{\das{P}}\mbox{''}\label{dPinTequivmsubd}
\end{equation}
and hence we can use the generalised truth values
$\TValM{{\ps{\das{P}}}\in\ps\TO^{\ket\psi}}$ or
$\TValM{\ps\w^{\ket\psi}\subseteq{\ps{\das{P}}}}$ interchangeably.

Thus, if desired, a truth object in quantum theory can be regarded
as a sub-object of $\Sig$, rather than a sub-object of $P\Sig$. In
a sense, these sub-objects, $\ps\w^{\ket\psi}$, of $\Sig$ are the
`closest' we can get to global elements of $\ps\Si$. This is why
we call them `pseudo-states'. However,  note that a pseudo-state
is not a \textit{minimal} element of the Heyting algebra
$\Subcl\Sig$ since these will include stalks that are empty sets,
something that is not possible for a pseudo-state. \footnote{Note
that the sub-objects $\ps\w^{\ket\psi}$ do not have any global
elements since any such would give a global element of $\Sig$ and,
of course, there are none. Thus if one is seeking examples of
presheaves with no global elements, the collection
$\ps\w^{\ket\psi}$, $\ket\psi\in\Hi$, afford many such.}

\subsubsection{Linguistic Implications}
The result \eq{dPinTequivmsubd} is very suggestive for a more
general development. In our existing treatment, in the formal
language $\L{S}$ we have concentrated on propositions of the form
``$\va{J}\in\va\TO$'' which, in a representation $\phi$, maps to
the arrow $\Val{\va{J}\in\va\TO}_\phi:P\Si_\phi\times
P(P\Si_\phi)\map\O_{\tau_\phi}$.  Here $\va{J}$ and $\va{\TO}$ are
variables of type $P\Si$ and $P(P\Si)$ respectively.

\displayE{11}{5}{4}{What is suggested by the discussion above is
that we could equally focus on terms of the form
``$\va\w\subseteq\va{J}$'', where both $\va\w$ and $\va{J}$ are
variables of type $P\Si$.}

\noindent Note that, in general, the $\phi$-representation of such
a term is of the form
\begin{equation}
\Val{\va{\w}\subseteq\va{J}}_\phi:P\Si_\phi\times
P\Si_\phi\map\O_{\tau_\phi} \label{ValwinK}
\end{equation}
where the `first slot' on the right hand side of the pairing in
\eq{ValwinK} is a truth-object (in pseudo-state form), and the
second correspond to a proposition represented by a sub-object of
$\Si_\phi$.

However, this raises the rather obvious question ``What \emph{is}
a pseudo-state?''. More precisely, we would like to know a generic
set of characteristic properties of those sub-objects of
$\Si_\phi$ that can be regarded as `pseudo-states'. A first step
would be to answer this question in the case of quantum theory. In
particular, are there any quantum pseudo-states that are
\emph{not} of the form $\ps\w^{\ket\psi}$ for some vector
$\ket\psi\in\Hi$?

In this context the localisation property expressed by
\eq{Def:TpsiGlobal} is rather suggestive. In the case that $\Hi$
has infinite dimension, de Groote has shown that there exist
quasi-points in $\PH$ that are not of the form $T^{\ket\psi}$ for
some $\ket\psi\in\Hi$ \cite{deG05c}.\footnote{However, he has also
shown that, in an appropriate topology, the set of all atomic
quasi-points is \emph{dense} in the set of all quasi-points. Of
course, none of these intriguing structures arise in a
finite-dimensional Hilbert space in anything other than a trivial
way. So, in that sense, it is unlikely that they will play any
fundamental role in explicating the topos representation of
quantum theory.} If $T$ is any such quasi-point,
\eq{Def:TpsiGlobal} suggests strongly that we define an associated
presheaf, $\ps{T}$, by
\begin{equation}
        \ps{T}_V=T\cap V\label{Def:psT}
\end{equation}
for all $V\in\Ob{\V{}}$. This construction seems natural enough
from a mathematical perspective, but we are not yet clear of the
physical significance of the existence of such `quasi
truth-objects'. The same applies to the associated `quasi
pseudo-state', $\ps\w^T$, defined by
\begin{equation}
        \ps\w^T_V:=\bigwedge\{\hat\a\in \ps{T}_V\}=
                \bigwedge\{\hat\a\in T\cap V\}
\end{equation}

\subsubsection{Time-Dependence and the Truth Object.}
As emphasised at the end of Section \ref{SubSec:PropLangPhys}, the
question of time dependence depends on the theory-type  being
considered. The structure of the language $\L{S}$ that has been
used so far is such that the time variable lies outside the
language. In this situation, the time dependence of the system can
be  implemented in several  ways.

For example, we can make the truth object time dependent, giving a
family of truth objects, $t\mapsto\ps\TO^t$, $t\in\mathR$. In the
case of classical physics, with the truth objects $\TO^s$,
$s\in\Si_\s$, the time evolution  comes from the time dependence,
$t\mapsto s_t$, of the microstate in accordance with the classical
equations of motion. This gives the family $t\mapsto \TO^{s_t}$ of
truth objects.

Something very similar  happens in quantum theory, and we acquire
a family, $t\mapsto \ps\TO^{\ket\psi_t}$, of truth objects, where
the states $\ket\psi_t$ satisfy the usual time-dependent
Schr\"odinger equation. Thus both classical and quantum truth
objects belong to a `Schr\"odinger picture' of time evolution. Of
course, there is a pseudo-state analogue of this in which we get a
one-parameter family, $t\mapsto\ps\w^{\ket\psi_t}$, of clopen
sub-objects of $\Sig$.

It is also possible to construct a `Heisenberg picture' where the
truth object is constant but the physical quantities and
associated propositions are time dependent. We will return to this
in Section \ref{Sec:Unitary}  when we discuss the use of unitary
operators.

\subsection{The Presheaf $P_{\operatorname{cl}}$(\underline{$\Sigma$}).}
\label{SubSec:PresheafP-clSig}
\subsubsection{The Definition of $P_{\operatorname{cl}}$(\underline{$\Sigma$}).}
We must now show that there really is a presheaf $\PSig$.

The easiest way of defining $\PSig$ is to start with the concrete
expression for the normal power object $P\Sig$ \cite{Gol84}.
First, if $\ps{F}$ is any presheaf over $\V{}$,  define the
\emph{restriction} of $\ps{F}$ to $V$ to be the functor
$\ps{F}\!\downarrow\!V$ from the category\footnote{The notation
$\downarrow\!\!V$ means the partially-ordered set of all
sub-algebras $V^\prime\subseteq V$.} $\downarrow\!\!V$ to $\Set$
that assigns to each $V_{1}\subseteq V$, the set $\ps{F}_{V_{1}}$,
and with the obvious induced presheaf maps.

Then, at each stage $V$,  $P\Sig_V$ is the set of natural
transformations from $\Sig\!\downarrow\!V$ to $\Om\!\downarrow\!
V$. These are in one-to-one correspondence with families of maps
 $\s:=\{\s_{V_{1}}:\Sig_{V_{1}}
\map\Om_{V_{1}}\mid V_{1}\subseteq V\}$, with the following
commutative diagram for all $V_2\subseteq V_1\subseteq V$:
\footnote{Note that any sub-object, $\ps{J}$ of $\Sig$, gives rise
to such a natural transformation from $\Sig\!\downarrow\!V$ to
$\Om\!\downarrow\!V$ for all stages $V$. Namely, for all
$V_1\subseteq V$, $\s_{V_1}:\Sig_{V_1}\map\Om_{V_1}$ is defined to
be the characteristic arrow
${\chi_{\ps{J}}}_{V_1}:\Sig_{V_1}\map\Om_{V_1}$ of the sub-object
$\ps{J}$ of $\Sig$. }
\begin{center}
\setsqparms[1`1`1`1;700`600]
\square[\Sig_{V_1}`\Om_{V_1}`\Sig_{V_2}`\Om_{V_2};
\s_{V_1}`\Sig(i_{V_2V_1})`\Om(i_{V_2V_1})`\s_{V_2}]
\end{center}
The presheaf maps are defined by
\begin{eqnarray}
        P\Sig(i_{V_1 V}):P\Sig_{V}&\map& P\Sig_{V_1}\\
         \s\ \ &\mapsto& \{\s_{V_2}\mid V_2\subseteq V_1\}
\end{eqnarray}
and the evaluation arrow ${\rm ev}:P\Sig\times\Sig\map\Om$, has
the form, at each stage $V$:
\begin{eqnarray}
        {\rm ev}_V:P\Sig_V\times\Sig_V&\map&\Om_V  \\
                        (\s,\l) &\mapsto& \s_V(\l)
\end{eqnarray}

Moreover, in general, given a map $\chi:\Sig_V\map\Om_V$, the
subset of $\Sig_V$ associated with the corresponding sub-object is
$\chi^{-1}(1)$, where $1$ is the unit (`truth') in the Heyting
algebra $\Om_V$.

This suggests strongly that an object, $\PSig$, in $\SetH{}$ can
be defined using the same definition of $P\Sig$ as above, except
that the family of maps  $\sigma:=\{\sigma_{V_{1}}:\Sig_{V_{1}}
\map\Om_{V_{1}}\mid V_{1}\subseteq V\}$ must be such that, for all
$V_1\subseteq V$, $\sigma_{V_1}^{-1}(1)$ is a \emph{clopen} subset
of the (extremely disconnected) Hausdorff space $\Sig_{V_1}$. It
is straightforward to check that such a restriction is consistent,
and that $\Subcl\Sig\simeq\Ga(\PSig)$ as required.

\subsubsection{The Monic Arrow From \G\ to $P_{\operatorname{cl}}$(\underline{$\Sigma$}).}
\label{SubSec:MonicGPSig} We define $\iota:\G\times\Sig\map\Om$,
with the power transpose $\name{\iota}:\G\map\PSig$, as follows.
First recall that in any topos, $\tau$  there is a bijection ${\rm
Hom}_\tau(A,C^B)\simeq {\rm Hom}_\tau(A\times B,C)$, and hence, in
particular, (using $P\Sig=\Om^{\Sig}$)
\begin{equation}
        {\rm Hom}_{\SetH{}}(\G, P\Sig)\simeq
        {\rm Hom}_{\SetH{}}(\G\times\Sig,\Om).
\end{equation}

Now let $\hat\alpha\in\PV$, and let $S_{\hat\alpha}:=
\{\l\in\Sig_V \mid\brak\l{\hat\alpha}=1\}$ be the clopen subset of
$\Sig_V$ that corresponds to the projector $\hat\alpha$ via the
spectral theorem; see \eq{Def:Salphahat}. Then we define
$\iota:\G\times\Sig\map\Om$ at stage $V$ by
\begin{equation}
 \iota_V(\hat\alpha,\l):=\{V^\prime\subseteq V\mid
    \Sig(i_{V^\prime\,V})(\l)\in
 S_{\G(i_{V^\prime\, V})(\hat\alpha)}\}\label{Def:iV}
\end{equation}
for all $(\hat\alpha,\l)\in \G_V\times\Sig_V$.

On the other hand, the basic result relating coarse-graining to
subsets of $\Sig$ is
\begin{equation}
S_{\G(i_{V^{\prime}\,V)}(\dasto{V}{\alpha})}=
        \Sig (i_{V^{\prime}\,V})(S_{\dasto{V}{\alpha}})
\end{equation}
for all $V^\prime\subseteq V$ and for all $\hat\alpha\in\G_V$. It
follows that
\begin{equation}
        \iota_V(\hat\alpha,\l):=\{V^\prime\subseteq V\mid
        \Sig(i_{V^\prime\,V})(\l)\in
        \Sig(i_{V^\prime\, V})(S_{\hat\alpha})\}
\end{equation}
for all $(\hat\alpha,\l)\in \G_V\times\Sig_V$. In this form is is
clear that $\iota_V(\hat\alpha,\l)$ is indeed a \emph{sieve} on
$V$; \ie\ an element of $\Om_V$.

The next step is to show that the collection of maps
$\iota_V:\G_V\times\Sig_V\map \Om_V$ defined in \eq{Def:iV}
constitutes a natural transformation from the object
$\G\times\Sig$ to the object $\Om$ in the topos $\SetH{}$. This
involves chasing around  a few commutative squares, and we will
spare the reader the ordeal. There is some subtlety,  since we
really want to deal with ${\rm Hom}_{\SetH{}}(\G,\PSig)$, not
${\rm Hom}_{\SetH{}}(\G,P\Sig)$; but all works in the end.

To prove that $\name{\iota}:\G\map\PSig$ is monic, it suffices to
show that the map $\name{\iota}_V:\G_V\map\PSig_V$ is injective at
all stages $V$. This is a straightforward exercise and the details
will not be given here.

The conclusion of this exercise is that, since
$\name{\iota}:\G\map\PSig$ is monic, the truth sub-objects
$\ps\TO^{\ket\psi}$ of $\G$ can also be regarded as sub-objects of
$\PSig$, and hence the truth value assignment in \eq{ValDasFinal}
is well-defined.

Finally then, for any given quantum state $\ket\psi$ the basic
proposition $\SAin\De$ can be assigned a generalised truth value
$\TVal{\Ain\De}{\ket\psi}$ in $\Ga\Om$, where $\tau:=\SetH{}$ is
the topos of presheaves over $\V{}$. This is defined at each
stage/context $V$ as
\begin{eqnarray}
\TVal{\Ain\De}{\ket\psi}_V&:=&\TValM{{\ps{\delta(\hat
E[A\in\De])}}\in\ps\TO^{\ket\psi}}_V
                                \nonumber       \\[4pt]
                &=&  \{V^\prime\subseteq V\mid
{\ps{\delta\big(\hat
E[A\in\De]\big)}}_{V^\prime}\in\ps\TO^{\ket\psi}_{V^\prime}\}
\end{eqnarray}

\subsection{Yet Another Perspective on the K-S\ Theorem}
\label{SubSec:ObpsWO} In classical physics, the pseudo-state
$\w^s\subseteq\S$ associated with the microstate $s\in\S$ is just
$\w^s:=\{s\}$. This gives the diagram
\begin{equation}
        \setsqparms[1`0`1`1;700`700] \label{ComDiagS}
        \square[\w^s` \S`\{*\}` P\S;
        {} ` {} ` \pi` \name{\w^s}]
\end{equation}
where $\name{\w^s}(*):=\{s\}$ and $\pi$ is the canonical map
\begin{eqnarray}
        \pi:\S&\longrightarrow&PS\nonumber \label{S->PS}\\
            s&\mapsto&\{s\}
\end{eqnarray}
The singleton $\{*\}$ is the terminal object in the category,
$\Set$, of sets, and the subset embedding $\w^s\map\S$ in
\eq{ComDiagS} is the categorical pull-back by $\pi$ of the monic
$\name{\w^s}:\{*\}\map P\S$.

In the quantum case, the analogue of the diagram \eq{ComDiagS} is
\begin{equation}
        \setsqparms[1`0`1`1;700`700] \label{ComDiagC2}
        \square[\ps\w^{\ket\psi}` \Sig`\ps{1}` P\Sig;
        {} ` {} ` \pi` \name{\ps\w^{\ket\psi}}]
\end{equation}
where the arrow $\pi:\Sig\map P\Sig$ has yet to be defined. To
proceed further, let us first return to  the set-theory map
\begin{eqnarray}
                X&\map& PX       \label{XmapPX}\\
                x&\mapsto& \{x\}                \nonumber
\end{eqnarray}
where $X$ is any set.

 We can think of \eq{XmapPX} as the power transpose,
 $\name{\beta}:X\map PX$, of the map $\beta:X\times X\map\{0,1\}$ defined by
\begin{equation}
       \beta(x,y):=\left\{\begin{array}{ll}
            1 & \mbox{\ if\ $x=y$;} \\
            0 & \mbox{\ otherwise.}     \label{Def:b(x,y)}
         \end{array}
        \right.
\end{equation}
In our topos case, the obvious definition for the arrow
$\pi:\Sig\map P\Sig$ is the power transpose $\name{\beta}:\Sig\map
P\Sig$, of the arrow $\beta:\Sig\times\Sig\map\Om$, defined by
\begin{equation}
        \beta_{V}(\l_1,\l_2):=
                \{V'\subseteq V\mid \l_1|_{V'}=\l_2|_{V'}\}
                \label{beta(l1l2)}
\end{equation}
for all stages $V$. Note that, in linguistic terms, the arrow
defined in \eq{beta(l1l2)} is just the representation in the
quantum topos $\SetH{}$, of the term`$\va\s_1=\va\s_2$', where
$\va\s_1$ and $\va\s_2$ are terms of type $\Si$; \ie
$\Val{\va\s_1=\va\s_2}:\Sig\times\Sig\map\Om$.

With this definition of $\pi$, the diagram in \eq{ComDiagC2}
becomes meaningful: in particular the monic
$\ps\w^{\ket\psi}\hookrightarrow\Sig$ is the categorical pull-back by $\pi$
of the monic $\name{\ps\w^{\ket\psi}} :\ps{1}\map P\Sig$.

There is, however, a significant difference between \eq{ComDiagC2}
and its classical analogue \eq{ComDiagS}. In the latter case, the
function $\name{\w^s}:\{*\}\map P\S $  can be `lifted' to a
function $\name{\w^s}^\uparrow:\{*\}\map \S$ to give a commutative
diagram: \ie\, such that
\begin{equation}
        \pi\circ\name{\w^s}^\uparrow=\name{\w^s}.
\end{equation}
Indeed, simply define
\begin{equation}
        \name{\w^s}^\uparrow(*):=s
\end{equation}

However, in the quantum case there can be no `lift'
$\name{\ps\w^{\ket\psi}}^\uparrow:1\map \Sig$, as this would
correspond to a global element of the spectral presheaf $\Sig$,
and of course there are none. Thus, from this perspective, the
Kochen-Specker theorem can be understood as asserting the
existence of an \emph{obstruction} to lifting the  arrow
$\name{\ps\w^{\ket\psi}}:1\map P\Sig$.

Lifting problems of the type
\begin{equation}
        \setsqparms[0`0`1`1;700`700] \label{ComDiagABC}
        \square[`A`C`B;
        {} ` {} ` \pi` \phi]
\end{equation}
occur in many places in mathematics. A special, but very
well-known, example of \eq{ComDiagABC} arises when trying to
construct cross-sections of a non-trivial principle fiber bundle
$\pi:P\map M$. In diagrammatic terms we have
\begin{equation}
        \setsqparms[0`0`1`1;700`700] \label{LiftPBun}
        \square[` P` M ` M;
        {} ` {} ` \pi`\id]
\end{equation}
A cross-section   of this bundle corresponds to a lifting of the
map $\id:M\map M$.

The obstructions to lifting $\id:M\map M$ through $\pi$ can be
studied in various ways. One technique is to decompose the bundle
$\pi: P\map M$ into a series of interpolating fibrations $P\map
P_1\map P_2\map\cdots M$ where each fibration $P_i\map P_{i+1}$
has the special property that the fiber is a particular
Eilenberg-McLane space (this is known as a `Postnikov tower'). One
then studies the sequential lifting of the function $\id:M\map M$,
\ie\ first try to lift it through the fibration $P_1\map M$;  if
that is successful try to lift it through $P_2\map P_1$; and so
on. Potential obstructions to performing these liftings appear as
elements of the cohomology groups $H^k(M;\pi^{k-1}(F))$,
$k=1,2,\ldots$, where $F$ is the fiber of the bundle.

We have long felt that it should possible to describe the
non-existence of global elements of $\Sig$ (\ie\ the
Kochen-Specker theorem) in some cohomological way, and the remark
above suggests one possibility. Namely, perhaps there is some
analogue of a `Postnikov factorisation' for the arrow
$\ps\pi:\Sig\map P\Sig$  that could give a cohomological
description of the obstructions to a global element of $\Sig$,
\ie\ to the lifting of a pseudo-state $\name{\ps\w^{\ket\psi}}:1
\map P\Sig$ through the arrow $\ps\pi:\Sig\map\ps P\Sig$ to give
an arrow $\ps{1}\map \Sig$.

Related to this is the question of if there is a `pseudo-state
object', $\ps\WO$, with the defining property that $\Ga\ps\WO$ is
equal to the set of all pseudo-states. Of course, to do this
properly requires a definition of a pseudo-state that goes beyond
the specific constructions of the objects $\ps\w^{\ket\psi}$,
$\ket\psi\in\Hi$. In particular, are there pseudo-states that are
\emph{not} of the form $\ps\w^{\ket\psi}$?

If such an object, $\ps\WO$ \emph{can} be found then $\ps\WO$ will
be a sub-object of $P\Sig$, and in the diagram in \eq{ComDiagC2}
one could then look to replace $P\Sig$ with $\ps\WO$.

\section{The de Groote Presheaves of Physical Quantities}
\label{Sec:deG}
\subsection{Background Remarks}
Our task now is to consider the representation of the local
language, $\L{S}$, in the case of quantum theory. We assume that
the relevant topos  is the same as that used for the propositional
language $\PL{S}$, \ie\ $\SetH{}$, but the emphasis is very
different.

From a physics perspective, the key symbols in $\L{S}$ are (i) the
ground-type symbols, $\Si$ and $\mathcal R$---the linguistic
precursors of the state object and the quantity-value object
respectively---and (ii) the function symbols $A:\Si\map\R$, which
are the precursors of physical quantities. In the quantum-theory
representation, $\phi$, of $\L{S}$,  the representation,
$\Si_\phi$, of $\Si$ is defined to be the spectral presheaf $\Sig$
in the topos  $\SetH{}$.

The critical question is to find the object, $\R_\phi$
(provisionally denoted as a presheaf $\ps{\R}$), in $\SetH{}$ that
represents $\R$, and is hence the quantity-value object. One might
anticipate that $\ps{\R}$ is just the real-number object in the
topos $\SetH{}$, but that turns out to be quite  wrong, and the
right answer cannot just be guessed. In fact, the correct choice
for $\ps{\R}$ is found indirectly by considering a related
question: namely, how to represent each function symbol
$A:\Si\map\R$, with a concrete arrow $A_\phi:\Si_\phi\map\R_\phi$
in  $\SetH{}$, \ie\ with a natural transformation
$\breve{A}:\Sig\map\ps{\R}$ between the presheaves $\Sig$ and
$\ps{\R}$.

Critical to this task are the daseinisation operations on
projection operators that were defined earlier as
\eq{Def:dasouter} and \eq{Def:dasinner}, and which are repeated
here for convenience: {\definition If $\hat P$ is a projection
operator, and $V\in\Ob{\V{}}$ is any context/stage,  we define:
\begin{enumerate}
\item The `outer daseinisation' operation is
\begin{equation}
        \dastoo{V}{P}:=\bigwedge\big\{\hat{\a}\in\PV\mid
          \P\preceq\hat\a\big\}.\label{Def:dasouterRep1}
\end{equation}
where `$\,\preceq$' denotes the usual ordering of projection
operators, and where $\PV$ is the set of all projection operators
in $V$.

\item Similarly, the `inner daseinisation' operation is defined in
the context $V$ as (c.f. \eq{Def:dasinner})
\begin{equation}
        \dastoi{V}{P}:=\bigvee\big\{\hat\beta\in\PV\mid
                \hat\beta\preceq \P\big\}. \label{Def:dasinnerRep1}
\end{equation}
\end{enumerate}
} \noindent Thus $\dastoo{V}{P}$ is the best approximation to $\P$
in $V$ from `above', being the smallest projection in $V$ that is
larger than or equal to $\P$. Similarly, $\dastoi{V}{P}$ is the
best approximation to $\P$ from `below', being the largest
projection in $V$ that is smaller than or equal to $\P$.

In Section \ref{SubSec:PresheafP-clSig}, we showed that the outer
presheaf is a sub-object of the power object $\PSig$ (in the
category $\SetH{}$), and hence that the global element $\daso{P}$
of $\G$ determines a (clopen) sub-object, $\ps{\daso{P}}$, of the
spectral presheaf $\Sig$. By these means, the quantum logic of the
lattice $\PH$ is mapped into the Heyting algebra of the set,
$\Subcl{\Sig}$, of clopen sub-objects of $\Sig$.

Our task now is to perform the second stage of the programme:
namely  (i) identify the quantity-value presheaf, $\ps\R$; and
(ii) show that any physical quantity can be represented by an
arrow from $\Sig$ to $\ps\R$.

\subsection{The Daseinisation of an Arbitrary Self-Adjoint Operator}
\subsubsection{Spectral Families and Spectral Order}
We now want to extend the daseinisation operations from
projections to arbitrary (bounded) self-adjoint operators. To this
end, consider first a bounded, self-adjoint operator, $\A$, whose
spectrum is purely discrete. Then the spectral theorem can be used
to write $\A=\sum_{i=1}^\infty a_i\P_i$ where $a_1,a_2,\ldots $
are the eigenvalues of $\A$, and $\P_1,\P_2,\ldots $ are the
spectral projection operators onto the corresponding eigenspaces.

A construction that comes immediately to mind is to use the
daseinisation operation on projections to define
\begin{equation}
\delta^o(\A)_V:= \sum_{i=1}^\infty a_i\,\daso{P_i}_V
                                                \label{Def:DasAWrong!}
\end{equation}
for each stage $V$. However, this procedure is rather unnatural.
For one thing, the  projections, $\P_i$, $i=1,2,\ldots$ form a
complete orthonormal set:
\begin{eqnarray}
        \sum_{i=1}^\infty \P_i&=&\hat 1,\\
        \P_i\P_j&=&\delta_{ij}\P_i,
\end{eqnarray}
whereas, in general, the collection of daseinised projections,
$\daso{P_i}_V$, $1=1,2,\ldots$ will not satisfy either of these
conditions. In addition, it is hard to see how the expression
$\daso{A}_V:= \sum_{i=1}^\infty a_i\,\daso{P_i}_V$ can be
generalised to operators, $\A$, with a continuous spectrum.

The answer to this conundrum lies in the work of de Groote. He
realised that although it is not  useful to daseinise the spectral
projections of an operator $\A$, it \emph{is} possible to
daseinise the \emph{spectral family} of $\A$ \cite{deG05}.

\paragraph{Spectral families.}
We first recall that a {\em spectral family} is a family of
projection operators $\hat E_\l$, $\l\in\mathR$, with the
following properties:
\begin{enumerate}
\item If $\l_2\leq\l _1$ then $\hat E_{\l_2}\preceq\hat E_{\l_1}$.
\item The net $\l\mapsto \hat E_\l$ of projection operators in the
lattice $\PH$  is bounded above by $\hat 1$, and below by $\hat
0$. In fact,
\begin{eqnarray}
                \lim_{\l\map\infty} \hat E_\l &=&\hat 1,\\
                \lim_{\l\map -\infty}\hat E_\l&=&\hat 0.
\end{eqnarray}
\item The map $\l\mapsto \hat E_\l$ is right-continuous:\footnote{
It is a matter of convention whether one chooses right-continuous
or left-continuous.}
\begin{equation}
  \bigwedge_{\epsilon\downarrow 0} \hat E_{\l+\epsilon}=\hat E_\l
\end{equation}
for all $\l\in\mathR$.
\end{enumerate}
The spectral theorem asserts that for any self-adjoint operator
$\A$, there exists a spectral family, $\l\mapsto \hat E^A_{\l}$,
such that
\begin{equation}
        \A=\int_\mathR \l\, d \hat E^A_{\l}   \label{SpTh}
\end{equation}
We are only concerned with bounded operators, and so the (weak
Stieljes) integral in \eq{SpTh} is really over the bounded
spectrum of $\A$ which, of course, is a compact subset of
$\mathR$. Conversely, given a bounded spectral family $\{\hat
E_\l\}_{\l\in\mathR}$,\footnote{I.e., there are $a,b\in\mathR$
such that $\hat E_\l=\hat 0$ for all $\l\leq a$ and $\hat
E_\l=\hat 1$ for all $\l\geq b$.} there is a bounded self-adjoint
operator $\A$ such that $\A=\int_\mathR \l\, d \hat E_{\l}$.

\paragraph{The spectral order.}
A key element for our work is the so-called {\em spectral order}
that was introduced in \cite{Ols71}.\footnote{The spectral order
was later reinvented by de Groote, see \cite{deG04}.} It is
defined as follows.  Let $\A$ and $\hat B$ be (bounded)
self-adjoint operators with spectral families $\{\hat
E^A_\l\}_{\l\in\mathR}$ and $\{\hat E^B_\l\}_{\l\in\mathR}$,
respectively. Then define:
\begin{equation}
        \A\preceq_s\hat B\mbox{ if and only if }
        \hat E^B_\l\preceq\hat E^A_\l \mbox{ for all $\l\in\mathR$}.
        \label{Def:ApreceqsB}
\end{equation}
It is easy to see that \eq{Def:ApreceqsB} defines a genuine
partial ordering on $\BH_\sa$ (the self-adjoint operators in
$\BH$). In fact, $\BH_\sa$ is a `boundedly complete' lattice with
respect to the spectral order, \ie\ each bounded set $S$ of
self-adjoint operators has a minimum $\bigwedge S\in\BH_\sa$ and a
maximum $\bigvee S\in\BH_\sa$ with respect to this order.

If $\hat P,\hat Q$ are projections, then
\begin{equation}
        \hat P\preceq_s\hat Q \mbox{ if and only if }
        \hat P\preceq\hat Q,             \label{[PQ]=0->s=}
\end{equation}
so the spectral order coincides with the usual partial order on
$\PH$. To ensure this, the `reverse' relation in
\eq{Def:ApreceqsB} is necessary, since the spectral family of a
projection $\hat P$ is given by
\begin{equation}
E_{\l}^{\hat P}=\left\{
\begin{tabular}
[c]{ll}%
$\hat{0}$ & if $\l<0$\\
$\hat{1}-\P$ & if $0\leq\l<1$\\
$\hat{1}$ & if $\l\geq1.$%
\end{tabular}
\right.
\end{equation}

If $\A,\hat B$ are self-adjoint operators such that (i) either
$\A$ or $\hat B$ is a projection, or (ii) $[\A,\hat B]=\hat 0$,
then $\A\preceq_s\hat B \mbox{ if and only if } \A\preceq\hat B$.
Here `$\preceq$' denotes the usual ordering on
$\BH_\sa$.\footnote{ The `usual' ordering is $\A\preceq\hat B$ if
$\bra\psi\A\ket\psi \leq \bra\psi\hat B\ket\psi$ for all vectors
$\ket\psi\in\mathcal H$.}

Moreover, if $\A,\hat B$ are arbitrary self-adjoint operators,
then $\A\preceq_s\hat B$ implies $\A\preceq\hat B$, but not vice
versa in general. Thus the spectral order is a partial order on
$\BH_\sa$ that is coarser than the usual one.

\subsubsection{Daseinisation of Self-Adjoint Operators.}
De Groote's crucial observation was the following. Let
$\l\mapsto\hat E_\l$ be a spectral family in $\PH$ (or,
equivalently, a self-adjoint operator $\A$). Then, for each stage
$V$, the following maps:
\begin{eqnarray}
\l&\mapsto& \bigwedge_{\mu>\l}\dastoo{V}{E_\mu} \label{dasoE}\\
\l&\mapsto& \dastoi{V}{E_\l}                    \label{dasiE}
\end{eqnarray}
also define spectral families.\footnote{The reason  \eq{dasoE} and
\eq{dasiE} have a different form is that
$\l\mapsto\dastoi{V}{E_\l}$ is right continuous whereas
$\l\mapsto\dastoo{V}{E_\l}$ is not. On the other hand, the family
$ \l\mapsto \bigwedge_{\mu>\l}\dastoo{V}{E_\mu}$ \emph{is} right
continuous.} These spectral families lie in $\PV$ and hence, by
the spectral theorem, define self-adjoint operators in $V$. This
leads to the definition of the two daseinisations of an arbitrary
self-adjoint operator:

{\definition Let $\A$ be an arbitrary self-adjoint operator. Then
the \emph{outer} and $\emph{inner}$ daseinisations of $\A$ are
defined at each stage $V$ as:
\begin{eqnarray}
\dastoo{V}{A}&:=&\int_\mathR \l\,
d\big(\delta^i_V(\hat E^A_{\l}) \big),\label{Def:dastooVA}\\
\dastoi{V}{A}&:=&\int_{\mathR}\l\, d
\big(\bigwedge_{\mu>\l}\delta^o_V(\hat E^A_{\mu})\big),
                                            \label{Def:dastoiVA}
\end{eqnarray}
respectively. }

\

Note that for all $\l\in\mathR$, and for all stages $V$, we have
\begin{equation}
       \dastoi{V}{E_\l}\preceq\bigwedge_{\mu>\l}\dastoo{V}{E_\mu}
                                            \label{dasiE<dasioE}
\end{equation}
and hence, for all $V$,
\begin{equation}
       \dastoi{V}{A}\preceq_s\dastoo{V}{A}.
\end{equation}
This explains why the `$i$' and `$o$' superscripts in
\eqs{Def:dastooVA}{Def:dastoiVA} are defined the way round that
they are.

Both outer daseinisation \eq{Def:dastooVA} and inner daseinisation
\eq{Def:dastoiVA} can be used to `adapt' a self-adjoint operator
$\A$ to contexts $V\in\Ob{\V{}}$ that do not contain $\A$. (On the
other hand, if $\A\in V$, then $\dastoo{V}{A}=\dastoi{V}{A}=\A$.)

\subsubsection{Properties of Daseinisation.}
We will now list  some useful properties of daseinisation.

{\bf 1.} It is clear that the outer, and inner, daseinisation
operations can be extended to situations where the self-adjoint
operator $\A$ does not belong to $\BH_\sa$, or where $V$ is not an
{\em abelian} sub-algebra of $\BH$. Specifically, let
$\mathcal{N}$ be an arbitrary von Neumann algebra, and let
$\mathcal S\subset\mathcal{N}$ be a proper von Neumann sub-algebra
such that $\hat 1_{\mathcal{N}}=\hat 1_{\mathcal{S}}=\hat 1$. Then
outer and inner daseinisation can be defined as the mappings
\begin{eqnarray}
   \delo:\mathcal{N}_\sa&\map&\mathcal{S}_\sa           \nonumber\\
        \A&\mapsto&\int_\mathR \l\, d
        \big(\deli_{\mathcal{S}}(\hat E^A_{\l}) \big),\\[8pt]
   \deli:\mathcal{N}_\sa&\map&\mathcal{S}_\sa   \nonumber\\
         \A&\mapsto&\int_{\mathR}\l\,
d\big(\bigwedge_{\mu>\l}\delo_{\mathcal{S}}(\hat E^A_{\mu})\big).
\end{eqnarray}

A particular case is  $\mathcal{N}=V$ and $\mathcal{S}=V^{\prime}$
for two contexts $V,V^{\prime}$ such that $V^{\prime}\subset V$.
Hence,  a self-adjoint operator can be restricted from one context
to a sub-context.

For the moment, we will let $\mathcal{N}$ be an arbitrary von
Neumann algebra, with $\mathcal{S}\subset\mathcal{N}$.

{\bf 2.} By construction,
\begin{equation}
  \dastoo{\mathcal S}{A}=\bigwedge\{\hat B\in
                \mathcal{S}_\sa\mid\hat B\succeq_s\A\},
\end{equation}
where the minimum is taken with respect to the spectral order;
\ie\ $\dastoo{\mathcal S}{A}$ is the smallest self-adjoint
operator in $\mathcal S$ that is spectrally larger than (or equal
to) $\A$. This implies $\dastoo{\mathcal S}{A}\succeq\A$ in the
usual order. Likewise,
\begin{equation}
  \dastoi{\mathcal S}{A}=
  \bigvee\{\hat B\in\mathcal S_\sa\mid\hat B\preceq_s\A\},
\end{equation}
so $\dastoi{\mathcal S}{A}$ is the largest self-adjoint operator
in $\mathcal S$ spectrally smaller than (or equal to) $\A$, which
 implies $\dastoi{\mathcal S}{A}\preceq\A$.

{\bf 3.} In general, neither $\dastoo{\mathcal S}{A}$ nor
$\dastoi{\mathcal S}{A}$ can be written as Borel functions of the
operator $\A$, since daseinisation changes the elements of the
spectral family, while a function merely `shuffles them around'.

{\bf 4.} Let $\A\in\mathcal N$ be self-adjoint. The spectrum,
$\sp{(\A)}$, consists of all $\l\in\mathR$ such that the spectral
family $\{\hat E^A_\l\}_{\l\in\mathR}$ is non-constant on any
neighbourhood of $\l$. By definition, outer daseinisation of $\A$
acts on the spectral family of $\A$ by sending $\hat E^A_\l$ to
$\hat E^{\dastoo{\mathcal S}{A}}_\l=\delta^i(\hat
E^A_\l)_{\mathcal S}$. If $\{\hat E^A_\l\}_{\l\in\mathR}$ is
constant on some neighbourhood of $\l$, then the spectral family
$\{\hat E^{\dastoo{\mathcal S}{A}}_\l\}_{\l\in\mathR}$ of
$\dastoo{\mathcal S}{A}$ is also constant on this neighbourhood.
This shows that
\begin{equation}                \label{sp(das(A))IsInsp(A)}
    {\rm sp}(\dastoo{\mathcal S}{A})\subseteq\sp(\A)
\end{equation}
for all self-adjoint operators $\A\in\mathcal{N}_\sa$ and all von
Neumann sub-algebras $\mathcal S$. Analogous arguments apply to
inner daseinisation.

Heuristically, this result implies that the spectrum of the
operator $\dastoo{\mathcal S}{A}$ is more degenerate than that of
$\A$; \ie\ the effect of daseinisation is to `collapse'
eigenvalues.

{\bf 5.} Outer and inner daseinisation are both non-linear
mappings. We will show this for projections explicitly. For
example, let $\hat{Q}:=\hat{1}-\P$. Then
$\delta^{o}(\hat{Q}+\P)_{\mathcal S}= \dastoo{\mathcal
S}{1}=\hat{1}$, while $\delta^{o}(\hat{1}-\P)_{\mathcal
S}\succ\hat{1}-\P$ and $\dastoo{\mathcal S}{P}\succ\P$ in general,
so $\delta^{o}(\hat {1}-\P)_{\mathcal S}+\dastoo{\mathcal S}{P}$
is the sum of two non-orthogonal projections in general (and hence
not equal to $\hat{1}$). For inner daseinisation, we have
$\delta^{i}(\hat{1}-\P)_{\mathcal S}\prec\hat{1}-\P$ and
$\dastoi{\mathcal S}{P}\prec\P$ in general, so $\delta^{i}
(\hat{1}-\P)_{\mathcal S}+\dastoi{\mathcal S}{P}\prec\hat{1}
=\delta^{i}(\hat{1}-\P+\P)_{\mathcal S}$ in general.

{\bf 6.} If $a\geq 0$, then $\delta^o(a\A)_{\mathcal
S}=a\dastoo{\mathcal S}{A}$ and $\delta^i(a\A)_{\mathcal
S}=a\dastoi{\mathcal S}{A}$. If $a<0$, then
$\delta^o(a\A)_{\mathcal S}=a\dastoi{\mathcal S}{A}$ and
$\delta^i(a\A)_{\mathcal S}=a\dastoo{\mathcal S}{A}$. This is due
the behaviour of spectral families under the mapping $\A\mapsto
-\A$.

{\bf 7.} Let $\A$ be a self-adjoint operator, and let $\hat
E[A\leq\l]=\hat E^A_\l$ be an element of the spectral family of
$\A$. From \eq{Def:dastooVA} we get
\begin{equation}
\hat E[\delta^o_{\mathcal S}(A)\leq\l]=\delta^i_{\mathcal
S}\big(\hat E[A\leq\l])
\end{equation}
and then
\begin{eqnarray}
     \hat E[\dastoo{{\mathcal S}}{A}>\l] &=& \hat 1-
     \hat E[\dastoo{{\mathcal S}}{A}\leq\l]\\
      &=&\hat 1-\delta^i_{\mathcal S}\big(\hat E[A\leq\l]\big)\\
     &=& \delta^o_{\mathcal S}\big(\hat 1-\hat E[A\leq\l]\big)
                                       \label{Eoi>doE=}
\end{eqnarray}
where we have used the general result that, for any projection
$\P$, we have $\hat 1-\dastoi{{\mathcal S}}{P}=\delta^o_{\mathcal
S}(\hat 1-\P)$. Then, \eq{Eoi>doE=} gives
\begin{equation}
        \hat E[\dastoo{{\mathcal S}}{A}>\l]=
        \delta^o\big(\hat E[A>\l]\big)_{\mathcal S}.
                                \label{E[dA>l]=}
\end{equation}

\subsubsection{The de Groote Presheaves}  We know that
$V\mapsto \dastoo{V}{P}$ and $V\mapsto\dastoi{V}{P}$ are global
elements of the outer presheaf, $\G$, and inner presheaf, $\H$,
respectively. Using the daseinisation operation for self-adjoint
operators, it is straightforward to construct analogous presheaves
for which $V\mapsto\dastoo{V}{A}$ and $V\mapsto\dastoi{V}{A}$ are
global elements. One of these presheaves was  briefly considered
in \cite{deG05}. We call these the `de Groote presheaves' in
recognition of the importance of de Groote's work.

{\definition The {\em outer de Groote presheaf}, $\dG$, is defined
as follows:
\begin{enumerate}
\item[(i)] On objects $V\in\Ob{\V{}}$: We define  $\dG_V:=V_{\rm sa}$, the
collection of self-adjoint members of $V$.

\item[(ii)] On morphisms $i_{V^{\prime}V}:V^{\prime }\subseteq V:$
The mapping $\dG(i_{V^{\prime}\, V}):\dG_V \map\dG_{V^{\prime}}$
is given by
\begin{eqnarray}
        \dG(i_{V^{\prime}\, V})(\A )&:=&\dastoo{V^{\prime}}{A}\\
&=&\int_\mathR \l\,d\big(\delta^i(\hat E^A_\l)_{V^{\prime}}\big)\\
&=&\int_\mathR \l\,d\big(\H(i_{V^\prime\,V})(\hat E^A_\l)\big)
\end{eqnarray}
for all $\A\in\dG_V$.
\end{enumerate}
} Here we used the  fact that the restriction mapping
$\H(i_{V^\prime\,V})$ of the inner presheaf $\H$ is the inner
daseinisation of projections
$\delta^i:\PV\map\mathcal{P}(V^{\prime})$.

{\definition The {\em inner de Groote presheaf}, $\dH$, is defined
as follows:
\begin{enumerate}
\item[(i)] On objects $V\in\Ob{\V{}}$:  We define $\dH_V:=V_{\rm sa}$, the
collection of self-adjoint members of $V$.

\item[(ii)] On morphisms $i_{V^{\prime}V}:V^{\prime }\subseteq V:$
The mapping $\dH(i_{V^{\prime}\, V}):\dH_V \map\dH_{V^{\prime}}$
is given by
\begin{eqnarray}
     \dH(i_{V^{\prime}\, V})(\A)&:=&\dastoi{V^{\prime}}{A}\\
     &=&\int_\mathR \l\,  d\big(\bigwedge_{\mu>\l}
     (\delta^o(\hat E^A_\mu)_{V^{\prime}}\big)\\
   &=&\int_\mathR \l\,  d\big(\bigwedge_{\mu>\l}
   (\G(i_{V^\prime\,V})(\hat E^A_\mu)\big)
\end{eqnarray}
for all $\A\in\dG_V$ (where $\G(i_{V^\prime\,V})=\delta^o:
\PV\map\mathcal{P}(V^{\prime})$).
\end{enumerate}
}

It is now clear that, by construction,
$\dastoo{}{A}:=V\mapsto\dastoo{V}{A}$ is a global element of
$\dG$, and $\dastoi{}{A}:=V\mapsto\dastoi{V}{A}$ is a global
element of $\dH$.

De Groote found an example of an element of $\Ga\dG$ that is
\emph{not} of the form $\dastoo{}{A}$ (as mentioned in
\cite{deG05}). The same example can be used to show that there are
global elements of the outer presheaf $\G$ that are not of the
form $\dastoo{}{P}$ for any projection $\P\in\PH$.

On the other hand, we have:
\begin{theorem}\label{DasOfsaOpsInjective}
The mapping
\begin{eqnarray}
      \deli:\BH_\sa&\map&\Ga\dH\\
                        \hat A&\mapsto&\dastoi{}{A}
\end{eqnarray}
from self-adjoint operators in $\BH$ to global sections of the
outer de Groote presheaf is injective. Likewise,
\begin{eqnarray}
    \delo:\BH_\sa&\map&\Ga\dG\\
                        \hat A&\mapsto&\dastoo{}{A}
\end{eqnarray}
is injective.
\end{theorem}

\begin{proof}
By construction, $\hat A\geq_s\dastoi{V}{A}$ for all
$V\in\Ob{\V{}}$. Since $\hat A$ is contained in at least one
context, so
\begin{equation}
       \hat A=\bigvee_{V\in\Ob{\V{}}}\dastoi{V}{A},
\end{equation}
where the maximum is taken with respect to the spectral order. If
$\dastoi{}{A}=\dastoi{}{B}$, then we have
\begin{equation}
  \hat A=\bigvee_{V\in\Ob{\V{}}}\dastoi{V}{A}
                        =\bigvee_{V\in\Ob{\V{}}}\dastoi{V}{B}=\hat B.
\end{equation}
Analogously, $\hat A\leq_s\dastoo{V}{A}$ for all $V\in\Ob{\V{}}$,
so
\begin{equation}
                        \hat A=\bigwedge_{V\in\Ob{\V{}}}\dastoo{V}{A}.
\end{equation}
If $\dastoo{}{A}=\dastoo{}{B}$, then we have
\begin{equation}
                        \hat A=\bigwedge_{V\in\Ob{\V{}}}\dastoo{V}{A}
                        =\bigwedge_{V\in\Ob{\V{}}}\dastoo{V}{B}=\hat B.
\end{equation}
\end{proof}

The same argument also holds more generally for arbitrary von
Neumann algebras, not just $\BH$.

\section{The Presheaves \ps{$\sp(\A)^{\succeq}$}, \ps{$\mathR^{\succeq}$}\ and \ps{$\mathR^\leftrightarrow$}}
\label{Sec:psSR}
\subsection{Background to the Quantity-Value Presheaf \ps{$\R$}.}
Our goal now is to construct a `quantity-value' presheaf $\ps{\R}$
with the property that inner and/or outer daseinisation of an
self-adjoint operator $\A$ can be used to define an arrow, \ie\ a
natural transformation, from $\Sig$ to $\ps{\R}$.\footnote{In
fact, we will define several closely related presheaves that can
serve as a quantity-value object.}

The arrow corresponding to a self-adjoint operator $\A\in\BH$ is
denoted for now by $\breve{A}:\Sig\map\ps{\R}$. At each stage $V$,
we need a mapping
\begin{eqnarray}
              \breve{A}_V:\Sig_V &\map& \ps{\R}_V\\
              \l &\mapsto& \breve{A}_V(\l),
\end{eqnarray}
and we make the basic assumption that this mapping is given by
evaluation. More precisely, $\l\in\Sig_V$ is a spectral
element\footnote{A `spectral element', $\l\in\Sig_V$ of $V$, is a
multiplicative, linear functional $\l:V\map\mathC$ with
$\brak\l{\hat 1}=1$, see also Def. \ref{Def_SpectralPresheaf}.} of
$V$ and hence can be evaluated on operators lying in $V$. And,
while $\A$ will generally not lie in $V$, both the inner
daseinisation $\dastoi{V}{A}$ and the outer daseinisation
$\dastoo{V}{A}$ do.

Let us start by considering the operators $\dastoo{V}{A}$, $V\in
\Ob{\V{}}$. Each of these is a self-adjoint operator in the
commutative von Neumann algebra $V$, and hence, by the spectral
theorem, can be represented by a function, (the Gel'fand
transform\footnote{This use of the `overline' symbol for the
Gel'fand transform should not be confused with our later use of
the same symbol to indicate a co-presheaf.})
$\GT{\dastoo{V}{A}}:\Sig_V\map {\rm sp}(\dastoo{V}{A})$, with
values in the spectrum ${\rm sp}(\dastoo{V}{A})$ of the
self-adjoint operator $\dastoo{V}{A}$. Since the spectrum of a
self-adjoint operator is a subset of $\mathR$, we can also write
$\GT{\dastoo{V}{A}}:\Sig_V\map\mathR$. The question now is whether
the collection of maps $\GT{\dastoo{V}{A}}:\Sig_V\map\mathR$,
$V\in\Ob{\V{}}$, can be regarded as an arrow from $\Sig$ to some
presheaf $\ps{\R}$.

To answer this we need to see how these operators behave as we go
`down a chain' of sub-algebras $V^\prime\subseteq V$. The first
remark is that if $V^\prime\subseteq V$ then
$\dastoo{V^\prime}{A}\succeq\dastoo{V}{A}$. When applied to the
Gel'fand transforms, this leads to the equation
\begin{equation}
\GT{\dastoo{V^\prime}{A}}(\l|_{V^{\prime}})\ge
\GT{\dastoo{V}{A}}(\l)
                        \label{dasBVVprime>V}
\end{equation}
for all $\l\in\Sig_V$, where $\l|_{V^\prime}$ denotes the
restriction of the spectral element $\l\in\Sig_V$ to the
sub-algebra $V^\prime\subseteq V$. However, the definition of the
spectral presheaf is such that
$\l|_{V^\prime}=\Sig(i_{V^\prime\,V})(\l)$, and hence
\eq{dasBVVprime>V} can be rewritten as
\begin{equation}
\GT{\dastoo{V^\prime}{A}}\big(\Sig(i_{V^\prime\,V})(\l) \big)
        \geq \GT{\dastoo{V}{A}}(\l) \label{dasBVVprime>V2}
\end{equation}
for all $\l\in\Sig_V$.

It is a standard result that the Dedekind real number object,
$\ps{\mathR}$, in a presheaf topos $\Set^{\mathcal C^{op}}$ is the
\emph{constant} functor from $\mathcal C^{op}$ to $\mathR$
\cite{MM92}. It follows that the family of Gel'fand transforms,
$\GT{\dastoo{V}{A}}$, $V\in\Ob{\V{}}$,  of the daseinised
operators $\dastoo{V}A$, $V\in\Ob{\V{}}$, does \emph{not} define
an arrow from $\Sig$ to $\ps{\mathR}$, as this would require an
equality in \eq{dasBVVprime>V2}, which is not true. Thus the
quantity-value presheaf, $\ps{\R}$, in the topos $\SetH{}$ is
\emph{not} the real-number object $\ps{\mathR}$, although clearly
$\ps{\R}$ has \emph{something} to do with the real numbers. We
must take into account the growth of these real numbers as we go
from $V$ to smaller sub-algebras $V^{\prime}$. Similarly, if we
consider inner daseinisation, we get a series of falling real
numbers.

The presheaf, $\ps{\R}$, that we will choose, and which will be
denoted by $\PR{\mathR}$, incorporates both aspects (growing and
falling real numbers).

\subsection{Definition of the Presheaves \ps{$\sp(\A)^{\succeq}$}, \ps{$\mathR^{\succeq}$}\ and \ps{$\mathR^\leftrightarrow$}.}
\label{SubSec:SR} The inapplicability of the real-number object
$\underline{\mathR}$ may seem strange at first,\footnote{Indeed,
it puzzled us for a while!} but actually it is not that
surprising. Because of the Kochen-Specker theorem, we do not
expect to be able to assign (constant) real numbers as values of
physical quantities, at least not globally. Instead, we draw on
some recent results of M. Jackson \cite{Jac06}, obtained as part
of his extensive study of measure theory on a topos of presheaves.
Here, we use a single construction in Jackson's thesis: the
presheaf of `order-preserving functions' over a partially ordered
set---in our case, $\V{}$. In fact, we will need both
order-reversing and order-preserving functions.

\begin{definition}
Let $(\mathcal{Q},\preceq)$ and $(\mathcal{P},\preceq)$ be
partially ordered
sets. A function%
\begin{equation}
                \mu:\mathcal{Q}\map\mathcal{P}%
\end{equation}
is  \emph{order-preserving} if $q_{1}\preceq q_{2}$ implies $\mu
(q_{1})\preceq\mu(q_{2})$ for all $q_1,q_2\in\mathcal Q$. It is
\emph{order-reversing} if $q_{1}\preceq q_{2}$ implies
$\mu(q_{1})\succeq\mu(q_{2})$.  We denote by
$\mathcal{OP(Q},\mathcal{P)}$ the set of order-preserving
functions $\mu:\mathcal{Q}\map\mathcal{P}$, and by
$\mathcal{OR(Q},\mathcal{P)}$ the set of order-reversing
functions.
\end{definition}
\noindent We note that if $\mu$ is order-preserving, then $-\mu$
is order-reversing, and vice versa.

Adapting Jackson's definitions slightly, if $\cal P$ is any
partially-ordered set, we have the following. {\definition The
\emph{${\cal P}$-valued presheaf, $\ps{{\cal P}}^\succeq$, of
order-reversing functions over $\V{}$} is defined as follows:
\begin{enumerate}
\item[(i)] On objects $V\in\Ob{\V{}}$:
\begin{equation}
     \ps{{\cal P}}^\succeq_V:=\{\mu:\downarrow\!\!V\map {\cal P}\mid
     \mu\in\mathcal{OR}(\downarrow\!\! V,{\cal P})\}
                                                    \label{Def:PGe}
\end{equation}
where $\downarrow\!\!V\subset\Ob{\V{}}$ is the set of all von
Neumann sub-algebras of $V$.

\item[(ii)] On morphisms $i_{V^{\prime}V}:V^{\prime }\subseteq V:$
The mapping $\ps{{\cal P}}^\succeq(i_{V^{\prime}\, V}):\ps{{\cal
P}}^\succeq_V \map \ps{{\cal P}}^\succeq_{V^{\prime}}$ is given by
\begin{equation}
   \ps{{\cal P}}^\succeq(i_{V^{\prime}\, V})(\mu):=
   \mu_{|_{V^\prime}}
\end{equation}
where $\mu_{|_{V^\prime}}$ denotes the restriction of the function
$\mu$ to $\downarrow\!\!V^\prime\subseteq\downarrow\!\!V$.
\end{enumerate}
} \noindent Jackson uses order-preserving functions with
$\mathcal{P} := [0,\infty)$ (the non-negative reals), with the
usual order $\leq$.

Clearly, there is an analogous definition of the ${\cal P}$-valued
presheaf, $\ps{{\cal P}}^\preceq$, of order-preserving functions
from $\downarrow\!\! V$ to ${\cal P}$. It can be shown that
$\ps{{\cal P}}^\succeq$ and $\ps{{\cal P}}^\preceq$ are isomorphic
objects in $\SetH{}$.

Let us first consider $\ps{{\cal P}}^\succeq$. For us, the key
examples for the partially ordered set ${\cal P}$ are (i)
$\mathR$, the real numbers with the usual order $\leq$, and (ii)
$\rm{sp}(\A)\subset\mathR$, the spectrum of some bounded
self-adjoint operator $\A$, with the order $\leq$ inherited from
$\mathR$. Clearly, the associated presheaf $\SpA$ is a sub-object
of the presheaf $\SR$.

Now let $\A\in\BH_\sa$, and let $V\in \Ob{\V{}}$. Then to each
$\l\in\Sig_{V}$ there is associated the function
\begin{equation}
\dasBVo{V}{A}(\l):\downarrow\!\!V\map \sp(\A),
\end{equation}
given by
\begin{eqnarray}
\left(\dasBVo{V}{A}(\l)\right)  (V^{\prime})
    &:=& \GT{\dastoo{V^{\prime}}{A}}(\Sig (i_{V^{\prime}V})(\l))\\
    &=& \GT{\dastoo{V^{\prime}}{A}}(\l|_{V^{\prime}})\\
    &=& \brak{\l|_{V^{\prime}}}{\dastoo{V^{\prime}}{A}}\\
    &=& \brak\l{\dastoo{V^{\prime}}{A}}         \label{order-reversing_function}
\end{eqnarray}
for all $V^{\prime}\subseteq V$. We note that as $V^{\prime}$
becomes smaller, $\dastoo{V^{\prime}}{A}$ becomes larger (or stays
the same) in the spectral order, and hence in the usual order on
operators. Therefore, $\dasBVo{V}{A}(\l):\downarrow\!\!
V\map\sp(\A)$ is an \emph{order-reversing function}, for each
$\l\in\Sig_V$.

It is worth noting that daseinisation of $\A$, i.e., the
approximation of the self-adjoint operator $\A$ in the
\emph{spectral} order, allows to define a function
$\dasBVo{V}{A}(\l)$ (for each $\l\in\Sig_{V}$) with values in the
spectrum of $\A$, since we have
$\sp(\dastoo{V}{A})\subseteq\sp(\A)$, see
\eq{sp(das(A))IsInsp(A)}. If we had chosen an approximation in the
usual \emph{linear} order on $\BH_\sa$, then the approximated
operators would not have a spectrum that is contained in $\sp(\A)$
in general.

Let
\begin{eqnarray}
                \dasBVo{V}{A}:\Sig_V &\map& \SpA_V\\
                \l &\mapsto& \dasBVo{V}{A}(\l)
\end{eqnarray}
denote the set of order-reversing functions from $\downarrow\!\!
V$ to $\sp(\A)$ obtained in this way. We then have the following,
fundamental, result which can be regarded as a type of
`non-commutative' spectral theorem in which each bounded,
self-adjoint operator $\A$ is mapped to an arrow from $\Sig$ to
$\SR$:
\begin{theorem}\label{Th:ST}
The mappings $\dasBVo{V}{A}$, $V\in\Ob{\V{}}$, are the components
of a natural transformation/arrow $\dasBo{A}:\Sig\map\SpA$.
\end{theorem}

\begin{proof}
We only have to prove that, whenever $V^{\prime}\subset V$, the
diagram
\begin{center}
\setsqparms[1`1`1`1;1000`700]                                 
\square[\Sig_V`\SpA_V`\Sig_{V^{\prime}}`\SpA_{V^{\prime}};    
\dasBVo{V}{A}```\dasBVo{V^{\prime}}{A}]                       
\end{center}

\noindent commutes. Here, the vertical arrows are the restrictions
of the relevant presheaves from the stage $V$ to
$V^\prime\subseteq V$.

In fact, the commutativity of the diagram follows directly from
the definitions. For each $\l \in\Sig_{V}$, the composition of the
upper arrow and the right vertical arrow gives
\begin{equation}
    (\dasBVo{V}{A}(\l))|_{V^{\prime}}=
    \dasBVo{V^{\prime}}{A}(\l|_{V^{\prime}}),
\end{equation}
which is the same function that we get by first restricting $\l$
from $\Sig_V$ to $\Sig_{V^{\prime}}$ and then applying
$\dasBVo{V^{\prime}}{A}$.
\end{proof}

\

In this way, to each physical quantity $\A$ in quantum theory
there is assigned a natural transformation $\dasBo{A}$ from the
state object $\Sig$ to the presheaf $\SpA$. Since $\SpA$ is a
sub-object of $\SR$ for each $\A$, $\dasBo{A}$ can also be seen as
a natural transformation/arrow from $\Sig$ to $\SR$. Hence the
presheaf $\SR$ in the topos $\SetH{}$ is one candidate for the
quantity-value object of quantum theory. Note that it follows from
Theorem \ref{DasOfsaOpsInjective} that the mapping
\begin{eqnarray}
    \theta:\BH_\sa &\map& \rm {Hom}_{\SetH{}}(\Sig,\SR)\\
                \A &\mapsto& \dasBo{A}
\end{eqnarray}
is injective. \footnote{Interestingly, these results all carry
over to  an arbitrary von Neumann algebra $\mathcal
N\subseteq\BH$. In this way, the formalism is flexible enough to
adapt to situations where we have symmetries (which can described
mathematically by a von Neumann algebra $\mathcal N$ that has a
non-trivial commutant) and super-selection rules (which
corresponds to $\mathcal N$ having a non-trivial centre).}

If $S$ denotes our quantum system, then, on the level of the
formal language $\L{S}$, we expect the mapping $A\map\hat A$ to be
injective, where $A$ is a function symbol of signature
$\Si\map\R$. It follows that we have obtained a a faithful
representation of these function symbols by arrows
$\dasBo{A}:\Sig\map\SR$ in the topos $\SetH{}$.

Similarly, there is an order-preserving function
\begin{equation}
                \dasBVi{V}{A}(\l): \;\downarrow\!\! V\map\sp(\A),
\end{equation}
that is defined for all $V^{\prime}\subseteq V$ by
\begin{eqnarray}
   \left(  \dasBVi{V}{A}(\l)\right)  (V^{\prime})
 &=& \GT{\delta^{i}(\A)_{V^{\prime}}}(\Sig(i_{V^{\prime}V})(\l))\\
 &=& \brak\l{\delta^{i}(\A)_{V^{\prime}}}.              \label{order-preserving_function}
\end{eqnarray}
Since $\dastoi{V^{\prime}}{A}$ becomes smaller (or stays the same)
as $V^{\prime}$ gets smaller, $\dasBVi{V}{A}(\l)$ indeed is an
\emph{order-preserving} function from $\downarrow\!\! V$ to
$\sp(\A)$ for each $\l\in\Sig_V$. Again, approximation in the
spectral order (in this case from below) allows us to define a
function with values in $\sp(\A)$, which would not be possible
when using the linear order.

Clearly, we can use the functions
$\dasBVi{V}{A}(\l),\;\l\in\Sig_{V}$, to define a natural
transformation $\dasBi{A}:\Sig\map\OP$ from the spectral presheaf,
$\Sig$, to the presheaf $\OP$ of real-valued, order-preserving
functions on $\downarrow\!\! V$. The components of $\dasBi{A}$ are
\begin{eqnarray}
\dasBVi{V}{A}:\Sig_V &\map& \ps{\rm{sp}(\A)^\preceq}_V\\
                \l &\mapsto& \dasBVi{V}{A}(\l).
\end{eqnarray}
It follows from Theorem \ref{DasOfsaOpsInjective} that the mapping
from self-adjoint operators to natural transformations $\dasBi{A}$
is injective.

The functions obtained from inner and outer daseinisation can be
combined to give yet another presheaf, and one that will be
particularly useful for the physical interpretation of these
constructions. The general definition is the following.
\begin{definition}
Let $\mathcal{P}$ be a partially-ordered set. The
\emph{$\mathcal{P}$-valued presheaf, $\PR{\mathcal P}$, of
order-preserving and order-reversing functions on $\V{}$} is
defined as follows:

(i) On objects $V\in Ob(\V{})$:
\begin{equation}
  \ps{\mathcal{P}^{\leftrightarrow}}_{V}
  :=\{(\mu,\nu)\mid\mu\in\mathcal{OP}(\downarrow\!\!V,\mathcal{P}),
\nu\in\mathcal{OR}(\downarrow\!\!V,\mathcal{P}),\mu\leq\nu\}
\end{equation}
where $\downarrow\!\! V\subset\Ob{\V{}}$ is the set of all
sub-algebras $V^{\prime}$ of $V$. Note that we introduce the
condition $\mu\leq\nu$, i.e., for all $V'\in\downarrow\!\!V$ we
demand $\mu(V')\leq\nu(V')$.

(ii) On morphisms $i_{V^{\prime}V}:V^{\prime}\subseteq V$:
\begin{eqnarray}
  \ps{\mathcal{P}^{\leftrightarrow}}(i_{V^{\prime}V}):\ps
 {\mathcal{P}^{\leftrightarrow}}_{V}  &\longrightarrow&
 \ps{\mathcal{P}^{\leftrightarrow}}_{V^{\prime}}\\
 (\mu,\nu)  &\longmapsto& (\mu|_{V^{\prime}},\nu|_{V^{\prime}}),
\end{eqnarray}
where $\mu|_{V^{\prime}}$ denotes the restriction of $\mu$ to
$\downarrow\!\! V^{\prime}\subseteq\downarrow\!\! V$, and
analogously for $\nu|_{V^{\prime}}$.
\end{definition}
Note that since we have the condition $\mu\leq\nu$ in (i), the
presheaf $\PR{\mathcal P}$ is not simply the product of the
presheaves $\ps{\mathcal{P}}^\succeq$ and
$\ps{\mathcal{P}}^\preceq$.

As we will discuss shortly, the presheaf, $\PR{\mathR}$, of
order-preserving and order-reversing, real-valued functions is
closely related to the `$k$-extension' of the presheaf $\SR$ (see
the Appendix for details of the $k$-extension procedure).

Now let
\begin{equation}
  \dasBV{V}{A}:=\left(  \dasBVi{V}{A}(\cdot),\dasBVo{V}{A}
  (\cdot)\right):
  \Sig_{V}\map\PR{\mathR}_{V}
\end{equation}
denote the set of all pairs of order-preserving and
order-reversing functions from $\downarrow\!\! V$ to $\mathR$ that
can be obtained from inner and outer daseinisation. It is easy to
see that we have the following result:

\begin{theorem}\label{Th:ST_3}
The mappings $\dasBV{V}{A}$, $V\in \Ob{\V{}}$, are the components
of a natural transformation $\dasB{A}:\Sig\map\PR{\mathR}$.
\end{theorem}
Again from Theorem \ref{DasOfsaOpsInjective}, the mapping from
self-adjoint operators to natural transformations, $\hat A\rightarrow%
\dasB{A}$, is injective.

Since $\dasBVi{V}{A}(\l)\leq\dasBVo{V}{A}(\l)$ for all
$\l\in\Sig_V$, we can interpret each pair
$(\dasBVi{V}{A}(\l),\dasBVo{V}{A}(\l))$ of values as an
\emph{interval}, which gives a first hint at the physical
interpretation.

\subsection{Inner and Outer Daseinisation from Functions on Filters}
\label{SubSec_DasFromFctsOnFilters} There is a close relationship
between inner and outer daseinisation, and certain functions on
the filters in the projection lattice $\PH$ of $\BH$. We give a
summary of these results here: details can be found in de Groote's
work \cite{deG05,deG05b}, the article \cite{Doe05b}, and a
forthcoming paper \cite{Doe07}. This subsection serves as a
preparation for the physical interpretation of the arrows
$\dasB{A}:\Sig\map\PR{\mathR}$.

\paragraph{Filter bases, filters and ultrafilters.} We first need
some basic definitions. Let $\mathbb{L}$ be a lattice with zero
element $0$. A subset $f$ of $\mathbb{L}$ is called a \emph{filter
base} if (i) $0\neq f$ and (ii) for all $a,b\in f$, there is a
$c\in f$ such that $c\leq a\wedge b$.

A subset $F$ of a lattice $\mathbb{L}$ with zero element $0$ is a
(proper) \emph{filter} (or \emph{dual ideal}) if (i) $0\notin F$,
(ii) $a,b\in F$ implies $a\wedge b\in F$ and (iii) $a\in F$ and
$b\geq a$ imply $b\in F$. In other words, a filter is an upper set
in the lattice $\mathbb{L}$ that is closed under finite minima.

By Zorn's lemma, every filter is contained in a maximal filter.
Obviously, such a maximal filter is also a maximal filter base.

Let $\mathbb{L}'$ be a sublattice of $\mathbb{L}$ (with common
$0$), and let $F'$ be a filter in $\mathbb{L}'$. Then $F'$, seen
as a subset of $\mathbb{L}$, is a filter base in $\mathbb{L}$. The
smallest filter in $\mathbb{L}$ that contains $F'$ is the
\emph{cone over} $F'$ \emph{in} $\mathbb{L}$:
\begin{equation}
                        \mathcal{C}_{\mathbb{L}}(F'):=\{b\in\mathbb{L}\mid\exists
                        a\in F':a\leq b\}.
\end{equation}
This is nothing but the upper set $\uparrow\!\!F'$ of $F'$ in
$\mathbb{L}$.

In our applications, $\mathbb{L}$ typically is the lattice $\PH$
of projections in $\BH$, and $\mathbb{L}'$ is the lattice
$\mathcal{P}(V)$ of projections in an abelian sub-algebra $V$.

If $\mathbb{L}$ is a \emph{Boolean} lattice, i.e., if it is a
distributive lattice with minimal element $0$ and maximal element
$1$, and a complement (negation)
$\neg:\mathbb{L}\rightarrow\mathbb{L}$ such that $a\vee\neg a=1$
for all $a\in\mathbb{L}$, then we define an \emph{ultrafilter}
$\tilde{F}$ to be a maximal filter in $\mathbb{L}$. An ultrafilter
$\tilde{F}$ is characterised by the following property: for all
$a\in\mathbb{L}$, either $a\in\tilde{F}$ or $\neg a\in\tilde{F}$.
This can easily be seen: we have $a\vee\neg a=1$ by definition.
Let us assume that $\tilde{F}$ is an ultrafilter and
$a\notin\tilde{F}$. This means that there is some $b\in\tilde{F}$
such that $b\wedge a=0$. Using distributivity of the lattice
$\mathbb{L}$, we get
\begin{equation}
b=b\wedge(a\vee\neg a)=(b\wedge a)\vee(b\wedge\neg a)=b\wedge\neg
a,
\end{equation}
so $b\leq\neg a$. Since $b\in\tilde{F}$ and $\tilde{F}$ is a
filter, this implies $\neg a\in\tilde{F}$. Conversely, if $\neg
a\notin\tilde{F}$, we obtain $a\in\tilde{F}$.

The projection lattice $\mathcal{P}(V)$ of an abelian von Neumann
algebra $V$ is a Boolean lattice. The maximal element is the
identity operator $\widehat{1}$ and, as we saw earlier, the
complement of a projection is given as
$\neg\hat\a=\hat{1}-\hat\a$. Each ultrafilter $\tilde{F}$ in
$\mathcal{P}(V)$ hence contains either $\hat\a$ or
$\hat{1}-\hat\a$ for all $\hat\a\in\mathcal{P}(V)$.

\paragraph{Spectral elements and ultrafilters.} Let $V\in
\Ob{\V{}}$, and let $\l\in\Sig_{V}$ be a spectral element of the
von Neumann algebra $V$. This means that $\l$ is a multiplicative
state of $V$. For all projections $\hat\a\in\mathcal{P}(V)$, we
have
\begin{equation}
    \brak\l{\hat\a}=\brak\l{\hat\a^{2}}=
        \brak\l{\hat\a}\brak\l{\hat\a},
\end{equation}
and so $\brak\l{\hat\a}\in\{0,1\}$. Moreover,
$\brak\l{\hat{0}}=0$, $\brak\l{\hat{1}}=1$, and if
$\brak\l{\hat\a}=0$, then $\brak\l{\hat{1}-\hat\a}=1$ (since
$\brak\l{\hat\a}+ \brak\l{\hat{1}-\hat\a}=\brak\l{\hat{1}}$).
Hence, for each $\hat\a \in\mathcal{P}(V)$ we have either
$\brak\l{\hat\a}=1$ or $\brak\l{\hat{1}-\hat\a}=1$. This shows
that the family
\begin{equation}
    F_{\l}:=\{\hat\a\in\mathcal{P}(V)\mid\brak\l{\hat\a}=1\}
\end{equation}
of projections is an ultrafilter in $\mathcal{P}(V)$. Conversely,
each $\l\in\Sig_{V}$ is uniquely determined by the set
$\{\brak\l{\hat\a}\mid\hat\a\in\mathcal{P}(V)\}$ and hence by an
ultrafilter in $\mathcal{P}(V)$. This shows that there is a
bijection between the set $\mathcal{Q}(V)$ of ultrafilters in
$\mathcal{P}(V)$ and the Gel'fand spectrum $\Sig_{V}$.

\paragraph{Observable and antonymous functions.} Let $\mathcal{N}$ be
a von Neumann algebra, and let $\mathcal{F(N)}$ be the set of
filters in the projection lattice $\mathcal{P(N)}$ of
$\mathcal{N}$. De Groote has shown \cite{deG05b} that to each
self-adjoint operator $\A\in\mathcal{N}$, there corresponds a,
so-called, `observable function' $f_{\A}:\mathcal{F(N)}
\map\sp(\A)$. If $\mathcal{N}$ is abelian, $\mathcal{N}=V$, then
$f_{\A}|_{\mathcal{Q}(V)}$ is just the Gel'fand transform of $\A$.
However, it is striking that $f_{\A}$ can be defined even if
$\mathcal{N}$ is \emph{non}-abelian; for us, the important example
is $\mathcal{N}=\BH$.

If $\{\hat{E}_{\mu}^{A}\}_{\mu\in\mathR}$ is the spectral family
of $\A$, then $f_{\A}$ is defined as
\begin{eqnarray}
        f_{\A}:\mathcal{F(N)} &\map& \sp(\A)\nonumber\\
        F &\mapsto& \inf\{\mu\in\mathR\mid\hat{E}_{\mu}^A\in F\}.
\end{eqnarray}
Conversely, given a bounded function $f:\mathcal{D(H)}\map\mathR$
with certain properties, one can find a unique self-adjoint
operator $\A\in\BH$ such that $f=f_{\A}$.

It can be shown that each observable function is completely
determined by its restriction to the space of maximal filters
\cite{deG05b}. Let $Q(\mathcal{N})$ denote the space of maximal
filters in $\mathcal{P(N)}$. The sets
\begin{equation}
                        \mathcal{Q}_\P(\mathcal{N}):=\{F\in\mathcal{Q(N)}\mid\P\in F\},
                        \ \ \P\in\mathcal{P(N)},
\end{equation}
form the base of a totally disconnected topology on
$\mathcal{Q(N)}$. Following de Groote, this space is called the
\emph{Stone spectrum} of $\mathcal{N}$. If $\mathcal{N}$ is
abelian, $\mathcal{N}=V$, then, upon the identification of maximal
filters (which are ultrafilters) in $\mathcal{P}(V)$ and spectral
elements in $\Sig_V$, the Stone spectrum $\mathcal{Q}(V)$ is the
Gel'fand spectrum $\Sig_{V}$ of $V$.

This shows that for an arbitrary von Neumann algebra
$\mathcal{N}$, the Stone spectrum $\mathcal{Q(N)}$ is a
generalisation of the Gel'fand spectrum (the latter is only
defined for abelian algebras). The observable function $f_{\A}$ is
a generalisation of the Gel'fand transform of $\A$.

We want to show that the observable function $f_{\dastoo{A}{V}}$
of the outer daseinisation of $\A$ to $V$ can be expressed by the
observable function $f_{\A}$ of $\A$ directly. Since this works
for all $V\in\Ob{\V{}}$, we obtain a nice encoding of all the
functions $f_{\dastoo{V}{A}}$ and hence of the self-adjoint
operators $\dastoo{V}{A}$. The result (already shown in
\cite{deG05}) is that, for all stages $V\in\Ob{\V{}}$ and all
filters $F$ in $\mathcal{F}(V)$,
\begin{equation}\label{Eq_f_dastooVA(I)=f_A(C(I))}
                f_{\dastoo{V}{A}}(F)=f_{\A}(\mathcal{C}_{\BH}(F)).
\end{equation}
We want to give an elementary proof of this. We need

\begin{lemma} \label{L_deltaT^-1(D)=Cone(D)}
Let $\mathcal{N}$ be a von Neumann algebra, $\mathcal{S}$ a von
Neumann sub-algebra of $\mathcal{N}$, and let
$\delta_{\mathcal{S}}^{i}:\mathcal{P(N)}\rightarrow\mathcal{P(S)}$
be the inner daseinisation map on projections. Then, for all
filters $F\in\mathcal{F(S)}$,
\begin{equation}
(\delta_{\mathcal{S}}^{i})^{-1}(F)=\mathcal{C}_{\mathcal{N}}(F).
\end{equation}
\end{lemma}

\begin{proof}
If $\hat{Q}\in F\subset\mathcal{P(S)}$, then $(\delta_{\mathcal{S}}%
^{i})^{-1}(\hat{Q})=\{\P\in\mathcal{P(N)}\mid\delta^{i}
(\P)_{\mathcal{S}}=\hat{Q}\}$. Let $\P\in\mathcal{P(N)}$ be such
that there is a $\hat{Q}\in F$ with $\hat{Q}\leq\P$, i.e.,
$\P\in\mathcal{C}_{\mathcal{N}}(F)$. Then
$\delta^{i}(\P)_{\mathcal{S}} \geq\hat{Q}$, which implies
$\delta^{i}(\P)_{\mathcal{S}}\in F$, since $F$ is a filter in
$\mathcal{P(S)}$. This shows that
$\mathcal{C}_{\mathcal{N}}(F)\subseteq(\delta_{\mathcal{S}}^{i})^{-1}(F)$.
Now let $\P\in\mathcal{P(N)}$ be such that there is no $\hat{Q}\in
F$ with $\hat{Q}\leq\P$. Since $\delta
^{i}(\P)_{\mathcal{S}}\leq\P$, there also is no $\hat {Q}\in F$
with $\hat{Q}\leq\delta^{i}(\P)_{\mathcal{S}}$, so
$\P\notin(\delta_{\mathcal{S}}^{i})^{-1}(F)$. This shows that
$(\delta
_{\mathcal{S}}^{i})^{-1}(F)\subseteq\mathcal{C}_{\mathcal{N}}(F)$.
\end{proof}

\

We now can prove

\begin{theorem} \label{P_f_AEncodesAllf_delta^o(A)}
Let $\A\in\mathcal{N}_{sa}$. For all von Neumann sub-algebras
$\mathcal{S}\subseteq\mathcal{N}$ and all filters
$F\in\mathcal{F(S)}$, we have%
\begin{equation}
f_{\delta^{o}(\A)_{\mathcal{S}}}(F)=f_{\A}(\mathcal{C}%
_{\mathcal{N}}(F)).
\end{equation}
\end{theorem}

\begin{proof}
We have%
\begin{align*}
f_{\delta^{o}(\A)_{\mathcal{S}}(F)}  &  =\inf\{\lambda\in
\mathR\mid\hat{E}_{\lambda}^{\delta^{o}(\A)_{\mathcal{S}}}\in
F\}\\
&  =\inf\{\lambda\in\mathR\mid\delta^{i}(\hat{E}_{\lambda}^{A}%
)_{\mathcal{S}}\in F\}\\
&  =\inf\{\lambda\in\mathR\mid\hat{E}_{\lambda}^{A}\in(\delta
_{\mathcal{S}}^{i})^{-1}(F)\}\\
&  =\inf\{\lambda\in\mathR\mid\hat{E}_{\lambda}^{A}\in\mathcal{C}%
_{\mathcal{N}}(F)\}\\
&  =f_{\A}(\mathcal{C}_{\mathcal{N}}(F)).
\end{align*}
The second equality is the definition of outer daseinisation (on
the level of spectral projections, see \eq{dasiE}). In the
penultimate step, we used Lemma \ref{L_deltaT^-1(D)=Cone(D)}.
\end{proof}

\

This clearly implies (\ref{Eq_f_dastooVA(I)=f_A(C(I))}). We saw
above that to each $\l\in\Sig_{V}$ there corresponds a unique
ultrafilter $F_{\l}\in\mathcal{Q}(V)$. Since $\dastoo{V}{A}\in
V_{\sa}$, the observable function $f_{\dastoo{V}{A}}$ is the
Gel'fand transform of $\delta ^{o}(\A)_{V}$, and so, upon
identifying the ultrafilter $F_{\l}$ with the spectral element
$\l$, we have
\begin{equation}\label{Eq_f_doVA(F_lambda)=lambda(doVA)}
f_{\dastoo{V}{A}}(F_{\l})=\overline{\delta^{o}(\A)_{V}}(\l)
                            =\brak\l{\dastoo{V}{A}}.
\end{equation}
From (\ref{Eq_f_dastooVA(I)=f_A(C(I))}) we have
\begin{equation}
   \brak\l{\dastoo{V}{A}}=f_{\dastoo{V}{A}}(F_{\l})=
                f_{\A}(\mathcal{C}_{\BH}(F_{\l}))
\end{equation}
for all $V\in\Ob{\V{}}$ and for all $\l\in\Sig_V$. In this sense,
the observable function $f_{\A}$ encodes all the outer
daseinisations $\dastoo{V}{A}$, $V\in\Ob{\V{}}$, of $\A$.

There is also a function, $g_{\A}$, on the filters in $\PH$ that
encodes all the inner daseinisations $\dastoi{V}{A}$, $V\in
\Ob{\V{}}$. This function is given for an arbitrary von Neumann
algebra $\mathcal{N}$ by
\begin{eqnarray}
    g_{\A}:\mathcal{F(N)} &\map& \sp(\A)\\
    F &\mapsto& \sup\{\l\in\mathR\mid\hat{1}-\hat{E}_{\l}^A\in F\}
\end{eqnarray}
and is called the \emph{antonymous function} of $\A$
\cite{Doe05b}. If $\mathcal{N}$ is abelian, then
$g_{\A}|_{\mathcal{Q}(V)}$ is the Gel'fand transform of $\A$ and
coincides with $f_{\A}$ on the space $\mathcal{Q}(V)$ of maximal
filters, i.e., ultrafilters in $\mathcal{P}(V)$. As functions on
$\mathcal{F}(V)$, $f_{\A}$ and $g_{\A}$ are different also in the
abelian case. For an arbitrary von Neumann algebra $\mathcal{N}$,
the antonymous function $g_{\A}$ is another generalisation of the
Gel'fand transform of $\A$.

There is a close relationship between observable and antonymous
functions \cite{deG05b,Doe05b}: for all von Neumann algebras
$\mathcal{N}$ and all self-adjoint operators
$\A\in\mathcal{N}_{\sa}$, it holds that
\begin{equation}
                        -f_{\A}=g_{-\A}.
\end{equation}

There is a relation analogous to
(\ref{Eq_f_dastooVA(I)=f_A(C(I))}) for antonymous functions: for
all $V\in\Ob{\V{}}$ and all filters $F$ in $\mathcal{F}(V)$,
\begin{equation}\label{Eq_g_dastoiVA(I)=g_A(C(I))}
                g_{\dastoi{V}{A}}(F)=g_{\A}(\mathcal{C}_{\BH}(F)).
\end{equation}

This follows from

\begin{theorem} \label{P_g_AEncodesAllg_delta^i(A)}
Let $\A\in\mathcal{N}_{sa}$. For all von Neumann sub-algebras
$\mathcal{S}\subseteq\mathcal{N}$ and all filters
$F\in\mathcal{F(S)}$,
we have%
\begin{equation}
                        g_{\delta^{i}(\A)_{\mathcal{S}}}(F)=g_{\A}(\mathcal{C}
                        _{\mathcal{N}}(F)).
\end{equation}

\end{theorem}

\begin{proof}
We have%
\begin{align*}
g_{\delta^{i}(\A)_{\mathcal{S}}}(F)=  &  \sup\{\lambda\in
\mathR\mid\hat{1}-\hat{E}_{\lambda}^{\delta^{i}(\hat
{A})_{\mathcal{S}}}\in F\}\\
=  &
\sup\{\lambda\in\mathR\mid\hat{1}-\bigwedge_{\mu>\lambda}\delta
^{o}(\hat{E}_{\mu}^{A})_{\mathcal{S}}\in F\}\\
=  & \sup\{\lambda\in\mathR\mid\hat{1}-\delta^{o}(\hat{E}_{\lambda
}^{A})_{\mathcal{S}}\in F\}\\
=  & \sup\{\lambda\in\mathR\mid\delta^{i}(\hat{1}-\hat{E}_{\lambda
}^{A})_{\mathcal{S}}\in F\}\\
=  &  \sup\{\lambda\in\mathR\mid\hat{1}-\hat{E}_{\lambda}^{A}%
\in(\delta_{\mathcal{S}}^{i})^{-1}(F)\}\\
=  &  \sup\{\lambda\in\mathR\mid\hat{1}-\hat{E}_{\lambda}^{A}%
\in\mathcal{C}_{\mathcal{N}}(F)\}\\
=  &  g_{\A}(\mathcal{C}_{\mathcal{N}}(F)),
\end{align*}
where in the penultimate step we used Lemma
\ref{L_deltaT^-1(D)=Cone(D)}.
\end{proof}

\

Let $\l\in\Sig_{V}$, and let $F_{\l}\in\mathcal{Q}(V)$ be the
corresponding ultrafilter. Since $\dastoi{V}{A}\in V$, the
antonymous function $g_{\dastoi{V}{A}}$ is the Gel'fand
transform of $\dastoi{V}{A}$, and we have%
\begin{equation}\label{Eq_g_diVA(F_lambda)=lambda(diVA)}
                g_{\dastoi{V}{A}}(F_{\l})=
\overline{\delta^{i}(\A)_{V}}(\l)=\brak\l{\dastoi{V}{A}}.
\end{equation}
From (\ref{Eq_g_dastoiVA(I)=g_A(C(I))}), we get
\begin{equation}
 \brak\l{\dastoi{V}{A}}=g_{\dastoi{V}{A}}(F_{\l})=
                        g_{\A}(\mathcal{C}_{\BH}(F_{\l}))
\end{equation}
for all $V\in \Ob{\V{}}$ and all  $\l\in\Sig_{V}$. Thus the
antonymous function $g_\A$ encodes all the inner daseinisations
$\dastoi{V}{A}$, $V\in\Ob{\V{}}$, of $\A$.

\subsection{A Physical Interpretation of the Arrow
$\dasB{A}:$\:\ps{$\Sigma$}$\:\map\:$\ps{$\mathR^\leftrightarrow$}}
\label{SubSubSec:PhysIntArrow}

Let $\ket\psi\in\mathcal{H}$ be a unit vector in the Hilbert space
of the quantum system. The expectation value of a self-adjoint
operator $\A\in \BH$ in the state $\ket\psi$ is given by
\begin{equation}
  \bra\psi\A\ket\psi=
\int_{-||\A||}^{||\A||}\l\, d\bra\psi\hat{E}_{\l}^A\ket\psi.
\end{equation}

In the discussion of truth objects in section
\ref{Sec:TruthValues}, we introduced the maximal filter
$T^{\ket\psi}$ in $\PH$,\footnote{Since $\PH$ is not distributive,
$T^{\ket\psi}$ is not an ultrafilter; \ie\ there are projections
$\P\in\PH$ such that neither $\P\in T^{\ket\psi}$ nor $\hat
1-\P\in T^{\ket\psi}$.} given by (cf. \eq{Def:TpsiGlobal})
\begin{equation}
T^{\ket\psi}:=\{\hat\a\in\PH\mid\hat\a\succeq\ketbra\psi\},
\end{equation}
where $\ketbra\psi$ is the projection onto the one-dimensional
subspace of $\mathcal H$ generated by $\ket\psi$. As shown in
\cite{Doe05b}, the expectation value $\bra\psi\A\ket\psi$ can be
written as
\begin{equation}
         \bra\psi\A\ket\psi=
       \int_{g_{\A}(T^{\ket\psi})}^{f_{\A}(T^{\ket\psi})}
\l\, d\bra\psi\hat{E}_{\l}^A\ket\psi.
\end{equation}

In an instrumentalist interpretation,\footnote{Which we avoid in
general, of course!} one would interpret $g_{\A}(T^{\ket\psi})$,
resp.\ $f_{\A}(T^{\ket\psi})$, as the smallest, resp.\ largest,
possible result of a measurement of the physical quantity $A$ when
the state is $\ket\psi$. If $\ket\psi$ is an eigenstate of $\A$,
then $\bra\psi\A\ket\psi$ is an eigenvalue of $\A$, and in this
case, $\bra\psi\A\ket\psi\in\sp(\A);$  moreover,
\begin{equation}
\bra\psi\A\ket\psi =g_{\A}(T^{\ket\psi})=f_{\A}(T^{\ket\psi}).
\end{equation}
If $\ket\psi$ is not an eigenstate of $\A$, then
\begin{equation}
g_{\A}(T^{\ket\psi})<\bra\psi\A\ket\psi<f_{\A}(T^{\ket\psi}).
\end{equation}
For details, see \cite{Doe05b}.

Let $V$ be an abelian sub-algebra of $\BH$ such that $\Sig_V$
contains the spectral element, $\l^{\ket\psi}$, associated with
$\ket\psi$.\footnote{This is the element defined by
$\l^{\ket\psi}(\A):=\bra\psi\A\ket\psi$ for all $\A\in V$. It is
characterised by the fact that $\l^{\ket\psi}(\ketbra\psi)=1$ and
$\l^{\ket\psi}(\hat{Q})=0$ for all $\hat{Q}\in\mathcal{P}(V)$ such
that $\hat{Q}\ketbra\psi=\hat{0}$. We have $\l^{\ket\psi}
\in\Sig_V$ if and only if $\ketbra\psi\in\mathcal{P}(V)$.} The
corresponding ultrafilter in $\mathcal{P}(V)$ consists of those
projections $\hat{\a}\in\mathcal{P}(V)$ such that $\hat\a\succeq
\ketbra\psi$. This is just the evaluation, $\ps\TO^{\ket\psi}_V$,
at stage $V$ of our truth object, $\ps\TO^{\ket\psi}$; see
\eq{TOpsi2b}.

Hence the cone $\mathcal{C}(\ps\TO^{\ket\psi}_V):=%
\mathcal{C}_{\BH}(\ps\TO^{\ket\psi}_V)$ consists of all
projections $\hat{R}\in\PH$ such that $\hat{R}\succeq\ketbra\psi$;
and so, for all stages $V$ such that
$\ketbra\psi\in\mathcal{P}(V)$ we have
\begin{equation}
        \mathcal{C}(\ps\TO^{\ket\psi}_V)=T^{\ket\psi}.
\end{equation}
This allows us to write the expectation value as
\begin{eqnarray}
   \bra\psi\A\ket\psi
   &=&\int_{g_{\A}(\mathcal{C}
   (\ps\TO^{\ket\psi}_V)}^{f_{\A}(\mathcal{C}
   (\ps\TO^{\ket\psi}_V)}\l\, d\bra\psi\hat{E}_{\l}^A\ket\psi\\
  &=& \int_{g_{\dastoi{V}{A}}
  (\ps\TO^{\ket\psi}_V)}^{f_{\dastoo{V}{A}}
  (\ps\TO^{\ket\psi}_V)}\l\, d\bra\psi\hat{E}_{\l}^A\ket\psi
\end{eqnarray}
for these stages $V$.

Equations (\ref{Eq_f_doVA(F_lambda)=lambda(doVA)}) and
(\ref{Eq_g_diVA(F_lambda)=lambda(diVA)}) show that
$f_{\dastoo{V}{A}}(\ps\TO^{\ket\psi}_V)=\bra\psi\dastoo{V}{A}\ket\psi$
and \mbox{$g_{\delta^{i}(\A)_{V}}(\ps\TO^{\ket\psi}_V)=
\bra\psi\dastoi{V}{A}\ket\psi$.} In the language of
instrumentalism, for stages $V$ for which
$\l^{\ket\psi}\in\Sig_{V}$,  the value
$\bra\psi\dastoi{V}{A}\ket\psi\in\sp(\A)$ is the smallest possible
measurement result for $\A$ in the quantum state $\ket\psi$; and
$\bra\psi\dastoo{V}{A}\ket\psi\in\sp(\A)$ is the largest possible
result.

These results depend on the fact that we use (inner and outer)
daseinisation, i.e., approximations in the spectral, not the
linear order.

If $\l\in\Sig_{V}$ is not of the form $\l=\l^{\ket\psi}$, for
some$\ket\psi\in\Hi$,  then the cone $\mathcal{C}(F_{\l})$ over
the ultrafilter $F_{\l}$ corresponding to $\l$ cannot be
identified with a vector in $\Hi$. Nevertheless, the quantity
$\mathcal{C}(F_{\l})$ is well-defined, and
(\ref{Eq_f_dastooVA(I)=f_A(C(I))}) and
(\ref{Eq_g_dastoiVA(I)=g_A(C(I))}) hold. If we go from $V$ to a
sub-algebra $V^{\prime}\subseteq V$, then $\delta^{i}(\hat
{A})_{V^{\prime}}\preceq\dastoi{V}{A}$ and $\delta^{o}(\hat
{A})_{V^{\prime}}\succeq\dastoo{V}{A}$, hence
\begin{eqnarray}
     \brak\l{\delta^{i}(\A)_{V^{\prime}}} &\leq&
     \brak\l{\delta^{i}(\A)_{V}},\\
     \brak\l{\delta^{o}(\A)_{V^{\prime}}} &\geq&
     \brak\l{\delta^{o}(\A)_{V}}
\end{eqnarray}
for all $\l\in\Sig_V$.

We can interpret the function
\begin{eqnarray}
  \dasBV{V}{A}:\Sig_{V} &\map& \PR{\mathR}_{V}\\
  \l &\mapsto& \dasBV{V}{A}(\l)=
           \left( \dasBVi{V}{A}(\l),\dasBVo{V}{A}(\l) \right)
\end{eqnarray}
as giving the `spread' or `range' of the physical quantity $A$ at
stages $V^{\prime}\subseteq V$. Each element $\l\in\Sig_{V}$ gives
its own `spread' $\dasBV{V}{A}(\l):\downarrow\!\! V\map
\sp(\A)\times\sp(\A)$. The intuitive idea is that at stage $V$,
given a point $\l\in\Sig_{V}$, the physical quantity $A$ `spreads
over' the subset of the spectrum, $\sp(\A)$, of $\A$  given by the
closed interval of $\sp(\A)\subset\mathR$ defined by
\begin{equation}
  \lbrack\dasBVi{V}{A}(\l)(V),\dasBVo{V}{A}(\l)(V)]
                \cap\sp(\A)=
  \lbrack\brak\l{\dastoi{V}{A}},\brak\l{\delta^{o}(\A))_{V}}]
                \cap\sp(\A),
\end{equation}
For a proper sub-algebra $V^{\prime}\subset V$, the spreading is
over the (potentially larger) subset
\begin{equation}
  \lbrack\dasBVi{V}{A}(\l)(V^{\prime}),
        \dasBVo{V}{A}(\l)(V^{\prime})]\cap\sp(\A)=
  \lbrack\brak\l{\delta^{i}(\A)_{V^{\prime}}},
         \brak\l{\delta^{o}(\A)_{V^{\prime}}}]\cap\sp(\A).
\end{equation}

All this is local in the sense that these expressions are defined
at a stage $V$ and for sub-algebras, $V^\prime$, of $V$, where
$\l\in\Sig_{V}$. No similar global construction or interpretation
is possible, since the spectral presheaf $\Sig$ has no global
elements, i.e., no points (while the \emph{set} $\Sig_{V}$ does
have points).

As we go down to smaller sub-algebras $V^{\prime}\subseteq V$, the
spread gets larger. This comes from the fact that $\A$ has to be
adapted more and more as we go to smaller sub-algebras
$V^{\prime}$. More precisely, $\A$ is approximated from below by
$\delta^{i}(\hat {A})_{V^{\prime}}\in V^{\prime}$ and from above
by $\delta^{o}(\hat {A})_{V^{\prime}}\in V^{\prime}$. This
approximation gets coarser as $V^{\prime}$ gets smaller, which
basically means that $V^{\prime}$ contains less and less
projections.

It should be remarked that $\dasB{A}$ does not assign actual
values to the physical quantity $A$, but rather the possible
\emph{range} of such values; and these are independent of any
state $\ket\psi$.  This is analogous to the classical case where
physical quantities are represented by  real-valued functions on
state space. The range of possible values is state-independent,
but the actual value possessed by a physical quantity \emph{does}
depend on the state of the system.

\paragraph{The quantity-value presheaf $\PR{\mathR}$ as the interval
domain.} Spitters and Heunen observed \cite{HeuSpit07} that the
presheaf $\PR{\mathR}$ is the \emph{interval domain} in out topos.
This object has mainly been considered in theoretical computer
science \cite{Esc96} and can be used to systematically encode
situations where real numbers are only known---or can only be
defined---up to a certain degree of accuracy. Approximation
processes can be well described using the mathematics of domain
theory. Clearly, this has close relations to our physical
situation, where the real numbers are spectral values of
self-adjoint operators and coarse-graining (or rather the inverse
process of fine-graining) can be understood as a process of
approximation.

\subsection{The value of a physical quantity in a quantum state}
\label{subsec_ValueOfPhysQuantInQState} We now want to discuss how
physical quantities, represented by natural transformations
$\dasB{A}$, acquire `values' in a given quantum state. Of course,
this is not as straightforward as in the classical case, since
from the Kochen-Specker theorem, we know that physical quantities
do not have real numbers as their values. As we saw, this is
related to the fact that there are no microstates, i.e., the
spectral presheaf has no global elements.

In classical physics, a physical quantity $A$ is represented by a
function $\breve{A}:\mathcal{S}\rightarrow\mathR$ from the state
space $\mathcal{S}$ to the real numbers. A point $s\in\mathcal{S}$
is a microstate, and the physical quantity $A$ has the value
$\breve{A}(s)$ in this state.

We want to mimic this as closely as possible in the quantum case.
In order to do so, we take a pseudo-state
\begin{equation}
      \w^{\ket\psi}:= \dasmap(\ketbra\psi)=V\mapsto \bigwedge
       \{\hat\a\in\G_V\mid\ketbra\psi\preceq\hat\a\}
\end{equation}
(see (\ref{Def:wpsi_2})) and consider it as a sub-object of
$\Sig$. This means that at each stage $V\in\Ob{\V{}}$, we consider
the set
\begin{equation}
   \ps\w^{\ket\psi}_{V}:=\{\l\in\Sig_V\mid\brak\l{\delta^o(\ketbra\psi)_V}=1\}
                        \subseteq\Sig_V.        \label{wurst_sub-object}
\end{equation}
Of course, the sub-object of $\Sig$ that we get simply is
$\ps{\delta^o(\ketbra\psi)}$. Sub-objects of this kind are as
close to microstates as we can get, see the discussion in section
\ref{SubSubSec:TOQT} and \cite{Doe07b}. We can then form the
composition
\begin{equation}
 \ps\w^{\ket\psi}\rightarrow\Sig\overset{\dasB{A}}{\longrightarrow}
                        \PR{\mathR},
\end{equation}
which is also denoted by $\dasB{A}(\ps\w^{\ket\psi})$. One can
think of this arrow as being the `value' of the physical quantity
$A$ in the state described by $\ps\w^{\ket\psi}$.

The first question is if we actually obtain a sub-object of
$\PR{\mathR}$ in this way. Let $V,V'\in\Ob{\V{}}$, $V'\subseteq
V$. We have to show that
\begin{equation}
               \PR{\mathR}(i_{V'V})(\dasB{A}_V(\ps\w^{\ket\psi}_{V}))%
                \subseteq\dasB{A}_{V'}(\ps\w^{\ket\psi}_{V'}).
\end{equation}
Let $\l\in\ps\w^{\ket\psi}_V$, then
\begin{equation}
                        \PR{\mathR}(i_{V'V})(\dasB{A}_V(\l))=(\dasB{A}_V(\l))|_{V'}
                        =\dasB{A}_{V'}(\l|_{V'}).
\end{equation}
By definition, we have
\begin{equation}
        \ps\w^{\ket\psi}_{V'}=\Sig(i_{V'V})(\ps\w^{\ket\psi}_V)
             =\{\l|_{V'}\mid\l\in\ps\w^{\ket\psi}_V\},
\end{equation}
that is, \emph{every} $\l'\in\ps\w^{\ket\psi}_{V'}$ is given as
the restriction of some $\l\in\ps\w^{\ket\psi}_V$. This implies
that we even obtain the equality
\begin{equation}
  \PR{\mathR}(i_{V'V})(\dasB{A}_V(\ps\w^{\ket\psi}_V))
                        =\dasB{A}_{V'}(\ps\w^{\ket\psi}_{V'}),
\end{equation}
so $\dasB{A}(\ps\w^{\ket\psi})$ is indeed a sub-object of
$\PR{\mathR}$.

\paragraph{Values as pairs of functions and eigenvalues.} At each
stage $V\in\Ob{\V{}}$, we have pairs of order-preserving and
order-reversing functions $\dasB{A}(\l)$, one function for each
$\l\in\w^{\ket\psi}_V$. If $\ket\psi$ is an eigenstate of $\hat A$
and $V$ is an abelian sub-algebra that contains $\hat A$, then
$\dastoi{V}{A}=\dastoo{V}{A}=\hat A$. Moreover,
$\ps\w^{\ket\psi}_V$ contains the single element
$\l_{\ketbra\psi}\in\Sig_V$, which is the pure state that assigns
$1$ to $\ketbra\psi$ and $0$ to all projections in
$\mathcal{P}(V)$ orthogonal to $\ketbra\psi$.

Evaluating $\dasB{A}(\ps\w^{\ket\psi})$ at $V$ hence gives a pair,
consisting of an order-preserving function
$\dasBi{A}_V(\l_{\ketbra\psi}): \downarrow\!\!V\rightarrow\sp(\hat
A)$ and an order-reversing function
$\dasBo{A}_V(\l_{\ketbra\psi}):\downarrow\!\!V\rightarrow\sp(\hat
A)$:
\begin{equation}
\dasB{A}_V(\ps\w^{\ket\psi}_{V})=(\dasBi{A}_V(\l_{\ketbra\psi}),
    \dasBo{A}_V(\l_{\ketbra\psi})).
\end{equation}
The value of both functions at stage $V$ is $\overline{\hat A}
(\l_{\ketbra\psi})=\la\l_{\ketbra\psi},\hat A\ra$, which is the
eigenvalue of $\hat A$ in the state $\ket\psi$. In this sense, we
get back the ordinary eigenvalue of $\hat A$ when the system is in
the eigenstate $\psi$.

\paragraph{A simple example.} We consider the value of the self-adjoint
projection operator $\ketbra\psi$, seen as (the representative of)
a physical quantity, in the (pseudo-)state $\ps\w^{\ket\psi}$. We
remark that $\sp(\ketbra\psi)=\{0,1\}$. By definition,
\begin{equation}
                        \breve{\delta}(\ketbra\psi)_V(\l)=(\breve{\delta}^i(\ketbra\psi)_V(\l),%
                        \breve{\delta}^o(\ketbra\psi)_V(\l))
\end{equation}
for all $V\in\Ob{\V{}}$ and all $\l\in\Sig_V$. In particular, the
function
\begin{equation}
                \breve{\delta}^o(\ketbra\psi)_V(\l):\downarrow\!\!V\rightarrow\{0,1\}
\end{equation}
is given as (see (\ref{order-reversing_function}), for all
$V'\subseteq V$,
\begin{equation}
 \breve{\delta}^o(\ketbra\psi)_V(\l)(V')=
                        \brak\l{\delta^o(\ketbra\psi)_{V'}}.
\end{equation}
If $\l\in\ps\w^{\ket\psi}_{V}$, then
$\brak\l{\delta^o(\ketbra\psi)_V}=1$, see
(\ref{wurst_sub-object}). Hence, for all
$\l\in\ps\w^{\ket\psi}_V$, we obtain, for all $V'\subseteq V$,
\begin{equation}
 \breve{\delta}^o(\ketbra\psi)_V(\l)(V')=
                        \brak\l{\delta^o(\ketbra\psi)_{V'}}=1.  \label{OR_function_is_constant_1}
\end{equation}
If we denote the constant function on $\downarrow\!\!V$ with value
$1$ as $1_{\downarrow V}$, then we can write
\begin{equation}
                        \breve{\delta}(\ketbra\psi)_V(\l)=(\breve{\delta}^i(\ketbra\psi)_V(\l),%
                        1_{\downarrow V})
\end{equation}
for all $V$ and all $\l\in\ps\w^{\ket\psi}_V$. The constant
function $1_{\downarrow V}$ trivially is an order-reversing
function from $\downarrow\!\!V$ to $\sp(\ketbra\psi)$. We now
consider the function
\begin{equation}
    \breve{\delta}^i(\ketbra\psi)_V(\l):\downarrow\!\!V
    \rightarrow\{0,1\}.
\end{equation}
It is given as (see (\ref{order-preserving_function}), for all
$V'\subseteq V$,
\begin{equation}
    \breve{\delta}^i(\ketbra\psi)_V(\l)(V')=%
                        \brak\l{\delta^i(\ketbra\psi)_{V'}}.
\end{equation}
If $\ketbra\psi\in\mathcal{P}(V')$, then, for all
$\l\in\ps\w^{\ket\psi}_{V'}$, we have
$\brak\l{\delta^i(\ketbra\psi)_{V'}}=\brak\l{\ketbra\psi}=1$. If
$\ketbra\psi\notin\mathcal{P}(V')$, then
$\delta^i(\ketbra\psi)_{V'}=\hat 0$, since
$\delta^i(\ketbra\psi)_{V}\preceq\ketbra\psi$ and $\ketbra\psi$ is
a projection onto a one-dimensional subspace, so
$\delta^i(\ketbra\psi)_{V'}$ must project onto the
zero-dimensional subspace.

Thus we get, for all $V$, for all $\l\in\ps\w^{\ket\psi}_V$ and
all $V'\subseteq V$:
\begin{equation}
                        \breve{\delta}^i(\ketbra\psi)_V(\l)(V')=
                        \left\{
                        \begin{tabular}
                        [c]{ll}%
                        $1$ & if $\ketbra\psi\in V'$\\
                        $0$ & if $\ketbra\psi\notin V'$\\
                        \end{tabular}
                        \ \right.               \label{order-preserving_for_onedim_proj}
\end{equation}
Summing up, we have completely described the `value'
$\breve\delta(\ketbra\psi)(\ps\w^{\ket\psi})$ of the physical
quantity described by $\ketbra\psi$ in the pseudo-state given by
$\ps\w^{\ket\psi}$.

There is an immediate generalisation of one part of this result:
Consider an arbitrary non-zero projection $\P\in\V{}$,\footnote{Of
course, if $\P$ is not a projection onto a one-dimensional
subspace, then it cannot be identified with a state.} the
corresponding sub-object $\ps{\daso{P}}$ of $\Sig$ obtained from
outer daseinisaion, and the sub-object $\dasB{P}(\ps{\daso{P}})$
of $\PR{\mathR}$. For all $V\in\Ob{\V{}}$ and all
$\l\in\ps{\daso{P}}_{V}$,  a completely analogous argument to the
one given above shows that for the order-reversing functions
$\dasBo{P}_{V}(\l):\downarrow\!\!V\rightarrow\{0,1\}$, we always
obtain the constant function $1_{\downarrow V}$.

The behaviour of the order-preserving functions
$\dasBi{P}_{V}(\l):\downarrow\!\!V\rightarrow\{0,1\}$ is more
complicated than in the case that $\P$ projects onto a
one-dimensional subspace. In general, $\P\notin V'$ does not imply
$\dasBi{P}_{V'}(\l)(V')=0$ for $\l\in\ps{\daso{P}}_{V}$, so the
analogue of (\ref{order-preserving_for_onedim_proj}) does not hold
in general.

\subsection{Properties of \ps{$\mathR^\leftrightarrow$}.}
From the perspective of our overall programme, Theorem
\ref{Th:ST_3} is a key result and shows that $\PR{\mathR}$ is a
possible choice for the quantity-value object for quantum theory.
To explore this further, we start by noting some elementary
properties of the presheaf $\PR{\mathR}$. Analogous arguments
apply to the presheaves $\OP$ and $\SR$.
\begin{enumerate}
                \item The presheaf $\PR{\mathR}$ has global elements: namely, pairs
                of order-preserving and order-reversing functions on the partially-ordered
                set $\Ob{\V{}}$ of objects in the category $\V{}$; \ie\ pairs of functions
                $(\mu,\nu):\Ob{\V{}}\map\mathR$ such that:
                \begin{equation}
                                \forall V_1,V_2\in\Ob{\V{}}, V_2\subseteq V_1: \mu(V_2)\leq\mu(V_1),\nu(V_2)\geq\nu(V_1).
                        \label{Def:GammaRD}
                \end{equation}

                \item
                \begin{enumerate}
                        \item Elements of $\Ga\PR{\mathR}$ can be added: \ie\ if $(\mu_1,\nu_1),(\mu_2,\nu_2)\in\Ga\PR{\mathR}$,
                        define $(\mu_1,\nu_1)+(\mu_2,\nu_2)$ at each stage $V$ by
                                \begin{equation}
                        ((\mu_1,\nu_1)+(\mu_2,\nu_2))(V^\prime):=(\mu_1(V^\prime)+\mu_2(V^\prime),\nu_1(V^\prime)+\nu_2(V^\prime))
                                                                                                \label{Def:mu+nu}
                                \end{equation} for all $V^\prime\subseteq V$. Note that if
                                $V_2\subseteq V_1\subseteq V$, then $\mu_1(V_2)\leq \mu_1(V_1)$ and
                                $\mu_2(V_2)\leq\mu_2(V_1)$, and so $\mu_1(V_2)+\mu_2(V_2)\leq\mu_1(V_1)
                                +\mu_(V_1)$. Likewise, $\nu_1(V_2)+\nu_2(V_2)\geq\nu_1(V_1)
                                +\nu_2(V_1)$. Thus the definition of $(\mu_1,\nu_1)+(\mu_2,\nu_2)$ in \eq{Def:mu+nu}
                                makes sense. Obviously, addition is commutative  and associative.

                                \item  However, it is \emph{not} possible to define `$(\mu_1,\nu_1)-(\mu_2,\nu_2)$'
                                in this way since the difference between two order-preserving
                                functions may not be order-preserving, nor need the difference
                                of two order-reversing functions be order-reversing. This
                                problem is addressed in Section \ref{Sec:PskSR}.

                                \item A zero/unit element can be defined for the additive
                                structure on $\Ga\PR{\mathR}$ as $0(V):=(0,0)$ for all $V\in\Ob{\V{}}$, where
                                $(0,0)$ denotes a pair of two copies of the function that is constantly $0$
                                on $\Ob{\V{}}$.

                                It follows from (a) and (c) that $\Ga\PR{\mathR}$ is a commutative monoid
                                (\ie\ a semi-group with a unit).
                \end{enumerate}

    The commutative monoid structure for $\Ga\PR{\mathR}$ is a
                reflection of the stronger fact that $\PR{\mathR}$ is a
        \emph{commutative-monoid} object in the topos $\SetH{}$.
        Specifically, there is an arrow
        \begin{eqnarray}
                        &+:\PR{\mathR}\times\PR{\mathR}\map\PR{\mathR},\\
                        &+_V((\mu_1,\nu_1),(\mu_2,\nu_2)):=(\mu_1+\mu_2,\nu_1+\nu_2)
        \end{eqnarray}
        for all $(\mu_1,\nu_1),(\mu_2,\nu_2)\in\PR{\mathR}_V$,
        and for all stages $V\in\Ob{\V{}}$. Here, $(\mu_1+\mu_2,\nu_1+\nu_2)$ denotes the
        real-valued function on $\downarrow\!\!V$ defined by
        \begin{equation}
                (\mu_1+\mu_2,\nu_1+\nu_2)(V^\prime):=(\mu_1(V^\prime)+\mu_2(V^\prime),\nu_1(V^\prime)+\nu_2(V^\prime))
                                                \label{Def:mu+nuFn}
    \end{equation}
    for all $V^\prime\subseteq V$.

    \item The real numbers, $\mathR$, form a ring, and so it is
    natural to see if a multiplicative structure can be put on
    $\Ga\PR{\mathR}$. The obvious `definition' would be, for all $V$,
    \begin{equation}
                (\mu_1,\nu_1)(\mu_2,\nu_2)(V):=(\mu_1(V)\mu_2(V),\nu_1(V)\nu_2(V))
                                                                        \label{Def:ab:=}
    \end{equation}
    for $(\mu_1,\nu_1),(\mu_2,\nu_2)\in\Ga\PR{\mathR}$. However, this
    fails because the right hand side of \eq{Def:ab:=} may not be a pair
    consisting of an order-preserving and an order-reversing function. This problem
    arises, for example, if $\nu_1(V)$ and $\nu_2(V)$ become negative: then, as $V$ gets
    smaller, the product $\nu_1(V)\nu_2(V)$ gets larger and thus defines an
    \emph{order-preserving} function.
 \end{enumerate}

\subsection{The Representation of Propositions From Inverse Images}
\label{_Sec_RepOfPropositions} In Section
\ref{SubSec:PropLangPhys}, we introduced a simple propositional
language, $\PL{S}$, for each system $S$, and discussed its
representations for the case of classical physics. Then, in
Section \ref{Sec:QuPropSpec} we analysed the, far more
complicated, quantum-theoretical representation of this language
in the set of clopen subsets of the spectral presheaf, $\Sig$, in
the topos $\SetH{}$. This gives a representation of the primitive
propositions $\SAin\De$ as sub-objects of $\Sig$:
\begin{equation}
        \piqt(\Ain\De):=\ps{\delo\big(\hat E[A\in\De]\big)}
                                                \label{Def:piAinD}
\end{equation}
where `$\delta^o$' is the (outer) daseinisation operation, and
$\hat E[A\in\De]$ is the spectral projection corresponding to the
subset $\De\cap\sp(\A)$ of the spectrum, $\sp(\A)$, of the
self-adjoint operator $\A$.

We now want to remark briefly on the nature, and representation,
of propositions using the `local' language $\L{S}$.

In any classical representation, $\s$, of $\L{S}$ in $\Set$, the
representation, $\R_\s$, of the quantity-value symbol $\R$ is
always just the real numbers $\mathR$. Therefore, it is simple to
take a subset $\De\subseteq\mathR$ of $\mathR,$ and construct the
propositions $\SAin\De$. In fact, if $A_\s:\Si_\s\map\mathR$ is
the representation of the function symbol $A$ with signature
$\Si\map\R$, then $A_\s^{-1}(\De)$ is a subset of the symplectic
manifold $\Si_\s$ (the representation of the ground type $\Si$).
This subset, $A_\s^{-1}(\De)\subseteq\Si_\s$, represents the
proposition $\SAin\De$ in the Boolean algebra of all (Borel)
subsets of $\Si_\s$.

We should consider the analogue of these steps in the
representation, $\phi$, of the same language, $\L{S}$, in the
topos $\tau_\phi:=\SetH{}$.  In fact, the  issues to be discussed
apply to a representation in \emph{any} topos.

We first note that  if $\Xi $ is a sub-object of $\R_\phi$, and if
$A_\phi:\Si_\phi\map{\mathcal R}_\phi$, then there is an
associated sub-object of $\Si_\phi$, denoted $A_\phi^{-1}(\Xi)$.
Specifically, if $\chi_{\Xi}:\R_\phi\map\O_{\tau_\phi}$ is the
characteristic arrow of the sub-object $\Xi$, then
$A_\phi^{-1}(\Xi)$ is defined to be the sub-object of $\Si_\phi$
whose characteristic arrow is $\chi_{\Xi}\circ
A_\phi:\Si_\phi\map\O_{\tau_\phi}$. These sub-objects are
analogues of the subsets, $A_\s^{-1}(\De)$, of the classical state
space $\Si_\s$: as such, they can represent propositions. In this
spirit, we could denote by $\SAin\Xi$ the proposition which the
sub-object $A_\phi^{-1}(\Xi)$ represents, although, of course, it
would be a mistake to interpret $\SAin\Xi$ as asserting that the
value of something lies in something else: in a general topos,
there are no such values.

In the case of quantum theory, the arrows
$A_\phi:\Si_\phi\map\R_\phi$ are of the form
$\dasB{A}:\Sig\map\PR{\mathR}$ where $\R_\phi:=\PR{\mathR}$. It
follows that the propositions in our $\L{S}$-theory are
represented by the sub-objects $\dasB{A}^{-1}(\ps{\Xi})$ of
$\Sig$, where $\ps{\Xi}$ is a sub-object of $\PR{\mathR}$.

To interpret such propositions, note first that in the
$\PL{S}$-propositions $\SAin\De$,  the range `$\De$' belongs to
the world that is external to the language. Consequently, the
\emph{meaning} of $\De$ is given  independently of $\PL{S}$. This
`externally interpreted' $\De$ is then inserted into the quantum
representation of $\PL{S}$ via the daseinisation of propositions
discussed in Section \ref{Sec:QuPropSpec}.

However, the situation is very different for the
$\L{S}$-propositions $\SAin\Xi$. Here, the quantity `$\Xi$'
belongs to the particular topos $\tau_\phi$, and hence it is
representation dependent. The implication is that the `meaning' of
$\SAin\Xi$  can only be discussed from `within the topos' using
the internal language that is associated with $\tau_\phi$, which,
we recall, carries the translation of $\L{S}$ given by the
topos-representation $\phi$.

From  a conceptual perspective, this situation is  `relational',
with the meanings of the various propositions being determined by
their relations to each other as formulated in the internal
language of the topos. Concomitantly, the  meaning of `truth'
cannot be understood using the \emph{correspondence theory} (much
favoured by instrumentalists) for there is nothing external to
which a proposition can `correspond'. Instead, what is needed is
more like a \emph{coherence} theory of truth in which  a whole
body of propositions is considered together \cite{Gray90}. This is
a fascinating subject, but further discussion  must be deferred to
later work.

\subsection{The relation between the formal languages $\L{S}$ and
$\PL{S}$} In the propositional language $\PL{S}$, we have symbols
$\SAin{\Delta}$ representing primitive propositions. In the
quantum case, such a primitive proposition is represented by the
outer daseinisation $\ps{\daso{P}}$ of the projection
corresponding to the proposition. (The spectral theorem gives the
link between propositions and projections.)

We now want to show that the sub-objects of $\Sig$ of the form
$\ps{\daso{P}}$ for $\P\in\PH$ also are part of the language
$\L{S}$. More precisely, we will show that $\ps{\daso{P}}$ can be
obtained as the inverse image of a certain sub-object of
$\PR{\mathR}$.

The sub-object of $\PR{\mathR}$ that we consider is
$\dasB{P}(\ps{\daso{P}})$. We take the inverse image of this
sub-object under the natural transformation
$\dasB{P}:\Sig\rightarrow\PR{\mathR}$. This means that we assume
that the language $\L{S}$ contains a function symbol
$P:\Sigma\rightarrow\R$ that is represented by the natural
transformation $\dasB{P}$.

One more remark: although it may look as if we put in from the
start the sub-object $\ps{\daso{P}}$ that we want to construct,
this is not the case: we can take the inverse of an
\emph{arbitrary} sub-object of $\PR{\mathR}$, and we happen to
choose $\dasB{P}(\ps{\daso{P}})$. Forming the inverse image of a
sub-object of $\PR{\mathR}$ under the natural transformation
$\dasB{P}$ is analogous to taking the inverse image $f^{-1}\{r\}$
of some real value $r$ under some real-valued function $f$. The
real value $r$ can be given as the value $r=f(x)$ of the function
at some element $x$ of its domain. This does not imply that
$f^{-1}\{r\}=\{x\}$: the inverse image may contain more elements
than just $\{x\}$. Likewise, we have to discuss whether the
inverse image $\dasB{P}^{-1}(\dasB{P}(\ps{\daso{P}}))$ equals
$\ps{\daso{P}}$ or is some larger sub-object of $\Sig$.

We start with the case that $\P=\ketbra\psi$, i.e., $\P$ is the
projection onto a one-dimensional subspace.

\begin{theorem}
The inverse image $\breve\delta(\ketbra\psi)^{-1}(\breve\delta
(\ketbra\psi)(\ps\w^{\ket\psi}))$ is $\ps\w^{\ket\psi}$.
\end{theorem}

\begin{proof}
Let $\ps{S}:=\breve\delta(\ketbra\psi)^{-1}(\breve\delta
(\ketbra\psi)(\ps\w^{\ket\psi}))$. In any case, we have
\begin{equation}
                        \ps{S}\supseteq\ps\w^{\ket\psi}.
\end{equation}
Let us assume that the inclusion is proper. Then there exists some
$V\in\Ob{\V{}}$ such that
\begin{equation}
   \ps{S}(V)\supset\ps\w^{\ket\psi}_{V}=\{\l\in\Sig_V\mid\la\l,
              \delta^o(\ketbra\psi)_V\ra=1\},
\end{equation}
which is equivalent to the existence of some $\l_0\in \ps{S}_V$
such that
\begin{equation}
   \la{\l_0},\delta^o(\ketbra\psi)_V\ra=0.
                        \label{S_larger_gives_lambda_with_0}
\end{equation}
By definition, we have $\ps{S}(V)=\{\l\in\Sig_V\mid\breve\delta%
(\ketbra\psi)_V(\l)\in\breve
\delta(\ketbra\psi)_V(\ps\w^{\ket\psi})\}$. For all
$\tilde\l\in\ps\w^{\ket\psi}_{V}$, it holds that
\begin{equation}
   \la\tilde\l,\delta^o(\ketbra\psi)_V\ra=1,
\end{equation}
see (\ref{OR_function_is_constant_1}). This implies that for all
$\l\in \ps{S}(V)$, we must have
$\la\l,\delta^o(\ketbra\psi)_V\ra=1$, which contradicts
(\ref{S_larger_gives_lambda_with_0}). Hence there cannot be a
proper inclusion $\ps{S}\supset\ps\w^{\ket\psi}$, and we rather
have equality
\begin{equation}
    \ps{S}(V)=\ps\w^{\ket\psi}_V
\end{equation}
for all $V\in\Ob{\V{}}$.
\end{proof}

The proof is based on the fact that the order-reversing functions
of the form $\breve\delta^o(\ketbra\psi)_V(\l):\,\downarrow\!\!V
\rightarrow\{0,1\}$, where $V\in\Ob{\V{}}$ and
$\l\in\ps\w^{\ket\psi}_{V}$, are constant functions $1_{\downarrow
V}$. The remark at the end of Section
\ref{subsec_ValueOfPhysQuantInQState} shows that this holds more
generally for arbitrary non-zero projections $\P$. Hence we
obtain:

\begin{corollary}
The inverse image $\dasB{P}^{-1}(\dasB{P}(\ps{\daso{P}}))$ is
$\ps{\daso{P}}$.
\end{corollary}

\section{Extending the Quantity-Value Presheaf to an Abelian Group Object}
\label{Sec:PskSR}
\subsection{Preliminary Remarks}
We have shown how each self-adjoint operator, $\A$, on the Hilbert
space $\Hi$ gives rise to an arrow $\dasB{A}:\Sig\map\PR{\mathR}$
in the topos $\SetH{}$. Thus, in the topos representation, $\phi$,
of $\L{S}$ for the theory-type `quantum theory', the arrow
$\dasB{A}:\Sig\map\PR{\mathR}$ is one possible
choice\footnote{Another choice is to use the presheaf $\SR$ as the
quantity-value object, or the isomorphic presheaf $\OP$.} for  the
representation, $A_\phi:\Si_\phi\map{\mathcal R}_\phi$, of the
function symbol, $A:\Si\map\mathcal R$.

This implies that the quantity-value object, $\R_\phi$, is the
presheaf, $\PR{\mathR}$. However, although such an identification
is possible, it does impose certain restrictions on the formalism.
These stem from the fact that $\PR{\mathR}$ is only a
\emph{monoid}-object in $\SetH{}$, and $\Ga\PR{\mathR}$ is only a
monoid, whereas the real numbers of standard physics are an
abelian group; indeed, they are a commutative ring.

In standard classical physics, $\Hom{\Set}{\Si_{\s}}{\mathR}$ is
the set of real-valued functions on the manifold $\Si_{\s}$; as
such, it possesses the structure of a commutative ring. On the
other hand, the set of arrows $\Hom{\SetH{}}{\Sig}{\PR{\mathR}}$
has only the structure of an additive monoid. This additive
structure is defined locally in the following way. Let
$\alpha,\beta\in\Hom{\SetH{}}{\Sig}{\PR{\mathR}}$. At each stage
$V\in\Ob{\V{}}$, $\alpha_V$ is a pair $(\mu_{1,V},\nu_{1,V})$,
consisting of a function $\mu_{1,V}$ from $\Sig_V$ to $\OP_V$, and
a function $\nu_{1,V}$ from $\Sig_V$ to $\SR_V$. For each
$\l\in\Sig_V$, one has an order-preserving function
$\mu_{1,V}(\l):\downarrow\!\!V\rightarrow\mathR$, and an
order-reversing function
$\nu_{1,V}(\l):\downarrow\!\!V\rightarrow\mathR$. We use the
notation $\alpha_V(\l):=(\mu_{1,V}(\l),\nu_{1,V}(\l))$.

Similarly, $\beta$ is given at each stage $V$ by a pair of
functions $(\mu_{2,V},\nu_{2,V})$, and for all $\l\in\Sig_V$, we
write $\beta_V(\l):=(\mu_{2,V}(\l),\nu_{2,V}(\l))$

We define, for all $V\in\Ob{\V{}}$, and all $\l\in\Sig_V$, (c.f.\
\eq{Def:mu+nuFn})
\begin{eqnarray}
                (\alpha+\beta)_V(\l) &=& ((\mu_{1,V}(\l),\nu_{1,V}(\l))+(\mu_{2,V}(\l),\nu_{2,V}(\l)))\\
                &:=& (\mu_{1,V}(\l)+\mu_{2,V}(\l),\nu_{1,V}(\l)+\nu_{2,V}(\l))\\
                &=& \alpha_V(\l)+\beta_V(\l),
                                        \label{Def:a+b}
\end{eqnarray}
It is clear that $(\alpha+\beta)_V(\l)$ is a pair consisting of an
order-preserving and an order-reversing function for all $V$ and
all $\l\in\Sig_V$, so that $\alpha+\beta$ is well
defined.\footnote{To avoid confusion we should emphasise that, in
general, the sum $\das{A}+\das{B}$ is \emph{not} equal to
$\delta(\A+\hat B)$.}

Arguably, the fact that $\Hom{\SetH{}}{\Sig}{\PR{\mathR}}$ is only
a monoid\footnote{An internal version of this result would show
that the exponential object $\PR{\mathR}^{\,\Sig}$ is a monoid
object in the topos $\SetH{}$. This could well be true, but we
have not studied it in detail.} is  a weakness in so far as we are
trying to make quantum theory `look' as much like classical
physics as possible. Of course, in more obscure applications such
as Planck-length quantum gravity, the nature of the quantity-value
object is very much open for debate. But when applied to regular
physics, we might like our formalism to look more like classical
physics than the monoid-only structure of
$\Hom{\SetH{}}{\Sig}{\PR{\mathR}}$.

The need for a subtraction, i.e. some sort of abelian group
structure on $\PR{\mathR}$, brings to mind the well-known
Grothendieck $k$-construction that is much used in algebraic
topology and other branches of pure mathematics. This gives a way
of `extending' an abelian semi-group to become an abelian group,
and this technique can be adapted to the present situation. The
goal is to construct a `Grothendieck completion',
$k(\PR{\mathR})$, of $\PR{\mathR}$ that is an abelian-group object
in the topos $\SetH{}$.\footnote{Ideally, we might like $\kSR$ or
$k(\PR{\mathR})$ to be a commutative-ring object, but this is not
true.}

Of course, we can apply the $k$-construction also to the presheaf
$\SR$ (or $\OP$, if we like). This comes with an extra advantage:
it is then possible to define the square of an arrow
$\dasBo{A}:\Sig\rightarrow\SR$, as is shown in the Appendix.
Hence, given arrows $\dasBo{A}$ and $\dasBo{A^2}$, we can define
an `intrinsic dispersion':\footnote{The notation used here is
potentially a little misleading. We have not given any meaning to
`$A^2$' in the language $\L{S}$; \ie\ in its current form, the
language does not give meaning to the square of a function symbol.
Therefore, when we write $\dasBo{A^2}$ this must be understood as
being the Gel'fand transform of the outer daseinisation of the
operator $\A^2$.}
\begin{equation}
        \nabla(\A):=\dasBo{A^2}-\dasBo{A}^2.   \label{Def:nabla}
\end{equation}
Since the whole $k$-construction is quite complicated and is not
used in this paper beyond the present section, we have decided to
put all the relevant definitions into the Appendix where it can be
read at leisure by anyone who is interested.

Interestingly, there is a close relation between $\PR{\mathR}$ and
$\kSR$, as shown in the next subsection.

\subsection{The relation between \underline{$\mathR^\leftrightarrow$}\ and $k$(\underline{$\mathR^\succeq$}).}
In Section \ref{SubSec:SR}, we considered the presheaf
$\PR{\mathR}$ of order-preserving and order-reversing functions as
a possible quantity-value object. The advantage  of this presheaf
is the symmetric utilisation of inner and outer daseinisation, and
the associated physical interpretation of arrows from $\Sig$ to
$\PR{\mathR}$.

It transpires that $\PR{\mathR}$ is closely related to $\kSR$.
Namely, for each $V$, we can define an equivalence relation
$\equiv$ on $\PR{\mathR}_V$ by
\begin{equation}\label{EquivRelInPR}
                (\mu_1,\nu_1) \equiv (\mu_2,\nu_2)
\makebox{ iff } \mu_1+\nu_1=\mu_2+\nu_2.
\end{equation}
Then $\PR{\mathR}/\equiv$ is isomorphic to $\kSR$ under the
mapping
\begin{equation}
    [\mu,\nu]\mapsto[\nu,-\mu]\in\kSR_V
\end{equation}
for all $V$ and all
$[\mu,\nu]\in(\PR{\mathR}/\equiv)_V$.\footnote{ This
identification also explains formula (\ref{EquivRelInPR}), which
may look odd at first sight. Recall that $[\mu,\nu]\in(\PR{\mathR}/%
\equiv)_V$ means that $\mu$ is order-preserving and $\nu$ is
order-reversing.}

However, there is a difference between the arrows that represent
physical quantities. The arrow $[\dasBo{A}]:\Sig\map\kSR$ is given
by first sending $\A\in\BH_\sa$ to $\dasBo{A}$ and then taking
$k$-equivalence classes---a construction that only involves outer
daseinisation. On the other hand, there is  an arrow
$[\dasB{A}]:\Sig\map\PR{\mathR}/\equiv$, given by first sending
$\A$ to $\dasB{A}$ and then taking the equivalence classes defined
in (\ref{EquivRelInPR}). This involves both inner and outer
daseinisation.

We can show that $[\dasBo{A}]$ uniquely determines $\A$ as
follows: Let
\begin{equation}
     [\dasBo{A}]:\Sig\map\kSR
\end{equation}
denote the natural transformation from the spectral presheaf to
the abelian group-object $\kSR$, given by first sending $\A$ to
$\dasBo{A}$ and then taking the $k$-equivalence classes at each
stage $V$. The monoid $\SR$ is embedded into $\kSR$ by sending
$\nu\in\SR_V$ to $[\nu,0]\in\kSR_V$ for all $V$, which implies
that $\A$ is also uniquely determined by
$[\dasBo{A}]$.\footnote{In an analogous manner, one can show that
the arrows $\dasBi{A}:\Sig\map\OP$ and $[\dasBi{A}]:\Sig\map
k(\OP)$ uniquely determine $\A$, and that the arrow
$\dasB{A}:\Sig\map\PR{\mathR}$ also uniquely determines $\A$.} We
note that, currently, it is an open question if $[\dasB{A}]$ also
fixes $\A$ uniquely.

We now have constructed several presheaves that are abelian group
objects within $\SetH{}$, namely $k(\PR{\mathR})$, $\kSR$ and
$\PR{\mathR}/\equiv$. The latter two are isomorphic presheaves, as
we have shown. All three presheaves can serve as the
quantity-value presheaf if one wants to have an abelian-group
object for this purpose. Intuitively, if the quantity-value object
is only an abelian-monoid object like $\PR{\mathR}$, then the
`values' can only be added, while in the case of an abelian-group
object, they can be added and subtracted.

\subsection{Algebraic properties of the potential quantity-value presheaves}
As matters stand, we have several possible choices for the
quantity-value presheaf, which is the representation for quantum
theory of the symbol $\R$ of the formal language $\L{S}$ that
describes our physical system. In this sub-section, we want to
compare the algebraic properties of these various presheaves. In
particular, we will consider the presheaves $\SR$, $\kSR$,
$\PR{\mathR}$ and $k(\PR{\mathR})$.\footnote{The presheaf $\OP$ is
isomorphic to $\SR$ and hence will not be considered separately.}

\paragraph{Global elements.} We first note that all these presheaves have global elements. For example, a global element of $\SR$ is given by an order-reversing function $\nu:\V{}\map\mathR$. As remarked in the Appendix, we have $\Ga\kSR\simeq k(\Ga\SR)$. Global elements of $\PR{\mathR}$ are pairs $(\mu,\nu)$ consisting of an order-preserving function $\mu:\V{}\map\mathR$ and an order-reversing function $\nu:\V{}\map\mathR$. Finally, it is easy to show that $\Ga k(\PR{\mathR})\simeq k(\Ga\PR{\mathR})$.

\paragraph{The real number object as a sub-object.} In a presheaf topos $\Set^{\mathcal C^{op}}$, the Dedekind real number object $\ps\mathR$ is the constant functor from $\mathcal C^{op}$ to $\mathR$. The presheaf $\ps\mathR$ is an internal field object (see e.g. \cite{Jst02}).

The presheaf $\SR$ contains the constant presheaf $\ps\mathR$ as a
sub-object: let $V\in\Ob{\V{}}$ and $r\in\ps\mathR_V\simeq
\mathR$. Then the function $c_{r,V}:\downarrow\!\!V\map\mathR$
that has the constant value $r$ is an element of $\SR_V$ since it
is an order-reversing function. Moreover, the global sections of
$\ps\mathR$ are given by constant functions $r:\V{}\map\mathR$,
and such functions are also global sections of $\SR$.

The presheaf $\SR$ can be seen as a sub-object of $\kSR$: let
$V\in\V{}$ and $\nu\in\SR_V$, then $[\nu,0]\in\kSR_V$. Thus the
real number object $\ps\mathR$ is also a sub-object of $\kSR$.

A real number $r\in\ps\mathR_V$ defines the pair
$(c_{r,V},c_{r,V})$ consisting of two copies of the constant
function $c_{r,V}:\downarrow\!\!V\map\mathR$. Since $c_{r,V}$ is
both order-preserving and order-reversing, $(c_{r,V},c_{r,V})$ is
an element of $\PR{\mathR}_V$ and hence $\ps\mathR$ is a
sub-object of $\PR{\mathR}$. Since $\PR{\mathR}$ is a sub-object
of $k(\PR{\mathR})$, the latter presheaf also contains the real
number object $\ps\mathR$ as a sub-object.

\paragraph{Multiplying with real numbers and vector space structure.} Let $c_r\in\Ga\ps\mathR$ be the constant function on $\V{}$ with value $r$. The global element $c_r$ of $\ps\mathR$ defines locally, at each $V\in\V{}$, a constant function $c_{r,V}:\ \downarrow\!\!V\map\mathR$. We want to consider if, and how, multiplication with these constant functions is defined in the various presheaves. We will call this `multiplying with a real number'.

Let $\mu\in\SR_V$. For all $V'\in\downarrow\!\!V$, we define the
product
\begin{equation}
                (c_{r,V}\mu)(V'):=c_{r,V}(V')\mu(V')=r\mu(V').
\end{equation}
If $r\geq 0$, then $r\mu:\downarrow\!\!V\map\mathR$ is an
order-reversing function again. However, if $r<0$, then $r\mu$ is
order-preserving and hence not an element of $\SR_V$. This shows
that for the presheaf $\SR$ only multiplication by non-negative
real numbers is well-defined.

However, if we consider $\kSR$ then multiplication with an
arbitrary real number is well-defined. For simplicity, we first
consider $r=-1$, i.e., negation. Let $[\nu,\ka]\in\SR_V$, then,
for all $V'\in\downarrow\!\!V$,
\begin{eqnarray}
                (c_{-1,V}[\nu,\ka])(V') &:=& c_{-1,V}(V')[\nu(V'),\ka(V')]\\
                &=& -[\nu(V'),\ka(V')]\\
                &=& [\ka(V'),\nu(V')],
\end{eqnarray}
so we have
\begin{equation}
                c_{-1,V}[\nu,\ka]=-[\nu,\ka]=[\ka,\nu].
\end{equation}
This multiplication with the real number $-1$ is, of course,
defined in such a way that it fits in with the additive group
structure on $\kSR$.

It follows that multiplying an element $[\nu,\ka]$ of $\kSR_V$
with an arbitrary real number $r$ can be defined as
\begin{equation}
                c_{r,V}[\nu,\kappa]:=\left\{
                \begin{tabular}
                [c]{ll}%
                $[c_{r,V}\nu,c_{r,V}\ka]=[r\nu,r\ka]$ & if $r\geq0$\\
                $-[c_{-r,V}\nu,c_{-r,V}\ka]=[-r\ka,-r\nu]$ & if $r<0.$%
                \end{tabular}
                \right.
\end{equation}
\begin{remark} In this way, the group object $\kSR$ in $\SetH{}$ becomes a vector space object, with the field object $\ps\mathR$ as the scalars.
\end{remark}

Interestingly, one can define multiplication with arbitrary real
numbers also for $\PR{\mathR}$, although this presheaf is not a
group object in $\SetH{}$. Let $(\mu,\nu)\in\PR{\mathR}_V$, so
that $\mu:\downarrow\!\!V\map\mathR$ is an order-preserving
function and $\nu:\downarrow\!\!V\map\mathR$ is order-reversing.
Let $r$ be an arbitrary real number. We define
\begin{equation}
                c_{r,V}(\mu,\nu):=\left\{
                \begin{tabular}
                [c]{ll}%
                $(c_{r,V}\mu,c_{r,V}\nu)=(r\mu,r\nu)$ & if $r\geq0$\\
                $(c_{r,V}\nu,c_{r,V}\mu)=(r\nu,r\mu)$ & if $r<0.$%
                \end{tabular}
                \right.
\end{equation}
This is well-defined since if $\mu$ is order-preserving, then
$-\mu$ is order-reversing, and if $\nu$ is order-reversing, then
$-\nu$ is order-preserving. For $r=-1$, we obtain
\begin{equation}
                c_{-1,V}(\mu,\nu)=-(\mu,\nu)=(-\nu,-\mu).
\end{equation}
But this does \emph{not} mean that $-(\mu,\nu)$ is an additive
inverse of $(\mu,\nu)$. Such inverses do not exist in
$\PR{\mathR}$, since it is not a group object. Rather, we get
\begin{equation}
                (\mu,\nu)+(-(\mu,\nu))=(\mu,\nu)+(-\nu,-\mu)=(\mu-\nu,\nu-\mu).                 \label{pseudo-inverse}
\end{equation}
If, for all $V'\in\downarrow\!\!V$, we interpret the absolute
value $|(\mu-\nu)(V')|$ as a measure of uncertainty as given by
the pair $(\mu,\nu)$ at stage $V'$, then we see from
(\ref{pseudo-inverse}) that adding $(\mu,\nu)$ and $(-\nu,-\mu)$
gives a pair $(\mu-\nu,\nu-\mu)\in\PR{\mathR}_V$ concentrated
around $(0,0)$, but with an uncertainty twice as large (for all
stages $V'$). We call $-(\mu,\nu)=(-\nu,-\mu)$ the
\emph{pseudo-inverse} of $(\mu,\nu)\in\PR{\mathR}$.

More generally, we can define a second monoid structure (besides
addition) on $\PR{\mathR}$, called \emph{pseudo-subtraction} and
given by
\begin{equation}
(\mu_1,\nu_1)-(\mu_2,\nu_2):=(\mu_1,\nu_1)+(-\nu_2,-\mu_2)=
(\mu_1-\nu_2,\nu_1-\mu_2).
\end{equation}
This operation has a neutral element, namely $(c_{0,V},c_{0,V})$,
for all stages $V\in\V{}$, which of course is also the neutral
element for addition. In this sense, $\PR{\mathR}$ is close to
being a group object. Taking equivalence classes as described in
(\ref{EquivRelInPR}) makes $\PR{\mathR}$ into a group object,
$\PR{\mathR}/\equiv$, isomorphic to $\kSR$.

Since multiplication with arbitrary real numbers is well-defined,
the presheaf $\PR{\mathR}$ is `almost a vector space object' over
$\ps\mathR$.

Elements of $k(\PR{\mathR})_V$ are of the form
$[(\mu_1,\nu_1),(\mu_2,\nu_2)]$. Multiplication with an arbitrary
real number $r$ is defined in the following way:
\begin{equation}
                c_{r,V}[(\mu_1,\nu_1),(\mu_2,\nu_2)]:=\left\{
                \begin{tabular}
                [c]{ll}%
                $[(r\mu_1,r\nu_1),(r\mu_2,r\nu_2)]$ & if $r\geq0$\\
                $[(-r\mu_2,-r\nu_2),(-r\mu_1,-r\nu_1)]$ & if $r<0.$%
                \end{tabular}
                \right.
\end{equation}
The additive group structure on $k(\PR{\mathR})$ implies
\begin{equation}
                -[(\mu_1,\nu_1),(\mu_2,\nu_2)]=[(\mu_2,\nu_2),(\mu_1,\nu_1)],
\end{equation}
so the multiplication with the real number $-1$ fits in with the
group structure. On the other hand, this negation is completely
different from the negation on $\PR{\mathR}$ (where
$-(\mu,\nu)=(-\nu,-\mu)$ for all $(\mu,\nu)\in\PR{\mathR}_V$ and
all $V\in\Ob{\V{}}$).

\begin{remark} The presheaf $k(\PR{\mathR})$ is a vector space
object in $\SetH{}$, with $\ps\mathR$ as the scalars.
\end{remark}

\section{The Role of Unitary Operators}
\label{Sec:Unitary}
\subsection{The Daseinisation of Unitary Operators}
Unitary operators play an important  role in the formulation of
quantum theory, and we need to understand the analogue of this in
our topos formalism.

Unitary operators arise in the context of both `covariance' and
`invariance'. In elementary quantum theory, the `covariance'
aspect comes the fact that if we have made the associations
\begin{eqnarray*}
\mbox{Physical state }&\mapsto&
        \mbox{state vector $\ket\psi\in\Hi$}\\[5pt]
\mbox{Physical observable }A &\mapsto&
        \mbox{self-adjoint operator $\A$ acting on $\Hi$}
\end{eqnarray*}
then the same physical predictions will be obtained if the
following associations are used instead
\begin{eqnarray}
\mbox{Physical state }&\mapsto&
     \mbox{state vector $\U \ket\psi\in\Hi$} \label{UCovar}\\[5pt]
\mbox{Physical observable }A &\mapsto&
    \mbox{self-adjoint operator }\U\A\U^{-1}
                \mbox{ acting on $\Hi$}\nonumber
\end{eqnarray}
for any unitary operator $\hat U$. Thus the mathematical
representatives of physical quantities are defined only up to
arbitrary transformations of the type above. In non-relativistic
quantum theory, this leads to the canonical commutation relations;
the angular-momentum commutator algebra; and the unitary time
displacement operator. Similar considerations in relativistic
quantum theory involve the Poincar\'e group.

The `invariance' aspect of unitary operators arises when the
operator commutes with the Hamiltonian, giving rise to conserved
quantities.

\paragraph{Daseinisation of unitary operators.} As a side remark,
we first consider the question if daseinisation can be applied to
a unitary operator $\U$. The answer is clearly `yes', via the
spectral representation:
\begin{equation}
        \U=\int_{\mathR}e^{i\l} d\hat E^U_\l
\end{equation}
where $\l\mapsto E^\U_\l$ is the spectral family for $\U$. Then,
in analogy with \eqs{Def:dastooVA}{Def:dastoiVA} we have the
following: {\definition The \emph{outer daseinisation},
$\dastoo{}U$, resp.\ the \emph{inner daseinisation}, $\dastoi{}U$,
of a unitary operator $\U$ are defined as follows:
\begin{eqnarray}
\dastoo{V}{U}&:=&\int_\mathR e^{i\l}\,
       d\big(\delta^i_V(\hat E^U_{\l}) \big),\label{Def:dastooVU}
       \\
       \dastoi{V}{U}&:=&\int_{\mathR}e^{i\l}\, d
       \big(\bigwedge_{\mu>\l}\delta^o_V(\hat E^U_{\mu})\big),
\label{Def:dastoiVU}
\end{eqnarray}
at each stage $V$. }

To interpret these entities\footnote{It would be possible to
`complexify' the presheaf $\kSR$ in order to represent unitary
operators as arrows from $\Sig$\ to $\mathC\kSR$.  Similar remarks apply to the presheaf $\PR{\mathR}$. However, there
is no obvious physical use for this procedure.} we need to
introduce a new presheaf defined as follows.

{\definition The {\em outer, unitary de Groote presheaf}, $\dOU$,
is defined by:
\begin{enumerate}
\item[(i)] On objects $V\in\Ob{\V{}}$:  $\dOU_V:=V_{\rm un}$, the
collection of unitary operators in $V$.

\item[(ii)] On morphisms $i_{V^{\prime}V}:V^{\prime }\subseteq V:$
The mapping $\dOU(i_{V^{\prime}\, V}):\dOU_V
\map\dOU_{V^{\prime}}$ is given by
\begin{eqnarray}
        \dOU(i_{V^{\prime}\, V})(\hat\alpha)&:=&
                \dastoo{V^{\prime}}{\alpha}\\
&=&\int_\mathR e^{i\l}\,  d
\big(\delta^i(\hat E^\alpha_\l)_{V^{\prime}}\big)\\
&=&\int_\mathR e^{i\l}\,  d\big(\H(i_{V^\prime\,V})(\hat
E^\alpha_\l)\big)
\end{eqnarray}
for all $\hat\alpha\in\dOU_V$.
\end{enumerate}
} \noindent Clearly, (i) there is an analogous definition of an
`inner', unitary de Groote presheaf; and (ii) the map $V\mapsto
\dastoo{V}{U}$ defines a global element of $\dOU$.

This definition has the interesting consequence that, at each
stage $V$,
\begin{equation}
        \delta^o(e^{i\A})_V=e^{i\dastoo{V}{A}}
\end{equation}
A particular example of this construction is the one-parameter
family of unitary operators, $t\mapsto e^{it\hat H}$, where $\hat
H$ is the Hamiltonian of the system.

Of course, in our case everything commutes. Thus suppose
$g\mapsto\U_g$ is a representation of a Lie group $G$ on the
Hilbert space $\Hi$. Then these operators can be daseinised to
give the map $g\mapsto\delta^o(\U_g)$, but generally this is not a
representation of $G$ (or of its Lie algebra) since, at each stage
$V$ we have
\begin{equation}
 \delta^o(\U_{g_1})_V\delta^o(\U_{g_2})_V =
 \delta^o(\U_{g_2})_V\delta^o(\U_{g_1})_V
\end{equation}
for all $g_1,g_2\in G$. Clearly, there is an analogous result for
inner daseinisation.

\subsection{Unitary Operators and Arrows in $\SetH{}$.}
\subsubsection{The Definition of $\ell_{\hat U}:\Ob{\V{}}\map\Ob{\V{}}$}
In classical physics, the analogue of unitary operators are
`canonical transformations'; \ie\ symplectic diffeomorphisms from
the state space $\S$ to itself. This suggests that should try to
associate arrows in $\SetH{}$ with each unitary operator $\hat U$.

Thus we want to see if unitary operators can  act on the objects
in $\SetH{}$. In fact, if $\UH$ denotes the group of all unitary
operators in $\Hi$, we would like to find a \emph{realisation} of
$\UH$ in the topos $\SetH{}$.

As a first step, if $\U\in\UH$ and $V\in\Ob{\V{}}$ is an abelian
von Neumann sub-algebra of $\BH$, let us define
\begin{equation}
        \ell_\U(V):=\{\U\A\U^{-1}\mid \A\in V\}.
                                \label{Def:ellU(V)}
\end{equation}
It is clear that $\ell_\U(V)$ is a unital, abelian algebra of
operators, and that self-adjoint operators are mapped into
self-adjoint operators.  Furthermore, the map
$\A\mapsto\U\A\U^{-1}$ is continuous in the weak-operator
topology, and hence, if $\{\A_i\}_{i\in I}$ is a weakly-convergent
net of operators in  $V$, then $\{\U\A_i\U^{-1}\}_{i\in I}$ is a
weakly-convergent net of operators in $\ell_\U(V)$, and vice
versa.

It follows that $\ell_\U(V)$ is an abelian von Neumann algebra
(\ie\ it is weakly closed), and hence $\ell_\U$ can be viewed as a
map $\ell_\U:\Ob{\V{}}\map\Ob{\V{}}$.  We note the following:
\begin{enumerate}
\item Clearly, for all $\U_1,\U_2\in\UH$,
\begin{equation}
        \ell_{\U_1}\circ\ell_{\U_2}=\ell_{\U_1\U_2}
        \label{RepUHV(H)}
\end{equation}
Thus $\U\mapsto\ell_\U$ is a realisation of the group $\UH$ as a
group of transformations of $\Ob{\V{}}$.

\item For all $U\in\UH$, $V$ and $\ell_\U(V)$ are
\emph{isomorphic} sub-algebras of $\BH$, and
$\ell_\U^{-1}=\ell_{\U^{-1}}$.

\item If $V^\prime\subseteq V$, then, for all $\U\in\UH$,
\begin{equation}
        \ell_\U(V^\prime)\subseteq\ell_\U (V).
\end{equation}
Hence, each transformation $\ell_\U$ preserves the
partial-ordering of the poset category $\V{}$.

From this it follows  that each $\ell_\U:\Ob{\V{}}\map\Ob{\V{}}$
is a \emph{functor} from the category $\V{}$ to itself.

\item One consequence of the order-preserving property of
$\ell_\U$ is as follows. Let  $S$ be a sieve of arrows on $V$,
\ie\ a collection of sub-algebras of $V$ with the property that if
$V^\prime\in S$, then, for all $V^{\prime\prime}\subseteq
V^\prime$ we have $V^{\prime\prime}\in S$. Then
\begin{equation}
        \ell_\U(S):=\{\ell_\U(V^\prime)\mid V^\prime\in S\}
        \label{Def:ellU(S)}
\end{equation}
is a sieve of arrows on $\ell_\U(V)$.\footnote{In the partially
ordered set $\V{}$, an arrow from $V^{\prime}$ to $V$ can be
identified with the sub-algebra $V^{\prime}\subseteq V$, since
there is exactly one arrow from $V^{\prime}$ to $V$.}
\end{enumerate}

\subsubsection{The Effect of $\ell_\U$ on Daseinisation}
We recall that if $\P$ is any projection, then the (outer)
daseinisation, $\dastoo{V}{P}$, of $\P$ at stage $V$ is
(\eq{Def:dasouter})
\begin{equation}
\dastoo{V}{P}:=\bigwedge\big\{\hat{Q}\in\mathcal{P}(V)\mid
\hat{Q}\succeq \P\big\}
\end{equation}
where we have resorted once more to using the propositional
language $\PL{S}$. Thus
\begin{eqnarray}
        \U\dastoo{V}{P}\U^{-1}&=&\U\bigwedge
\big\{\hat{Q}\in\mathcal{P}(V)\mid\hat{Q}\succeq \P\big\}\U^{-1} \nonumber\\
        &=& \bigwedge\big\{\U\hat{Q}\U^{-1}\in\mathcal{P}
(\ell_\U(V))\mid\hat{Q}\succeq \P\big\} \nonumber\\
        &=& \bigwedge\big\{\U\hat{Q}\U^{-1}\in\mathcal{P}
(\ell_\U(V))\mid\U\hat{Q}\U^{-1}\succeq \U\P\U^{-1}\big\} \nonumber\\
        &=& \delo(\U\P\U^{-1})_{\ell_\U(V)}             \label{Udas(P)U_1=}
\end{eqnarray}
where we used the fact that the map $\hat Q\mapsto\U\hat Q\U^{-1}$
is weakly continuous.

Thus we have the important result
\begin{equation}
\U\dastoo{V}{P}\U^{-1}=\delo(\U\P\U^{-1})_{\ell_\U(V)}
                \label{Utransdas}
\end{equation}
for all unitary operators $\U$, and for all stages $V$. There is
an analogous result for inner daseinisation.

Equation \eq{Utransdas} can be applied to the de Groote presheaf
$\dG$ to give
\begin{equation}
\U\dastoo{V}{A}\U^{-1}=\delta^{o}(\U\A\U^{-1})_{\ell_\U(V)}
\end{equation}
for unitary operators $\U$, and all stages $V$.

We recall that the truth sub-object, $\ps\TO^{\ket\psi}$, of the
outer presheaf, $\G$, is defined at each stage $V$ by (cf
\eq{Def:dasinner})
\begin{eqnarray}
 \ps\TO^{\ket\psi}_V&:=&\{\hat\alpha\in \G_V\mid
                {\rm Prob}(\hat\alpha;\ket\psi)=1\}\label{TOpsi1Rep1}
                                                \nonumber\\[2pt]
        &=&\{\hat\alpha\in \G_V\mid
                \bra\psi\hat\alpha\ket\psi=1\}     \label{TOpsi2Rep1}
\end{eqnarray}
The neo-realist, physical interpretation of $\ps\TO^{\ket\psi}$ is
that the `truth' of the proposition represented by $\hat P$  is
\begin{eqnarray}
\TValM{\daso{P}\in\ps\TO^{\ket\psi}}_V &:=&\{V^\prime\subseteq
V\mid
\dastoo{V^\prime}{P}\in\ps\TO^{\ket\psi}_{V^\prime}\}  \\
        &=&\{V^\prime\subseteq V\mid \bra\psi\dastoo{V^\prime}{P}
                        \ket\psi=1\}\label{TVDas}
\end{eqnarray}
for all stages $V$. We then get
\begin{eqnarray}
            \nonumber &\ell_\U&\big(\TValM{\daso{P}\in\ps\TO^{\ket\psi}}_V\\
            &=& \ell_\U\{V^\prime\subseteq V\mid\bra\psi
        \dastoo{V^\prime}{P}\ket\psi=1\}\\
            &=& \{\ell_\U(V^\prime)\subseteq\ell_\U(V)\mid\bra\psi
        \dastoo{V^\prime}{P}\ket\psi=1\}\\
            &=& \{\ell_\U(V^\prime)\subseteq\ell_\U(V)\mid\bra\psi
        \U^{-1}\U\dastoo{V^\prime}{P}\U^{-1}\U\ket\psi=1\}\\
            &=& \{\ell_\U(V^\prime)\subseteq
                \ell_\U(V)\mid\bra\psi\U^{-1}\delo(\U\P\U^{-1})_{\ell_\U(V)}
        \U\ket\psi=1\}\\
            &=& \TValM{\delo(\U\P\U^{-1})\in\ps\TO^{\U\ket\psi}}_{\ell_\U (V).}
\end{eqnarray}
Thus we get the important result
\begin{equation}
\TValM{\delo(\U\P\U^{-1})\in\ps\TO^{\U\ket\psi}}_{\ell_\U(V)}=
        \ell_\U\big(\TValM{\daso{P}\in\ps\TO^{\ket\psi}}_V\big).
                \label{[]UCovar}
\end{equation}
This can be viewed as the topos analogue of the statement in
\eq{UCovar} about the invariance of the results of quantum theory
under the transformations $\ket\psi\mapsto\U\ket\psi$,
$\A\mapsto\U\A\U^{-1}$. Of course, there is a pseudo-state
analogue of all these expressions involving the sub-objects
$\ps\w^{\ket\psi}$, $\ket\psi\in\Hi$.

\subsubsection{The $\U$-twisted Presheaf}
Let us return once more to the definition \eq{Def:ellU(V)} of the
functor $\ell_\U:\V{}\map\V{}$. As we shall see later, any such
functor induces a `geometric morphism' from $\SetH{}$ to
$\SetH{}$. The exact definition is not needed here: it suffices to
remark that part of this geometric morphism is an arrow
$\ell_\U^*:\SetH{}\map\SetH{}$ defined by
\begin{equation}
        \ps{F}\mapsto\ell_\U^*\ps{F}:=\ps{F}\circ\ell_\U.
\end{equation}

Note that, if $\U_1,\U_2\in\UH$ then, for all presheaves $\ps{F}$,
\begin{eqnarray}
        \ell_{\U_2}^*(\ell_{\U_1}^*\ps{F})&=&
                (\ell_{\U_1}^*\ps{F})\circ\ell_{\U_2}=
        (\ps{F}\circ\ell_{\U_1})\circ\ell_{\U_2}\nonumber\\
        &=&\ps{F}\circ(\ell_{\U_1}\circ\ell_{\U_2})=
                        \ps{F}\circ\ell_{\U_1\U_2}\nonumber\\
&=&\ell_{\U_1\U_2}^*\ps{F}.
\end{eqnarray}
Since this is true for all functors $\ps{F}$ in $\SetH{}$, we
deduce that
\begin{equation}
        \ell_{\U_2}^*\circ\ell_{\U_1}^*=\ell_{\U_1\U_2}^*
\end{equation}
and hence the map $\U\mapsto\ell_\U^*$ is an (anti-)representation
of the group $\UH$ by arrows in the topos $\SetH{}$.

Of particular interest to us are the presheaves $\ell_\U^*\Sig$
and $\ell_U ^*\kSR$. We denote them by $\Sig^\U$ and $\kSR^\U$
respectively and say that they are `$\U$-twisted'.

{\theorem For each $\U\in\UH$, there is a natural isomorphism
$\iota:\Sig\map\Sig^\U$ as given in the following diagram
\begin{center}
\setsqparms[1`1`1`1;1000`700]
\square[\Sig_V`\Sig_V^\U`\Sig_{V^{\prime}}`\Sig_{V^\prime}^\U;
\iota^\U_V`\Sig(i_{V^\prime V})`\Sig^\U(i_{V^\prime
V})`\iota^\U_{V^\prime}]
\end{center}
where, at each stage $V$,
\begin{equation}
        (\iota^U_V(\l))(\A):=\brak\l{\U^{-1}\A\U}
\end{equation}
for all $\l\in\Sig_V$, and all $\A\in V_{\sa}$.}

The proof, which just involves chasing round the diagram above
using the basic definitions, is not included here.

Even simpler is the following theorem: {\theorem For each
$\U\in\UH$, there is a natural isomorphism
$\ka^\U:\SR\map(\SR)^\U$ whose components
$\ka^\U_V:\SR_V\map(\SR)^\U_V$ are given by
\begin{equation}
\ka^\U_V(\mu)(\ell_\U(V^\prime)):=\mu(V^\prime) \label{Def:kRRU}
\end{equation}
for all $V^\prime\subseteq V$.}

Here, we recall $\mu\in\SR_V$ is a function
$\mu:\downarrow\!\!V\map\mathR$ such that if $V_2\subseteq
V_1\subseteq V$ then $\mu(V_2)\geq\mu(V_1)$, \ie\ an
order-reversing function. In \eq{Def:kRRU} we have used the fact
that there is a bijection between the sets
$\downarrow\!\!\ell_\U(V)$ and $\downarrow\!\!V$.

Finally, {\theorem We have the following commutative diagram:
\begin{center}
                \setsqparms[1`1`1`1;1000`700]
                \square[\Sig`\Sig^\U`\SR`\SR^\U.;
                \iota^\U`\dasB{A}`\breve{\delta}(\U^{-1}\A\U)`\ka^\U]
\end{center}
}

\paragraph{The analogue of unitary operators for a general topos.}
It is  interesting to reflect on  the analogue of the above
constructions for a general topos. It soon becomes  clear that,
once again, we encounter the antithetical concepts of `internal'
and `external'.

For example, in the discussion above, the unitary operators and
the group $\UH$ lie outside the topos $\SetH{}$ and enter directly
from the underlying, standard quantum formalism. As such, they are
external to both the languages $\PL{S}$ and $\L{S}$.  We
anticipate that notions of `covariance' and `symmetry'  have
applications well beyond those in classical physics and quantum
physics. However,  at the very least, in a general topos one would
presumably replace the external $\UH$ with an internal group
object in the topos concerned. And, of course,  the notion of
`symmetry' is closely related to the concept to time, and time
development, which opens up a Pandora's box of possible
speculation. These issues are important, and await further
development.

\section{The Category of Systems}
\label{Sec:CatSys}
\subsection{Background Remarks}
We now return to the more general aspects of our theory, and study
its application  to a \emph{collection} of systems, each one of
which may be associated with a  different topos. For example, if
$S_1,S_2$ is a pair of systems, with associated topoi $\tau(S_1)$
and $\tau(S_2)$, and if $S_1$ is a sub-system of $S_2$, then we
wish to consider how $\tau(S_1)$ is related to $\tau(S_2)$.
Similarly, if a composite system is formed from a pair of systems
$S_1,S_2$, what relations are there between the topos of the
composite system and the topoi of the constituent parts?

Of course, in one sense, there is  only one true `system', and
that is the universe as a whole. Concomitantly,  there  is just
one local language, and one topos. However, in practice,  the
science community divides the universe conceptually into portions
that are sufficiently simple to be amenable to theoretical and/or
empirical discussion. Of course, this division is not unique,  but
it must be such that the coupling between portions is weak enough
that,  to a good approximation, their theoretical models can be
studied in isolation from each other. Such an essentially
isolated\footnote{The ideal monad has no windows.} portion of the
universe is called a `sub-system'. By an abuse of language,
sub-systems of the universe are usually called `systems'  (so that
the universe as a whole is one super-system), and then we can talk
about `sub-systems' of these systems; or `composites' of them; or
sub-systems of the composite systems, and so on.

In practice,  references by physicists to systems and
sub-systems\footnote{The word `sub-system' does not only mean a
collection of objects that is spatially localised. One could also
consider sub-systems of field systems by focussing on a just a few
modes of the fields as is done, for example, in the
Robertson-Walker model for cosmology. Another possibility would be
to use fields localised in some fixed space, or space-time region
provided that this is consistent with the dynamics.}   do not
generally signify \emph{actual} sub-systems of the real universe
but rather idealisations of possible systems. This is what a
physics lecturer  means when he or she starts a lecture by saying
``Consider a point particle moving in three dimensions.....''.

To develop these ideas further we need   mathematical control over
the  systems of interest, and their interrelations. To this end,
we start by focussing on some collection, $\Sys$, of physical
systems to which a particular theory-type is deemed to be
applicable. For example, we could consider   a collection of
systems that are to be discussed using the methodology of
classical physics; or systems to be discussed using standard
quantum theory; or whatever. For completeness, we  require that
every sub-system of a system in $\Sys$ is itself a member of
$\Sys$, as is every composite of members of $\Sys$.

We shall assume that  the systems in $\Sys$ are all associated
with local languages of the type discussed earlier, and that they
all have the \emph{same} set of ground  symbols which, for the
purposes of the present discussion, we take to be just $\Si$ and
$\R$. It follows that the languages $\L{S}$, $S\in\Sys$, differ
from each other only in the set of function symbols $\F{S}$; \ie\
the set of \emph{physical quantities}.

As a simple example of the system-dependence of the set of
function symbols  let system $S_1$ be a point particle moving in
one dimension, and let the set of physical quantities be
$\F{S_1}=\{x,p,H\}$. In the language $\L{S_1}$, these
function-symbols represent the position, momentum, and energy of
the system respectively. On the other hand, if $S_2$ is a particle
moving in three dimensions, then in the language $\L{S_2}$ we
could have $\F{S_2}= \{x,y,z,p_x,p_y,p_z,H\}$ to allow for
three-dimensional position and momentum. Or, we could decide to
add angular momentum as well, to give the set $\F{S_2}=
\{x,y,z,p_x,p_y,p_z,J_x,J_y,J_z,H\}$.

\subsection{The Category $\Sys$}
\subsubsection{The Arrows and Translations for the Disjoint Sum
$S_1\sqcup S_2$.}\label{SubSubSec:ATDS}
 The use of local languages is central to our
overall topos scheme, and therefore we need to understand, in
particular, (i) the relation between the languages $\L{S_1}$ and
$\L{S_2}$ if $S_1$ is a sub-system of $S_2$; and (ii) the relation
between $\L{S_1}$, $\L{S_2}$ and $\L{S_1\di S_2}$, where $S_1\di
S_2$ denotes the composite of systems $S_1$ and $S_2$.

These discussions can be made more precise by regarding $\Sys$ as
a category whose objects are the systems.\footnote{To control the
size of  $\Sys$  we  assume that the collection of objects/systems
is  a \emph{set} rather than a more general class.} The arrows in
$\Sys$ need  to cover two basic types of relation: (i) that
between $S_1$ and $S_2$ if $S_1$ is a `sub-system' of $S_2$; and
(ii) that between a composite system, $S_1\di S_2$, and its
constituent systems, $S_1$ and $S_2$.

This may seem straightforward but, in fact,  care is needed since
although the idea  of a `sub-system'  seems intuitively clear, it
is hard to give a physically acceptable definition that is
universal. However,  some insight into this idea can be gained by
considering its meaning in classical physics. This is very
relevant for the general scheme since one of our main goals is to
make all theories `look' like classical physics in the appropriate
topos.

To this end, let $S_1$ and $S_2$ be classical systems whose state
spaces are the symplectic manifolds $\S_1$ and $\S_2$
respectively. If $S_1$ is deemed to be a sub-system of $S_2$, it
is natural to require that $\S_1$ is a \emph{sub-manifold} of
$\S_2$, \ie\ $\S_1\subseteq\S_2$. However, this condition cannot
be used as a \emph{definition} of a `sub-system' since the
converse may not be true: \ie\ if $\S_1\subseteq\S_2$, this does
not necessarily mean that, from a physical perspective, $S_1$
could, or would, be said to be a sub-system of $S_2$.\footnote{
For example, consider the diagonal sub-manifold
$\De(\S)\subset\S\times\S$ of the symplectic manifold $\S\times\S$
that represents the composite $S\di S$ of two copies of a system
$S$. Evidently, the states in $\De(\S)$ correspond to the
situation in which both copies of $S$\ `march together'. It is
doubtful if this would be recognised physically as a sub-system.}

On the other hand, there are situations where being a sub-manifold
clearly \emph{does} imply being a physical sub-system. For
example, suppose the state space $\S$ of a system $S$ is a
disconnected manifold with two components $\S_1$ and $\S_2$, so
that $\S$ is the disjoint union, $\S_1\coprod\S_2$, of the
sub-manifolds $\S_1$ and $\S_2$.  Then it  seems physically
appropriate to say that the system $S$ itself is disconnected, and
to write $S=S_1\sqcup S_2$ where the symplectic manifolds that
represent the sub-systems $S_1$ and $S_2$ are $\S_1$ and $\S_2$
respectively.

One reason why it is reasonable to call $S_1$ and $S_2$
`sub-systems' in this particular situation is that any continuous
dynamical evolution of a state point in $\S\simeq\S_1\sqcup \S_2 $
will always lie in either one component or the other. This
suggests that perhaps, in general, a necessary condition for a
sub-manifold $\S_1\subseteq\S_2$ to represent a physical
sub-system is that the dynamics of the system $S_2$ must be such
that $\S_1$ is  mapped into itself  under the dynamical evolution
on $\S_2$; in other words, $\S_1$ is a
\emph{dynamically-invariant} sub-manifold of $\S_2$. This
correlates with the idea mentioned earlier that sub-systems are
weakly-coupled with each other.

However, such a dynamical restriction is not something that should
be coded into the languages, $\L{S_1}$ and $\L{S_2}$: rather, the
dynamics is to be  associated with the \emph{representation} of
these languages in the appropriate topoi.

Still, this caveat does not apply to the disjoint sum $S_1\sqcup
S_2$ of two systems $S_1,S_2$, and we will assume that, in
general, (\ie\ not just in classical physics) it is legitimate to
think of $S_1$ and $S_2$ as being sub-systems of $S_1\sqcup S_2$;
something that we indicate by defining  arrows $i_1:S_1\map
S_1\sqcup S_2$, and $i_2:S_2\map S_1\sqcup S_2$ in $\Sys$.

To proceed further it is important to understand the connection
between the putative arrows in the category $\Sys$, and the
`translations' of the associated languages. The first step is to
consider what can be said about the relation between $\L{S_1\sqcup
S_2}$, and $\L{S_1}$ and $\L{S_2}$. All three languages share the
same ground-type symbols, and so what we are concerned with is the
relation between the function symbols of signature $\Si\map\R$ in
these languages.

By considering what is meant intuitively by the disjoint sum, it
seems plausible that each physical quantity for the system
$S_1\sqcup S_2$ produces a physical quantity for $S_1$, and
another one for $S_2$. Conversely, specifying a pair of physical
quantities---one for $S_1$ and one for $S_2$---gives a physical
quantity for $S_1\sqcup S_2$. In other words,
\begin{equation}
        \F{S_1\sqcup S_2}\simeq \F{S_1}\times\F{S_2}\label{FS1sumS2}
\end{equation}
However, it is important not to be too dogmatic about statements
of this type since in non-classical theories  new possibilities
can arise that are counter to intuition.

Associated with \eq{FS1sumS2} are the maps $\L{i_1}:\F{S_1\sqcup
S_2}\map\F{S_1}$ and $\L{i_2}:\F{S_1\sqcup S_2}\map\F{S_2}$,
defined as  the projection maps of the product. In the theory of
local languages, these transformations are essentially
\emph{translations} \cite{Bell88} of $\L{S_1\sqcup S_2}$ in
$\L{S_1}$ and $\L{S_2}$ respectively; a situation that we denote
$\L{i_1}:\L{S_1\sqcup S_2}\map\L{S_1}$, and $\L{i_2}:\L{S_1\sqcup
S_2}\map\L{S_2}$.

To be more precise, these operations are translations if, taking
$\L{i_1}$ as the explanatory example, the map
$\L{i_1}:\F{S_1\sqcup S_2}\map \F{S_1}$ is supplemented with the
following map from the ground symbols of $\L{S_1\sqcup S_2}$ to
those of $\L{S_1}$:
\begin{eqnarray}
        \L{i_1}(\Si)&:=&\Si,            \label{tSigma}\\
        \L{i_1}(\R)&:=&\R,              \label{tR}    \\
        \L{i_1}(1)&:=&1,                \label{t1}    \\
        \L{i_1}(\Omega)&:=&\Omega.       \label{tOmega}
\end{eqnarray}
Such a translation map is then extended  to all type symbols using
the definitions
\begin{eqnarray}
        \L{i_1}(T_1\times T_2\times\cdots\times T_n)&=&
        \L{i_1}(T_1)\times \L{i_1}(T_2)  \times\cdots\times
        \L{i_1}(T_n), \\[2pt]
        \L{i_1}(PT)&=&P[\L{i_1}(T)]
\end{eqnarray}
for all finite $n$ and all type symbols $T,T_1,T_2,\ldots,T_n$.
This, in turn, can be extended inductively to all terms in the
language. Thus, in our case, the translations act trivially on all
the type symbols.

\paragraph{Arrows in $\Sys$ \emph{are} translations.}
Motivated by this argument  we now turn everything around and, in
general, \emph{define} an arrow $j:S_1\map S$ in the category
$\Sys$ to mean that there is some \emph{physically meaningful} way
of transforming the physical quantities in $S$ to physical
quantities in $S_1$. If, for any pair of systems $S_1,S$ there is
more than one such transformation, then there will be more than
one arrow from $S_1$ to $S$.

To make this more precise, let $\Loc$ denote the collection of all
(small\footnote{This means that the collection of symbols is a
set, not a more general class.}) local languages. This is a
category whose objects are the local languages, and whose arrows
are translations between languages.  Then our basic assumption is
that the association $S\mapsto\L{S}$ is a covariant functor from
$\Sys$ to $\Loc^{\rm op}$, which we denote as ${\cal
L}:\Sys\map\Loc^{\rm op}$.

Note that the  combination of a pair of arrows in $\Sys$  exists
in so far as the associated translations can be combined.

\subsubsection{The Arrows and Translations for the
Composite System $S_1\di S_2$.}\label{SubSubSec:ArrTranComp} Let
us now consider the composition $S_1\di S_2$ of a pair of systems.
In the case of classical physics,  if $\S_1$ and $\S_2$ are the
symplectic manifolds that represent the systems $S_1$ and $S_2$
respectively, then the manifold that represents the composite
system is the cartesian product $\S_1\times\S_2$. This is
distinguished by the existence of the two projection functions
$\pr_1:\S_1\times \S_2\map \S_1$ and $\pr_2:\S_1\times
\S_2\map\S_2$.

It seems reasonable to impose the same type of structure on
$\Sys$: \ie\ to require there to be arrows $p_1:S_1\di S_2\map
S_1$ and $p_2:S_1\di S_2\map S_2$ in $\Sys$. However, bearing in
mind the definition above, these arrows $p_1,p_2$ exist if, and
only if, there are corresponding translations
$\L{p_1}:\L{S_1}\map\L{S_1\di S_2}$, and
$\L{p_2}:\L{S_2}\map\L{S_1\di S_2}$. But there \emph{are} such
translations: for if $A_1$ is a physical quantity for system
$S_1$, then $\L{p_1}(A_1)$ can be defined as that same physical
quantity, but now regarded as pertaining to the combined system
$S_1\di S_2$; and analogously for system $S_2$.\footnote{For
example, if $A$ is the energy of particle $1$, then we can talk
about this energy in the combination of a pair of particles. Of
course, in---for example---classical physics there is no reason
why the energy of particle $1$ should be \emph{conserved} in the
composite system, but that, dynamical, question is a different
matter.}  We shall denote this translated quantity,
$\L{p_1}(A_1)$, by $A_1\di 1$.

Note that we do \emph{not} postulate any simple relation between
$\F{S_1\di S_2}$ and $\F{S_1}$ and $\F{S_2}$; \ie\ there is no
analogue of \eq{FS1sumS2} for combinations of systems.

The definitions above of the basic arrows suggest that we might
also want to impose the following conditions:
\begin{enumerate}
\item The arrows $i_1:S_1\map S_1\sqcup S_2$, and
$i_2:S_2\map S_1\sqcup S_2$ are \emph{monic} in $\Sys$.

\item The arrows $p_1:S_1\di S_2\map S_1$ and
$p_2:S_1\di S_2\map S_2$ are \emph{epic} arrows in $\Sys$.
\end{enumerate}
However, we do \emph{not} require that $S_1\cup S_2$ and $S_1\di
S_2$ are the co-product and product, respectively, of $S_1$ and
$S_2$ in the category $\Sys$.

\subsubsection{The Concept of `Isomorphic' Systems.}
We also need to decide  what it means to say that two systems
$S_1$ and $S_2$ are \emph{isomorphic}, to be denoted $S_1\simeq
S_2$. As with the concept of sub-system, the notion of isomorphism
is to some extent a matter of definition rather than obvious
physical structure, albeit with the expectation that isomorphic
systems in $\Sys$ will correspond to isomorphic local languages,
and be represented by isomorphic mathematical objects in any
concrete realisation of the axioms: for example, by isomorphic
symplectic manifolds in classical physics.

To a considerable extent, the physical meaning of `isomorphism'
depends on whether one is dealing with actual physical systems, or
idealisations of them. For example,  an electron confined in a box
in Cambridge is presumably isomorphic to one  confined in the same
type of box in London, although they are not the same physical
system. On the other hand, when a lecturer says ``Consider an
electron trapped in a box....'', he/she is referring to an
idealised system.

One could, perhaps, say that an idealised system is an
\emph{equivalence class} (under isomorphisms)\ of real systems,
but even working only with  idealisations does not entirely remove
the need for the concept of isomorphism.

For example, in classical mechanics, consider the (idealised)\
system $S$ of a point particle moving in a box, and let $1$ denote
the `trivial system' that consists of just a single  point with no
internal or external degrees of freedom. Now consider the  system
$S\di 1$. In classical mechanics this is represented by the
symplectic manifold $\S\times\{*\}$, where $\{*\}$\ is a single
point, regarded as a zero-dimensional manifold. However,
$\S\times\{*\}$ is isomorphic to the manifold $\S,$ and it is
clear physically that the system $S\di 1$ is isomorphic to the
system $S$.  On the other hand, one cannot say that $S\di 1$ is
literally \emph{equal} to $S$, so the concept of `isomorphism'
needs to be maintained.

One thing that \emph{is} clear is that if $S_1\simeq S_2$ then
$\F{S_1}\simeq \F{S_2}$, and if any other non-empty sets of
function symbols are present, then they too must be isomorphic.

Note that when introducing a trivial system, $1$, it necessary to
specify its local language, $\L{1}$. The set of function symbols
$\F{1}$ is not completely empty since, in classical physics, one
does have a preferred physical quantity, which is just the number
$1$. If one asks what is meant in general by the `number $1$' the
answer is not trivial since, in the reals $\mathR$, the number $1$
is the multiplicative identity. It would be possible to add the
existence of such a unit to the axioms for $\R$ but this would
involve introducing a multiplicative structure and we do not know
if there might be physically interesting topos representations
that do not have this feature.

For the moment then, we will  say that the trivial system has just
a single physical quantity, which in classical physics translates
to the number $1$.  More generally, for the language $\L{1}$ we
specify that $\F{1}:=\{I\}$, \ie\ $\F{1}$ has just a single
element, $I$, say. Furthermore, we add the axiom
\begin{equation}
 :\forall \va{s}_1\forall \va{s}_2,I(\va{s}_1)=I(\va{s}_2),
\end{equation}
where $\va{s}_1$ and $\va{s}_2$ are variables of type $\Si$. In
fact, it seems natural to add such a trivial quantity to the
language $\L{S}$ for \emph{any} system $S$, and from now on we
will assume that this has been done.

A related issue is that, in classical physics, if $A$ is a
physical quantity, then so is $rA$ for any $r\in\mathR$. This is
because the set of classical quantities
$A_\s:\Si_\s\map\R_\s\simeq\mathR$ forms a \emph{ring} whose
structure derives from the ring structure of $\mathR$. It would be
possible to add ring axioms for $\R$ to the language $\L{S}$, but
this is too strong, not least because, as shown earlier, it fails
in quantum theory. Clearly, the general question of axioms for
$\R$ needs more thought: a task for later work.

If desired, an `\emph{empty}' system, $0$, can be added too, with
$\F{0}:=\emptyset$. This, so called, `pure language', $\L{0}$, is
an initial object in the category $\Loc$.

\subsubsection{An Axiomatic Formulation of the Category $\Sys$}
Let us now  summarise, and clarify, our list of axioms for a
category $\Sys$:
\begin{enumerate}
        \item  The collection $\Sys$ is a small category
    whose  objects are the systems of interest (or, if desired,
    isomorphism classes of such systems) and whose
    arrows are defined as above.

        Thus the fundamental property of an arrow $j:S_1\map S$ in $\Sys$
        is that it induces, and is essentially \emph{defined by}, a
        translation $\L{j}:\L{S}\map \L{S_1}$. Physically, this corresponds
        to  the physical quantities for system $S$ being `pulled-back' to
        give physical quantities for system $S_1$.

        Arrows of particular interest are those associated with
        `sub-systems' and `composite systems', as discussed above.

        \item The axioms for a category are satisfied because:
        \begin{enumerate}
            \item Physically, the ability to form composites of
            arrows follows from the concept of `pulling-back' physical
            quantities. From a mathematical perspective,
            if $j:S_1\map S_2$ and $k:S_2\map S_3$,
            then the translations give functions $\L{j}:\F{S_2}
                \map \F{S_1}$ and $\L{k}:\F{S_3} \map \F{S_2}$. Then clearly
                $\L{j}\circ \L{k}:\F{S_3}\map \F{S_1}$, and this can thought of as
                the translation corresponding to the arrow $k\circ j:S_1\map S_3$.

            The associativity of the law of arrow combination can be
            proved in a similar way.

           \item We add by hand a special arrow ${\rm id}_S : S \map S$
           which is defined to correspond to the translation $\L{\id_S}$ that
          is given by the identity map on $\F{S}$. Clearly, ${\rm id}_S :
         S \map S$ acts an an identity morphism should.
        \end{enumerate}

        \item For any pair of systems $S_1,S_2$, there is a
        \emph{disjoint sum}, denoted $S_1\sqcup S_2$. The disjoint sum has
        the following properties:
        \begin{enumerate}
            \item For all systems $S_1,S_2,S_3$ in $\Sys$:
            \begin{equation}
                        (S_1\sqcup S_2)\sqcup S_3\simeq S_1\sqcup (S_2\sqcup S_3).
            \end{equation}

            \item For all systems $S_1, S_2$ in $\Sys$:
        \begin{equation}
                    S_1\sqcup S_2 \simeq S_2\sqcup S_1.
            \end{equation}

        \item There are  arrows in $\Sys$:
        \begin{equation}
                        i_1:S_1\map S_1\sqcup S_2 \mbox{\ \ and }
                        i_2:S_2\map S_1\sqcup S_2
        \end{equation}
          that are associated with translations in the sense discussed in
          Section \ref{SubSubSec:ATDS}. These are associated with the
          decomposition
                \begin{equation}
                \F{S_1\sqcup S_2}\simeq\F{S_1}\times\F{S_2}.
                                \label{F(S1+S2)=FS1xFS2}
                \end{equation}
        \end{enumerate}

        We assume that if $S_1,S_2$ belong to $\Sys$, then $\Sys$ also
        contains $S_1\sqcup S_2$.

        \item For any given pair of systems $S_1,S_2$, there is a
        \emph{composite} system in $\Sys$, denoted\footnote{The product
        operation in a monoidal category is often written `$\otimes$'.
        However,  a different symbol  has been used here to avoid
        confusion with existing usages in physics of the tensor product
        sign `$\otimes$'.} $S_1\di S_2$, with the following properties:
        \begin{enumerate}
            \item For all systems $S_1,S_2,S_3$ in $\Sys$:
            \begin{equation}
                        (S_1\di S_2)\di S_3 \simeq
                 S_1\di (S_2\di S_3).\label{(Sys: Assoc)}
            \end{equation}

            \item For all systems $S_1, S_2$ in $\Sys$:
        \begin{equation}
                    S_1\di S_2 \simeq S_2 \di S_1.            \label{(SysSym)}
            \end{equation}

            \item There are  arrows in $\Sys$:
        \begin{equation}
                        p_1:S_1\di S_2 \map S_1\mbox{ and }
                        p_2:S_1\di S_2 \map S_2
        \end{equation}
          that are associated with translations in the sense discussed
           in Section \ref{SubSubSec:ArrTranComp}.
        \end{enumerate}
        We assume that if $S_1,S_2$ belong to $\Sys$, then $\Sys$ also
        contains the composite system $S_1\di S_2$.

        \item It seems physically reasonable to add the axiom
        \begin{equation}
            (S_1\sqcup S_2)\di S\simeq (S_1\di S)\sqcup (S_2\di S)
                        \label{S1cupS2diS}
        \end{equation}
        for all systems $S_1,S_2,S$. However, physical intuition can
        be a dangerous thing, and so, as with most of these axioms,
        we are not dogmatic, and feel free to change them as
        new insights emerge.

        \item There is a trivial system, $1$,  such that for all systems
        $S$, we have
        \begin{equation}
        S \di 1\simeq  S \simeq1 \di S
                                \label{(Sys: Unit)}
        \end{equation}

        \item It may be convenient to postulate an `empty system', $0$,
        with the properties
        \begin{eqnarray}
        S\di 0&\simeq& 0\di S\simeq 0 \\
        S\sqcup 0&\simeq& 0\sqcup S\simeq S
        \end{eqnarray}
        for all systems $S$.

        Within the meaning given to arrows in $\Sys$, $0$ is a
        \emph{terminal object} in $\Sys$. This is because the empty set of
        function symbols of signature $\Si\map\R$ is a subset of any other
        set of function symbols of this signature.
\end{enumerate}

It might seem tempting to postulate that composition laws are
well-behaved with respect to arrows. Namely, if $j:S_1\map S_2$,
then, for any $S$, there is an arrow $S_1\di S\map S_2\di S$ and
an arrow $S_1\sqcup S\map S_2\sqcup S$.\footnote{A more accurate
way of capturing this idea is to say that the operation
$\Sys\times\Sys\map\Sys$ in which
\begin{equation}
\langle S_1,S_2\rangle\mapsto S_1\di S_2 \label{(Sys: BiFunc)}
\end{equation}
is a \emph{bi-functor} from $\Sys\times\Sys$ to $\Sys$. Ditto for
the operation in which $\la S_1,S_2\ra\mapsto S_1\sqcup S_2$.}

In the case of the disjoint sum, such an arrow can be easily
constructed using \eq{F(S1+S2)=FS1xFS2}. First split the function
symbols in $\F{S_1\sqcup S}$ into $\F{S_1}\times\F{S}$ and the
function symbols in $\F{S_2\sqcup S}$ into $\F{S_2}\times\F{S}$.
Since there is an arrow $j:S_1\map S_2$, there is a translation
$\L{j}:\L{S_2}\map \L{S_1}$, given by a mapping
$\L{j}:\F{S_2}\map\F{S_1}$. Of course, then there is also a
mapping $\L{j}\times
\L{\id_S}:\F{S_2}\times\F{S}\map\F{S_1}\times\F{S}$, \ie\ a
translation between $\L{S_2\sqcup S}$ and $\L{S_1\sqcup S}$. Since
we assume that there is an arrow in $\Sys$ whenever there is a
translation (in the opposite direction), there is indeed  an arrow
$S_1\sqcup S\map S_2\sqcup S$.

In the case of the composition, however, this would require a
translation $\L{S_2\di S}\map \L{S_1\di S}$, and this cannot be
done in general since we have no \emph{prima facie} information
about the  set of function symbols $\F{S_2\di S}$. However, if we
restrict the arrows in $\Sys$ to be those associated with
sub-systems, combination of systems, and compositions of such
arrows, then it is easy to see that the required translations
exist (the proof of this makes essential use of \eq{S1cupS2diS}).

If we make this restriction of arrows, then the axioms
\eq{(SysSym)}, \eqs{(Sys: Unit)}{(Sys: BiFunc)}, mean that,
essentially, $\Sys$ has the structure of a \emph{symmetric
monoidal}\footnote{In the actual definition of a monoidal category
the two isomorphisms in \eq{(Sys: Unit)} are separated from each
other, whereas  we have identified them. Further more, these
isomorphism are required to be natural. This seems a correct thing
to require in our case, too.} category in which the monoidal
product operation is `$\di$', and the left and right unit object
is $1$. There is also a monoidal structure associated with the
disjoint sum `$\sqcup$', with $0$ as the unit object.

We say `essentially' because in order to comply with all the
axioms of a monoidal category, $\Sys$ must satisfy certain
additional, so-called, `coherence' axioms. However, from a
physical perspective these are very plausible statements about (i)
how the unit object $1$ intertwines with the $\di$-operation; how
the null object intertwines with the $\sqcup$-operation; and (iii)
certain properties of quadruple products (and disjoint sums) of
systems.

\paragraph{A simple example of a category $\Sys$.}
It might be helpful at this point to give a simple example of a
category $\Sys$. To that end, let $S$ denote a point particle that
moves in three dimensions, and let us suppose that $S$ has no
sub-systems other than the trivial system $1$. Then  $S\di S$ is
defined to be a pair of particles moving in three dimensions, and
so on. Thus the objects in our category are $1$, $S$, $S\di S$,
$\ldots$,  $S\di S\di\cdots S$ $\ldots$ where the `$\di$'
operation is formed any finite number of times.

At this stage, the only arrows are those that are associated with
the constituents of a composite system. However, we could
contemplate adding to the systems the disjoint sum $S\sqcup (S\di
S)$ which is a system that is either one particle or two particles
(but, of course, not both at the same time). And, clearly, we
could extend this to $S\sqcup (S\di S)\sqcup (S\di S\di S)$, and
so on. Each of these disjoint sums comes with its own arrows, as
explained above.

Note that this particular category of systems has the property
that it can be treated using either classical physics or quantum
theory.

\subsection{Representations of $\Sys$ in Topoi}
We assume that all the systems in $\Sys$ are to be treated with
the same theory type. We also assume that systems in $\Sys$ with
the \emph{same} language are to be represented in the same topos.
Then we define:\footnote{As emphasised already, the association
$S\mapsto \L{S}$ is generally not one-to-one: \ie\ many systems
may share the same language. Thus, when we come discuss the
representation of the language $\L{S}$ in a topos, the extra
information about the system $S$  is used in fixing the
representation.} {\definition\label{Defn:TopReal} A \emph{topos
realisation} of $\Sys$ is an association, $\phi$, to each system
$S$ in $\Sys$, of a triple $\phi(S)=\la\rho_{\phi,
S},\L{S},\tau_\phi(S)\ra$ where:

\begin{enumerate}
        \item[(i)] $\tau_\phi(S)$ is the topos in which the
        theory-type applied to system $S$ is to be realised.

       \item[(ii)] $\L{S}$ is the local language in $\Loc$
       that is associated with $S$.
        This is not dependent on the realisation $\phi$.

\item[(iii)]$\rho_{\phi, S}$ is
       a representation of the local language $\L{S}$ in the
       topos $\tau_\phi(S)$. As a more descriptive piece of notation
        we write $\rho_{\phi,S}:\L{S}\sq\tau_\phi(S)$.
        The key part of this representation is the map
\begin{equation}
        \rho_{\phi,S}:\F{S}\map\Hom{\tau_\phi(S)}{\Si_{\phi,S}}
        {\R_{\phi,S}}
\end{equation}
where $\Si_{\phi,S}$ and $\R_{\phi,S}$ are the state object and
quantity-value object, respectively, of the representation $\phi$
in the topos $\tau_\phi(S)$. As a convenient piece of notation we
write $A_{\phi,S}:=\rho_{\phi,S}(A)$ for all $A\in\F{S}$.
\end{enumerate}
} \noindent This definition is only partial; the possibility of
extending it will be discussed shortly.

Now, if $j:S_1\map S$ is an arrow in $\Sys$, then there is a
translation arrow ${\cal L}{(j)}:\L{S}\map\L{S_1}$. Thus we have
the beginnings of a commutative diagram
\begin{equation}
\setsqparms[1`1`-1`1;1000`700]\label{ComDiag}
\square[S_1`\la\rho_{\phi,S_1},\L{S_1},\tau_\phi(S_1)
\ra`S`\la\rho_{\phi,S},\L{S},\tau_\phi(S)\ra; \phi`j`?\times{\cal
L}(j)\times?`\phi]
\end{equation}
However, to be useful, the arrow on the right hand side of this
diagram should refer to some relation between (i) the topoi
$\tau_\phi(S_1)$ and $\tau_\phi(S)$; and (ii) the realisations
$\rho_{\phi,S_1}:\L{S_1}\sq\tau_\phi(S_1)$ and
$\rho_{\phi,S}:\L{S}\sq\tau_\phi(S)$: this is the significance of
the two `?' symbols in the arrow written `$?\times\L{j}\times?$'.

Indeed, as things stand, Definition \ref{Defn:TopReal} says
nothing  about relations between the topoi representations of
different systems in $\Sys$. We are particularly interested in the
situation where there are two different systems $S_1$ and $S$ with
an arrow $j:S_1\map S$ in $\Sys$.

We know that the arrow $j$ is associated with a translation
$\L{j}:\L{S}\map\L{S_1}$, and an attractive possibility,
therefore, would be to seek, or postulate, a `covering' map
$\phi(\L{j}): \Hom{\tau_\phi(S)}{\Si_{\phi,S}}{\R_{\phi,S}} \map
\Hom{\tau_\phi(S_1)}{\Si_{\phi,S_1}}{\R_{\phi,S_1}}$ to be
construed as a topos representation of the translation
$\L{j}:\L{S}\map \L{S_1}$, and hence of the  arrow $j:S_1\map S$
in $\Sys$.

This raises the questions of what properties these `translation
representations' should possess in order to justify saying that
they `cover' the translations. A minimal requirement is that if
$k:S_2\map S_1$ and $j:S_1\map S$, then the map $\phi(\L{j\circ
k}): \Hom{\tau_\phi(S)}{\Si_{\phi,S}}{\R_{\phi,S}} \map
\Hom{\tau_\phi(S_2)}{\Si_{\phi,S_2}}{\R_{\phi,S_2}}$ factorises as
\begin{equation}
        \phi(\L{j\circ k})=\phi(\L{k})\circ\phi(\L{j}).
                                                \label{phiL(jcirck)}
\end{equation}
We also require that
\begin{equation}
\phi(\L{\id_S})=\id: \Hom{\tau_\phi(S)}{\Si_{\phi,S}}{\R_{\phi,S}}
\map \Hom{\tau_\phi(S)}{\Si_{\phi,S}}{\R_{\phi,S}}
\label{phiL(idS)}
\end{equation}
for all systems $S$.

The conditions \eq{phiL(jcirck)} and \eq{phiL(idS)} seem eminently
plausible, and they are not particularly strong. A far more
restrictive axiom would be to require the following diagram to
commute:
\begin{equation}\label{ComDLphi0}
        \setsqparms[1`1`1`1;1200`700]
        \square[\F{S}`\Hom{\tau_\phi(S)}{\Si_{\phi,S}}
        {\R_{\phi,S}}`\F{S_1}`\Hom{\tau_\phi(S_1)}
        {\Si_{\phi,S_1}}{\R_{\phi,S_1}};
        \rho_{\phi,S}`\L{j}`\phi(\L{j})`\rho_{\phi,S_1}]
\end{equation}
At first sight, this requirement seems very appealing. However,
caution is needed when postulating `axioms' for a theoretical
structure in physics. It is easy to get captivated by the
underlying mathematics and to assume, erroneously, that  what is
mathematically elegant is necessarily true in the physical theory.

The translation $\phi(\L{j})$ maps an arrow from $\Si_{\phi,S}$ to
$\R_{\phi,S}$ to an arrow from $\Si_{\phi,S_1}$ to
$\R_{\phi,S_1}$. Intuitively, if $\Si_{\phi,S_1}$ is a `much
larger' object than $\Si_{\phi,S}$ (although since they lie in
different topoi, no direct comparison is available), the
translation can only be `faithful' on some part of
$\Si_{\phi,S_1}$ that can be identified with (the `image' of)
$\Si_{\phi,S}$. A concrete example of this will show up in the
treatment of composite quantum systems, see Subsection
\ref{SubSec:TranslCompQT}. As one might expect, a form of
entanglement plays a role here.

\subsection{Classical Physics in This  Form}
\subsubsection{The Rules so Far.}
Constructing maps $\phi(\L{j}):
\Hom{\tau_\phi(S)}{\Si_{\phi,S}}{\R_{\phi,S}} \map
\Hom{\tau_\phi(S_1)}{\Si_{\phi,S_1}}{\R_{\phi,S_1}}$ is likely to
be complicated when $\tau_\phi(S)$ and $\tau_\phi(S_1)$ are
different topoi, and so we begin  with the example of classical
physics, where the topos is always $\Set$.

In general, we are interested in the relation(s) between the
representations $\rho_{\phi,S_1}: \L{S_1}\sq\tau_\phi(S_1)$ and
$\rho_{\phi,S}:\L{S}\sq\tau_\phi(S)$ that is associated with an
arrow $j:S_1\map S$ in $\Sys$. In classical physics, we only have
to study the relation between the representations $\rho_{\s,S_1}:
\L{S_1}\sq\Set$ and $\rho_{\s,S}:\L{S}\sq\Set$.

Let us summarise what we have said so far (with $\s$ denoting the
$\Set$-realisation of classical physics):
\begin{enumerate}
    \item For any system $S$ in $\Sys$, a representation
        $\rho_{\s,S}:\L{S}\sq\Set$ consists of the following ingredients.
        \begin{enumerate}
                \item The ground symbol $\Si$  is
                represented by a symplectic  manifold,
                $\Si_{\s,S}:=\rho_{\s,S}(\Si)$, that serves as the
                classical state space.

                \item For all systems $S$, the ground symbol $\cal R$ is
                represented by the real numbers $\mathR$, \ie\
                $\R_{\s,S}=\mathR$, where $\R_{\s,S}:=\rho_{\s,S}(\R)$.

                \item Each function symbol $A:\Si\map\R$ in $\F{S}$ is
                represented by a function $A_{\s,S}=\rho_{\s,S}(A):
                \Si_{\s,S}\map\mathR$ in the set of functions\footnote{In
                practice, these functions are required to be measurable
                with respect to the Borel structures on the symplectic
                manifold $\Si_\s$ and $\mathR$. Many of the functions
                will also be smooth, but we will not go into such
                details here.} $C(\Si_{\s,S}, \mathR)$.
        \end{enumerate}

        \item The trivial system is mapped to a singleton set $\{*\}$
    (viewed as a zero-dimensional symplectic manifold):
    \begin{equation}
            \Si_{\s,1} := \{*\}.
    \end{equation}
        The empty system is represented by the empty set:
        \begin{equation}
        \Si_{\s,0}:=\emptyset.
        \end{equation}

        \item Propositions about the system $S$ are represented by (Borel)
        subsets of the state space $\Si_{\s,S}$.

        \item The composite system $S_1\di S_2$
        is represented by the
    Cartesian product $\Si_{\s,S_1}\times \Si_{\s,S_2}$; \ie\
    \begin{equation}
            \Si_{\s,\,S_1\di S_2}\simeq \Si_{\s,S_1}\times\Si_{\s,S_2}.
                                            \label{sigmaS1S2}
    \end{equation}

    The disjoint sum $S_1\sqcup S_2$ is represented by the disjoint
        union $\Si_{\s,S_1}\coprod \Si_{\s,S_2}$;\ie\
    \begin{equation}
        \Si_{\s,S_1\sqcup S_2}\simeq\Si_{\s,S_1}{\textstyle\coprod}\Si_{\s,S_2}.
     \end{equation}

    \item Let $j : S_1\map S$ be an arrow in $\Sys$. Then
        \begin{enumerate}
            \item There is a  translation map $\L{j}:\F{S}\map \F{S_1}$.

                \item  There is a symplectic function $\s(j):\Si_{\s,S_1}
        \map \Si_{\s,S}$ from the symplectic manifold $\Si_{\s,S_1}$
                to the symplectic manifold $\Si_{\s,S}$.
    \end{enumerate}
\end{enumerate}

The existence of this function $\s(j):\Si_{\s,S_1} \map
\Si_{\s,S}$ follows directly from the properties of sub-systems
and composite systems in classical physics. It is discussed in
detail below in Section \eq{SubSubSec:DetailsTrans}. As we shall
see, it underpins the classical realisation of our axioms.

These properties of the arrows stem from the fact that the
linguistic function symbols in $\F{S}$ are represented by
real-valued functions in $C(\Si_{\s,S},\mathR)$. Thus we can write
$\rho_{\s,S}:\F{S}\map C(\Si_{\s,S},\mathR)$, and similarly
$\rho_{\s,S_1}:\F{S_1}\map C(\Si_{\s,S_1},\mathR)$. The diagram in
\eq{ComDLphi0} now becomes
\begin{equation}\label{ComDLsClass}
        \setsqparms[1`1`1`1;1000`700]
        \square[\F{S}`C(\Si_{\s,S},\mathR)`\F{S_1}`C(\Si_{\s,S_1},\mathR);
        \rho_{\s,S}`\L{j}`\s(\L{j})`\rho_{\s,S_1}]
\end{equation}
and, therefore, the question of interest is if there is a
`translation representation' function
$\s(\L{j}):C(\Si_{\s,S},\mathR)\map C(\Si_{\s,S_1},\mathR)$ so
that this diagram commutes.

Now, as stated above, a physical quantity, $A$, for the system $S$
is represented in classical physics by a real-valued function
$A_{\s,S}=\rho_{\s,S}(A) :\Si_{\s, S}\map\mathR$. Similarly, the
representation of $\L{j}(A)$  for $S_1$ is given by a function
$A_{\s,S_1}:=\rho_{\s,S_1}(A):\Si_{\s,S_1}\map \mathR$. However,
in this classical case we also have the function
$\s(j):\Si_{\s,S_1}\map \Si_{\s,S}$, and it is clear that we can
use it to define
$[\rho_{\s,S_{1}}(\L{j}(A)](s):=\rho_{\s,S}(A)\big(\s(j)(s)\big)$
for all $s\in\Si_{\s,S_1}$. In other words
\begin{equation}
        \rho_{\s,S_1}\big(\L{j}(A)\big)=\rho_{\s,S}(A)\circ\s(j)
\end{equation}
or, in simpler notation
\begin{equation}
        \big((\L{j}(A)\big)_{\s,S_1}=A_{\s,S}\circ\s(j).
\end{equation}
But then it is clear that a translation-representation function
$\s(\L{j}):C(\Si_{\s,S},\mathR) \map C(\Si_{\s,S_1},\mathR)$ with
the desired property of making \eq{ComDLsClass} commute can be
defined by
\begin{equation}
\s(\L{j})(f):=f\circ\s(j)       \label{ClassPB}
\end{equation}
for all $f\in C(\Si_{\s,S},\mathR)$; \ie\ the function
$\s(\L{j})(f):\Si_{\s,S_1}\map\mathR$ is the usual pull-back of
the function $f:\Si_{\s,S}\map\mathR$ by the function
$\s(j):\Si_{\s,S_1}\map\Si_{\s,S}$. Thus, in the case of classical
physics, the commutative diagram in \eq{ComDiag} can be completed
to give
\begin{equation}                                \label{ComDiagC3}
        \setsqparms[1`1`-1`1;1000`700]
        \square[S_1`\la\rho_{\s,S_1},\L{S_1},
        \Set\ra`S`\la\rho_{\s,S},\L{S},\Set\ra;
        \s`j`\s(\L{j})\times\L{j}\times\id`\s]
\end{equation}

\subsubsection{Details of the Translation Representation.}
\label{SubSubSec:DetailsTrans}
\paragraph{The translation representation for a disjoint sum of classical
systems.} We first consider arrows of the form
\begin{equation}
S_{1}\overset{i_{1}}{\map}S_{1}\sqcup S_{2}\overset{i_{2}
}{\leftarrow}S_{2}
\end{equation}
from the components $S_{1}$, $S_{2}$ to the disjoint sum
$S_{1}\sqcup S_{2}$. The systems $S_{1}$, $S_{2}$ and $S_{1}\sqcup
S_{2}$ have symplectic manifolds $\Si_{\s,S_{1}}$,
$\Si_{\s,S_{2}}$ and $\Si_{\s,S_{1}\sqcup
S_{2}}=\Si_{\s,S_{1}}\coprod\Si_{\s,S_{2}}$. We write $i:=i_{1}$.

Let $S$ be a classical system. We assume that the function symbols
$A\in \F{S}$ in the language $\L{S}$ are in bijective
correspondence with an appropriate subset of the functions
$A_{\s,S}\in C(\Si_{\s,S},\mathR)$.\footnote{Depending on the
setting, one can assume that $\F{S}$ contains function symbols
corresponding bijectively to measurable, continuous or smooth
functions.}

There is an obvious translation representation. For if
$A\in\F{S_1\sqcup S_2}$, then since $\Si_{\s,S_1\sqcup
S_2}=\Si_{\s,S_1}\coprod\Si_{\s,S_1}$, the associated function
$A_{\s, S_1\sqcup S_2}:\Si _{\s,S_1\sqcup S_2} \map\mathR$ is
given by a pair of functions $A_{1}\in C(\Si_{\s,S_1},\mathR)$ and
$A_{2}\in C(\Si_{\s,S_2},\mathR)$; we write $A_{\s, S_1\sqcup
S_2}=\la A_1,A_2\ra$. It is natural to demand that the translation
representation $\s(\L{i})(A_{\s, S_1\sqcup S_2})$ is $A_1$.  Note
that what is essentially being discussed here is the
classical-physics representation of the relation \eq{FS1sumS2}.

The canonical choice for $\s(i)$ is
\begin{eqnarray}
  \s(i):\Si_{\s,S_1}  &  \map&\Si_{\s,S_1\sqcup S_2}=
                         \Si_{\s,S_1}{\textstyle\coprod}\Si_{\s,S_2}\\
                s_{1}  &  \mapsto& s_1.
\end{eqnarray}
Then the pull-back along $\s(i)$,
\begin{eqnarray}
\s(i)^*:C(\Si_{\s,S_1\sqcup S_2},\mathR)  & \map& C(\Si_{\s,S_1},\mathR)\\
A_{\s,S_1\sqcup S_2}  &  \mapsto& A_{\s,S_1\sqcup S_2} \circ
\s(i),
\end{eqnarray}
maps (or `translates') the topos representative $A_{\s,S_1\sqcup
S_2} =\la A_1, A_2\ra$ of the function symbol $A\in\F{S_1\sqcup
S_2}$ to a real-valued function $A_{\s,S_1\sqcup S_2} \circ \s(i)$
on $\Si_{\s,S_1}$. This function is clearly equal to $A_1$.

\paragraph{The translation in the case of a composite classical system.}
We now consider arrows in $\mathbf{Sys}$ of the form
\begin{equation}
S_{1}\overset{p_{1}}{\leftarrow}S_{1}\di S_{2}\overset{p_2
}{\map}S_2
\end{equation}
from the composite classical system $S_1\di S_2$ to the
constituent systems $S_1$ and $S_2$. Here, $p_1$ signals that
$S_{1}$ is a constituent of the composite system $S_1\di S_2$,
likewise $p_2$. The systems $S_{1}$, $S_{2}$ and $S_{1}\di S_{2}$
have symplectic manifolds $\Si_{\s,S_{1}}$, $\Si_{\s,S_{2}}$ and
$\Si_{\s,S_{1}\di S_{2} }=\Si_{\s,S_{1}}\times\Si_{\s,S_{2}}$,
respectively; \ie\  the state space of the composite system
$S_{1}\di S_{2}$ is the cartesian product of the state spaces of
the components. For typographical simplicity in what follows we
denote $p:=p_{1}$.

There is a canonical translation $\L{p}$ between the languages
$\L{S_1}$ and $\L{S_{1}\di S_{2}}$ whose representation is the
following. Namely, if $A$ is in $\F{S_1}$, then the corresponding
function $A_{\s,S_1}\in C(\Si_{\s,S_{1} },\mathR)$ is translated
to a  function $\s(\L{p}) (A_{\s,S_1})\in C(\Si_{\s,S_{1}\di
S_{2}},\mathR)$ such that
\begin{equation}
        \s(\L{p})(A_{\s,S_1})(s_1,s_2)=A_{\s,S_1}(s_1)
\end{equation}
for all $(s_1,s_2)\in\Si_{\s,S_1}\times\Si_{\s,S_2}$.

This natural translation representation is based on the fact that,
for the symplectic manifold $\Si_{\s,S_1\di
S_2}=\Si_{\s,S_1}\times\Si_{\s,S_2}$, each point
$s\in\Si_{\s,S_1\di S_2}$ can be identified with a pair,
$(s_1,s_2)$, of points $s_1\in\Si_{\s,S_1}$ and
$s_2\in\Si_{\s,S_2}$. This is possible since the cartesian product
$\Si_{\s,S_1}\times\Si_{\s,S_2}$ is a product in the categorial
sense and hence has projections $\Si_{\s,S_1}\leftarrow
\Si_{\s,S_1}\times\Si_{\s,S_2}\map\Si_{\s,S_2}$. Then the
translation representation of functions is constructed in a
straightforward manner. Thus, let
\begin{eqnarray}
\s(p):\Si_{\s,S_1}\times\Si_{\s,S_2}  &  \map&\Si_{\s,S_1}\nonumber\\
        (s_1,s_2)  &  \mapsto& s_1
\end{eqnarray}
be the canonical projection. Then, if $A_{\s,S_1}\in
C(\Si_{\s,S_1},\mathR)$, the function
\begin{equation}
A_{\s,S_1}\circ\s(p)\in
C(\Si_{\s,S_{1}}\times\Si_{\s,S_{2}},\mathR)
\end{equation}
is such that, for all
$(s_1,s_2)\in\Si_{\s,S_1}\times\Si_{\s,S_2}$,
\begin{equation}
A_{\s,S_1}\circ\s(p)(s_{1},s_{2})=A_{\s,S_1}(s_1).
\end{equation}
Thus we can define
\begin{equation}
        \s(\L{p})(A_{\s,S_1}):=A_{\s,S_1}\circ\s(p).
\end{equation}
Clearly, $\s(\L{p})(A_{\s,S_1})$ can be seen as the representation
of the function symbol $A\di 1\in\F{S_1\di S_2}$.

\section{Theories of Physics in a General Topos}
\label{Sec:ToposAxioms}
\subsection{The Pull-Back Operations}
\subsubsection{The Pull-Back of Physical Quantities.}
Motivated by the above, let us try now to see what can be said
about the scheme in general. Basically, what is involved is the
topos representation of translations of languages. To be more
precise, let $j:S_1\map S$ be an arrow in $\Sys$, so that there is
a translation $\L{j}:\L{S}\map \L{S_1}$ defined by the translation
function $\L{j}:\F{S}\map \F{S_1}$. Now suppose that the systems
$S$ and $S_1$ are represented in the topoi $\tau_\phi(S)$ and
$\tau_\phi(S_1)$ respectively. Then, in these representations, the
function symbols of signature $\Si\map\cal R$  in $\L{S}$ and
$\L{S_1}$ are represented by elements of
$\Hom{\tau_\phi(S)}{\Si_{\phi,S}}{\R_{\phi,S}}$ and
$\Hom{\tau_\phi(S_1)}{\Si_{\phi,S_1}}{\R_{\phi,S_1}}$
respectively.

Our task is to find a  function
\begin{equation}
 \phi(\L{j}):\Hom{\tau_\phi(S)}{\Si_{\phi,S}}
 {\R_{\phi,S}}\map\Hom{\tau_\phi(S_1)}{\Si_{\phi,S_1}}{\R_{\phi,S_1}}
\end{equation}
that can be construed as the topos representation of the
translation $\L{j}:\L{S}\map \L{S_1}$, and hence of the  arrow
$j:S_1\map S$ in $\Sys$. We are particularly interested in seeing
if $\phi(\L{j})$ can be chosen so that the following diagram, (see
\eq{ComDLphi0}) commutes:
\begin{equation}                                \label{ComDLphi}
        \setsqparms[1`1`1`1;1200`700]
        \square[\F{S}`\Hom{\tau_\phi(S)}{\Si_{\phi,S}}
        {\R_{\phi,S}}`\F{S_1}`\Hom{\tau_\phi(S_1)}
        {\Si_{\phi,S_1}}{\R_{\phi,S_1}};
        \rho_{\phi,S}`\L{j}`\phi(\L{j})`\rho_{\phi,S_1}]
\end{equation}
However, as has been emphasised already, it is not clear that one
\emph{should} expect to find a function $\phi(\L{j}):
\Hom{\tau_\phi(S)}{\Si_{\phi,S}}{\R_{\phi,S}} \map
\Hom{\tau_\phi(S_1)}{\Si_{\phi,S_1}}{\R_{\phi,S_1}}$  with this
property. The existence and/or properties of such a function will
be dependent on the theory-type, and it seems unlikely that much
can be said in general about the diagram \eq{ComDLphi}.
Nevertheless, let us see how far we \emph{can} get in discussing
the existence of such a function in general.

Thus, if $\mu\in\Hom{\tau_\phi(S)}{\Si_{\phi,S}}{\R_{\phi,S}}$,
the critical question is if there is some `natural' way whereby
this arrow can be `pulled-back' to give an element
$\phi(\L{j})(\mu)\in\Hom{\tau_\phi(S_1)}
{\Si_{\phi,S_1}}{\R_{\phi,S_1}}$.

The first pertinent remark is that $\mu$ is an arrow in the topos
$\tau_\phi(S)$, whereas the sought-for pull-back will be an arrow
in the topos $\tau_\phi(S_1)$, and so we need a mechanism for
getting from one topos to the other (this problem, of course, does
not arise in classical physics since the topos of every
representation is always $\Set$).

The obvious way of implementing this change of topos is via some
\emph{functor}, $\tau_\phi(j)$ from $\tau_\phi(S)$\ to
$\tau_\phi(S_1)$. Indeed, given such  a functor, an arrow
$\mu:\Si_{\phi,S}\map\R_{\phi,S}$ in $\tau_\phi(S)$ is transformed
to the arrow
\begin{equation}
\tau_\phi(j)(\mu):
\tau_\phi(j)(\Si_{\phi,S})\map\tau_\phi(j)(\R_{\phi,S})
                                \label{taut(A)}
\end{equation}
 in $\tau_\phi(S_1)$.

To convert this to  an arrow from $\Si_{\phi,S_1}$ to
$\R_{\phi,S_1}$, we need to supplement \eq{taut(A)} with a pair of
arrows $\phi(j),\beta_\phi(j)$ in $\tau_\phi(S_1)$ to get the
diagram:
\begin{equation}                                \label{PBPhimu}
        \setsqparms[1`-1`1`0;1000`450]
        \square[\tau_\phi(j)(\Si_{\phi,S})`\tau_\phi(j)(\R_{\phi,S})
                `\Si_{\phi,S_1}`\R_{\phi,S_1};
        \tau_\phi(j)(\mu)`\phi(j)`\beta_\phi(j)`]
\end{equation}
The pull-back, $\phi(\L{j})(\mu)\in
\Hom{\tau_\phi(S_1)}{\Si_{\phi,S_1}}{\R_{\phi,S_1}}$, with respect
to these choices can then be defined as
\begin{equation}
  \phi(\L{j})(\mu):=\beta_\phi(j)\circ \tau_\phi(j)(\mu)\circ\phi(j).
                \label{Def:phi(t)(A)}
\end{equation}
It follows that a key part of the construction of a topos
representation, $\phi$, of $\Sys$ will be to specify the functor
$\tau_\phi(j)$ from $\tau_\phi(S)$ to $\tau_\phi(S_1)$, and the
arrows $\phi(j):\Si_{\phi,S_1}\map\tau_\phi(j)(\Si_{\phi,S})$ and
$\beta_\phi(j):\tau_\phi(j)(\R_{\phi,S})\map\R_{\phi,S_1}$ in the
topos $\tau_\phi(S_1)$. These need to be defined in such a way as
to be consistent with a chain of arrows $S_2\map S_1\map S$.

When applied to the representative $A_{\phi,S}:\Si_{\phi,S}\map
\R_{\phi,S}$ of a physical quantity $A\in\F{S}$, the diagram
\eq{PBPhimu}  becomes (augmented with the upper half)
\begin{equation}                                \label{PBPhiA}
        \xext=3000
        \yext=3000
        \setsqparms[1`-1`1`1;1200`450]
        \square[\tau_\phi(j)(\Si_{\phi,S})`\tau_\phi(j)
        (\R_{\phi,S})`\Si_{\phi,S_1}`\R_{\phi,S_1};
        `\phi(j)`\beta_\phi(j)`\phi(\L{j})(A_{\phi,S})]
        \setsqparms[1`1`1`0;1200`450]
        \putsquare(-1470,450)[\Si_{\phi,S}`\R_{\phi,S}``;
        A_{\phi,S}`\tau_\phi(j)`\tau_\phi(j)`\tau_\phi(j)(A_{\phi,S})]
\end{equation}
The commutativity of \eq{ComDLphi} would then require
\begin{equation}
        \phi(\L{j})(A_{\phi,S})=(\L{j}A)_{\phi,S_1}\label{ClassCom}
\end{equation}
or, in a more expanded notation,
 \begin{equation}
        \phi(\L{j})\circ\rho_{\phi,S}=\rho_{\phi,S_1}\circ\L{j},
                                                                \label{ClassComEx}
\end{equation}
where both the left hand side and the right hand side of
\eq{ClassComEx} are mappings from $\F{S}$ to $\Hom{\tau_\phi(S_1)}
{\Si_{\phi,S_1}}{\R_{\phi,S_1}}$.

Note that the analogous diagram in classical physics is simply
\begin{equation}
        \setsqparms[1`-1`1`1;1000`450]
        \square[\Si_{\s,S}`\mathR`\Si_{\s,S_1}`\mathR;
        A_{\s,S}`\s(j)`\id`\s(\L{j})(A_{\s,S})]
\end{equation}
and the commutativity/pull-back condition \eq{ClassCom} becomes
\begin{equation}
\s(\L{j})(A_{\s,S})=(\L{j}A)_{\phi,S_1}
\end{equation}
which is satisfied by virtue of \eq{ClassPB}.

It is clear from the above that the arrow
$\phi(j):\Si_{\phi,S_1}\map\tau_\phi(j)(\Si_{\phi,S})$ can be
viewed as the topos analogue of the map
$\s(j):\Si_{\s,S_1}\map\Si_{\s,S}$ that arises in classical
physics whenever there is an arrow $j:S_1\map S$.

\subsubsection{The Pull-Back of Propositions.}
More insight can be gained into the nature of the triple
$\la\tau_\phi(j),\phi(j),\beta_\phi(j)\ra$ by considering the
analogous  operation for propositions. First, consider  an arrow
$j:S_1\map S$ in $\Sys$ in classical physics. Associated with this
there is (i) a translation $\L{j}:\L{S}\map \L{S_1}$; (ii) an
associated translation mapping $\L{j}:\F{S}\map \F{S_1}$; and
(iii) a symplectic function $\s(j):\Si_{\s,S_1}\map \Si_{\s,S}$.

Let $K$ be a (Borel) subset of the state space, $\Si_{\s,S}$;
hence $K$ represents a proposition about the system $S$. Then
$\s(j)^*(K):=\s(j)^{-1}(K)$ is a subset of $\Si_{\s,S_1}$ and, as
such, represents a proposition about the system $S_1$.  We say
that $\s(j)^*(K)$ is the \emph{pull-back} to $\Si_{\s,S_1}$ of the
$S$-proposition represented by $K$. The existence of such
pull-backs is part of the consistency of the representation of
propositions in classical mechanics, and it is important to
understand what the analogue of this is in our topos scheme.

Consider the general case with the two systems $S_1,S$ as above.
Then let $K$ be a proposition, represented as a sub-object of
$\Si_{\phi,S}$, with a monic arrow $i_K:
K\hookrightarrow\Si_{\phi,S}$. The question now is if the triple
$\la\tau_\phi(j),\phi(j),\beta_\phi(j)\ra$ can be used to pull $K$
back to give a proposition in $\tau(S_1)$, \ie\ a sub-object of
$\Si_{\phi, S_1}$?

The first requirement is that the functor
$\tau_\phi(j):\tau_{\phi}(S)\map \tau_\phi(S_1)$ should
\emph{preserve monics}. In this case, the monic arrow $i_K:
K\hookrightarrow\Si_{\phi,S}$ in $\tau_\phi(S)$ is transformed to
the monic arrow
\begin{equation}
        \tau_\phi(j)(i_K):\tau_\phi(j)(K)
        \hookrightarrow\tau_\phi(j)(\Si_{\phi,S})
\end{equation}
in $\tau_\phi(S_1)$; thus $\tau_\phi(j)(K)$ is a sub-object of
$\tau_\phi(j)(\Si_{\phi,S})$ in $\tau_\phi(S_1)$. It is a property
of a topos that the pull-back of a monic arrow is monic ; \ie\ if
$M\hookrightarrow Y$ is monic, and if $\psi:X\map Y$, then
$\psi^{-1}(M)$ is a sub-object of $X$.  Therefore, in the case of
interest, the monic arrow $
\tau_\phi(j)(i_K):\tau_\phi(j)(K)\hookrightarrow
\tau_\phi(j)(\Si_{\phi,S})$ can be pulled back along
$\phi(j):\Si_{\phi,S_{1}}\map \tau_\phi(j)(\Si_{\phi,S})$ (see
diagram \eq{PBPhiA}) to give the monic
$\phi(j)^{-1}(\tau_\phi(j)(K))\subseteq\Si_{\phi,S_1}$. This is a
candidate for the pull-back of the proposition represented by the
sub-object $K\subseteq \Si_{\phi,S}$.

In conclusion, propositions can be pulled-back provided that the
functor $\tau_\phi(j):\tau_{\phi}(S)\map \tau_\phi(S_1)$ preserves
monics. A sufficient way of satisfying this requirement is for
$\tau_\phi(j)$ to be left-exact. However, this raises the question
of ``where do left-exact functors come from?''.

\subsubsection{The Idea of a Geometric Morphism.}
\label{SubSubSec:GeomMorph} It transpires that there \emph{is} a
natural source of left-exact functors, via the idea of a
\emph{geometric morphism}. This fundamental concept in topos
theory is defined as follows \cite{MM92}.

\begin{definition}
A \emph{geometric morphism} $\phi:{\cal F}\map{\cal E}$ between
topoi $\cal F$ and $\cal E$ is a pair of functors $\phi^*:{\cal
E}\map{\cal F}$ and $\phi_*:{\cal F}\map \cal E$ such that
\begin{enumerate}
    \item[(i)] $\phi^*\dashv \phi_*$, \ie\ $\phi^*$ is left adjoint
    to $\phi_*$;
    \item[(ii)] $\phi^*$ is left exact, \ie\ it preserves
    all finite limits.
\end{enumerate}
The morphism $\phi^*:{\cal E}\map{\cal F}$ is called the
\emph{inverse image} part of the geometric morphism $\varphi$;
$\phi_*:{\cal F}\map \cal E$ is called the \emph{direct image}
part. \end{definition}

Geometric morphisms are very important because they are the topos
equivalent of continuous functions. More precisely, if $X$ and $Y$
are topological spaces, then any continuous function $f:X\map Y$
induces a geometric morphism between the topoi ${\rm Sh}(X)$ and
${\rm Sh}(Y)$ of sheaves on $X$ and $Y$ respectively.  In
practice, just as the arrows in the category of topological spaces
are continuous functions, so in any category whose objects are
topoi, the arrows are normally defined to be geometric morphisms.
In our case, as we shall shortly see, all the examples of
left-exact functors that arise in the quantum case do, in fact,
come from geometric morphisms. For these reasons, from now on we
will postulate that any arrows between our topoi arise from
geometric morphisms.

One central property of a geometric morphism is that it preserves
expressions written in terms of geometric logic. This greatly
enhances the attractiveness of assuming from the outset that the
internal logic of the system languages, $\L{S}$, is restricted to
the sub-logic afforded by geometric logic.

\emph{En passant}, another key result for us is the following
theorem (\cite{MM92} p359):
\begin{theorem}\label{Th:phiC->D}
If $\varphi:{\cal C}\map {\cal D}$ is a functor between categories
$\cal C$ and $\cal D$, then it induces a geometric morphism (also
denoted $\varphi$)
\begin{equation}
        \varphi:\SetC{{\cal C}}\map\SetC{{\cal D}}
\end{equation}
for which the functor $\varphi^*:\SetC{{\cal D}}\map\SetC{{\cal
C}}$ takes a functor $\ps{F}:{\cal D}\map\Set$ to the functor
\begin{equation}
   \varphi^*(\ps{F}):=\ps{F}\circ\varphi^\op    \label{Def:phi*(F)}
\end{equation}
from $\cal C$ to $\Set$.

In addition, $\varphi^*$ has a left adjoint $\varphi_!$; \ie\
$\varphi_!\dashv\varphi^*$.
\end{theorem}
We will use this important theorem in several crucial places.

\subsection{The Topos Rules for Theories of Physics}
\label{SubSec:GeneralToposAxioms} We will now present  our general
rules for using topos theory in the mathematical representation of
physical systems and their theories. {\definition The category
${\cal M}(\Sys)$ is the following:
\begin{enumerate}
\item The objects of ${\cal M}(\Sys)$ are the topoi that are to be
used in representing the systems in $\Sys$.

\item The arrows from $\tau_1$ to $\tau_2$ are defined to be
the  geometric morphisms from $\tau_2$ to $\tau_1$. Thus the
inverse part, $\varphi^*$, of an arrow
$\varphi^*:\tau_1\map\tau_2$ is a left-exact functor from $\tau_1$
to $\tau_2$.
\end{enumerate}
}

{\definition The  rules for using topos theory are as follows:
}\label{D_GeneralToposRules}
\begin{enumerate}
\item A \emph{topos realisation}, $\phi$, of $\Sys$ in
${\cal M}(\Sys)$ is an assignment, to each system $S$ in $\Sys$,
of a triple $\phi(S)=\la\rho_{\phi, S},\L{S},\tau_\phi(S)\ra$
where:
\begin{enumerate}
        \item $\tau_\phi(S)$ is the topos in  ${\cal M}(\Sys)$ in
         which the physical theory of  system $S$ is to be realised.

       \item $\L{S}$ is the local language that is associated with $S$.
        This is independent of the realisation, $\phi$, of $\Sys$
        in ${\cal M}(\Sys)$.

\item $\rho_{\phi,S}:\L{S}\sq\tau_\phi(S)$ is
       a representation of the local language $\L{S}$ in the
       topos $\tau_\phi(S)$.

\item In addition, for each arrow $j:S_1\map S$ in $\Sys$ there
is a triple  $\la\tau_\phi(j)$,$\phi(j)$, $\beta_\phi(j)\ra$ that
interpolates between $\rho_{\phi,S}:\L{S}\sq\tau_\phi(S)$ and
$\rho_{\phi,S_1}:\L{S_1}\sq\tau_\phi(S_1)$; for details see below.
\end{enumerate}

\item
\begin{enumerate}
\item The representations, $\rho_{\phi,S}(\Si)$ and
$\rho_{\phi,S}(\R)$,  of the ground symbols $\Si$ and $\R$ in
$\L{S}$ are denoted $\Si_{\phi,S}$ and $\R_{\phi, S}$,
respectively. They are known as the `state object' and
`quantity-value object' in $\tau_\phi(S)$.

\item The representation by $\rho_{\phi,S}$ of each function symbol
$A:\Si\map\R$ of the system $S$ is an arrow, $\rho_{\phi,S}(A):
\Si_{\phi,S}\map\R_{\phi,S}$ in $\tau_\phi(S)$; we will usually
denote this arrow as $A_{\phi, S}:\Si_{\phi,S}\map\R_{\phi,S}$.

\item Propositions about the system $S$ are represented by sub-objects
of $\Si_{\phi,S}$. These will typically be of the form
$A_{\phi,S}^{-1}(\Xi)$, where $\Xi$ is a sub-object of
$\R_{\phi,S}$.\footnote{Here, $A_{\phi,S}^{-1}(\Xi)$ denotes the
sub-object of $\Si_{\phi,S}$ whose characteristic arrow is
$\chi_{\Xi}\circ A_{\phi,S}:\Si_{\phi,S}\map
\Omega_{\tau_\phi(S)}$, where $\chi_{\Xi}:\R_{\phi,S}\map
\Omega_{\tau_\phi(S)}$ is the characteristic arrow of the
sub-object $\Xi$.}
\end{enumerate}

\item Generally, there are  no `microstates' for the system $S$;
\ie\  no global elements (arrows $1\map \Si_{\phi,S}$) of the
state object $\Si_{\phi,S}$; or, if there are any, they may not be
enough to determine $\Si_{\phi,S}$ as an object in $\tau_\phi(S)$.

Instead, the role of a state is played by a `truth sub-object'
$\TO$ of $P\Si_{\phi,S}$.\footnote{In classical physics, the truth
object corresponding to a micro-state $s$ is the collection of all
propositions that are true in the state $s$.} If
$J\in\Sub{\Si_{\phi,S}}\simeq\Ga(P\Si_{\phi,S})$, the `truth of
the proposition represented by $J$' is defined to be
\begin{equation}
\TValM{J\in\TO}=
        \Val{\tilde{J}\in\tilde\TO}_\phi\,
      \circ   \la\name{J},\name\TO\ra
\end{equation}
See Section \ref{SubSub:TruthObjects} for full information on the
idea of a `truth object'. Alternatively, one may use pseudo-states
rather than truth objects, in which case the relevant truth values
are of the form $\nu(\w\subseteq J)$.

\item  There is a `unit object' $1_{{\cal M}(\Sys)}$ in ${\cal
M}(\Sys)$ such that if $1_\Sys$ denotes the trivial system in
$\Sys$ then, for all topos realisations $\phi$,
    \begin{equation}
        \tau_\phi(1_\Sys)=1_{{\cal M}(\Sys)}.
    \end{equation}

Motivated by the results for quantum theory (see Section
\ref{SubSubSec:TransSum}), we postulate that the unit object
$1_{{\cal M}(\Sys)}$ in ${\cal M}(\Sys)$ is  the category of sets:
    \begin{equation}
            1_{{\cal M}(\Sys)}=\Set.
    \end{equation}

\item To each arrow $j:S_1\map S$ in $\Sys$, we have the following:
\begin{enumerate}
\item There is a translation $\L{j}:\L{S}\map\L{S_1}$.
This is specified by a map between function symbols:
$\L{j}:\F{S}\map\F{S_1}$.

\item With the translation
$\L{j}:\F{S}\map\F{S_1}$ there is associated a corresponding
function
\begin{equation}
        \phi(\L{j}):\Hom{\tau_\phi(S)}{\Si_{\phi,S}}{\R_{\phi,S}}
        \map\Hom{\tau_\phi(S_1)}{\Si_{\phi,S_1}}{\R_{\phi,S_1}}.
\end{equation}
These may, or may not, fit together in the commutative diagram:
\begin{equation}           \label{Diag_CommOfTranslationAndRep}
        \setsqparms[1`1`1`1;1200`700]
        \square[\F{S}`\Hom{\tau_\phi(S)}{\Si_{\phi,S}}
        {\R_{\phi,S}}`\F{S_1}`\Hom{\tau_\phi(S_1)}
        {\Si_{\phi,S_1}}{\R_{\phi,S_1}};
        \rho_{\phi,S}`\L{j}`\phi(\L{j})`\rho_{\phi,S_1}]
\end{equation}
\item The function  $\phi(\L{j}): \Hom{\tau_\phi(S)}{\Si_{\phi,S}}
{\R_{\phi,S}} \map\Hom{\tau_\phi(S_1)}{\Si_{\phi,S_1}}
{\R_{\phi,S_1}}$ is built from the following ingredients. For each
topos realisation $\phi$, there is a triple
$\la\nu_\phi(j),\phi(j),\beta_\phi(j)\ra$ where:
\begin{enumerate}
\item[(i)] $\nu_\phi(j): \tau_\phi(S_1)\map\tau_\phi(S)$ is a geometric morphism; \ie\ an arrow in the category  ${\cal M}(\Sys)$ (thus
$\nu_\phi(j)^*:\tau_\phi(S)\map\tau_\phi(S_1)$ is left exact).

\vspace{5pt} {\bf N.B.} To simplify the notation a little we will
denote $\nu_\phi(j)^*$ by $\tau_\phi(j)$. This is sensible in so
far as, for the most part, only the inverse part of $\nu_\phi(j)$
will be used in our constructions. \vspace{5pt}

\item[(ii)]
$\phi(j):\Si_{\phi,S_1}\map \tau_\phi(j)\big(\Si_{\phi,S}\big)$ is
an arrow in the topos $\tau_\phi(S_1)$.

\item[(iii)]
$\beta_\phi(j):\tau_\phi(j)\big(\R_{\phi,S}\big)
\map\R_{\phi,S_1}$ is an arrow in the topos $\tau_\phi(S_{1})$.
\end{enumerate}
These fit together in the diagram
\begin{equation}                                \label{PBPhiA2}
        \xext=3000
        \yext=3000
        \setsqparms[1`-1`1`1;1200`450]
        \square[\tau_\phi(j)(\Si_{\phi,S})`\tau_\phi(j)(\R_{\phi,S})
        `\Si_{\phi,S_1}`\R_{\phi,S_1};
        `\phi(j)`\beta_\phi(j)`\phi(\L{j})(A_{\phi,S})]
        \setsqparms[1`1`1`0;1200`450]
        \putsquare(-1470,450)[\Si_{\phi,S}`\R_{\phi,S}``;
        A_{\phi,S}`\tau_\phi(j)`\tau_\phi(j)`\tau_\phi(j)(A_{\phi,S})]
\end{equation}
The arrows $\phi(j)$ and $\beta_\phi(j)$ should behave
appropriately under composition of arrows in $\Sys$.

The commutativity of the diagram \eq{Diag_CommOfTranslationAndRep}
is equivalent to the relation
\begin{equation}
\phi(\L{j})(A_{\phi,S})=[\L{j}(A)]_{\phi,S_1}\label{CDpullback}
\end{equation}
for all $A\in\F{\phi,S}$. As we keep emphasising, the satisfaction
or otherwise of this relation will depend on the theory-type and,
possibly, the representation $\phi$.

\item If a proposition in $\tau_\phi(S)$ is represented by the
monic arrow, $K\hookrightarrow\Si_{\phi,S}$, the `pull-back' of
this proposition to $\tau_\phi(S_1)$ is defined to be
$\phi(j)^{-1}\big(\tau_\phi(j)(K)\big)\subseteq\Si_{\phi,S_1}$.
\end{enumerate}

\item
\begin{enumerate}
\item If $S_1$ is a sub-system of $S$, with an associated arrow
$i:S_1\map S$ in $\Sys$ then, in the diagram in \eq{PBPhiA2}, the
arrow $\phi(j):\Si_{\phi,S_1}\map\tau_\phi(j)(\Si_{\phi,S})$ is a
monic arrow in $\tau_\phi(S_{1})$.

In other words, $\Si_{\phi,S_1}$ is a sub-object of
$\tau_\phi(j)(\Si_{\phi,S})$, which is denoted
\begin{equation}
\Si_{\phi,S_1}\subseteq
        \tau_\phi(j)(\Si_{\phi,S}).  \label{SubSysSS}
\end{equation}

We may also want to conjecture
\begin{equation}
\R_{\phi,S_1}\simeq
            \tau_\phi(j)\big(\R_{\phi,S}\big).  \label{RsimR}
\end{equation}

\item Another possible conjecture is the following: if $j:S_1\map S$
is an epic arrow in $\Sys$, then, in the diagram in \eq{PBPhiA2},
the arrow $\phi(j):\Si_{\phi,S_1}\map \tau_\phi(j)(\Si_{\phi,S})$
is an epic arrow in $\tau_\phi(S_1)$.

In particular, for the epic arrow $p_1:S_1\di S_2\map S_1$, the
arrow $\phi(p_1):\Si_{\phi,S_1\di S_2}\map
\tau_\phi\big(\Si_{\phi,S_1}\big)$ is an epic arrow in the topos
$\tau_\phi(S_1\di S_2)$.
\end{enumerate}
\end{enumerate}

One should not read Rule 2.\ above as implying that the choice of
the state object and quantity-value object are \emph{unique} for
any given system $S$. These objects would at best be selected only
up to isomorphism in the topos $\tau(S) $. Such morphisms in the
 $\tau(S)$\footnote{Care is needed not to confuse morphisms
in the topos $\tau(S)$ with morphisms in the category ${\cal
M}(\Sys)$ of topoi. An arrow from the object $\tau(S)$ to itself
in the category ${\cal M}(\Sys)$ is a geometric morphism in the
topos $\tau(S)$. However, not every arrow in $\tau(S)$ need arise
in this way, and an important role can be expected to be played by
arrows of this second type. A good example is when $\tau(S)$ is
the category of sets, $\Set$. Typically,
$\tau_\phi(j):\Set\map\Set$ is the identity, but there are many
morphisms from an object $O$ in $\Set$ to itself: they are just
the functions from $O$ to $O$.} can be expected to play a key role
in developing the topos analogue of the important idea of a
\emph{symmetry}, or \emph{covariance} transformation of the
theory.

In the example of classical physics,  for all systems we have
$\tau(S)=\Set$ and $\Si_{\s,S}$ is a symplectic manifold, and the
collection of all symplectic manifolds is a category. It would be
elegant if we could assert that, in general, for a given
theory-type the possible state objects in a given topos $\tau$
form the objects of an \emph{internal} category in $\tau$.
However, to make such a statement would require a general theory
of state-objects and, at the moment, we do not have such a thing.

From a more conceptual viewpoint we note that the `similarity' of
our axioms to those of standard classical physics is reflected in
the fact that (i) physical quantities are represented by arrows
$A_{\phi,S}:\Si_{\phi,S}\map\R_{\phi,S}$; (ii) propositions are
represented by sub-objects of $\Si_{\phi,S}$; and (iii)
propositions are assigned truth values. Thus any theory satisfying
these axioms `looks' like classical physics, and has an associated
neo-realist interpretation.

\section{The General Scheme applied to Quantum Theory}
\label{Sec:ReviewQT}
\subsection{Background Remarks}
We now want to study the extent to which our `rules' apply to the
topos representation of quantum theory.

For a quantum system with (separable) Hilbert space $\Hi$, the
appropriate topos (what we earlier called $\tau_\phi(S)$) is
$\SetH{}$: the category of presheaves over the category (actually,
partially-ordered set) $\V{}$ of unital, abelian von Neumann
sub-algebras of the algebra, $\BH$,\ of bounded operators on
$\Hi$.

A particularly important object in $\SetH{}$ is the \emph{spectral
presheaf} $\Sig$, where, for each $V$, $\Sig_V$ is defined to be
the Gel'fand spectrum of the abelian algebra $V$. The sub-objects
of $\Sig$ can be identified as the topos representations of
propositions, just as the subsets of $\S$ represent propositions
in classical physics.

In Sections \ref{Sec:psSR} and \ref{Sec:PskSR}, several closely
related choices for a quantity-value object $\R_\phi$ in $\SetH{}$
were discussed. In order to keep the notation simpler, we
concentrate here on the presheaf $\SR$ of real-valued,
order-reversing functions. All results hold analogously if the
presheaf $\PR{\mathR}$ (which we actually prefer for giving a
better physical interpretation) is used.\footnote{Since the
construction of the arrows $\dasB{A}:\Sig\map\PR{\mathR}$ involves
both inner and outer daseinisation, we would have double work with
the notation, which we avoid here.}

Hence, physical quantities $A:\Si\rightarrow\R$, which correspond
to self-adjoint operators $\A$, are represented by natural
transformations/arrows $\dasBo{A}:\Sig\rightarrow\SR$.  The
mapping $\A\mapsto\dasBo{A}$ is injective. For brevity, we write
$\dasB{A}:=\dasBo{A}$.\footnote{Note that this is \emph{not} the
same as the convention used earlier, where $\dasB{A}$ denoted a
different natural transformation!}

\subsection{The Translation Representation for  a Disjoint
Sum of Quantum Systems}\label{SubSubSec:TransSum} Let $\Sys$ be a
category whose objects are systems that can be treated using
quantum theory. Let $\L{S}$ be the local language of a system $S$
in $\Sys$ whose quantum Hilbert space is denoted $\Hi_S$. We
assume that to each function symbol, $A:\Si\map\R$, in $\L{S}$
there is associated a self-adjoint operator $\hat
A\in\mathcal{B(H} _{S}),$\footnote{More specifically, one could
postulate that the elements of $\F{S}$ are associated with
self-adjoint operators in some unital von Neumann sub-algebra of
$\mathcal{B(H}_{S})$.} and that the map
\begin{eqnarray}
         \F{S} &\rightarrow& \BH_\sa\\
        A &\mapsto& \hat A
\end{eqnarray}
is injective (but not necessarily surjective, as we will see in
the case of a disjoint sum of quantum systems).

We consider first arrows of the form
\begin{equation}
S_{1}\overset{i_{1}}{\rightarrow}S_{1}\sqcup S_{2}\overset{i_{2}
}{\leftarrow}S_{2}
\end{equation}
from the components $S_{1}$, $S_{2}$ to a disjoint sum
$S_{1}\sqcup S_{2}$; for convenience we write $i:=i_{1}$.  The
systems $S_{1}$, $S_{2}$ and $S_{1}\sqcup S_{2}$ have the Hilbert
spaces $\Hi_1$, $\Hi_2$ and $\Hi_1\oplus\Hi_2$, respectively.

As always, the translation $\L{i}$ goes in the opposite direction
to the arrow $i$, so
\begin{equation}
\L{i}:\F{S_1\sqcup S_2}\map \F{S_1}.
\end{equation}
Then our first step is find an  `operator translation' from the
relevant self-adjoint operators in $\Hi_1\oplus\Hi_2$ to those in
$\Hi_1$,

To do this, let $A$ be a function symbol in $\F{S_1\sqcup S_2}$.
In Section \ref{SubSubSec:ATDS}, we argued that $\F{S_1\sqcup
S_2}\simeq\F{S_1}\times\F{S_2}$ (as in \eq{FS1sumS2}), and hence
we introduce the notation $A=\la A_1, A_2\ra$, where
$A_1\in\F{S_1}$ and $A_2\in\F{S_2}$. It is then natural to assume
that the quantisation scheme is such that the operator, $\hat A$,
on $\Hi_1\oplus\Hi_2$ can be decomposed as $\hat A=\hat
A_1\oplus\hat A_2$, where the operators $\hat A_1$ and $\hat A_2$
are defined on $\Hi_1$ and $\Hi_2$ respectively, and correspond to
the function symbols $A_1$ and $A_2$.\footnote{It should be noted
that our scheme does not use all the self-adjoint operators on the
direct sum $\Hi_1\oplus\Hi_2$: only the `block diagonal' operators
of the form $\hat A=\hat A_1\oplus\hat A_2$ arise.} Then the
obvious operator translation is $\hat A\mapsto\hat
A_1\in\mathcal{B(H}_{1})_{sa}$.

We now consider the general rules in the  Definition
\ref{D_GeneralToposRules} and see to what extent they apply in the
example of quantum theory.

\textbf{1.} As we have stated several times, the topos
$\tau_{\phi}(S)$ associated with a quantum system $S$ is
\begin{equation}
\tau_{\phi}(S)=\Set^{\mathcal{V(H}_{S})^{op}}.
\end{equation}
Thus (i) the objects of the category $\mathcal{M}(\Sys)$ are topoi
of the form $\Set^{\mathcal{V(H}_{S})^{op}}$, $S\in\Ob\Sys$; and
(ii) the arrows between two topoi are defined to be  geometric
morphisms. In particular, to each arrow $j:S_1\map S$ in $\Sys$
there must correspond a geometric morphism
$\nu_\phi(j):\tau_\phi(S_1)\map\tau_\phi(S)$ with associated
left-exact functor
$\tau_\phi(j):=\nu_\phi(j)^*:\tau_\phi(S)\map\tau_\phi(S_{1})$. Of
course, the existence of these functors in the quantum case has
yet to be shown.

\textbf{2.} The realisation $\rho_{\phi,S}:\L{S}\rightsquigarrow
\tau_\phi(S)$ of the language $\L{S}$ in the topos $\tau_\phi(S)$
is given as follows. First, we define the state object
$\Si_{\phi,S}$ to be  the spectral presheaf,
$\Sig^{\mathcal{V(H}_{S})}$, over $\mathcal{V(H}_{S}\mathcal{)}$,
the context category of $\mathcal{B(H}_{S})$. To keep the notation
brief, we will denote\footnote{Presheaves are always denoted by
symbols that are underlined.}
 $\Sig^{\mathcal{V(H}_{S})}$ as $\Sig^{\Hi_{S}}$.

Furthermore, we define the quantity-value object, $\R_{\phi,S}$,
to be the presheaf $\SR{}^{\;\Hi_{S}}$ that was defined in Section
\ref{Sec:psSR} . Finally, we define
\begin{equation}
A_{\phi,S}:=\dasB{A},
\end{equation}
for all $A\in\F{S}$. Here $\dasB{A}:\Sig^{\Hi_{S}
}\map\SR{}^{\Hi_S}$ is constructed using the Gel'fand transforms
of the (outer) daseinisation of $\hat A$, for details see later.

\textbf{3.} The truth object $\ps\TO^{\ket\psi}$ corresponding to
a pure state $\ket\psi$ was discussed in Section
\ref{SubSubSec:TOQT}. Alternatively, we have the pseudo-state
$\ps\w^{\ket\psi}$.

\textbf{4.} Let $\Hi=\mathC$ be the one-dimensional Hilbert space,
corresponding to the trivial quantum system $1$. There is exactly
one abelian sub-algebra of $\mathcal{B}(\mathC)\simeq\mathC$,
namely $\mathC$ itself. This leads to
\begin{equation}
\tau_{\phi}(1_{\Sys})=\Set^{\{*\}}\simeq\Set=1_{\mathcal{M}(\Sys)}.
\end{equation}

\textbf{5.} Let $A\in \F{S_{1}\sqcup S_{2}}$ be a function symbol
for the system $S_1\sqcup S_2$. Then, as discussed above, $A$ is
of the form $A=\la A_1, A_2\ra$ (compare equation \eq{FS1sumS2}),
which corresponds to a self-adjoint operator $\hat
A_1\oplus\hat{A}_2\in\mathcal{B(H}_{1}\oplus\Hi_2)_{\sa}$. The
topos representation of $A$ is the natural transformation
$\breve\delta(\la A_1, A_2\ra):\Sig^{\Hi_1\oplus\Hi_2}
\map\SR{}^{\;\Hi_1\oplus\Hi_2}$, which is defined at each stage
$V\in\Ob{\mathcal{V(H}_{1}\oplus\Hi_2)}$ as
\begin{eqnarray}
\breve\delta(\la A_1,A_2\ra)_V:\Sig^{\Hi_1\oplus\Hi_2}_V
  &\map&\SR{}^{\;\Hi_1\oplus\Hi_2}_V \nonumber  \\
\l &\mapsto&\{V^\prime\mapsto \brak{\l|_{V^{\prime}}} {\delta(
\hat{A_1}\oplus\hat{A_2})_{V^{\prime}}}    \mid
V^{\prime}\subseteq V\} \label{Def:dasB(A1+A2)}
\end{eqnarray}
where the right hand side \eq{Def:dasB(A1+A2)} denotes an
order-reversing function.

We will need the following:
\begin{lemma}
\label{L_DaseinisationAndDirectSum}Let
$\hat{A_1}\oplus\hat{A_2}\in\mathcal{B(H}_{1}\oplus\Hi_2)_{\sa}$,
and let $V=V_1\oplus V_2\in\Ob{\mathcal{V}(\Hi_1\oplus\Hi_2)}$
such that $V_1\in\Ob{\V{1}}$ and $V_2\in\Ob{\V{2}}$. Then
\begin{equation}
\delta(\hat{A_1}\oplus\hat{A_2})_{V}=\delta(\hat{A_1})_{V_1}
\oplus\delta(\hat{A_2})_{V_2}. \label{das(A1+A2)=}
\end{equation}

\end{lemma}

\begin{proof}
Every projection $\hat Q\in V$ is of the form $\hat Q=\hat
Q_{1}\oplus\hat Q_{2}$ for unique projections $\hat
Q_{1}\in\mathcal{P(H}_{1})$ and $\hat Q_{2}
\in\mathcal{P(H}_{2})$. Let $\P\in\PH$ be of the form
$\P=\P_1\oplus\P_2$ such that $\P_1\in\mathcal{P(H}_1)$ and
$\P_2\in\mathcal{P(H}_1)$. The largest projection in $V$ smaller
than or equal to $\P$, \ie\ the inner daseinisation of $\P$ to
$V$, is
\begin{equation}
\dastoi{V}{P}=\hat{Q}_{1}\oplus\hat{Q}_{2},
\end{equation}
where $\hat{Q}_{1}\in\mathcal{P(}V_1)$ is the largest projection
in $V_1$ smaller than or equal to $\P_{1}$, and $\hat{Q}_{2}%
\in\mathcal{P(}V_2)$ is the largest projection in $V_2$ smaller
than or equal to $\P_{2}$, so
\begin{equation}
\dastoi{V}{P}=\delta(\P_{1})_{V_1}\oplus\delta(\P_{2})_{V_2}.
\end{equation}
This implies $\delta(\hat{A}\oplus\hat{B}
)_{V}=\delta(\hat{A})_{V_1}\oplus\delta(\hat{B})_{V_2}$, since
(outer) daseinisation of a self-adjoint operator just means inner
daseinisation of the projections in its spectral family, and all
the projections in the spectral family of $\A\oplus\hat B$ are of
the form $\P=\P_1\oplus\P_2$.
\end{proof}

As discussed in Section \ref{Sec:ToposAxioms}, in order to mimic
the construction that we have in the classical case, we need  to
pull back the arrow/natural transformation $\breve\delta(\la
A_1,A_2\ra): \Sig^{\Hi_1\oplus\Hi_2}
\map\SR{}^{\;\Hi_1\oplus\Hi_2}$  to obtain an arrow from
$\Sig^{\Hi_1}$ to $\SR{}^{\;\Hi_1}$. Since we decided that the
translation on the level of operators sends
$\hat{A_1}\oplus\hat{A_2}$ to $\hat{A_1}$, we expect that this
arrow from $\Sig^{\Hi_1}$ to $\SR{}^{\;\Hi_1}$ is $\dasB{A_1}$. We
will now show how this works.

The presheaves $\Sig^{\Hi_1\oplus\Hi_2}$ and $\Sig^{\Hi_1}$ lie in
different topoi, and in order to `transform' between them we need
we need a (left-exact) functor from the topos
$\Set^{\mathcal{V(\Hi}_1\oplus\Hi_2)^{op}}$to the topos
$\SetH{1}$: this is the functor
$\tau_\phi(j):\tau_\phi(S)\map\tau_\phi(S_1)$ in \eq{PBPhiA2}. One
natural place to look for such a functor is as the inverse-image
part of a geometric morphism from $\Set^{\mathcal{V(H}_1)^{op}}$
to $\Set^{\mathcal{V(H}_{1}\oplus\Hi_2)^{\op}}$. According to
Theorem \ref{Th:phiC->D}, one source of such a geometric morphism,
$\mu$, is a functor
\begin{equation}
        m:\V{1}\map \mathcal{V(H}_{1}\oplus\Hi_2),
\end{equation}
and the obvious choice for this is
\begin{equation}
        m(V):=V\oplus \mathC\hat 1_{\Hi_2}
\end{equation}
for all $V\in\Ob{\V{1}}$. This function from $\Ob{\V{1}}$ to
$\Ob{\mathcal{V(H}_{1}\oplus\Hi_2)}$ is clearly order preserving,
and hence $m$ is a genuine functor.

Let $\mu:\Set^{\mathcal{V(H}_{1})^\op}\map
\Set^{\mathcal{V(H}_1\oplus\Hi_2)^\op}$ denote the geometric
morphism induced by $m$. The inverse-image functor of $\mu$ is
given by
\begin{eqnarray}
\mu^*:\Set^{\mathcal{V(H}_1\oplus\Hi_2)^\op}&\map&
                \Set^{\mathcal{V(H}_{1})^\op}         \\
                \ps{F}  &  \mapsto& \ps{F}\circ m^\op.
\end{eqnarray}
This means that, for all $V\in \Ob{\V{1}}$, we have
\begin{equation}
(\mu^*\ps{F}^{\Hi_1\oplus\Hi_2})_V=\ps{F}^{\Hi_1\oplus\Hi_2}_{m(V)}=
\ps{F}^{\Hi_1\oplus\Hi_2}_{V\oplus \mathC\hat
1_{\Hi_2}}.\label{mu*Si+2}
\end{equation}
For example, for the spectral presheaf we get
\begin{equation}
(\mu^*\Sig^{\Hi_1\oplus\Hi_2})_V=\Sig^{\Hi_1\oplus\Hi_2}_{m(V)}=
        \Sig^{\Hi_1\oplus\Hi_2}_{V\oplus \mathC\hat 1_{\Hi_2}}.
\end{equation}
This is the functor that is denoted
$\tau_\phi(j):\tau_\phi(S_1)\map\tau_\phi(S)$ in  \eq{PBPhiA2}.

We next need to find an arrow
$\phi(i):\Sig^{\Hi_1}\map\mu^*\Sig^{\Hi_1\oplus\Hi_2}$ that is the
analogue of the arrow
$\phi(j):\Si_{\phi,S_1}\map\tau_\phi(j)(3\Si_{\phi,S})$ in
\eq{PBPhiA2}.

For each $V$, the set $(\mu^*\Sig^{\Hi_1\oplus\Hi_2})_V=
\Sig^{\Hi_1\oplus\Hi_2}_{V\oplus \mathC\hat 1_{\Hi_2}}$ contains
two types of spectral elements $\l$: the first type are those $\l$
such that $\brak\l{\hat 0_{\Hi_1}\oplus\hat 1_{\Hi_2}}=0$. Then,
clearly, there is some $\tilde\l\in\Sig^{\Hi_1}_V$ such that
$\brak{\tilde\l}{\A}=\brak\l{\A\oplus\hat 0_{\Hi_2}}=
\brak\l{\A\oplus\hat1_{\Hi_2}}$ for all $\A\in V_\sa$. The second
type of spectral elements $\l\in\Sig^{\Hi_1\oplus\Hi_2}_{V\oplus
\mathC\hat 1_{\Hi_2}}$ are such that $\brak\l{\hat
0_{\Hi_1}\oplus\hat 1_{\Hi_2}}=1$. In fact, there is exactly one
such $\l$, and we denote it by $\l_0$. This shows that
$\Sig^{\Hi_1\oplus\Hi_2}_{V\oplus \mathC\hat 1_{\Hi_2}}
\simeq\Sig^{\Hi_1}_V\cup\{\l_0\}$. Accordingly, at each stage $V$,
the mapping $\phi(i)$ sends each $\tilde\l\in\Sig^{\Hi_1}_V$ to
the corresponding $\l\in\Sig^{\Hi_1\oplus\Hi_2}_{V\oplus
\mathC\hat 1_{\Hi_2}}$.

The presheaf $\SR{}^{\;\Hi_1\oplus\Hi_2}$ is given at each stage
$W\in\Ob{\mathcal{V}(\Hi_1\oplus\Hi_2)}$ as the order-reversing
functions $\nu:\downarrow\!\!W\rightarrow\mathR$, where
$\downarrow\!\!W$ denotes the set of unital, abelian von Neumann
sub-algebras of $W$. Let $W=V\oplus\mathC\hat 1_{\Hi_2}$. Clearly,
there is a bijection between the sets $\downarrow\!\!
W\subset\Ob{\mathcal{V}(\Hi_1\oplus\Hi_2)}$ and $\downarrow\!\! V
\subset\Ob{\V{}}$. We can thus identify
\begin{equation}
        (\mu^*\SR{}^{\;\Hi_1\oplus\Hi_2})_V=
        \SR{}^{\;\Hi_1\oplus\Hi_2}_{V\oplus \mathC\hat 1_{\Hi_2}}
        \simeq\SR^{\Hi_1}_V
\end{equation}
for all $V\in\Ob{\V{}}$. This gives an isomorphism $\beta_\phi(i):
\mu^*\SR{}^{\;\Hi_1\oplus\Hi_2}\map\SR{}^{\;\Hi_1}$, which
corresponds to the arrow
$\beta_\phi(j):\tau_\phi(j)(\R_{\phi,S})\map\R_{\phi,S_1}$ in
\eq{PBPhiA2}.

Now consider the arrow $\breve\delta(\la A_1, A_2\ra):
\Sig^{\Hi_1\oplus\Hi_2}\map\SR{}^{\;\Hi_1\oplus\Hi_2}$. This is
the analogue of the arrow $A_{\phi,S}:\Si_{\phi,S}\map\R_{\phi,S}$
in \eq{PBPhiA2}. At each stage
$W\in\Ob{\mathcal{V}(\Hi_1\oplus\Hi_2)}$, this arrow is given by
the (outer) daseinisation $\delta(\A_1\oplus\A_2)_{W^{\prime}}$
for all $W^{\prime}\in\downarrow\!\!W$. According to Lemma
\ref{L_DaseinisationAndDirectSum}, we have
\begin{equation}
        \delta(\hat{A}_1\oplus\hat{A}_2)_{V\oplus\mathC\hat 1_{\Hi_2}}=
\dasto{V}{A_1}\oplus\delta(A_2)_{\mathC\hat 1_{\Hi_2}}
        =\dasto{V}{A_1}\oplus\rm{max(sp}(\A_2))\hat 1_{\Hi_2}
\end{equation}
for all $V\oplus\mathC\hat
1_{\Hi_2}\in\Ob{\mathcal{V}(\Hi_1\oplus\Hi_2)}$. This makes clear
how the arrow
\begin{equation}
\mu^*(\breve\delta(\la A_1, A_2\ra)):\mu^*\Sig^{\Hi_1\oplus\Hi_2}
\map\mu^*\SR{}^{\;\Hi_1\oplus\Hi_2}
\end{equation}
is defined. Our conjectured pull-back/translation representation
is
\begin{equation}
        \phi(\L{i})\big(\breve\delta(\la A_1, A_2\ra)\big):=
    \beta_\phi(i)\circ\mu^*(\breve\delta(\la A_1, A_2\ra))
    \circ\phi(i):\Sig^{\Hi_1}\map\SR{}^{\;\Hi_1}.
\end{equation}
Using the definitions of $\phi(i)$ and $\beta_\phi(i)$, it becomes
clear that
\begin{equation}
\beta_\phi(i)\circ\mu^*(\breve\delta(\la A_1, A_2\ra))
\circ\phi(i)=\dasB{A_1}.
\end{equation}
Hence,  the commutativity condition in \eq{CDpullback} is
satisfied for arrows in $\Sys$ of the form $i_{1,2}:S_{1,2}\map
S_1\sqcup S_2$.

\subsection{The Translation Representation for Composite
Quantum Systems}\label{SubSec:TranslCompQT} We now consider arrows
in $\Sys$ of the form
\begin{equation}
S_{1}\overset{p_{1}}{\leftarrow}S_{1}\di S_{2}\overset{p_{2}
}{\rightarrow}S_{1},
\end{equation}
where the quantum systems $S_{1}$, $S_{2}$ and $S_{1}\di S_{2}$
have the Hilbert spaces $\Hi_1$, $\Hi_2$ and $\Hi_1\otimes \Hi_2$,
respectively.\footnote{As usual, the composite system $S_{1}\di
S_{2}$ has as its Hilbert space the tensor product of the Hilbert
spaces of the components.}

The canonical translation\footnote{As discussed in Section
\ref{SubSubSec:ArrTranComp}, this translation, $\L{p_1}$,
transforms a physical quantity $A_1$ of system $S_{1}$ into a
physical quantity $A_1\di1$, which is the `same' physical quantity
but now seen as a part of the composite system $S_{1}\di S_{2}$.
The symbol $1$ is the trivial physical quantity: it is represented
by the operator $\hat{1}_{\Hi_2}$.} $\L{p_1}$ between the
languages $\L{S_1}$ and $\L{S_{1}\di S_{2}}$ (see Section
\ref{SubSubSec:ArrTranComp}) is such that if $A_1$ is a function
symbol in $\F{S_1}$, then the corresponding operator
$\hat{A_1}\in\mathcal{B(H}_1)_{\sa}$ will be `translated' to the
operator $\hat{A_1}\otimes\hat{1}_{\Hi_2}\in
\mathcal{B(H}_{1}\otimes\Hi_2)$. By assumption, this corresponds
to the function symbol $A_1\di1$ in $\F{S_{1}\di S_{2}}$.

\subsubsection{Operator Entanglement and Translations.} We should be
cautious about what to expect from this translation when we
represent a physical quantity $A:\Si\map\R$ in $\F{S_1}$ by an
arrow between presheaves, since there are no canonical projections
\begin{equation}
\Hi_1\leftarrow\Hi_1\otimes\Hi_2 \rightarrow\Hi_2,
\end{equation}
and hence no canonical projections
\begin{equation}
\Sig^{\Hi_1}\leftarrow\Sig^{\Hi_1\otimes\Hi_2}\rightarrow
\Sig^{\Hi_2}
\end{equation}
from the spectral presheaf of the composite system to the spectral
presheaves of the components.\footnote{On the other hand, in the
classical case, there \emph{are} canonical projections
\begin{equation}
\Si_{\s,S_{1}}\leftarrow\Si_{\s,S_{1}\di S_{2}}
\rightarrow\Sigma_{\s,S_{2}}
\end{equation}
because the symplectic manifold $\Si_{\s,S_{1}\di S_{2}}$ that
represents the composite system is the cartesian product
$\Si_{\s,S_{1}\di S_{2}}=\Si_{\s,S_{1}}\times\Si_{\s,S_{2}}$,
which is a product in the categorial sense and hence comes with
canonical projections.}

This is the point where a form of \emph{entanglement} enters the
picture. The spectral presheaf $\Sig^{\Hi_1\otimes\Hi_2}$ is a
presheaf over the context category
$\mathcal{V}(\Hi_1\otimes\Hi_2)$ of $\Hi_1\otimes\Hi_2$. Clearly,
the context category $\V{1}$ can be embedded into
$\mathcal{V}(\Hi_1\otimes\Hi_2)$ by the mapping $V_1\mapsto
V_1\otimes\mathC\hat 1_{\Hi_2}$, and likewise $\V{2}$ can be
embedded into $\mathcal{V}(\Hi_1\otimes\Hi_2)$. But not every
$W\in\Ob{\mathcal{V}(\Hi_1\otimes\Hi_2)}$ is of the form
$V_1\otimes V_2$.

This comes from the fact that not all vectors in
$\Hi_1\otimes\Hi_2$ are of the form $\psi_1\otimes \psi_2$, hence
not all projections in $\mathcal{P}(\Hi_1\otimes\Hi_2)$ are of the
form $\P_{\psi_1}\otimes\P_{\psi_2}$, which in turn implies that
not all $W\in\mathcal{V}(\Hi_1\otimes\Hi_2)$ are of the form
$V_1\otimes V_2$. There are more contexts, or world-views,
available in $\mathcal{V}(\Hi_1\otimes\Hi_2)$ than those coming
from $\V{1}$ and $\V{2}$. We call this `operator entanglement'.

The topos representative of $\hat{A_1}$ is
$\dasB{A_1}:\Sig^{\Hi_1}\map \SR{}^{\;\Hi_1}$, and the
representative of $\hat{A_1}\otimes\hat{1}_{\Hi_2}$ is
$\breve\delta(A_1\di 1):
\Sig^{\Hi_1\otimes\Hi_2}\map\SR{}^{\;\Hi_1\otimes\Hi_2}$. At
sub-algebras $W\in\Ob{\mathcal{V(H}_{1}\otimes\Hi_2)}$ which are
\emph{not} of the form $W=V_1\otimes V_2$ for any
$V_1\in\Ob{\mathcal{V(H}_{1})}$ and
$V_2\in\Ob{\mathcal{V(H}_{2})}$, the daseinised operator
$\delta(\hat{A_1}_{W}\otimes\hat{1}_{\Hi_2})\in W_{\sa}$ will not
be of the form
$\delta(\hat{A_1})_{V_1}\otimes\delta(\hat{A_1})_{V_2}$.\footnote{Currently,
it is even an open question if
$\delta(\hat{A_1}_{W}\otimes\hat{1}_{\Hi_2})=
\delta(\hat{A_1})_{V_1}\otimes\hat{1}_{\Hi_2}$ if $W=V_1\otimes
V_2$ for a non-trivial algebra $V_2$.} On the other hand, it is
easy to see that  $\delta(\hat{A_1}\otimes\hat{1}_{\Hi_2})_{W}=
\delta(\hat{A_1})_{V_1}\otimes\hat{1}_{\Hi_2}$ if
$W=V_1\otimes\mathC\hat 1_{\Hi_2}$.

Given a physical quantity $A_1$, represented by the arrow
$\dasB{A_1}:\Sig^{\Hi_1}\rightarrow\SR^{\Hi_1}$, we can (at best)
expect that the translation of this arrow into an arrow from
$\Sig^{\Hi_1\otimes\Hi_2}$ to $\SR^{\Hi_1\otimes\Hi_2}$ coincides
with the arrow $\dasB{A_1\di 1}$ on the `\emph{image}' of
$\Sig^{\Hi_1}$ in $\Sig^{\Hi_1\otimes\Hi_2}$. This image will be
constructed below using a certain geometric morphism. As one might
expect, the image of $\Sig^{\Hi_1}$ is a presheaf $\ps P$ on
$\mathcal{V}(\Hi_1\otimes\Hi_2)$ such that $\ps
P_{V_1\otimes\mathC\hat 1_{\Hi_2}}\simeq\Sig^{\Hi_1}_{V_1}$ for
all $V_1\in\V{1}$, \ie\ the presheaf $\ps P$ can be identified
with $\Sig^{\Hi_1}$ exactly on the image of $\V{1}$ in
$\mathcal{V}(\Hi_1\otimes\Hi_2)$ under the embedding $V_1\mapsto
V_1\otimes\mathC\hat 1_{\Hi_2}$. At these stages, the translation
of $\dasB{A_1}$ will coincide with $\dasB{A_1\di 1}$. At other
stages $W\in\mathcal{V}(\Hi_1\otimes\Hi_2)$, the translation
cannot be expected to be the same natural transformation as
$\dasB{A\di 1}$ in general.

\subsubsection{A Geometrical Morphism and a Possible Translation}
The most natural approach to a translation is the following. Let
$W\in\Ob{\mathcal{V(H}_{1}\otimes\Hi_2)}$, and define $V_{W}
\in\Ob{\V{1}}$ to be the \emph{largest} sub-algebra of
$\mathcal{B(H}_{1})$ such that
$V_{W}\otimes\,\mathC\hat{1}_{\Hi_2}$ is a sub-algebra of $W$.
Depending on $W$, $V_{W}$ may, or may not, be the trivial
sub-algebra $\mathC\hat{1}_{\Hi_1}$. We note that if
$W^{\prime}\subseteq W$, then
\begin{equation}
V_{W^{\prime}}\subseteq V_{W}, \label{VW'<VW}
\end{equation}
but $W^{^{\prime}}\subset W$ only implies $V_{W^{\prime}}\subseteq
V_{W}$.

The trivial algebra $\mathC\hat{1}_{\Hi_1}$ is not an object in
the category $\V{1}$. This is why we introduce the `augmented
context category' $\V{1}_*$, whose objects are those of $\V{1}$
united with $\mathC\hat{1}_{\Hi_1}$, and with the obvious
morphisms ($\mathC\hat{1}_{\Hi_1}$ is a sub-algebra of all
$V\in\V{1}$).

Then there is a functor $n:{\mathcal{V(H}_{1}\otimes\Hi_2})\map
\V{1}_*$, defined as follows. On objects,
\begin{eqnarray}
        n:\Ob{\mathcal{V(}\Hi_1\otimes\Hi_2)} &\map& \Ob{\V{1}_*}                  \nonumber\\
    W &\mapsto& V_{W},
\end{eqnarray}
and if $i_{W^{\prime}W}:W^\prime\map W$ is an arrow in
$\mathcal{V(}\Hi_1 \otimes\Hi_2)$, we define $n(i_{W^{\prime}W}):=
i_{V_{W^{\prime}}V_{W}}$ (an arrow in $\V{1}_*$); if
$V_{W^{\prime}}=V_{W}$, then $i_{V_{W^{\prime}}V_{W}}$ is the
identity arrow $\operatorname*{id}_{V_{W}}$.

Now let
\begin{equation}
\nu:\Set^{\mathcal{V(H}_{1}\otimes\Hi_2)^{\op}}\map\Set^{{(\V{1}}_*)^{\rm
op}}
\end{equation}
denote the geometric morphism induced by $\pi$. Then the
(left-exact) inverse-image functor
\begin{equation}
\nu^*:\Set^{{(\V{1}}_*)^{\rm
op}}\map\Set^{\mathcal{V(H}_{1}\otimes\Hi_2)^\op}
\end{equation}
acts on a presheaf $\ps{F}\in\Set^{{(\V{1}}_*)^{\rm op}}$ in the
following way. For all $W\in\Ob{\mathcal{V(H}_{1}\otimes\Hi_2)}$,
we have
\begin{equation}
(\nu^*\ps{F})_W=\ps{F}_{n^{\op}(W)}=\ps{F}_{V_{W}}
\end{equation}
and
\begin{equation}
(\nu^*\ps{F})(i_{W^{\prime}W})=
                \ps{F}(i_{V_{W^{\prime}}V_{W}})
\end{equation}
for all arrows $i_{W^{\prime}W}$ in the category
$\mathcal{V(H}_{1}\otimes \Hi_2)$.\footnote{We remark, although
will not prove it here, that the inverse-image presheaf
$\nu^*\ps{F}$ coincides with the direct image presheaf
$\phi_*\ps{F}$ of $\ps{F}$ constructed from the geometric morphism
$\phi$ induced by the functor
\begin{eqnarray}
\ka:\mathcal{V(H}_{1})&\map&\mathcal{V(H}_{1}\otimes\Hi_2)
                                                \nonumber\\
V  &  \mapsto& V\otimes\mathC\hat{1}_{\Hi_2}.
\end{eqnarray}
Of course, the inverse image presheaf $\beta^*\ps{F}$ is much
easier to construct.}

In particular, for all $W\in\mathcal{V(H}_{1}\otimes\Hi_2)$, we
have
\begin{eqnarray}
(\nu^*\Sig^{\Hi_1})_W&=&\Sig^{\Hi_1}_{V_{W}},\label{Defnu*Sig}\\
(\nu^*\SR{}^{\;\Hi_1})_W  & =&\SR{}^{\;\Hi_1}_{V_{W}}.
\end{eqnarray}
Since $V_W$ can be $\mathC\hat 1_{\Hi_1}$, we have to extend the
definition of the spectral presheaf $\Sig^{\Hi_1}$ and the
quantity-value presheaf $\SR^{\Hi_1}$ such that they become
presheaves over $\V{1}_*$ (and not just $\V{1}$). This can be done
in a straightforward way: the Gel'fand spectrum $\Sig_{\mathC\hat
1_{\Hi_1}}$ of $\mathC\hat 1_{\Hi_1}$ consists of the single
spectral element $\l_1$ such that $\brak{\l_1}{\hat 1_{\Hi_1}}=1$.
Moreover, $\mathC\hat 1_{\Hi_1}$ has no sub-algebras, so the
order-reversing functions on this algebra correspond bijectively
to the real numbers $\mathR$.

Using these equations, we see that the arrow
$\dasB{A_1}:\Sig^{\Hi_1}\map\SR{}^{\;\Hi_1}$ that corresponds to
the self-adjoint operator $\hat{A_1}\in\mathcal{B(H}_{1})_{\sa}$
gives rise to the arrow
\begin{equation}
\nu^*(\dasB{A_1}):\nu^*\Sig^{\Hi_1}\map\nu^*\SR{}^{\;\Hi_1}.
                                        \label{dum2}
\end{equation}
In terms of our earlier notation, the functor
$\tau_\phi(p_1):\SetH{1}\map
\Set^{\mathcal{V(H}_{1}\otimes\Hi_2)^\op}$ is $\nu^*$, and the
arrow in \eq{dum2} is the arrow
$\tau_\phi(j)(A_{\phi,S}):\tau_\phi(j)(\Si_{\phi,S})\map
\tau_\phi(j)(\R_{\phi,S})$ in \eq{PBPhiA2} with $j:S_1\map S$
being replaced by  $p:S_1\di S_2\map S_1$, which is the arrow in
$\Sys$ whose translation representation we are trying to
construct.

The next arrow we need is the one denoted
$\beta_\phi(j):\tau_\phi(j)(\R_{\phi,S})\map\R_{\phi,S_1}$ in
\eq{PBPhiA2}. In the present case, we define
$\beta_\phi(p):\nu^*\SR{}^{\;\Hi_1}
\map\SR{}^{\;\Hi_1\otimes\Hi_2}$ as follows. Let
$\alpha\in(\nu^*\SR{}^{\;\Hi_1})_W\simeq\SR{}^{\;\Hi_1}_{V_W}$ be
an order-reversing real-valued function on $\downarrow\!\! V_W$.
Then we define an order-reversing function
$\beta_\phi(p)(\alpha)\in \SR{}^{\;\Hi_1\otimes\Hi_2}_W$ as
follows. For all $W^\prime\subseteq W$, let
\begin{equation}
        [\beta_\phi(p)(\alpha)](W^\prime):=\alpha(V_{W^\prime})
\end{equation}
which, by virtue of \eq{VW'<VW}, is an order-reversing function
and hence a member of $\SR{}^{\;\Hi_1\otimes\Hi_2}_W$.

We also need an arrow in
$\Set^{\mathcal{V(H}_{1}\otimes\Hi_2)^\op}$ from
$\Sig^{\Hi_1\otimes\Hi_2}$ to $\nu^*\Sig^{\Hi_1}$, where
$\nu^*\Sig^{\Hi_1}$ is defined in \eq{Defnu*Sig}. This is the
arrow denoted
$\phi(j):\Si_{\phi,S_1}\map\tau_\phi(j)(\Si_{\phi,S})$ in
\eq{PBPhiA2}.

The obvious choice is to restrict
$\l\in\Sig^{\Hi_{1}\otimes\Hi_2}_W$ to the sub-algebra
$V_{W}\otimes\mathC \hat{1}_{\Hi_2}\subseteq W$, and to identify
$V_{W}\otimes\mathC \hat{1}_{\Hi_1}\simeq
V_{W}\otimes\hat{1}_{\Hi_1}\simeq V_W$ as von Neumann algebras,
which gives
$\Sig^{\Hi_1\otimes\Hi_2}_{V_{W}\otimes\mathC\hat{1}_{\Hi_2}}
\simeq\Sig^{\Hi_1}_{V_{W}}$. Let
\begin{eqnarray}
\phi(p)_W:\Sig^{\Hi_1\otimes\Hi_2}_W &\map&
                            \Sig^{\Hi_1}_{V_{W}}\nonumber\\
                                \l &  \mapsto&\l|_{V_{W}}
\end{eqnarray}
denote this arrow at stage\ $W$. Then
\begin{equation}
\beta_\phi(p)\circ\nu^*(\dasB{A_1})\circ\phi(p):
\Sig^{\Hi_1\otimes\Hi_2} \map\SR{}^{\;\Hi_1\otimes\Hi_2}
\end{equation}
is a natural transformation which is defined for all
$W\in\Ob{\mathcal{V(H}_{1}\otimes\Hi_2)}$ and all $\l\in W$ by
\begin{eqnarray}
\left(\beta_\phi(p)\circ\nu^*(\dasB{A_1})\circ\phi(p)\right)_W(\l)
     &=& \nu^*(\dasB{A})(\l|_{V_{W}})\\
     &=& \{V^\prime\mapsto\brak{\l|_{V^{\prime}}}{\dasto{V^{\prime}}{A}}\mid
V^{\prime}\subseteq V_{W}\}\hspace{1.1cm}
\end{eqnarray}
This is clearly an order-reversing real-valued function on the set
$\downarrow\!\! W$ of sub-algebras of $W$, \ie\ it is an element
of $\SR{}^{\;\Hi_1\otimes\Hi_2}_W$. We define
$\beta_\phi(p)\circ\nu^*(\dasB{A_1})\circ\phi(p)$ to be the
translation representation, $\phi(\L{p})(\dasB{A_1})$ of
$\dasB{A_1}$ for the composite system.

Note that, by construction, for each $W$, the arrow
$(\beta_\phi(p)\circ\nu^*(\dasB{A_1})\circ\phi(p))_W$ corresponds
to the self-adjoint operator
$\delta(\hat{A_1})_{V_{W}}\otimes\hat{1}_{\Hi_2}\in W_{\sa}$,
since
\begin{equation}
\brak{\l|_{V_{W}}}{\delta(\hat{A_1})_{V_{W}}}=
\brak\l{\delta(\hat{A_1})_{V_{W} }\otimes\hat{1}_{\Hi_2}}
\end{equation} for all $\l\in\Sig^{\Hi_1\otimes\Hi_2}_W$.

\paragraph{Comments on these results.}
This is about as far as we can get with the arrows associated with
the composite of two quantum systems.  The results above can be
summarised in the equation
\begin{equation}                    \label{Eq_QCompSysTranslation}
        \phi(\L{p})(\dasB{A_1})_W=
\breve\delta(A_1)_{V_{W}}\otimes\hat{1}_{\Hi_2}
\end{equation}
for all contexts $W\in\Ob{\mathcal{V(H}_{1}\otimes\Hi_2)}$. If
$W\in\Ob{\mathcal{V}(\Hi_1\otimes\Hi_2)}$ is of the form
$W=V_1\otimes\mathC\hat 1_{\Hi_2}$, \ie\ if $W$ is in the image of
the embedding of $\V{1}$ into $\mathcal{V}(\Hi_1\otimes\Hi_2)$,
then $V_W=V_1$ and the translation formula gives just what one
expects: the arrow $\dasB{A_1}$ is translated into the arrow
$\dasB{A_1\di 1}$ at these stages, since
$\delta(\hat{A_1}\otimes\hat{1}_{\Hi_2})_{V_1\otimes\mathC\hat
1_{\Hi_2}}=
\delta(\hat{A_1})_{V_1}\otimes\hat{1}_{\Hi_2}$.\footnote{To be
precise, both the translation $\phi(\L{p})(\dasB{A_1})_W$, given
by (\ref{Eq_QCompSysTranslation}), and $\dasB{A\di 1}_W$ are
mappings from $\Sig^{\Hi_1\otimes\Hi_2}_W$ to
$\SR^{\Hi_1\otimes\Hi_2}_W$. Each
$\l\in\Sig^{\Hi_1\otimes\Hi_2}_W$ is mapped to an order-reversing
function on $\downarrow\!\! W$. The mappings
$\phi(\L{p})(\dasB{A_1})_W$ and $\dasB{A\di 1}_W$ coincide at all
$W^{\prime}\in\downarrow\!\! W$ that are of the form
$W^{\prime}=V^{\prime}\otimes\mathC\hat 1_{\Hi_2}$.}

If $W\in\Ob{\mathcal{V}(\Hi_1\otimes\Hi_2)}$ is not of the form
$W=V_1\otimes\mathC\hat 1_{\Hi_2}$, then it is relatively easy to
show that
\begin{equation}
      \delta(\hat{A_{1}}\otimes\hat{1}_{\Hi_2})_{W}\neq
      \delta(\hat{A_1})_{V_{W}}\otimes \hat{1}_{\Hi_2}
\end{equation}
in general. Hence
\begin{equation}
    \phi(\L{p})(\dasB{A_1}) \neq  \breve\delta(A_1\di 1), \label{phiLpdAn=}
\end{equation}
whereas, intuitively, one might have expected equality. Thus the
`commutativity' condition \eq{ClassCom} is not satisfied.

In fact, there appears to be \emph{no} operator $\hat
B\in\mathcal{B(H}_{1}\otimes\Hi_2)$ such that
$\phi(\L{p})(\dasB{A_1})=\dasB{B}$. Thus the quantity,
$\beta_\phi(p)\circ\nu^*(\dasB{A_1})\circ\phi(p)$, that is our
conjectured pull-back, is  an arrow in
$\Hom{\Set^{\mathcal{V(H}_{1}}\otimes\Hi_2)^{\op}}
{\Sig^{\Hi_1\otimes\Hi_2}}{\SR{}^{\;\Hi_1\otimes\Hi_2}}$ that is
not of the form $A_{\phi,S_1\di S_2}$ for any physical quantity
$A\in\F{S_1\di S_2}$.

Our current understanding is that this translation is `as good as
possible': the arrow
$\dasB{A_1}:\Sig^{\Hi_1}\rightarrow\SR^{\Hi_1}$ is translated into
an arrow from $\Sig^{\Hi_1\otimes\Hi_2}$ to
$\SR^{\Hi_1\otimes\Hi_2}$ that coincides with $\dasB{A_1}$ on
those part of $\Sig^{\Hi_1\otimes\Hi_2}$ that can be identified
with $\Sig^{\Hi_1}$. But $\Sig^{\Hi_1\otimes\Hi_2}$ is much
larger, and it is not simply a product of $\Sig^{\Hi_1}$ and
$\Sig^{\Hi_2}$. The context category
$\mathcal{V}(\Hi_1\otimes\Hi_2)$ underlying
$\Sig^{\Hi_1\otimes\Hi_2}$ is much richer than a simple product of
$\V{1}$ and $\V{2}$. This is due to a kind of operator
entanglement. A translation can at best give a faithful picture of
an arrow, but it cannot possibly `know' about the more complicated
contextual structure of the larger category.

Clearly, both technical and interpretational  work remain to be
done.


\section{Characteristic Properties of $\Si_\phi$, $\R_\phi$ and
$\TO/\w$} \label{Sec:CharPropsObjects}
\subsection{The State Object $\Si_\phi$}
A major motivation for our work is the desire to find mathematical
structures with whose aid genuinely new types of theory can be
constructed. Consequently, however fascinating the
`toposification' of quantum theory may be, this particular theory
should not be allowed to divert us too much from the main goal.
However, it is also important to see if any general lessons can be
learnt from what  has been done so far. This is likely to be
crucial in the construction of new theories.

In developing the topos version of quantum theory we have
constructed concrete objects in the topos to function as the state
object and quantity-value object. We have also seen how each
quantum vector state gives a precise truth object, or
`pseudo-state'.

The challenging question now is what, if anything, can be said in
general about these key ingredients in our scheme. Thus, ideally,
we would be able to specify characteristic properties for
$\Si_\phi$, $\R_\phi$, and the truth objects/pseudo-states. A
related problem is to understand if there is an object,
$\WO_\phi$, of all truth objects/pseudo-states, and, if so, what
are its defining properties. Any such universal properties could
be coded into the structure of the language, $\L{S}$, of the
system, hence ensuing that they are present in all topos
representations of $S$.
 In particular, should a symbol $\WO$ be added to $\L{S}$ as the linguistic precursor of an object of pseudo-states?

So far, we only know of two explicit examples of
physically-relevant topos representations of a system language,
$\L{S}$: (i) the representation of classical physics in $\Set$;
and (ii) the representation of quantum physics in topoi of the
form $\SetH{}$. This does provide  much to go on when it comes to
speculating on characteristic properties of the key objects
$\Si_\phi$ and $\R_\phi$. From this perspective, it would be
helpful  if there is an alternative  way of finding the quantum
objects $\Sig$ and $\PR{\mathR}$ in addition to the one provided
by the  approach that we have adopted. Fortunately, this has been
done recently by Heunen and Spitters \cite{HeuSpit07}; as we shall
see in Section \ref{SubSubSec:HeunSpit}, this does provide more
insight into a possible generic structure for $\Si_\phi$.

\subsubsection{An Analogue of a Symplectic Structure or Cotangent Bundle?}
Let us start with the state object $\Si_\phi$.  In classical
physics, this is a symplectic manifold; in quantum theory it is
the spectral presheaf $\Sig$ in the topos $\SetH{}$. Does this
suggest any  properties for $\Si_\phi$ in general?

One possibility is that the state-object, $\Si_\phi$, has some
sort of `symplectic structure'. If taken literally, this phrase
suggests synthetic differential geometry (SDG): a theory that is
based on the existence in certain topoi (not $\Set$) of genuine
`infinitesimals'. However, this seems unlikely for the quantum
topoi $\SetH{}$ and  we would probably need to extend these topoi
considerably in order to incorporate SDG. Thus when we say
``\ldots some sort of symplectic structure'', the phrase `some
sort' has to be construed rather broadly.

We suspect that, with this caveat, the state object $\Sig$ may
have such a structure, particularly for those quantum systems that
come from quantising a given classical system. However, at the
moment this is still a conjecture. We are currently studying
systems whose classical state space is  the cotangent bundle,
$T^*Q$, of a configuration space $Q$. We think that the quantum
analogue of this space is a certain presheaf, $\ps{M_Q}$, that is
associated with the  \emph{maximal} commutative sub-algebra,
$M_Q\in\Ob{\V{}}$, generated by the smooth, real-valued functions
on $Q$. This is currently work in progress.

But even if  the quantum state-object \emph{does} have a remnant
`symplectic structure', it is debatable if this should be
axiomatised in general. Symplectic structures arise in classical
physics because the underlying equations of motion are
\emph{second-order} in the `configuration' variables $q$, and
hence first-order in the pair $(q,p)$, where $p$ are the `momentum
variables'.

However if, say, Newton's equations of gravity had been
third-order in $q$, this would lead to triples $(q,p,a)$ ($a$ are
`acceleration' variables) and symplectic structure would not be
appropriate.

\subsubsection{$\Si_\phi$ as a Spectral Object: the Work of
Heunen and Spitters} \label{SubSubSec:HeunSpit}
 Another way of understanding the state
object $\Si_\phi$ is suggested by the recent work of Heunen and
Spitters \cite{HeuSpit07}. They start with a non-commutative
$C^*$-algebra, $\cal A$, of observables in some `ambient topos',
$\mathfrak{S}$---in our case, this is $\Set$---and then proceed
with the following steps:
\begin{enumerate}
\item Construct the poset category\footnote{This notation
has been chosen to suggest more clearly the analogues with our
topos constructions that use the base category $\V{}$. It is not
that used by Heunen and Spitters.}  ${\cal V}(\mathcal{A})$ of
commutative sub-algebras of $\mathcal{A}$.

\item Construct the topos, $\mathfrak{S}^{{\cal V}(\mathcal{A})}$
of covariant functors (\ie\ \emph{co}-presheaves) on the
category/poset ${\cal V}(\mathcal{A})$. \footnote{They affirm that
the operation $\mathcal{A}\mapsto\mathfrak{S}^{{\cal
V}(\mathcal{A})}$
 defines a functor from the category of $C^*$-algebras in $\mathfrak{S}$
to the category of elementary topoi and geometric morphisms.}

\item Construct the `tautological' co-presheaf
$\cps{\cal A}$ in which $\cps{\cal A}(V):=V$ for each commutative
sub-algebra, $V$, of $\cal A$. Then if $i_{V_1V_2}:V_1\subseteq
V_2$, the associated arrow $\cps{\cal A}(i_{V_1V_2}):\cps{\cal
A}(V_1)\map\cps{\cal A}(V_2)$ is just the inclusion map of
 $\cps{\cal A}(V_1)$ in $\cps{\cal A}(V_2)$.

\item They show that $\cps{\cal A}$ has the structure of a
commutative (pre\footnote{The `pre' refers to the fact that the
algebra is defined using only the co-presheaf, $\cps{\mathbb{Q}}$,
of complexified rationals. This co-presheaf is constructed from
the natural-number object, $\cps\mathN$, only. This restriction is
necessary because it is not possible to define the \emph{norm} of
a $C^*$-algebra in purely constructive terms.}) $C^*$-algebra
inside the topos

\item Using a recent, very important, result of
Banacheswski and Mulvey \cite{BanMul06}, Heunen and Spitters show
that the spectrum, $\cps{\Si}$, of the commutative algebra
$\cps{\cal A}$ can be computed internally, and that it has the
structure of an internal locale in $\mathfrak{S}^{{\cal
V}(\mathcal{A})}$.

\item They then show that, in the case of quantum theory, $\cps{\Si}$  is essentially our
spectral object, $\Sig$, but viewed as a co-presheaf of
\emph{locales}, rather than as a presheaf of topological spaces.

\end{enumerate}

Thus Heunen and Spitters differ from us in that (i) they work in a
general ambient topos $\mathfrak{S}$, whereas we use $\Set$; (ii)
they use $C^*$-algebras rather than von Neumann
algebras\footnote{One problem with $C^*$-algebras is that they
very often do not contain enough projectors; and, of course, these
are the entities that represent propositions. This obliges Heunen
and Spitters to move from a $C^*$-algebra to a $AW^*$-algebra,
which is just an abstract version of a von Neumann algebra.}; and
(iii) they use covariant rather than contravariant functors.

The fact that they recover what is (essentially) our spectral
presheaf is striking. Amongst other things, it suggests a possible
axiomatisation of the state object, $\Si_\phi$. Namely, we could
require that in \emph{any} topos representation, $\phi$, the state
object is (i) the spectrum of some internal, commutative (pre)
$C^*$-algebra; and (ii) the spectrum has the structure of an
internal locale in the topos $\tau_\phi$.

It is not currently clear whether or not it makes physical sense
to always require $\Si_\phi$ to be the spectrum of an internal
algebra. However, even in the contrary case it still makes sense
to explore the possibility that $\Si_\phi$ has the `topological'
property of being an internal locale. This opens up many
possibilities, including that of constructing the (internal)
topos, $\Sh{\Si_\phi}$, of sheaves over $\Si_\phi$.

\subsubsection{Using Boolean Algebras as the Base Category}
As remarked earlier, there are several possible choices for the
base category over which the set-valued functors are defined. Most
of our work has been based on the category, $\V{}$, of commutative
von Neumann sub-algebras of $\BH$. As indicated above, the
Heunen-Spitters constructions use the category of commutative
$C^*$-algebras. More abstractly, one can start with any
$AW^*$-algebra or $C^*$-algebra.

However, as discussed briefly in Section
\ref{SubSubSec:BoolAlgBase}, another possible choice is  the
category, $\BlH$,  of all Boolean sub-algebras of the lattice of
projection operators on $\Hi$.  The ensuing topos,
$\Set^{\BlH^\op}$, or $\Set^{\BlH}$, is interesting in its own
right, but particularly so when combined with the ideas of Heunen
and Spitters. As applied to the category $\BlH$, their work
suggests that we first construct the tautological co-presheaf
$\cps{\BlH}$ which associates to each $B\in\Ob{\BlH}$, the Boolean
algebra $B$. Viewed internally in the topos $\Set^{\BlH}$, this
co-presheaf is a Boolean-algebra object. We conjecture that the
spectrum of $\cps{\BlH}$ can be obtained in a constructive way
using the internal logic of $\Set^{\BlH}$. If so, it seems clear
that, after using the locale trick  of \cite{HeuSpit07}, this
spectrum will essentially be the same as our dual presheaf
$\ps{D}$.

Thus, in this approach, the state object is the spectrum of an
internal \emph{Boolean-algebra}, and daseinisation maps the
projection operators in $\Hi$ into elements of this algebra. This
reinforces still further our claim that quantum theory looks like
classical physics in an appropriate topos. This raises some
fascinating possibilities. For example, we make the following:
\displayE{10}{5}{4}{{\bf Conjecture:} The subject of quantum
computation is equivalent to the study of `classical' computation
in the quantum topos $\SetH{}$.}

\subsubsection{Application to Other Branches of Algebra}
\label{SubSubSec:AppAlgebra} It is clear that the scheme discussed
above could fruitfully be extended to various branches of algebra.
Thus, if $\mathfrak{A}$ is any algebraic structure\footnote{We are
assuming that the ambient topos is $\Set$, but other choices could
be considered.},  we can consider the category
$\mathcal{V}(\mathfrak{A})$ whose objects are the commutative
sub-algebras of $\mathfrak{A}$, and whose arrows are algebra
embeddings (or, slightly more generally, monomorphisms). One can
then consider the topos, $\SetC{\mathcal{V}(\mathfrak{A})}$, of
all set-valued, contravariant functors on
$\mathcal{V}(\mathfrak{A})$; alternatively, one might look at the
topos, $\Set^{\mathcal{V}(\mathfrak{A})}$, of covariant functors.

For this structure to be mathematically interesting it is
necessary that  the abelian sub-objects of $\mathfrak{A}$ have a
well-defined spectral structure. For example, let $\mathfrak{A}$
be any locally-compact topological group.  Then the
\emph{spectrum} of any commutative (locally-compact) subgroup $A$
is just the Pontryagin dual of $A$, which is itself a
locally-compact, commutative group. The \emph{spectral presheaf}
of $\mathfrak{A}$ can then be defined as the object,
$\Si_{\mathfrak{A}}$, in $\SetC{\mathcal{V}(\mathfrak{A})}$ that
is constructed in the obvious way (\ie\ analogous to the way in
which $\Sig$ was constructed) from this collection of Pontryagin
duals.

We conjecture that a careful analysis would show that, for at
least some structures of this type:
\begin{enumerate}
\item There is a `tautological' object, $\cps{\mathfrak{A}}$,
in the topos $\Set^{\mathcal{V}(\mathfrak{A})}$ that is associated
with the category  $\mathcal{V}(\mathfrak{A})$.
\item Viewed internally, this tautological  object is a commutative algebra.
\item This object has a spectrum that can be constructed internally, and is essentially the spectral presheaf, $\Si_{\mathfrak{A}}$, of $\mathfrak{A}$.
\end{enumerate}

It seems clear that, in general, the spectral presheaf,
$\Si_{\mathfrak{A}}$, is a potential candidate as the basis for
 non-commutative spectral theory.

\subsubsection{The Partial Existence of Points of $\Si_\phi$}
One of the many intriguing features of topos theory is that it
makes sense to talk about entities that only `partially exist'.
One can only speculate on what would have been Heidegger's
reaction had he been told that the answer to ``What is a thing?''
is ``Something that partially exists''. However, in the realm of
topos theory the notion of `partial existence' lies easily with
the concept of propositions that are only `partly true'.

A particularly interesting example is the existence, or otherwise,
of `points' (\ie\ global elements) of the state object $\Si_\phi$.
If $\Si_\phi$ has no global elements (as is the case for the
quantum  spectral presheaf, $\Sig$) it may still have `partial
elements'. A partial element is defined to be an arrow
$\xi:U\map\Si_\phi$ where the object $U$ in the topos $\tau_\phi$
is a sub-object of the terminal object $1_{\tau_\phi}$. Thus there
is a monic $U\hookrightarrow 1_{\tau_\phi}$ with the property that
the arrow $\xi:U\map\Si_\phi$ cannot be extended to an arrow
$1_{\tau_\phi}\map\Si_\phi$. Studying the obstruction to such
extensions could be another route to finding a cohomological
expression of the Kochen-Specker theorem.

Pedagogically, it is attractive to say that the non-existence of a
global element of $\Sig$ is analogous to the non-existence of a
cross-section of the familiar `double-circle', helical covering of
a single circle, $S^1$. This principal $\mathZ_2$-bundle over
$S^1$ is non-trivial, and hence has no cross-sections.

However,  \emph{local} cross-sections do exist, these being
defined as sections of the bundle restricted to any open subset of
the base space $S^1$. In fact, this bundle is \emph{locally
trivial}; \ie\ each point $s\in S^1$ has a neighbourhood $U_s$
such that the restriction of the bundle to $U_s$ is trivial, and
hence sections of the bundle restricted to $U_s$ exist.

There is an analogue of local triviality in the topos quantum
theory where $\tau_\phi=\SetH{}$. Thus, let $V$ be any object in
$\V{}$ and define $\downarrow\!\!V:=\{V_1\in\Ob{\V{}}\mid
V_1\subseteq V\}$. Then $\downarrow\!\!V$ is like a
`neighbourhood' of $V$; indeed, that is precisely what it is if
the poset $\Ob{\V{}}$ is equipped with the topology generated by
the lower sets. Furthermore, given any presheaf $\ps{F}$ in
$\tau_\phi$, the restriction, $\ps{F}\!\!\downarrow\!\!V$, to $V$,
can be defined as in Section \ref{SubSec:PresheafP-clSig}. It is
easy to see that, for all stages $V$, the presheaf
$\ps{F}\!\downarrow\!V$ \emph{does} have global elements. In this
sense, every presheaf in $\SetH{}$ is `locally trivial'.
Furthermore, to each $V$ there is associated a sub-object
$\ps{U}^V$ of $\ps{1}$ such that each global element of
$\ps{F}\!\downarrow\!V$ corresponds to a partial element of
$\ps{F}$.

Thus, for the topos $\SetH{}$, there is a precise sense in which
the spectral presheaf has `local elements', or `points that
partially exist'. However, it is not clear to what extent such an
assertion can, or should, be made for a general topos $\tau_\phi$.
Certainly, for any presheaf topos, $\SetC{C}$, one can talk about
`localising' with respect to the objects in the base category $C$,
but the situation for a more general topos is less clear.

\subsubsection{The Work of Corbett \emph{et al}}
Another interesting question is whether these different ways of
seeing the state-object relate at all to the work of Corbett and
his collaborators \cite{AdelCor05,Corbett07}.

For some time Corbett has been studying  what he calls `quantum'
real numbers, or `qr-numbers', as another way of a obtaining a
`realist' interpretation of quantum theory. The first step is to
take the space of states, ${\cal E}_S$, of a quantum system (where
a state is viewed as a positive linear functional on an
appropriate $C^*$-algebra, ${\cal A}$) and  equip it with the
weakest topology such that the functions
$\A\mapsto\tr{\A\hat\rho}$ are continuous for all states
$\hat\rho\in{\cal E}_S$. Then a `qr-number' is  defined as a
global element of the sheaf of germs of continuous real-valued
functions on ${\cal E}_S$. Put another way, a qr-number is a
(Dedekind) real number in the topos, $\Sh{{\cal E}_S}$, of sheaves
over ${\cal E}_S$. The fundamental physical postulate is then:
\begin{enumerate}
\item The `numerical values' of a physical
quantity, $A$, are given by the qr-numbers $a_Q(U) :=
\tr{\A\rho}_{\hat\rho\in U}$ where $U$ is an open subset of ${\cal
E}_S$.

\item Every physical quantity has a qr-number value at all times.
\item Every physical quantity has an open subset of ${\cal E}_S$
associated with it at all times. This is the extent to which the
quantity can be said to `exist'.
\end{enumerate}
Evidently this theory also is  `contextual', with the contexts now
being identified with the open sets of ${\cal E}_S$.

There seems a possible link between these ideas and the work of
Heunen and Spitter. The latter construct the (internal in the
topos $\Set^{\V{}}$) topos, $\Sh{\cps{\Si}}$, of sheaves over the
locale $\cps{\Si}$ and then show that each bounded self-adjoint
operator in $\Hi$ is represented by an analogue of an `interval
domain\footnote{This is rather like a Dedekind real number except
that the overlap axiom is missing.}' in this topos. This is their
analogue of our non-commutative spectral theory in which $\A$ is
represented by an arrow $\daso{A}:\Sig\map\SR$ (or an arrow from
$\Sig$ to $\PR{\mathR}$, if one prefers that choice of
quantity-value object).

It would be interesting to see if there is any relation between
the real numbers in the Corbett presheaf, $\Sh{{\cal E}_S}$ and
the interval domains in the Heunen-Spitters presheaf
$\Sh{{\cps\Si}}$.  Roughly speaking, we can say that Corbett
\emph{et al} assign exact values to physical quantities by making
the state `fuzzy', whereas we (and Heunen \&\ Spitter) keep the
state sharp, but ascribe `fuzzy' values to physical quantities.
Clearly, there are some interesting questions here for further
research.

\subsection{The Quantity-Value Object $\R_\phi$}
Let us turn now to  the quantity-value object $\R_\phi$. This
plays a key role in the representation of any physical quantity,
$A$, by an arrow $A_\phi: \Si_\phi\map\R_\phi$. In so far as a
`thing' is a bundle of properties, these properties refer to
values of physical quantities, and so the nature of these `values'
is of central importance.

We anticipate that $\R_\phi$ has many global elements
$1_{\tau_\phi}\map\R_\phi$, and these can be interpreted as the
possible `values' for physical quantities. If $\Si_\phi$ also has
global elements/microstates $s:1_{\tau_\phi}\map\Si_\phi$, then
these combine with any arrow $A_\phi:\Si_\phi\map\R_\phi$ to give
global elements of $\R_\phi$. It seems reasonable to refer to the
element, $A_\phi\circ s:1_{\tau_\phi}\map\R_\phi$ as the `value'
of $A$ when the microstate is $s$. However, our expectation is
that, in general, $\Si_\phi$ may well have no global elements, in
which case the interpretation of  $A_\phi:\Si_\phi\map\R_\phi$ in
terms of values is somewhat subtler.  This has to be done
internally using the language $\L{\tau_\phi}$ associated with the
topos $\tau_\phi$: the overall logical structure is a nice example
of a `coherence' theory of truth \cite{Gray90}.

As far as axiomatic properties of  $\R_\phi$ are concerned, the
minimal requirement is presumably that it should have some
ordering property that arises in all topos representations of the
system $S$. This universal property  could be coded into the
internal language, $\L{S}$, of $S$. This implements our intuitive
feeling that, in so far as the concept of `value' has any meaning,
it must be possible to say that the value of one quantity is
`larger' (or `smaller') than that of another. It seems reasonable
to expect this relation to be transitive, but that is about all.
In particular, we see no reason to suppose that this relation will
always correspond to a \emph{total} ordering: perhaps there are
pairs of physical quantities whose `values' simply cannot be
compared at all. Thus, tentatively, we can augment $\L{S}$ with
the axioms for a poset structure on $\R$.

Beyond this simple ordering property, it becomes less clear what
to assume about the quantity-value object. The example of quantum
theory shows that it is wrong to automatically equate $\R_\phi$
with the real-number object $\mathR_\phi$ in the topos
$\tau_\phi$. Indeed, we believe that is almost always the case.

However, this makes it harder to know what to assume of $\R_\phi$.
The quantum case shows that $\R_\phi$ may have considerably fewer
algebraic properties than the real-number object $\mathR_\phi$. On
a more `topological' front it is attractive to assume that
$\R_\phi$ is an internal locale in the topos $\tau_\phi$. However,
one should be cautious when conjecturing about $\R_\phi$ since our
discussion of various possible quantity-value objects in quantum
theory depended closely on the specific details of the spectral
structure in this topos.

A more general perspective is given by the work of Heunen and
Spitters \cite{HeuSpit07}. We recall that their starting point is
a (non-commutative) $C^*$-algebra, $\mathcal{A}$ in an ambient
topos $\mathfrak{S}$. Then they construct  the topos of
co-presheaves, $\mathfrak{S}^{\mathcal{V}(\mathcal{A})}$, and show
that $\cps{\mathcal{A}}$ is an internal, pre $C^*$-algebra in this
topos. Finally, they construct the spectrum, $\cps{\Si}$, of
$\cps{\mathcal{A}}$ and show that it is an internal locale.

Having shown that $\cps\Si$ has the structure of a locale, it is
rather natural to consider the (internal) topos of sheaves,
$\Sh{\cps\Si}$, over $\cps\Si$, and then construct the
`interval-domain' object in this topos\footnote{If $X$ is any
topological space it is well-known that the real numbers in the
topos $\Sh{X}$ are in one-to-one correspondence with elements of
the space, $C(X,\mathR)$, of continuous, real-valued functions on
$X$}. In the case of quantum theory, they show that this is
related to what we have called $\ps\mathR^\leftrightarrow$.

This approach might be a useful tool when looking for ways of
axiomatising $\R_\phi$. Thus, if in any topos representation
$\phi$, we assume that the state object $\Si_\phi$ is an internal
local in $\tau_\phi$, we can construct the internal topos
$\Sh{\Si_\phi}$ and consider its interval-domain number object. It
remains to be seen if this has any generic use in practice.

\subsection{The Truth Objects $\TO$, or Pseudo-State Object $\WO_\phi$}
The truth objects in a topos representation are certain
sub-objects of $P\Si_\phi.$\ Their construction will be very
theory-dependent, as are the pseudo-states, and  the pseudo-state
object, $\WO_\phi$, if there is one.  Each proposition about the
physical system is represented by a sub-object
$J\subseteq\Si_\phi$, and given a truth-object $\TO\subseteq
P\Si_\phi$, the generalised truth value of the proposition is
$\Val{J\in\TO}\in\Ga\O_\phi$; in terms of pseudo-states, $\w$, the
generalised truth values are of the form $\Val{\w\subseteq
J}\in\Ga\O_\phi$.

The key properties of the quantum truth objects,
$\ps\TO^{\ket\psi}$, (or pseudo-states $\ps\w^{\ket\psi}$) can
easily be emulated if one is dealing with a more general base
category, ${\cal C}$, so that the topos concerned is $\SetC{C}$.
However, it is not clear what, if anything, can be said about the
structure of truth objects/pseudo-states in a more generic topos
representation.

An attractive possibility is that there is a general analogue of
\eq{ComDiagC2} in the form
\begin{equation}
        \setsqparms[1`0`1`1;700`700] \label{ComDiagphi}
        \square[\w` \Si_\phi` 1_{\tau_\phi}` P\Si_\phi;
        {} ` {} ` \pi_\phi` \name{\w}]
\end{equation}
and that obstructions to the existence of global elements of the
state-object $\Si_\phi$ can be studied with the aid of this
diagram. If there is a pseudo-state object $\WO_\phi$ then this
could be a natural replacement for $P\Si_\phi$ in this diagram.

\section{Conclusion}
In this long article we have developed the idea that, for any
given theory-type (classical physics, quantum physics,
DI-physics,\ldots) the theory of a particular physical system,
$S$, is to be constructed in the framework of a certain,
system-dependent, topos. The central idea is that a local
language, $\L{S}$, is attached to each system $S$, and that the
application of a given theory-type to $S$ is equivalent to finding
a representation, $\phi$, of $\L{S}$ in a topos $\tau_\phi(S)$;
this is equivalent to finding a translation of $\L{S}$ into the
internal language associated with $\tau_\phi(S)$; or a functor to
$\tau_\phi(S)$ from the topos associated with $\L{S}$.

Physical quantities are represented by arrows in the topos from
the state object $\Si_{\phi,S}$ to the quantity-value object
$\R_{\phi,S}$, and propositions are represented by sub-objects of
the state object. The idea of a `truth sub-object' of
$P\Si_{\phi,S}$ (or a `pseudo-state' sub-object of $\Si_{\phi,S}$)
then leads to a neo-realist interpretation of propositions in
which each proposition is assigned a truth value that is a global
element of the sub-object classifier $\Omega_{\tau_\phi(S)}$.  In
general, neo-realist statements about the world/system $S$ are to
be expressed in the internal language of the topos $\tau_\phi(S)$.
Underlying this is the intuitionistic, deductive logic provided by
the local language $\L{S}$.

These axioms are based on ideas from the topos representation of
quantum theory, which we have discussed in depth. Here, the topos
involved is $\SetH{}$: the topos of presheaves over the base
category $\V{} $ of commutative, von Neumann sub-algebras of the
algebra, $\BH$, of all bounded operators on the quantum Hilbert
space, $\Hi$. Each such sub-algebra can be viewed as a context in
which the theory can be viewed from a classical perspective. Thus
a context can be described  as a `classical snap-shot', or `window
on reality', or `world-view'/weltanschauung. Mathematically, a
context is a `stage of truth':  a concept that goes back to
Kripke's use of a presheaf topos as a model of his intuitionistic
view of time and process.

We have shown how the process of `daseinisation' maps projection
operators (and hence equivalent classes of propositions) into
sub-objects of the state-object $\Sig$. We have also shown how
this can be extended to an arbitrary, bounded self-adjoint
operator, $\A$. This produces an arrow $\breve{A}:\Sig\map\ps\R$
where the minimal choice for the quantity-value object, $\ps\R$,
is the object $\SR$. We have also argued that, from a physical
perspective, it is more attractive to choose $\PR{\mathR}$ as the
quantity-value presheaf. The significance of these results is
enhanced considerably by the alternative, Heunen \& Spitters
derivation of $\Sig$ as the spectrum of an internal, commutative
algebra. These, and related, results all encourage the idea that
quantum theory can be viewed as classical theory but in a topos
other than the topos of sets, $\Set$.

Every classical system uses the same topos, $\Set$. However, in
general, the topos will be system dependent as, for example, is
the case with the quantum topoi of the form $\SetH{}$, where
$\Hi$\ is the Hilbert space of the system. This leads to the
problem of understanding how the topoi for a class of systems
behave under the action of taking a sub-system, or combining a
pair of systems to give a single composite system. We have
presented a set of axioms that capture the general ideas we are
trying to develop. Of course, these axioms are not cast in stone,
and are still partly `experimental' in nature. However, we have
shown that classical physics exactly fits our suggested scheme,
and that quantum physics `almost' does: `almost' because of the
issues concerning the translation representation of the arrows
associated with compositions of systems that were discussed in
Section \ref{SubSec:TranslCompQT}.

An important challenge for future work is to show that our general
topos scheme can be used to develop genuinely new theories of
physics, not just to rewrite old ones in a new language. Of
particular  interest  is the problem with which we motivated the
scheme in the first place: namely, to find tools for constructing
theories that go beyond quantum theory and  which do not use
Hilbert spaces, path integrals, or any of the other familiar
entities in which the continuum real and/or complex numbers play a
fundamental role.

As we have discussed, the topoi for quantum systems are of the
form $\SetH{}$, and hence  embody contextual logic in a
fundamental way. One way of going `beyond' quantum theory, while
escaping the \emph{a priori} imposition of continuum concepts, is
to use presheaves over a more general `category of contexts',
$\mathcal C$, \ie \ develop the theory in the topos
$\Set^{{\mathcal C}^{{\rm op}}}$. Such a structure embodies
contextual, multi-valued logic in an intrinsic way, and in that
sense might be said to encapsulate one of the fundamental insights
of quantum theory. However, and unlike in quantum theory, there is
no obligation to use the real or complex numbers in the
construction of the category $\mathcal C$.

Indeed, early on in this work we noted  that  real numbers  arise
in theories of physics in three different (but related) ways: (i)
as the values of physical quantities; (ii) as the values of
probabilities; and (iii) as a fundamental ingredient in models of
space and time. The first of these is now subsumed into the
quantity-value object $\R_\phi$, and which now has no \emph{a
priori} relation to the real number object in $\tau_\phi$.  The
second source of real numbers has gone completely since we no
longer have probabilities of propositions but rather generalised
truth values whose values lie in $\Ga\O_{\tau_\phi}$. The third
source is also no long binding since models of space and time in a
topos could depend on many things: for example, infinitesimals.

Of course, although true, these remarks do not of themselves give
a concrete example of a theory that is `beyond quantum theory'. On
the other hand, these ideas certainly  point in a novel direction,
and one at which, almost certainly, we would not have arrived  if
the challenge to `go beyond quantum theory' had been construed
only in terms of trying to generalise Hilbert spaces, path
integrals, and the like.

From a more general perspective, other types of topoi are possible
realms for the construction of physical theories. One
 simple, but mathematically rich example arises from
the theory of $M$-sets. Here, $M$ is a monoid and, like all
monoids, can be viewed as a category with a single object, and
whose arrows are  the elements of $M$.  Thought of as a category,
a  monoid is `complementary' to  a partially-ordered set. In a
monoid, there is only one object, but plenty of arrows from that
object to itself; whereas in a partially-ordered set there are
plenty of objects, but at most one arrow between any pair of
objects.  Thus a partially-ordered set is  the most economical
category with which to capture the concept of `contextual logic'.
On the other hand, the logic associated with a monoid is
non-contextual as there is only one object in the category.

It is easy to see that a functor from $M$ to $\Set$ is just an
`$M$-set': \ie\  a set on which $M$  acts as a monoid of
transformations. An arrow between two such $M$-sets is  an
equivariant map between them. In physicists' language, one would
say that the topos $\Set^M$---usually denoted $BM$--- is the
category of the `non-linear realisations' of $M$.

The sub-object classifier, $\O_{BM}$,  in  $BM$ is the collection
of left ideals in $M$; hence, many of the important constructions
in the topos can be handled using the language of algebra.  The
topos $BM$ is one of the simplest to define and work with and, for
that reason, it is a popular source of examples in texts on topos
theory. It would be intriguing to experiment with constructing
model theories of physics using one of these simple topoi. One
possible use of $M$-sets is discussed in \cite{Isham05b} in the
context of reduction of the state vector, but there will surely be
others.

\paragraph{Is there `un gros topos'?} It is clear that there are
many other topics for future research. A question that is of
particular interest is if there is a \emph{single} topos within
which all systems of  a given theory-type can be discussed. For
example, in the case of quantum theory the relevant topoi are of
the form $\SetH{}$, where $\Hi$ is a Hilbert space, and  the
question is whether all such topoi can be gathered together to
form a single topos (what Grothendieck termed `un gros topos')
within which all quantum systems can be discussed.

There are well-known examples of such constructions in the
mathematical literature. For example, the category, ${\rm Sh}(X)$,
of sheaves on a topological space $X$ is a topos, and there are
collections $\bf T$ of topological spaces which form a
Grothendieck site, so that the topos ${\rm Sh}(\bf{T})$ can be
constructed. A particular object  in ${\rm Sh}(\bf{T})$ will then
be a sheaf over $\bf T$ whose stalk over any object $X$ in $\bf T$
will be the topos  ${\rm Sh}(X)$.

For our purposes, the ideal situation would be if the various
categories of systems, $\Sys$, can be chosen in such a way that
${\cal M}(\Sys)$ is a site. Then the topos of sheaves, ${\rm
Sh}({\cal M}(\Sys)),$ over this site would provide a common topos
in which all systems of this theory type---\ie\ the objects of
$\Sys$---can be discussed. We do not know if this is possible, and
it is a natural subject for future study.

\paragraph{Some more speculative lines of future research.} At a
conceptual level, one motivating desire for the entire research
programme was to find a formalism that would always give some sort
of `realist' interpretation, even in the case of quantum theory
which is normally presented in an instrumentalist way. But this
particular example raises an interesting point because the
neo-realist interpretation takes place in the topos $\SetH{}$,
whereas the instrumentalist interpretation works in the familiar
topos $\Set$ of sets, and one might wonder how universal is the
use of a pair of topoi in this way.

Another, related, issue concerns the representation of the
$\PL{S}$-propositions of the form $\SAin\De$  discussed in Section
\ref{Sec:QuPropSpec}. This serves as a bridge between the
`external' world of a background spatial structure, and the
internal world of the topos. This link is not present with the
$\L{S}$ language whose propositions are purely internal terms of
type $\O$ of the form `$A(\va{s})\in\va{\De}$'. In a topos
representation, $\phi$, of $\L{S}$, these become propositions of
the form `$A\in\Xi$', where $\Xi$\ is a sub-object of $\R_\phi$.

In general, if we have an example of our axioms working
neo-realistically in a topos $\tau$, one might wonder if there is
an `instrumentalist' interpretation of the same theory in a
different topos, $\tau_i$, say? Of course, the word
`instrumentalism' is used metaphorically here, and any serious
consideration of such a pair $(\tau,\tau_i)$ would require a lot
of very careful thought.

However, if a pair $(\tau,\tau_i)$ does exist, the question then
arises of whether  there is a \emph{categorial} way of linking the
neo-realist and instrumentalist interpretations: for example, via
a functor $I:\tau\map\tau_i$. If so, is this related to some
analogue of the daseinisation operation that produced the
representation of the $\PL{S}$-propositions, $\SAin\De$ in quantum
theory? Care is needed in discussing such issues since informal
set theory is  used as a meta-language in constructing a topos,
and one has to be careful not to confuse this with the existence,
or otherwise, of an `instrumentalist' interpretation of any given
representation.

If such a functor, $I:\tau\map\tau_i$, did exist then one could
speculate  on the possibility of finding an `interpolating chain'
of functors
\begin{equation}
        \tau\map\tau^1\map\tau^2\map
        \cdots\map\tau^n\map \tau_i
\end{equation}
which could be interpreted conceptually as corresponding to an
interpolation between the philosophical views of  realism and
instrumentalism!

Even more speculatively one might wonder if ``one person's realism
is another person's instrumentalism''. More precisely, given a
pair $(\tau,\tau_i)$ in the sense above, could there be cases in
which the topos $\tau$  carries a neo-realist interpretation of a
theory with respect to an instrumentalist interpretation  in
$\tau_i$, whilst being the carrier of an instrumentalist
interpretation with respect to the neo-realism of a `higher'
topos; and so on? For example, is there some theory whose
`instrumentalist manifestation` takes place in the topos
$\SetH{}$?

On the other hand,  one might want to say that `instrumentalist'
interpretations always take place in the world of classical set
theory, so that  $\tau_i$ should always be chosen to be $\Set$. In
any event, it would be interesting to study the quantum case more
closely to see if there are any categorial relations between the
formulation in $\SetH{}$ and the instrumentalism interpretation in
$\Set$. It can be anticipated that the action of daseinisation
will play an important role here.

All this is, perhaps\footnote{To be honest, the `perhaps' should
really be replaced by `highly'.}, rather speculative but there is
a more obvious situation in which a double-topos structure will be
necessary, irrespective of philosophical musings on
instrumentalism. This is if one wants to discuss the `classical
limit' of some topos theory. In this case this limit will exist in
the topos, $\Set$, and this must be used in addition to the topos
of the basic theory. A good example of this, of course, is the
topos of quantum theory discussed in this article. If, \emph{pace}
Landsmann, one thinks of `quantisation' as a functor from $\Set$
to $\SetH{}$, then the classical limit will perhaps involve a
functor going in the opposite direction.

\paragraph{Implications for quantum gravity.}
A serious claim stemming from our work is that a successful theory
of quantum gravity should be constructed in some topos $\cal
U$---the `topos of the universe'---that is \emph{not} the topos of
sets. All entities of physical interest will be represented in
this topos, including models for space-time (if there are any at a
fundamental level in quantum gravity) and, if relevant, loops,
membranes etc. as well as incorporating the anticipated
generalisation of quantum theory.

Such a theory of quantum gravity will have a neo-realist
interpretation in the topos $\cal U$, and hence would be
particularly useful in the context of quantum cosmology. However,
in practice, physicists  divide the world up into smaller, more
easily handled, chunks, and each of them would correspond to what
earlier we have called a `system' and, correspondingly, would have
its own topos. Thus $\cal U$ is something like the `gros topos' of
the theory, and would combine together the individual
`sub-systems' in a categorial way. Of course, it is most unlikely
that there is any preferred way of dividing the universe up into
bite-sized chunks, but this is  not problematic as the ensuing
relativism is naturally incorporated into the idea of a
Grothendieck site.


\section{Appendix 1: Some Theorems and Constructions Used in the Main Text}

\subsection{Results on Clopen Subobjects of \underline{$\Sigma$}.}
\begin{theorem}{The collection, $\Subcl\Sig$, of all clopen
sub-objects of $\Sig$ is a Heyting algebra.}\label{Th:SubclSig}
\end{theorem}
\begin{proof} First recall how a Heyting algebra structure is placed
on the set, $\Sub{\Sig}$, of all sub-objects of $\Sig$.
\paragraph{The `$\lor$'- and `$\land$'-operations.}
Let $\ps{S},\ps{T}$ be two sub-objects of $\Sig$. Then the
`$\lor$' and `$\land$' operations are defined by
\begin{eqnarray}
(\ps{S}\lor\ps{T})_V & :=&\ps{S}_V\cup \ps{T}_V\\[2pt]
(\ps{S}\land \ps{T})_V& :=&\ps{S}_V\cap \ps{T}_V
\end{eqnarray}
for all contexts $V$. It is easy to see that if $\ps{S}$ and
$\ps{T}$ are clopen sub-objects of $\Sig$, then so are
$\ps{S}\lor\ps{T}$ and $\ps{S}\land \ps{T}$.

\paragraph{The zero and unit elements.}
The zero element in the Heyting algebra $\Sub{\Sig}$ is the empty
sub-object $\ps{0}:=\{\varnothing_{V}\mid V\in\Ob{\V{}}\}$, where
$\varnothing_{V}$ is the empty subset of $\Sig_V$. The unit
element in $\Sub{\Sig}$ is $\Sig$.  It is clear that both $\ps{0}$
and $\Sig$ are clopen sub-objects of $\Sig$.

\paragraph{The `$\Rightarrow$'-operation.}
The most interesting part is the definition of the implication
$\ps{S}\Rightarrow\ps{T}$. For all $V\in\Ob{\V{}}$, it is given by
\begin{eqnarray}
(\ps{S}\Rightarrow\ps{T})_V&:=&\{\l\in\Sig_V\mid\forall\,
V^{\prime}\subseteq V,\mbox{ if }
                                                   \nonumber\\
    &{}&\hspace{1cm} \Sig(i_{V^{\prime}V})(\l)\in
    \ps{S}_{V^\prime} \mbox{ then }\Sig(i_{V^{\prime}V})(\l)\in
            \ps{T}_{V^{\prime}}\}                       \\[2pt]
    &=&\{\l\in\Sig_V\mid \forall V^{\prime}
        \subseteq V,\mbox{ if}                \nonumber\\
    &{}&\hspace{1cm} \l|_{V^{\prime}}\in \ps{S}_{V^\prime}
        \mbox{ then }\l|_{V^{\prime}}\in\ps{T}_{V^{\prime}}\}.
\end{eqnarray}

Since $\lnot\ps{S}:=\ps{S}\Rightarrow\ps{0}$, the expression for
negation follows from the above as
\begin{eqnarray}
(\lnot \ps{S})_V & =&\{\l\in\Sig_V\mid\forall\,
        V^{\prime}\subseteq V,\
\Sig(i_{V^{\prime}V})(\l)\notin \ps{S}_{V^{\prime}}\}\\[5pt]
        & =&\{\l\in\Sig_V\mid\forall\,V^{\prime}\subseteq V, \,
              \l|_{V^{\prime}}\notin\ps{S}_{V^{\prime}}\}.
\end{eqnarray}
We rewrite the formula for negation as
\begin{equation}
(\lnot \ps{S})_V=\bigcap_{V^{\prime}\subseteq
V}\big\{\l\in\Sig_V\mid \l|_{V^{\prime}}\in
\ps{S}_{V^{\prime}}^c\big\}           \label{negS(V)=}
\end{equation}
where $\ps{S}_{V^{\prime}}{}^c$ denotes the complement of
$\ps{S}_{V^{\prime}}$ in $\Sig_{V^{\prime}}$. Clearly,
$\ps{S}_{V^{\prime}}{}^c$ is clopen in $\Sig_{V^{\prime}}$ since
$\ps{S}_{V^{\prime}}$ is clopen. Since the restriction
$\Sig(i_{V^{\prime}V}):\Sig_V\map\Sig_{V^{\prime}}$ is continuous
and surjective\footnote{See proof of Theorem \ref{Theorem:
SG=SigS} below.}, it is easy to see that the inverse image
$\Sig(i_{V^{\prime}V})^{-1}(\ps{S}_{V^{\prime}}{}^{c})$ is clopen
in $\Sig_V$. Clearly,
\begin{equation}
\Sig(i_{V^{\prime}V}
)^{-1}(\ps{S}_{V^{\prime}}{}^{c})=\big\{\l\in\Sig_V\mid
\l|_{V^{\prime}}\in \ps{S}_{V^{\prime}}{}^{c}\big\}
\end{equation}
and so, from \eq{negS(V)=} we have
\begin{equation}
(\lnot \ps{S})_V=\bigcap_{V^{\prime}\subseteq V}
\Sig(i_{V^{\prime}V})^{-1}(\ps{S}_{V^{\prime}}{}^{c})
                                        \label{lnotS(V)=}
\end{equation}

The problem is that we want $(\lnot\ps{S})_V$ to be a
\emph{clopen} subset of $\Sig_V$. Now the right hand side of
\eq{lnotS(V)=} is the intersection of a family, parameterised by
$\{V^\prime\mid V^\prime\subseteq V\},$ of clopen sets. Such an
intersection is always closed, but it is only guaranteed to be
open if $\{V^\prime\mid V^\prime\subseteq V\}$ is a finite set,
which of course may  not be the case.

If $V^{\prime\prime}\subseteq V^{\prime}$ and
$\l|_{V^{\prime\prime}}\in \ps{S}_{V^{\prime\prime}}{}^{c}$, then
$\l|_{V^{\prime}}\in \ps{S}_{V^{\prime}}{}^{c}$. Indeed, if we had
$\l|_{V^{\prime}}\in \ps{S}_{V^{\prime}}$, then
$(\l|_{V^{\prime}})|_{V^{\prime\prime}}= \l|_{V^{\prime\prime}}\in
\ps{S}_{V^{\prime\prime}}$ by the definition of a sub-object, so
we would have a contradiction. This implies
$\Sig(i_{V^{\prime\prime}V}
)^{-1}(\ps{S}_{V^{\prime\prime}}{}^{c})\subseteq
\Sig(i_{V^{\prime}V})^{-1}(S_{V^{\prime}}{}^{c})$, and hence the
right hand side of \eq{lnotS(V)=} is a decreasing net of clopen
subsets of $\Sig_V$ which  converges to something, which we take
as the subset of $\Sig_V$ that is to be $(\neg\ps{S})_V$.

Here we have used the fact that the set of clopen subsets of
$\Sig_V$ is a complete lattice, where the minimum of a family
$(U_i)_{i\in I}$ of clopen subsets is defined as the interior of
$\bigcap_{i\in I}U_i$.  This leads us to define
\begin{eqnarray}
(\lnot \ps{S})_V  & :=&\mbox{int}\bigcap_{V^{\prime}\subseteq V}
\Sig(i_{V^{\prime}V} )^{-1}(\ps{S}_{V^{\prime}}{}^{c}) \\
    & =&\mbox{int}\bigcap_{V^{\prime}\subseteq V}\big\{\l\in
\Sig_V\mid \l|_{V^{\prime}}\in(S_{V^{\prime}}{}^{c})\big\}
\end{eqnarray}
as the negation in $\Subcl\Sig$. This modified definition
guarantees that $\lnot\ps{S}$ is a \emph{clopen} sub-object. A
straightforward extension of this method gives a consistent
definition of $\ps{S}\Rightarrow\ps{T}$.

\noindent This concludes the proof of the theorem.
\end{proof}

The following theorem shows the relation between the restriction
mappings of the outer presheaf $\G$ and those of the spectral
presheaf $\Sig$. We basically follow de Groote's proof of Prop.
3.22 in \cite{deG05c} and show that this result, which uses quite
a different terminology, actually gives the desired relation.

\begin{theorem}
\label{Theorem: SG=SigS} Let $V,V'\in\Ob{\V{}}$ such that
$V'\subset V$. Then
\begin{equation}
                        S_{\G(i_{V^{\prime}V})(\dastoo{V}{P})}= \Sig (i_{V^{\prime}
                        V})(S_{\dastoo{V}{P}}).
\end{equation}
\end{theorem}

\begin{proof}
First of all, to simplify notation, we can replace $\dastoo{V}{P}$
by $\P$ (which amounts to the assumption that $\P\in\PV$. This
does not play a r\^ole for the current argument). By definition,
$\G(i_{V^{\prime}V})\big(\P\big)=\dastoo{V^\prime}{P}$, so we have
to show that $S_{\dastoo{V^\prime}{P}}=\Sig
(i_{V^{\prime}V})(S_{\P})$ holds.

If $\l\in S_{\P}$, then $\l(\P)=1$, which implies $\l(\hat Q)=1$
for all $\hat Q\geq\P$. In
particular, $\l(\dastoo{V'}{P})=1$, so $\Sig (i_{V^{\prime}V})(\l)=%
\l|_{V'}\in\S_{\dastoo{V^\prime}{P}}$. This shows that $\Sig
(i_{V^{\prime}V})(S_{\P})\subseteq S_{\dastoo{V^\prime}{P}}$.

To show the converse inclusion, let $\l'\in S_{\dastoo{V'}{P}}$,
which means that $\l'(\dastoo{V'}{P})=1$. We have $\P\in{\G(i_{V'V})}^{-1}%
(\dastoo{V'}{P})$. Let
\begin{equation}
                        F_{\l'}:=\{\hat Q\in\mathcal{P}(V')\mid\l'(\hat Q)=1\}
                        =\l'^{-1}(1)\cap\mathcal{P}(V').
\end{equation}
As shown in section \ref{SubSec_DasFromFctsOnFilters}, $F_{\l'}$
is an ultrafilter in the projection lattice
$\mathcal{P}(V')$.\footnote{In general, each ultrafilter $F$ in
the projection lattice of an abelian von Neumann algebra $V$
corresponds to a unique element $\l_F$ of the Gel'fand spectrum of
$V$. The ultrafilter is the collection of all those projections
that are mapped to $1$ by $\l$, i.e., $F=\l_F^{-1}(1)\cap\PV$.}
The idea is to show that $F_{\l'}\cup\P$ is a filter base in $\PV$
that can be extended to an ultrafilter, which corresponds to an
element of the Gel'fand spectrum of $V$.

Let us assume that $F_{\l'}\cup\P$ is not a filter base in $\PV$.
Then there exists some $\hat Q\in F_{\l'}$ such that
\begin{equation}
                        \hat Q\land\P=\hat Q\P=\hat 0,
\end{equation}
which implies $\P\leq\hat 1-\hat Q$, so
\begin{equation}
                        \G(i_{V'V})(\P)=\dastoo{V'}{P}\leq\G(i_{V'V})
                        (\hat 1-\hat Q)=\hat 1-\hat Q
\end{equation}
and hence we get the contradiction
\begin{equation}
                        1=\l'(\dastoo{V'}{P})\leq\l'(\hat 1-\hat Q)=0.
\end{equation}
By Zorn's lemma, the filter base $F_{\l'}\cup\P$ is contained in
some (not necessarily unique) maximal filter base in $\PV$. Such a
maximal filter base is an ultrafilter and thus corresponds to an
element $\l$ of the Gel'fand spectrum $\Sig_V$ of $V$. Since $\P$
is contained in the ultrafilter, we have $\l(\P)=1$, so
$\l\in\S_{\P}$. By construction, $\Sig
(i_{V^{\prime}V})(\l)=\l|_{V'}=\l'\in S_{\dastoo{V'}{P}}$, the
element of $\Sig_{V'}$ we started from. This shows that
$S_{\dastoo{V^\prime}{P}}\subseteq\Sig (i_{V^{\prime}V})%
(S_{\P})$, and we obtain
\begin{equation}
                        S_{\dastoo{V^\prime}{P}}=\Sig(i_{V^{\prime}V})(S_{\P}).
                        \label{Openness}
\end{equation}

It is well-known that every state $\l'\in\Sig_{V'}$ is of the form
$\l'=\Sig(i_{V'V})(\l)=\l|_{V'}$ for some $\l\in\Sig_V$. This
implies
\begin{equation}
                        \Sig(i_{V'V})^{-1}(S_{\dastoo{V'}{P}})=S_{\dastoo{V'}{P}}%
                        \subseteq\Sig_V. \label{Continuity}
\end{equation}
Note that on the right hand side, $S_{\dastoo{V'}{P}}$ (and not
$S_{\P}$, which is a smaller set in general) shows up.
\end{proof}

De Groote has shown in \cite{deG05c} that for any unital abelian
von Neumann algebra $V$, the clopen sets $S_{\hat Q}$, $\hat
Q\in\PV$, form a base of the Gel'fand topology on $\Sig_V$.
Formulas (\ref{Openness}) and (\ref{Continuity}) hence show that
the restriction mappings
\begin{eqnarray}
                        \nonumber \Sig(i_{V^{\prime}V}):\Sig_V &\rightarrow& \Sig_{V'}\\
                        \nonumber \l &\mapsto& \l|_{V'}
\end{eqnarray}
of the spectral presheaf are open and continuous. Using
continuity, it is easy to see that $\Sig (i_{V^{\prime}V})$ is
also closed: let $C\subseteq\Sig_{V}$ be a closed subset. Since
$\Sig_{V}$ is compact, $C$ is compact, and since $\Sig
(i_{V^{\prime}V})$ is continuous, $\Sig
(i_{V^{\prime}V})(C)\subseteq\Sig_{V^{\prime}}$ is compact, too.
However, $\Sig_{V^{\prime}}$ is Hausdorff, and so
$\Sig(i_{V^{\prime}V})(C)$ is closed in $\Sig_{V^{\prime}}$.

\subsection{The Grothendieck $k$-Construction for an Abelian Monoid}
Let us briefly review the  Grothendieck construction for an
abelian monoid $M$.

{\definition A \emph{group completion} of $M$ is an abelian group
$k(M)$ together with a monoid map $\theta:M\map k(M)$ that is
universal. Namely, given any monoid morphism $\phi:M\map G$, where
$G$ is an abelian group, there exists a unique \emph{group}
morphism $\phi^\prime: k(M)\map G$  such that $\phi$ factors
through $\phi^\prime$; \ie\ we have the commutative diagram
\begin{center}
   \Vtrianglep<1`1`-1;400>[M`G`k(M);\phi`\theta`\phi^{\prime}]
\end{center}
with $\phi=\phi^\prime\circ\theta$. }

\noindent It is easy to see that any such $k(M)$ is unique up to
isomorphism.

To prove existence, first take the set of all pairs $(a,b)\in
M\times M$, each of which is to be thought of heuristically as
$a-b$. Then,  note that \emph{if} inverses existed in $M$, we
would have  $a-b=c-d$ if and only if $a+d=c+b$. This suggests
defining an equivalence relation on $M\times M$ in the following
way:
\begin{equation}
  (a,b)\equiv (c,d)  \makebox{ iff $\exists g\in M$ such that }
        a+d+g=b+c+g.     \label{Def:k(M)}
\end{equation}

{\definition The \emph{Grothendieck completion} of an abelian
monoid $M$ is the pair $(k(M),\theta)$ defined as follows:
\begin{enumerate}
\item[(i)] $k(M)$  is the set of equivalence classes $[a,b]$, where
the equivalence relation is defined in \eq{Def:k(M)}. A group law
on $k(M)$ is defined by
\begin{eqnarray}
 &&{\rm (i)}\   [a,b]+[c,d]:=[a+c,b+d],           \\[2pt]
 &&{\rm (ii)}\          0_{k(M)}:=[0_M,0_M],      \\[2pt]
 &&{\rm (iii)}\  -[a,b]:=[b,a],
\end{eqnarray}
where $0_M$ is the unit in the  abelian monoid $M$.

\item[(ii)] The map $\theta:M\map k(M)$ is defined by
\begin{equation}
\theta(a):=[a,0]
\end{equation}
for all $a\in M$.
\end{enumerate}
} It is straightforward to show that (i) these definitions are
independent of the representative elements in the equivalence
classes; (ii) the axioms for a group are satisfied; and (iii) the
map $\theta$ is universal in the sense mentioned above.

It is also clear that $k$ is a \emph{functor} from the category of
abelian monoids to the category of abelian groups. For, if
$f:M_1\map M_2$ is a morphism between abelian monoids,  define
$k(f):k(M_1)\map k(M_2)$ by $k(f)[a,b]:=[f(a),f(b)]$ for all
$a,b\in M_1$.

\subsection{Functions of Bounded Variation and $\Ga$\ps{$\mathR^{\succeq}$}}
These techniques will now be applied to the set, $\Ga\SR$, of
global elements of $\SR$. We could equally well consider
$\Ga\PR{\mathR}$ and its $k$-extension, but this would just make
the notation more complex, so in this and the following
subsections, we will mainly concentrate on $\Ga\SR$ (resp. $\SR$).
The results can easily be extended to $\Ga\PR{\mathR}$ (resp.
$\PR{\mathR}$).

It was discussed in Section \ref{SubSec:SR} how global elements of
$\PR{\mathR}$ are in one-to-one correspondence with pairs
$(\mu,\nu)$ consisting of an order-preserving and an
order-reversing function on the category $\V{}$; \ie\ with
functions $\mu:\Ob{\V{}}\map\mathR$ such that, for all
$V_1,V_2\in\Ob{\V{}}$, if $V_2\subseteq V_1$ then
$\mu(V_2)\leq\mu(V_1)$ and $\nu:\Ob{\V{}}\map\mathR$ such that,
for all $V_1,V_2\in\Ob{\V{}}$, if $V_2\subseteq V_1$ then
$\nu(V_2)\geq\nu(V_1)$; see \eq{Def:GammaRD}. The monoid law on
$\Ga\PR{\mathR}$ is given by \eq{Def:mu+nuFn}.

Clearly, global elements of $\SR$ are given by order-reversing
functions $\nu:\V{}\rightarrow\mathR$, and $\Ga\SR$ is an abelian
monoid in the obvious way. Hence the Grothendieck construction can
be applied to give an abelian group $k(\Ga\SR)$. This is defined
to be the set of equivalence classes $[\nu,\ka]$ where
$\nu,\ka\in\Ga\SR$, and where $(\nu_1,\ka_1)\equiv(\nu_2,\ka_2)$
if, and only if, there exists $\alpha\in\Ga\SR$, such that
\begin{equation}
                \nu_1+\ka_2+\alpha=\ka_1+\nu_2+\alpha      \label{Def:equivkGR}
\end{equation}
Since $\Ga\SR$ has a cancellation law, we have
$(\nu_1,\ka_1)\equiv(\nu_2,\ka_2)$ if, and only if,
\begin{equation}
                \nu_1+\ka_2=\ka_1+\nu_2.
\end{equation}
Intuitively, we can think of $[\nu,\ka]$ as being `$\nu-\ka$', and
embed $\Ga\SR$ in $k(\Ga\SR)$ by $\nu\mapsto[\nu,0]$. However,
$\nu,\ka$ are $\mathR$-valued functions on $\Ob{\V{}}$ and hence,
in this case, the expression `$\nu-\ka$' also has a \emph{literal}
meaning: \ie\ as the function $(\nu-\ka)(V):=\nu(V)-\ka(V)$ for
all $V\in\Ob{\V{}}$.

This is not just a coincidence of notation. Indeed, let
$F\big(\Ob{\V{}},\mathR\big)$ denote the set of all real-valued
functions on $\Ob{\V{}}$. Then we can construct the map,
\begin{eqnarray}
          j:k(\Ga\SR)&\map   & F\big(\Ob{\V{}},\mathR\big)
                                        \label{Def:jk->F} \\
           {[}\nu,\ka]&\mapsto&  \nu-\ka  \nonumber
\end{eqnarray}
which is  well-defined on equivalence classes.

It is  easy to see that the map in \eq{Def:jk->F} is injective.
This raises the question of  the \emph{image} in
$F\big(\Ob{\V{}},\mathR\big)$ of the map $j$: \ie\ what types of
real-valued function on $\Ob{\V{}}$ can be written as the
difference between two order-reversing functions?

For functions $f:\mathR\map\mathR$, it is a standard result that a
function can be written as the difference between two monotonic
functions if, and only if, it has bounded variation. The natural
conjecture is that a similar result applies here. To show this, we
proceed as follows.

Let $f:\Ob{\V{}}\map\mathR$ be a real-valued function on the set
of objects in the category $\V{}$. At each $V\in\Ob{\V{}}$,
consider a finite chain
\begin{equation}
    C:=\{V_0,V_1,V_2,\ldots,\ V_{n-1},V\mid V_0\subset
    V_1\subset V_2\subset\cdots\subset V_{n-1}\subset V\}
\end{equation}
of proper subsets, and define the \emph{variation} of $f$ on this
chain to be
\begin{equation}
        V_f(C):=\sum_{j=1}^n|f(V_j)-f(V_{j-1})|
\end{equation}
where we set $V_n:=V$. Now take the supremum of $V_f(C)$ for all
such chains $C$. If this is finite, we say that $f$ has a
\emph{bounded variation} and define
\begin{equation}
        I_f(V):=\sup_C V_f(C)
\end{equation}

Then it is clear that (i) $V\mapsto I_f(V)$ is an order-preserving
function on $\Ob{\V{}}$; (ii) $f-I_f$ is an order-reversing
function on $\Ob{\V{}}$; and (iii) $-I_f$ is an order-reversing
function on $\Ob{\V{}}$. Thus, any function, $f$, of bounded
variation can be written as
\begin{equation}
        f\equiv(f-I_f)-(-I_f)
\end{equation}
which is the difference of two order-reversing functions; \ie\ $f$
can be expressed as the difference of two elements of $\Ga\SR$.

Conversely, it is a straightforward modification of the proof for
functions \mbox{$f:\mathR\map\mathR$}, to show that if
$f:\Ob{\V{}}\map\mathR$ is the difference of two order-reversing
functions, then $f$ is of bounded variation. The conclusion is
that $k(\Ga\SR)$ is in bijective correspondence with the set,
$\BV$, of functions  $f:\Ob{\V{}}\map\mathR$ of bounded variation.

\subsection{Taking Squares in $k(\Ga$\ps{$\mathR^{\succeq}$)}}
We can now think of $k(\Ga\SR)$ in two ways: (i) as the set of
equivalence classes $[\nu,\ka]$, of elements $\nu,\ka\in\Ga\SR$;
and (ii) as the set, $\BV$, of differences $\nu-\ka$ of such
elements.

As expected, $\BV$ is an abelian group. Indeed: suppose $\alpha
=\nu_1-\ka_1$ and $\beta=\nu_2-\ka_2$ with
$\nu_1,\nu_2,\ka_1,\ka_2\in \Ga\SR$, then
\begin{equation}
        \alpha+\beta=(\nu_1+\nu_2)-(\ka_1+\ka_2)
\end{equation}
Hence $\alpha+\beta$ belongs to $\BV$ since $\nu_1+\nu_2$ and
$\ka_1 +\ka_2$ belong to $\Ga\SR$.

\paragraph{The definition of $[\nu,0]^2$.} We will now show how to
take the square of elements of $k(\Ga\SR)$ that are of the form
$[\nu,0]$. Clearly, $\nu^2$ is well-defined as a function on
$\Ob{\V{}}$, but it may not belong to $\Ga\SR$. Indeed, if
$\nu(V)<0$ for any $V$, then the function $V\mapsto\nu^2(V)$ can
get smaller as $V$ gets smaller, so it is order-preserving instead
of order-reversing.

This suggests the following strategy. First, define functions
$\nu_+$ and $\nu_-$ by
\begin{equation}
\nu_+(V):= \left\{\begin{array}{ll}
            \nu(V) & \mbox{\ if\ $\nu(V)\geq 0$} \\[2pt]
            0 & \mbox{\ if\ $\nu(V)< 0$}
         \end{array}
        \right.
\end{equation}
and
\begin{equation}
\nu_-(V):= \left\{\begin{array}{ll}
            0 & \mbox{\ if\ $\nu(V)\geq0$} \\[2pt]
            \nu(V) & \mbox{\ if\ $\nu(V)<0$.}
         \end{array}
        \right.
\end{equation}
Clearly, $\nu(V)=\nu_+(V)+\nu_-(V)$ for all $V\in\Ob{\V{}}$. Also,
for all $V$, $\nu_+(V)\nu_-(V)=0$, and hence
\begin{equation}
        \nu(V)^2=\nu_+(V)^2+\nu_-(V)^2 \label{l2=lL2+lR2}
\end{equation}
However, (i) the function $V\mapsto\nu_+(V)^2$ is order-reversing;
and (ii) the function $V\mapsto\nu_-(V)^2$ is order-preserving.
But then $V\mapsto -\nu_-(V)^2$ is order-reversing. Hence, by
rewriting \eq{l2=lL2+lR2} as
\begin{equation}
        \nu(V)^2=\nu_+(V)^2-(-\nu_-(V)^2)
\end{equation}
we see that the function $V\mapsto \nu^2(V):=\nu(V)^2$ is an
element of $\BV$.

In terms of $k(\Ga\SR)$, we can define
\begin{equation}
        [\nu,0]^2:=[\nu_+^2, -\nu_-^2]
\end{equation}
which belongs to $k(\Ga\SR)$. Hence, although there exist
$\nu\in\Ga\SR$ that have no square in $\Ga\SR$, such global
elements of $\SR$ \emph{do} have squares in the $k$-completion,
$k(\Ga\SR)$. On the level of functions of bounded variation, we
have shown that the square of a monotonic (order-reversing)
function is a function of bounded variation.

On the other hand, we cannot take the square of an arbitrary
element $[\nu,\ka]\in\Ga\SR$, since the square of a function of
bounded variation need not be a function of bounded
variation.\footnote{We have to consider functions like $(\nu_+
+\nu_- -(\ka_+ +\ka_-))^2$, which contains terms of the form
$\nu_+\ka_-$ and $\nu_-\ka_+$: in general, these are neither
order-preserving nor order-reversing.}

\subsection{The Object $k$(\ps{$\mathR^{\succeq}$}) in the Topos $\SetH{}$.}
\subsubsection{The Definition of $k$(\ps{$\mathR^{\succeq}$}).}
The next step is to translate these results about the set
$k(\Ga\SR)$ into the construction of an object $\kSR$ in the topos
$\SetH{}$. We anticipate that, if this can be done, then
$k(\Ga\SR)\simeq\Ga\kSR$.

As was discussed in Section \eq{SubSec:SR}, the presheaf $\SR$ is
defined at each stage $V$\ by
\begin{equation}
 \SR_V:=\{\nu:\downarrow\!\!V\map\mathR
\mid\nu\in\mathcal{OR}(\downarrow\!\!V,\mathR)\}.
\end{equation}
If $i_{V^\prime V}:V^\prime\subseteq V$, then the presheaf map
from $\SR_V$ to $\SR_{V^\prime}$ is just the restriction of the
order-reversing functions from $\downarrow\!\! V$ to
$\downarrow\!\! V^\prime$.

The first step in constructing $\kSR$ is to define an equivalence
relation on pairs of functions, $\nu,\ka\in\SR_V$, for each stage
$V$, by saying that $(\nu_1,\ka_1) \equiv (\nu_2,\ka_2)$ if, and
only, there exists $\a\in\SR_V$ such that
\begin{equation}
        \nu_1(V^\prime)+\ka_2(V^\prime)+\alpha(V^\prime)=
        \ka_1(V^\prime)+\nu_2(V^\prime)+\alpha(V^\prime)
\end{equation}
for all $V^\prime\subseteq V$.

{\definition The presheaf  $\kSR$ is defined over the category
$\V{}$ in the following way.
\begin{enumerate}
\item[(i)] On objects $V\in\Ob{\V{}}$:
\begin{equation}
 \kSR_V:=\{[\nu,\ka]\mid\nu,\ka\in\mathcal{OR}(\downarrow\!\!V,\mathR)\},
\end{equation}
where $[\nu,\ka]$ denotes the $k$-equivalence class of
$(\nu,\ka)$.

\item[(ii)] On morphisms $i_{V^\prime V}:V^\prime\subseteq V$: The
arrow $\kSR(i_{V^\prime V}):\kSR_V\map\kSR_{V^\prime}$ is given by
$\big(\kSR(i_{V^\prime
V})\big)([\nu,\ka]):=[\nu|_{V^\prime},\ka|_{V^\prime}]$ for all
$[\nu,\ka]\in\kSR_V$.
\end{enumerate}
}

It is straightforward to show that $\kSR$ is an abelian
group-object in the topos $\SetH{}$. In particular, an arrow
$+:\kSR\times\kSR\map\kSR$ is defined at each stage $V$ by
\begin{equation}
  +_V\big([\nu_1,\ka_1],[\nu_2,\ka_2]\big):=
        [\nu_1+\nu_2,\ka_1+\ka_2]
\end{equation}
for all
$\big([\nu_1,\ka_1],[\nu_2,\ka_2]\big)\in\kSR_V\times\kSR_V$. It
is easy to see that (i) $\Ga\kSR\simeq k(\Ga\SR)$; and (ii) $\SR$
is a sub-object of $\kSR$ in the topos $\SetH{}$.

\subsubsection{The Presheaf $k$(\ps{$\mathR^{\succeq}$})
as the Quantity-Value Object.} We can now identify $\kSR$ as a
possible quantity-value object in $\SetH{}$. To each bounded,
self-adjoint operator $\A$, there is an arrow
$[\dasBo{A}]:\Sig\map\kSR$, given by first sending $\A\in\BH_\sa$
to $\dasBo{A}$ and then taking $k$-equivalence classes. More
precisely, one takes the monic $\iota:\SR\hookrightarrow\kSR$ and
then constructs $\iota\circ\dasBo{A}:\Sig\map\kSR.$

Since, for each stage $V$, the elements in the image of
$[\dasBo{A}]_V= (\iota\circ\dasBo{A})_V$ are of the form
$[\nu,0]$, $\nu\in\SR_V$,  their square is well-defined. From a
physical perspective, the use of $\kSR$ rather than $\SR$ renders
possible the definition of things like the `intrinsic dispersion',
$\nabla(\A):=\dasBo{A^2}-\dasBo{A}^2$; see \eq{Def:nabla}.

\subsubsection{The square of an arrow $[\dasBo{A}]$.}
An arrow $[\dasBo{A}]:\Sig\map\kSR$ is constructed by first
forming the outer daseinisation $\dasBo{A}$ of $\A$, which is an
arrow from $\Sig$ to $\SR$, and then composing with the monic
arrow from $\SR$ to $\kSR$. Since only outer daseinisation is
used, for each $V\in\V{}$ and each $\l\in\Sig_V$ one obtains an
element of $\kSR_V$ of the form $[\daso{A}_V(\l),0]$. We saw how
to take the square of these functions, and applying this to all
$\l\in\Sig_V$ and all $V\in\V{}$, we get the square
$[\dasBo{A}]^2$ of the arrow $[\dasBo{A}]$.

If we consider an arrow of the form
$\dasB{A}:\Sig\map\PR{\mathR}$, then the construction involves
both inner and outer daseinisation, see Theorem \ref{Th:ST_3}. For
each $V$ and each $\l\in\Sig_V$, we obtain a pair of functions
$(\dasi{A}_V(\l),\daso{A}_V(\l))$, which are both not constantly
$0$ in general. There is no canonical way to take the square of
these in $\PR{\mathR}_V$. Going to the $k$-extension
$k(\PR{\mathR})$ of $\PR{\mathR}$ does not improve the situation,
so we cannot define the square of an arrow $\dasB{A}$ (or
$[\dasB{A}]$ in general.

\section{Appendix 2: A Short Introduction to the Relevant Parts of Topos Theory}
\subsection{What is a Topos?}
It is impossible to give here more than the briefest of
introductions to topos theory. At the danger of being highly
imprecise, we restrict ourselves to mentioning some aspects of
this well-developed mathematical theory and give a number of
pointers to the literature. The aim merely is to give a very rough
idea of the structure and internal logic of a topos.

There are a number of excellent textbooks on topos theory, and the
reader should consult at least one of them. We found the following
books useful: \cite{LR03,Gol84,MM92,Jst02,Bell88,LamScott86}.

Topos theory is a remarkably rich branch of mathematics which can
be approached from a variety of different viewpoints. The basic
area of mathematics  is category theory; where, we recall, a
category consists of a collection of \emph{objects} and a
collection of \emph{morphisms} (or \emph{arrows}).

In the special case of the category of sets, the objects are sets,
and a morphism is a function between a pair of sets. In general,
each morphism $f$ in a category is associated with a pair of
objects\footnote{The collection of all objects in category, $\cal
C$, is denoted $\Ob{\cal C}$. The collection of arrows from $B$ to
$A$ is denoted $\Hom{\cal C}{B}{A}$. We will only be interested in
`small' categories in which both these collections are sets
(rather than the, more general, classes.)}, known as its `domain'
and  `codomain', and is written as $f:B\map A$ where $B$ and $A$
are the domain and codomain respectively. Note that this arrow
notation is used even if $f$ is not a function in the normal
set-theoretic sense. A key ingredient in the definition of a
category is that if $f:B\map A$ and $g:C\map B$ (\ie\ the codomain
of $g$ is equal to the domain of $f$) then $f$ and $g$ can be
`composed' to give an arrow $f\circ g:C\map A$; in the case of the
category of sets, this is just the usual composition of functions.

A simple  example of a category is given by  any partially-ordered
set (`poset') $\cal C$: (i) the objects are defined to be the
elements of $\cal C$; and (ii) if $p,q\in\cal C$, a morphism from
$p$ to $q$ is defined to exist if, and only if, $p\preceq q$ in
the poset structure.  Thus, in a poset regarded as a category,
there is at most one morphism between any pair of objects
$p,q\in\cal C$; if it exists, we shall write this morphism as
$i_{pq}:p\map q$. This example is important for us in form of the
`category of contexts', $\V{}$, in quantum theory. The objects in
$\V{}$ are the commutative, unital\footnote{`Unital' means that
all these algebras contain the identity operator $\hat 1\in\BH$.}
von Neumann sub-algebras of the algebra, $\BH$, of all bounded
operators on the Hilbert space $\Hi$.

\paragraph{Topoi as mathematical universes.} Every (elementary) topos
$\tau$ can be seen as a \emph{mathematical universe}. As a
category, a topos $\tau$ possesses a number of structures that
generalise constructions that are possible in the category,
$\Set$, of sets and functions.\footnote{More precisely,
\emph{small} sets and functions between them. Small means that we
do not have proper classes. One must take care in these
foundational issues to avoid problems like Russell's paradox.}
Namely, in $\Set$, we can construct new sets from given ones in
several ways. Specifically, let $S,T$ be two sets, then we can
form the cartesian product $S\times T$, the disjoint union
$S\amalg T$ and the exponential $S^{T}$---the set of all functions
from $T$ to $S$.

These constructions turn out to be fundamental, and they can all
be phrased in an abstract, categorical manner, where they are
called the `product', `co-product' and `exponential',
respectively. By  definition, in a topos $\tau$, these operations
always exist. The first and second of these properties are called
`finite completeness' and `finite co-completeness', respectively.

One consequence of the existence of finite limits is that each
topos, $\tau$, has a \emph{terminal object}, denoted by $1_\tau$.
This is characterised by the property that for any object $A$ in
the topos $\tau$, there exists exactly one arrow from $A$ to
$1_\tau$. In $\Set$, any one-element set $1=\{*\}$ is
terminal.\footnote{Like many categorical constructions, the
terminal object is fixed only up to isomorphism: all one-element
sets are isomorphic to each other, and any of them can serve as a
terminal object. Nonetheless, one speaks of \emph{the} terminal
object.}

Of course, $\Set$ is a topos, too, and it is precisely the topos
which usually plays the r\^{o}le of our mathematical universe,
since we construct our mathematical objects starting from sets and
functions between them. As a slogan, we have: a topos $\tau$ is a
category with `certain crucial' properties that are similar to
those in $\Set$. A very nice and gentle introduction to these
aspects of topos theory is the book \cite{LR03}. Other good
sources are \cite{Gol84,McL71}.

In order to `do mathematics', one must also have a logic,
including a deductive system. Each topos comes equipped with an
\emph{internal logic}, which is of \emph{intuitionistic} type. We
will now very briefly sketch the main characteristics of
intuitionistic logic and the mathematical structures in a topos
that realise this logic.

\paragraph{The sub-object classifier.}
Let $X$ be a set, and let $P(X)$ be the power set of $X$; \ie\ the
set of subsets of $X$. Given a subset $K\in P(X)$, one can ask for
each point $x\in X$ whether or not it lies in $K$. Thus there is
the \emph{characteristic function} $\chi_K:X\rightarrow \{0,1\}$
of $K$, which is defined as
\begin{equation}            \label{Def:chiK}
\chi_{K}(x):=\left\{
\begin{tabular}
[c]{ll}
$1$ & if $x\in K$\\
$0$ & if $x\notin K$
\end{tabular}
\ \right.
\end{equation}
for all $x\in X$; cf.\ \eq{TVxinKcl}. The two-element set
$\{0,1\}$ plays the r\^{o}le of a set of \emph{truth-values} for
propositions (of the form ``$x\in K$''). Clearly, $1$ corresponds
to `true', $0$ corresponds to `false', and there are no other
possibilities. This is an argument about sets, so it takes place
in, and uses the logic of, the topos $\Set$ of sets and functions.
$\Set$ is a \emph{Boolean topos}, in which the familiar two-valued
logic and the axiom ($*$) hold. (This does not contradict the fact
that the internal logic of topoi is intuitionistic, since Boolean
logic is a special case of intuitionistic logic.)

In an arbitrary topos, $\tau$, there is a special object
$\O_\tau$, called the \emph{sub-object classifier}, that takes the
r\^{o}le of the set $\{0,1\}\simeq \{{\rm false,true}\}$ of
truth-values. Let $B$ be an object in the topos, and let $A$ be a
sub-object of $B$. This means that there is a monic $A\map
B$,\footnote{A \emph{monic} is the categorical version of an
injective function. In the topos $\Set$, monics exactly are
injective functions.} (this is the categorical generalisation of
the inclusion of a subset $K$ into a larger set $X$). As in the
case of $\Set$, we can also characterise $A$ as a sub-object of
$B$ by an arrow from $B$ to the sub-object classifier $\O_\tau$;
in $\Set$, this arrow is the characteristic function $\chi
_{K}:X\map\{0,1\}$ of \eq{Def:chiK}. Intuitively, this
`characteristic arrow' from $B$ to $\O_\tau$ describes how $A$
`lies in' $B$. The textbook definition is:
\begin{definition}
In a category $\tau$ with finite limits, a \emph{sub-object
classifier} is an object $\O_\tau$, together with a monic ${\rm
true} :1_\tau\map\O_\tau$, such that to every monic $m:A\map B$ in
$\tau$ there is a unique arrow $\chi_A:B\map\O_\tau$ which, with
the given monic, forms a pullback square
\begin{center}
\setsqparms[1`2`2`1;700`700] \square[A`1_\tau`B`\O_\tau;`m`{\rm
true}`\chi_A]
\end{center}
\end{definition}

In $\Set$, the arrow ${\rm true}:1\map\{0,1\}$ is given by ${\rm
true}(*)=1$. In general, the sub-object classifier, $\O_\tau$,
need not be a set, since it is an object in the topos $\tau$, and
the objects of $\tau$ need not be sets. Nonetheless, there is an
abstract notion of \emph{elements} (or \emph{points}) in category
theory that we can use. Then the elements of $\O_\tau$ are the
truth-values available in the internal logic of our topos $\tau$,
just like `false' and `true', the elements of $\{{\rm false,
true}\}$, are the truth-values available in the topos $\Set$.

To understand the abstract notion of elements, let us consider
sets for a moment. Let $1=\{*\}$ be a one-element set, the
terminal object in $\Set$. Let $S$ be a set and consider an arrow
$e$ from $1$ to $S$. Clearly, (i) $e(*)\in S$ is an element of
$S$; and (ii) the set of all functions from $1$ to $S$ corresponds
exactly to the set of all elements of $S$.

This idea can be generalised to any category that has a terminal
object $1$. More precisely, an \emph{element} of an object $A$ is
defined to be an arrow from $1$ to $A$ in the category. For
example, in the definition of the sub-object classifier the arrow
`${\rm true}:1_\tau\map\O_\tau$' is an element of $\O_\tau$. It
may happen that an object $A$ has no elements, i.e., there are no
arrows $1_\tau\map A$. It is common to consider arrows from
subobjects $U$ of $A$ to $A$ as \emph{generalised elements}.

As mentioned above, the elements of the sub-object classifier,
understood as the arrows $1_\tau\map\O_\tau$, are the
truth-values. Moreover, the set of these arrows forms a Heyting
algebra (see, for example, section 8.3 in \cite{Gol84}). This is
how (the algebraic representation of) intuitionistic logic
manifests itself in a topos. Another, closely related fact is that
the set, $\Sub{A}$, of sub-objects of any object $A$ in a topos
forms a Heyting algebra.

\paragraph{The definition of a topos.} Let us pull together these
various remarks and list
the most important properties of a topos, $\tau$, for our
purposes:
\begin{enumerate}
\item There is a terminal object $1_\tau$ in $\tau$. Thus,
given any object $A$ in the topos, there is a unique arrow
$A\map 1_\tau$.

For any object $A$ in the topos, an arrow $1_\tau\map A$ is called
a \emph{global element} of $A$. The set of all global elements of
$A$ is denoted $\Ga A$.

Given $A,B\in\Ob{\tau}$, there is a product $A\times B$ in $\tau$.
In fact, a topos always has \emph{pull-backs}, and the product is
just a special case of this.\footnote{The conditions in 1.\ above
are equivalent to saying that $\tau$ is finitely complete.}

\item There is an initial object $0_\tau$ in $\tau$. This means
that given any object $A$ in the topos, there is a unique arrow
$0_\tau\map A$.

Given $A,B\in\Ob{\tau}$, there is a co-product $A\sqcup B$ in
$\tau$.  In fact, a topos always has \emph{push-outs}, and the
co-product is just a special case of this.\footnote{The conditions
in 2.\ above are equivalent to saying that $\tau$ is finitely
co-complete.}

\item There is \emph{exponentiation}: \ie\ given objects $A,B$ in
$\tau$ we can form an object $A^B$, which is the topos analogue of
the set of functions from $B$ to $A$ in set theory. The definitive
property of exponentiation is that, given any object $C$, there is
an isomorphism
\begin{equation}
\Hom{\tau}{C}{A^B}\simeq \Hom{\tau}{C\times B}{A}\label{Def:exp}
\end{equation}
that is natural in $A$ and $C$; \ie\ it is `well-behaved' under
morphisms of the objects involved.

\item There is a sub-object classifier $\O_\tau$.
\end{enumerate}

\subsection{Presheaves on a Poset}
\label{SubSec:presheaves-poset}

To illustrate the main ideas, we will first give a few definitions
from the theory of presheaves on a partially ordered set (or
`poset'); in the case of quantum theory, this poset is the space
of `contexts' in which propositions are asserted. We shall then
use these ideas to motivate the definition of a presheaf on a
general category.  Only the briefest of treatments is given here,
and the reader is referred to the standard literature for more
information \cite{Gol84,MM92}.

A \emph{presheaf} (also known as a \emph{varying set\/}) $\ps{X}$
on a poset $\cal C$ is a function that assigns to each $p\in\cal
C$, a set $\ps{X}_p$; and to each pair $p\preceq q$ (\ie\
$i_{pq}:p\map q$), a map $\ps{X}_{qp}:\ps{X}_q\map \ps{X}_p$ such
that (i) $\ps{X}_{pp}:\ps{X}_p\map\ps{X}_p$ is the identity map
${\rm id}_{{\ps{X}_p}}$ on $\ps{X}_p$, and (ii) whenever $p\preceq
q\preceq r$, the composite map
$\ps{X}_r\stackrel{\ps{X}_{rq}}\longrightarrow
\ps{X}_q\stackrel{\ps{X}_{qp}}\longrightarrow \ps{X}_p$ is equal
to
 $\ps{X}_r\stackrel{\ps{X}_{rp}}\longrightarrow \ps{X}_p$, \ie
\begin{equation}
        \ps{X}_{rp}= \ps{X}_{qp}\circ\ps{X}_{rq}. \label{Xrp=XqpXrq}
\end{equation}
The notation $\ps{X}_{qp}$ is shorthand for the more cumbersome
$\ps{X}(i_{pq})$; see below in the definition of a functor.

An \emph{arrow}, or \emph{natural transformation} $\eta:\ps{X}\map
\ps{Y}$ between two presheaves $\ps{X},\ps{Y}$ on $\cal C$ is a
family of maps $\eta_p:\ps{X}_p\map \ps{Y}_p$, $p\in\cal C$, that
satisfy the intertwining conditions
\begin{equation}
        \eta_p\circ\ps{X}_{qp}=\ps{Y}_{qp}\circ\eta_q
\end{equation}
whenever $p\preceq q$. This is equivalent to the commutative
diagram
\begin{equation}                                                                \label{Def:eta}
                \setsqparms[1`1`1`1;1000`700]
    \square[\ps{X}_q`\ps{X}_p`\ps{Y}_q`\ps{Y}_p;
    \ps{X}_{qp}`\eta_q`\eta_p`\ps{Y}_{qp}]
\end{equation}

It follows from these basic definitions, that a sub-object of a
presheaf $\ps{X}$ is a presheaf $\ps{K}$, with an arrow
$i:\ps{K}\map \ps{X}$ such that (i) $\ps{K}_p\subseteq \ps{X}_p$
for all $p\in\cal C$; and (ii) for all $p\preceq q$, the map
$K_{qp}:\ps{K}_q\map \ps{K}_p$ is the restriction of
$\ps{X}_{qp}:\ps{X}_q\map \ps{X}_p$ to the subset
$\ps{K}_q\subseteq\ps{X}_q$. This is shown in the commutative
diagram
\begin{equation}                                                                \label{cd}
                \setsqparms[1`1`1`1;1000`700]
    \square[\ps{K}_q`\ps{K}_p`\ps{X}_q`\ps{X}_p;
    \ps{K}_{qp}```\ps{X}_{qp}]
\end{equation}
where the vertical arrows are subset inclusions.

The collection of all presheaves on a poset $\cal C$ forms a
category, denoted $\SetC{\cal C}$.  The arrows/morphisms between
presheaves in this category the arrows (natural transformations)
defined above.

\subsection{Presheaves on a General Category}
\label{SubSec:presheaves-gen-cat} The ideas sketched above admit
an immediate generalization to the theory of presheaves on an
arbitrary `small' category $\cal C$ (the qualification `small'
means that the collection of objects is a genuine set, as is the
collection of all arrows/morphisms between any pair of objects).
To make the necessary definition we first need the idea of a
`functor':

\paragraph{The idea of a functor.}
A central concept is that of a `functor' between a pair of
categories $\cal C$ and $\cal D$. Broadly speaking, this is an
arrow-preserving function from one category to the other. The
precise definition is as follows.

\begin{definition}
\begin{enumerate}

\item {A \emph{covariant functor} $\fu{F}$ from a category $\cal
C$ to a category $\cal D$ is a function that assigns
    \begin{enumerate}
        \item to each $\cal C$-object $A$, a $\cal D$-object
        $\fu{F}_A$;

        \item {to each $\cal C$-morphism $f:B\map A$, a
$\cal D$-morphism $\fu{F}(f):\fu{F}_B\map \fu{F}_A$ such that
$\fu{F}(\id_A)={\rm id}_{\fu{F}_A}$; and, if $g:C\map B$, and
$f:B\map A$ then
    \begin{equation}
        \fu{F}(f\circ g)=\fu{F}(f)\circ
                \fu{F}(g).     \label{Def:covfunct}
    \end{equation}
        }
    \end{enumerate}
    }

\item {A {\em contravariant functor\/} $\fu{X}$ from a category
$\cal C$ to a category $\cal D$ is a function that assigns
\begin{enumerate} \item to each $\cal C$-object $A$, a $\cal
D$-object $\fu{X}_A$;

    \item {to each $\cal C$-morphism $f:B\map A$, a $\cal
D$-morphism $\fu{X}(f):\fu{X}_A\map \fu{X}_B$ such that
$\fu{X}(\id_A)=\id_{\fu{X}_A}$; and, if $g:C\map B$, and $f:B\map
A$ then
    \begin{equation}
        \fu{X}(f\circ g)=\fu{X}(g)\circ\fu{X}(f).
                        \label{Def:confunct}
    \end{equation}
        }
    \end{enumerate}
    }
\end{enumerate}
\end{definition}

The connection with the idea of a presheaf on a poset is
straightforward. As mentioned above, a poset $\cal C$ can be
regarded as a category in its own right, and it is clear that a
presheaf on the poset $\cal C$ is the same thing as a
contravariant functor $\ps{X}$ from the category $\cal C$ to the
category $\Set$ of normal sets. Equivalently, it is a covariant
functor from the `opposite' category\footnote{The `opposite' of a
category $\cal C$ is a category, denoted ${\cal C}^\op$, whose
objects are the same as those of $\cal C$, and whose morphisms are
defined to be the opposite of those of $\cal C$; \ie\ a morphism
$f:A\map B$ in ${\cal C}^\op$ is said to exist if, and only if,
there is a morphism $f:B\map A$ in $\cal C$.} ${\cal C}^{\rm op}$
to $\Set$. Clearly,  \eq{Xrp=XqpXrq} corresponds to the
contravariant condition  \eq{Def:confunct}. Note that
mathematicians usually call the objects in $\cal C$ `stages of
truth', or just `stages'. For us they are `contexts', `classical
snap-shops', or `world views'.

\paragraph{Presheaves on an arbitrary category $\cal C$.}
These remarks motivate the definition of a presheaf on an
arbitrary small category $\cal C$: namely, a {\em presheaf\/} on
$\cal C$ is a covariant functor\footnote{Throughout this series of
papers, a presheaf is indicated by a letter that is underlined.}
$\ps{X}:{\cal C}^\op\map\Set$ from ${\cal C}^\op$ to the category
of sets. Equivalently, a presheaf is a contravariant functor from
$\cal C$ to the category of sets.

We want to make the collection of presheaves on $\cal C$ into a
category, and therefore we need to define what is meant by a
`morphism' between two presheaves $\ps{X}$ and $\ps{Y}$.  The
intuitive idea is that such a morphism from $\ps{X}$ to $\ps{Y}$
must give a `picture' of $\ps{X}$ within $\ps{Y}$. Formally, such
a morphism is defined to be a \emph{natural transformation}
$N:\ps{X}\map\ps{Y}$, by which is meant a family of maps (called
the \emph{components} of $N$) $N_A:\ps{X}_A\map\ps{Y}_A$,
$A\in\Ob{\cal C}$, such that if $f:B\map A$ is a morphism in $\cal
C$, then the composite map $\ps{X}_{A}
\stackrel{N_A}\longrightarrow\ps{Y}_A\stackrel{\ps{Y}(f)}
\longrightarrow\ps{Y}_B$ is equal to $\ps{X}_A
\stackrel{\ps{X}(f)}\longrightarrow\ps{X}_B\stackrel{N_B}
\longrightarrow \ps{Y}_A$. In other words, we have the commutative
diagram
\begin{equation}                                                                \label{cdNT}
                \setsqparms[1`1`1`1;1000`700]
    \square[\ps{X}_A`\ps{X}_B`\ps{Y}_A`\ps{Y}_B;
    \ps{X}(f)`N_A`N_B`\ps{Y}(f)]
\end{equation}
of which \eq{Def:eta} is clearly a special case. The category of
presheaves on $\cal C$ equipped with these morphisms is denoted
$\SetC{\cal C}$.

The idea of a sub-object generalizes in an obvious way. Thus we
say that $\ps{K}$ is a \emph{sub-object} of $\ps{X}$ if there is a
morphism in the category of presheaves (\ie\ a natural
transformation) $\iota:\ps{K}\map\ps{X}$ with the property that,
for each $A$, the component map $\iota_A:\ps{K}_A\map\ps{X}_A$ is
a subset embedding, \ie\ $\ps{K}_A\subseteq \ps{X}_A$. Thus, if
$f:B\map A$ is any morphism in $\cal C$, we get the analogue of
the commutative diagram \eq{cd}:
\begin{equation}                                                                \label{subobject}
                \setsqparms[1`1`1`1;1000`700]
    \square[\ps{K}_A`\ps{K}_B`\ps{X}_A`\ps{X}_B;
    \ps{K}(f)```\ps{X}(f)]
\end{equation}
where, once again, the vertical arrows are subset inclusions.

The category of presheaves on $\cal C$, $\Set^{{\cal C}^{\rm
op}}$, forms a topos. We do not need the full definition of a
topos; but we do need the idea, mentioned in Section
\ref{SubSec:presheaves-poset}, that a topos has a sub-object
classifier $\O$, to which we now turn.

\paragraph{Sieves and the sub-object classifier $\Om$.}
Among the key concepts in presheaf theory is that of a `sieve',
which plays a central role in the construction of the sub-object
classifier in the topos of presheaves on a category $\cal C$.

A {\em sieve\/} on an object $A$ in $\cal C$ is defined to be a
collection $S$ of morphisms $f:B\map A$ in $\cal C$ with the
property that if $f:B\map A$ belongs to $S$, and if $g:C\map B$ is
any morphism with co-domain $B$, then $f\circ g:C\map A$ also
belongs to $S$. In the simple case where $\cal C$ is a poset, a
sieve on $p\in\cal C$ is any subset $S$ of $\cal C$ such that if
$r\in S$ then (i) $r\preceq p$, and (ii) $r'\in S$ for all
$r'\preceq r$; in other words, a sieve is nothing but a {\em
lower\/} set in the poset.

The presheaf $\Om:{\cal C}\map \Set$ is now defined as follows. If
$A$ is an object in $\cal C$, then $\Om_A$ is defined to be the
set of all sieves on $A$; and if $f:B\map A$, then
$\Om(f):\Om_A\map\Om_B$ is defined as
\begin{equation}
{\ps{\O}}(f)(S):= \{h:C\map B\mid f\circ h\in S\}
                                \label{Def:Om(f)}
\end{equation}
for all $S\in\Om_A$; the sieve $\Om(f)(S)$ is often written as
$f^*(S)$, and is known as the {\em pull-back\/} to $B$ of the
sieve $S$ on $A$ by the morphism $f:B\map A$.

It should be noted that if $S$ is a sieve on $A$, and if $f:B\map
A$ belongs to $S$, then from the defining property of a sieve we
have
\begin{equation}
        f^*(S):=\{h:C\map B\mid f\circ h\in S\}=
\{h:C\map B\}=:\ \downarrow\!\!B     \label{f*S}
\end{equation}
where $\downarrow\!\!B$ denotes the {\em principal\/} sieve on
$B$, defined to be the set of all morphisms in $\cal C$ whose
codomain is $B$.

If $\cal C$ is a poset, the pull-back operation corresponds to a
family of maps $\Om_{qp}:\Om_q\map\Om_p$ (where $\Om_p$ denotes
the set of all sieves/lower sets on $p$ in the poset) defined by
$\Om_{qp}=\Om(i_{pq})$ if $i_{pq}:p\map q$ ({\em i.e.}, $p\preceq
q$). It is straightforward to check that if $S\in\Om_q$, then
\begin{equation}
\Om_{qp}(S):=\downarrow\!{p}\cap S \label{Def:Omqp}
\end{equation}
where $\downarrow\!{p}:=\{r\in{\cal C}\mid r\preceq p\}$.

A crucial property of sieves is that the set $\Om_A$ of sieves on
$A$ has the structure of a Heyting algebra. Specifically, the unit
element $1_{\Om_A}$ in $\Om_A$ is the principal sieve
$\downarrow\!\!A$, and the null element $0_{\Om_A}$ is the empty
sieve $\emptyset$. The partial ordering in $\Om_A$ is defined by
$S_1\preceq S_2$ if, and only if, $S_1\subseteq S_2$; and the
logical connectives are defined as:
\begin{eqnarray}
    && S_1\land S_2:=S_1\cap S_2    \label{Def:S1landS2}\\
    && S_1\lor S_2:=S_1\cup S_2     \label{Def:S1lorS2} \\
    &&S_1\Rightarrow S_2:=\{f:B\map A\mid
    \mbox{$\forall g:C\map B$ if $f\circ g\in S_1$ then
                $f\circ g\in S_2$}\}\hs{30}
\end{eqnarray}
As in any Heyting algebra, the negation of an element $S$ (called
the {\em pseudo-complement\/} of $S$) is defined as $\neg
S:=S\Rightarrow 0$; so that
\begin{equation}
    \neg S:=\{f:B\map A\mid \mbox{for all
$g:C\map B$, $f\circ g\not\in S$} \}.    \label{Def:negS}
\end{equation}

It can be shown that the presheaf $\Om$ is a sub-object classifier
for the topos $\SetC{\cal C}$. That is to say, sub-objects of any
object $\ps{X}$ in this topos ({\em i.e.}, any presheaf on $\cal
C$) are in one-to-one correspondence with morphisms
$\chi:\ps{X}\map {\ps{\O}}$. This works as follows. First, let
$\ps K$ be a sub-object of $\ps{X}$ with an associated
characteristic arrow $\chi_{\ps{K}}:\ps{X}\map{\ps{\O}}$. Then, at
any stage $A$ in $\cal C$, the `components' of this arrow,
$\chi_{\ps{K}A}:\ps{X}_A\map\Om_A$, are defined as
\begin{equation}
    \chi_{\ps{K} A}(x):=\{f:B\map A\mid \ps{X}(f)(x)\in
                                \ps{K}_B\} \label{Def:chiKA}
\end{equation}
for all $x\in \ps{X}_A$. That the right hand side of
\eq{Def:chiKA} actually {\em is\/} a sieve on $A$ follows from the
defining properties of a sub-object.

Thus, in each `branch' of the category $\cal C$ going `down' from
the stage $A$, $\cha{\ps{K}}{}_A(x)$ picks out the first member
$B$ in that branch for which $\ps{X}(f)(x)$ lies in the subset
$\ps{K}_B$, and the commutative diagram \eq{subobject} then
guarantees that $\ps{X}(h\circ f)(x)$ will lie in $\ps{K}_C$ for
all $h:C\map B$.  Thus each stage  $A$ in $\cal C$ serves as a
possible context for an assignment to each $x\in \ps{X}_A$ of a
generalised truth value---a sieve belonging to the Heyting algebra
$\Om_A$.  This is the sense in which contextual, generalised truth
values arise naturally in a topos of presheaves.

There is a converse to \eq{Def:chiKA}: namely, each morphism
$\chi:\ps{X}\map{\ps{\O}}$ ({\em i.e.}, a natural transformation
between the presheaves $\ps{X}$ and ${\ps{\O}}$) defines a
sub-object $\ps{K}^\chi$ of $\ps{X}$ via
\begin{equation}
    \ps{K}^\chi_A:=\chi_A^{-1}\{1_{\Om_A}\}.
                            \label{Def:KchiA}
\end{equation}
at each stage $A$.

\paragraph{Global elements of a presheaf.}
For the category of presheaves on $\cal C$, a terminal object
$\ps{1}:{\cal C}\map \Set$ can be defined by $\ps{1}_A:=\{*\}$ at
all stages $A$ in $\cal C$; if $f:B\map A$ is a morphism in $\cal
C$ then $\ps{1}(f):\{*\}\map\{*\}$ is defined to be the map
$*\mapsto *$. This is indeed a terminal object since, for any
presheaf $\ps{X}$, we can define a unique natural transformation
$N:\ps{X}\map\ps{1}$ whose components
$N_A:\ps{X}(A)\map\ps{1}_A=\{*\}$ are the constant maps $x\mapsto
*$ for all $x\in\ps{X}_A$.

As a morphism $\ga:\ps{1}\map\ps{X}$ in the topos $\SetC{\cal C}$,
a global element corresponds to a choice of an element
$\ga_A\in\ps{X}_A$ for each stage  $A$ in $\cal C$, such that, if
$f:B\map A$, the `matching condition'
\begin{equation}
    \ps{X}(f)(\ga_A)=\ga_B \label{Def:global}
\end{equation}
is satisfied.

\section*{Acknowledgements} This research was
supported by grant RFP1-06-04 from The Foundational Questions
Institute (fqxi.org). AD gratefully acknowledges financial support
from the DAAD.

This work is also supported in part by the EC Marie Curie Research
and Training Network ``ENRAGE'' (European Network on Random
GEometry) MRTN-CT-2004-005616.

We are both very grateful to Professor Hans de Groote for his
detailed and insightful comments on our work.

CJI expresses his gratitude to Jeremy Butterfield for the lengthy,
and most enjoyable, collaboration in which were formulated the
early ideas about using topoi to study quantum theory.

\end{document}